\documentclass[a4paper,UKenglish,cleveref, autoref, thm-restate]{lipics-v2021}
\pdfoutput=1
\usepackage{kbordermatrix}
\usepackage{todonotes}
\usepackage{hyperref}

\usepackage{enumerate}
\usepackage{amsthm}
\usepackage{verbatim}
\usepackage{amsmath}
\usepackage{amssymb}
\usepackage{mathtools}
\usepackage[capitalise]{cleveref}
\usepackage{euler}

\usepackage{macros}
\allowdisplaybreaks

\hideLIPIcs  

\title{Tight Bounds for Connectivity Problems Parameterized by Cutwidth} 

\author{Narek Bojikian}{Humboldt-Universität zu Berlin, Germany}{bojikian@hu-berlin.de}{https://orcid.org/0000-0003-1072-4873}{}

\author{Vera Chekan}{Humboldt-Universität zu Berlin, Germany}{vera.chekan@informatik.hu-berlin.de}{https://orcid.org/0000-0002-6165-1566}{Supported by the DFG Research Training Group 2434 “Facets of Complexity”.}

\author{Falko Hegerfeld}{Humboldt-Universität zu Berlin, Germany}{hegerfeld@informatik.hu-berlin.de}{https://orcid.org/0000-0003-2125-5048}{}

\author{Stefan Kratsch}{Humboldt-Universität zu Berlin, Germany}{kratsch@informatik.hu-berlin.de}{https://orcid.org/0000-0002-0193-7239}{}

\authorrunning{N.\ Bojikian, V.\ Chekan, F.\ Hegerfeld, and S.\ Kratsch} 

\Copyright{Narek Bojikian, Vera Chekan, Falko Hegerfeld, and Stefan Kratsch} 



\ArticleNo{0}

\nolinenumbers 

\begin{document}

\maketitle

\begin{abstract}
In this work we start the investigation of tight complexity bounds for connectivity problems parameterized by cutwidth assuming the Strong Exponential-Time Hypothesis (SETH).
Van Geffen et al.~\cite{GeffenJKM20} posed this question for \Poct\ and \Pfvs. We answer it for these two and four further problems, namely \Pcvc, \Pcds, \Pst, and \Pcoct. For the latter two problems it sufficed to prove lower bounds that match the running time inherited from parameterization by treewidth; for the others we provide faster algorithms than relative to treewidth and prove matching lower bounds.
For upper bounds we first extend the idea of Groenland et al.~[STACS~2022] to solve what we call \emph{coloring-like problems}. Such problems are defined by a symmetric matrix $M$ over $\FF_2$ indexed by a set of colors. The goal is to count the number (modulo some prime $p$) of colorings of a graph such that $M$ has a $1$-entry if indexed by the colors of the end-points of any edge. We show that this problem can be solved faster if $M$ has small rank over $\FF_p$. We apply this result to get our upper bounds for \pcvc\ and \pcds.
The upper bounds for \poct\ and \pfvs\ use a subdivision trick to get below the bounds that matrix rank would yield.
\end{abstract}

\newpage

\section{Introduction}

Parameterized complexity studies the complexity of (typically $\mathsf{NP}$-hard) computational problems
in a finer way where aside from the input size $n$, other values, called parameters and frequently denoted by $k$, are considered in the running time. This might be for example the solution size or some structural properties of input. The class of fixed-parameter tractable problems ($\mathsf{FPT}$) contains problems that admit an
algorithm with running time $\mathcal{O}\big(f(k)n^c\big)$ for a computable
function $f$ and a constant $c$.
Since $f$ is only required to be computable, it might grow very rapidly, e.g.,
$f(k) = 100^k$ or $k^k$, or even the power tower function are allowed by this
definition. 
Therefore, from a practical point of view, a problem being $\mathsf{FPT}$ does not say much about the solvability of the problem in reasonable running time.

This inspired the search for better functions $f$. Since the problems we deal
with are already $\mathsf{NP}$-hard, stronger hardness conjectures than $\mathsf{P} \neq \mathsf{NP}$ have been assumed to show that
some function $f$ is essentially optimal for a problem and a specific parameter.
For example, assuming the Exponential Time Hypothesis (ETH)
\cite{ImpagliazzoP01}, it has been shown that many problems do not admit
algorithms with single-exponential running time \cite{CyganNPPRW11,
LokshtanovMS182}, i.e., a running time of the form $\ostar\big( c^k\big)$ for some
constant $c$. Since some problems in $\mathsf{FPT}$ are known to admit an algorithm with
single-exponential running time, an even stronger conjecture known as the Strong
Exponential Time Hypothesis (SETH) was stated and used to prove that some base $c$ is optimal for such a problem under this conjecture. This conjecture claims, roughly
speaking, that \textsc{SAT} cannot be solved much better than brute-forcing. Such
lower bounds were easier to prove for structural parameters on graphs, where the
value of the parameter reflects how well-structured or interconnected a graph
is. Upper bounds for such parameters usually rely on dynamic programming
employing some tricks and advanced techniques like fast subset convolution
\cite{BjorklundHKK07}, rank-based methods \cite{BodlaenderCKN15,CyganKN18,CurticapeanLN18}, the isolation lemma~\cite{MulmuleyVV87} etc.
For many classical problems parameterized by \twdth{}, it has been shown that
the optimal running time under SETH is single-exponential and the base of the
exponent is known, e.g.,~\cite{LokshtanovMS18}. There is a special class of problems related to this question.
These are the \emph{connectivity problems}. Even though the class is not well-defined, all
problems in this class impose connectivity constraints on the structure of a
solution. In these problems we usually look for  a set of vertices (or edges)
that it is either connected itself and has some
further properties (e.g., \Pcvc{} or \Pcds{}), or such that the input graph satisfies a certain
disconnectivity requirement after its removal (e.g., \Pfvs{}). For a
long time, the existence of single-exponential algorithms for connectivity
problems parameterized by \twdth{} or other structural parameters remained open.
But then a breakthrough work of Cygan et al.\ introduced a new view on
such problems called \emph{cut \& count}~\cite{CyganNPPRW11}. 
This technique is randomized and it reduces connectivity problems to counting certain bipartitioned solutions modulo two.

Tight bounds for problems parameterized by treewidth and pathwidth have been widely studied (e.g., \cite{CyganNPPRW11,CyganKN18,CurticapeanLN18,LokshtanovMS18}). An optimal dynamic programming algorithm traverses a tree or a
path decomposition in a bottom-up manner and utilizes the fact that every bag of the decomposition is a small vertex
separator. 
Therefore, it is also natural to study parameters that are based on edge separators.
Imagine that the vertices of the graph are put on the line in some fixed order and the edges are drawn as $x$-monotone
curves. The \ctwdth{} of this arrangement is then the maximal number of edges crossing any vertical line. The \ctwdth{}
of the graph is then the smallest \ctwdth{} of such an arrangement. Note that for any vertical line, the set of edges
crossing it separates vertices lying on different sides of this line from each other. Therefore, \ctwdth{} is an analogue of pathwidth based on edge separators. In fact, pathwidth can be defined in an analogous way to \ctwdth{}, that is,
we count the number of vertices on one side of the cut that have neighbors on the other side of the cut
\cite{Kinnersley92}. This also shows that pathwidth is upper-bounded by \ctwdth{}.

Since \ctwdth{} is an upper bound for pathwidth, an algorithm running in time $\ostar\big(f(\pw)\big)$ also runs in time $\ostar\big(f(\ctw)\big)$. In particular, single-exponential solvability transfers from pathwidth to \ctwdth{}. However, it is possible that the optimal dependence on cutwidth is smaller than for pathwidth.
Tight bounds for problems parameterized by \ctwdth{} have already been determined for example for \Pcol\ \cite{JansenN18}, \Pqcol{} \cite{GroenlandMNS22},
\Pis\ and \Pds\ \cite{GeffenJKM20}, and for various factor problems \cite{MarxSS21}. 
However, for connectivity problems, this question remained open. Van Geffen et al.~asked for the complexity of \Poct, \Pfvs, and \Phc\ ~\cite{GeffenJKM20}. We start this investigation in our work. 
Generally, tight bounds have also been studied for other decompositional parameters as well: treedepth (e.g., \cite{HegerfeldK22}),
clique-width (e.g., \cite{Lampis20}) etc. 

\subparagraph*{Our contribution}

Van Geffen et al.\ \cite{GeffenJKM20} asked for the exact complexity of \Poct\ (\poct) and \Pfvs\ (\pfvs)
parameterized by \ctwdth{} under SETH. In this work, we answer this question for these two problems in addition to four
other connectivity problems.
For two of the problems, we show the optimal base of exponent is the same for the parameterizations by treewidth and cutwidth. For the remaining problems, the base is smaller for cutwidth.
The right column of \cref{table:results} contains the tight bounds for these problems parameterized by cutwidth and summarizes the results of our work. 
The middle column contains the analogous results for treewidth for comparison \cite{CyganNPPRW11,LokshtanovMS18}. 

\begin{table}
\begin{center}
    \begin{tabular}{ |c|c|c|c| }
     \hline
     \Pcvc{} (CVC) & $3^{\tw}$ & $2^{\ctw}$\\
     \Pcds{} (CDS) & $4^{\tw}$ &$3^{\ctw}$\\
     \Poct{} (OCT) & $3^{\tw}$ &$2^{\ctw}$\\
     \Pfvs{} (FVS) & $3^{\tw}$ & $2^{\ctw}$\\
     \Pst{} (ST)       & $3^{\tw}$ & $3^{\ctw}$\\
     \Pcoct{} (COCT) & $4^{\tw}$ & $4^{\ctw}$\\
     \hline
    \end{tabular}
    \end{center}
    \caption{Tight bounds for parameterizations by treewidth and cutwidth.\label{table:results}}
\end{table}

\subparagraph{Organization}
We begin this work with a brief summary of the used notation. 
In \cref{app:general-approach} we define coloring-like problems and provide a general framework to solve these problems efficiently. 
In \cref{subsec:cvc-upper-bound} and \cref{app:cds} we apply this framework to solve \Pcvc{} and \Pcds{}, respectively.
After that, in \cref{app:oct} we present an algorithm for \Poct{} based on edge-subdivision.
In \cref{app:fvs-ub} we also apply edge-subdivision to accelerate the \cnc{} approach for the \Pfvs{} problem.Next in \cref{app:lb} we provide our lower-bound constructions for \Pst{} (\cref{append:st}), 
\Pcds{} (\cref{app:cds-lb}), \Pcvc{} (\cref{app:cvc-lb}), \Poct{} (\cref{app:oct-lb}), \Pfvs{} (\cref{app:fvs-lb}), and \Pcoct{} (\cref{app:coct}).
We conclude in \cref{sec:conclusion} by providing possible directions of further research in this area. 

\section{Preliminaries}\label{sec:notation}

For $n \in \NNN$, let $[n] = \{1, 2, \dots, n\}$ and $[n]_0 = [n] \cup \{0\}$. For a function $f\colon X \to Y$ and a set $S
\subseteq X$, with $f_{|_{S}}$ we denote the function $f_{|_{S}}\colon S \to Y$ such that $f_{|_{S}}(x) = f(x)$ for all $x \in
S$. If $Y \not\subseteq \mathbb{Z}$ we define $f(S) = \{y \in Y \mid \exists x \in S\colon f(x) = y\}$. With $f^{-1}(y)$ we denote the set
$\{x \in X \mid f(x)= y\}$. We abuse the notation for injective functions and consider the singleton~$\{x\}$ as an element $x \in X$.
For a function $f\colon X \to Y$ with $Y \subseteq \ZZ$ and a subset $S \subseteq X$, with $f(S)$ we denote the value $f(S) = \sum_{s \in S}
f(s)$. Further, let $f\colon X \times \ZZ^2 \to Y$ be a function, sometimes we call it a \emph{table}. We call $X$ the
\emph{domain} of $f$ and denote it with $\dom(f)$. If $0 \in Y$, the \emph{support} of $f$ is the set $\supp(f) = \bigl\{x \in X \bigm\vert \exists \order, \weight \in \ZZ\colon f(x, \order, \weight) \neq 0\bigr\}$.

Apart from $[n]$, we will use several other interpretations of square brackets. First, Iverson's bracket notation: for a
predicate $p$, the value $[p]$ is equal to $1$ if $p$ is true and $0$ otherwise. 
Sometimes, for space reasons, we will also write $1_p$ instead of $[p]$.
Second, for a function $f\colon X \to Y$,
element $x^*$, and element $y \in Y$, with $f[x^* \mapsto y]$ we denote the function  $f[x^* \mapsto y]\colon X \cup \{x^*\}
\to Y$, where $f[x^* \mapsto y](x^*) = y$ and $f[x^* \mapsto y](x) = f(x)$ for $x\neq x^*$.
Note that both $x^* \in X$ and $x^* \notin X$ are allowed by this definition.
Finally, for two elements $x, y$, with~$[x \mapsto y]$ we denote the unique function in $\{y\}^{\{x\}}$.

In this work, we will only consider simple loopless undirected graphs.
A \emph{linear arrangement}~$\ell$ of a graph $G = (V, E)$ is a bijection $\ell\colon [n] \to V$. For $i \in [n]$, let $v_i = \ell(i)$ and
with~$E_i$ we refer to the set of edges $E_i = \bigl\{\{v_j, v_k\} \in E \bigm\vert j \leq i, k > i\bigm\}$ called the \emph{$i$th cut}
(also called the \emph{cut at} $v_i$ or \emph{between} $v_i$ and $v_{i+1}$). If an edge $e$ belongs to $E_i$, we also say that it
\emph{crosses} the $i$th cut. We say that two edges $e \neq e'$ \emph{overlap} on $\ell$, if there exists $i \in [n]$ such that $e, e' \in E_i$.
The \emph{cutwidth}
$\ctw(\ell)$ of $\ell$ is defined as $\ctw(\ell) = \max_{i \in [n]} |E_i|$. The \emph{cutwidth} $\ctw(G)$ of a
graph $G$ is then the smallest cutwidth over all linear arrangements of $G$. Let $\ell$ now denote a fixed linear arrangement of $G$. For $i \in [n]$, we use the following notation:
	\begin{itemize}
		\item $V_i = \{v_1, v_2, \dots, v_i\}$, $G_i = G[V_i]$,
		\item $X_i = \bigl\{v_j \bigm\vert j \leq i, \exists k > i \colon \{v_j, v_k\} \in E_i\bigr\} \cup \{v_i\}$,
		\item $Y_i = \bigl\{v_j \bigm\vert j > i, \exists k \leq i\colon \{v_j, v_k\} \in E_i\bigr\}$,
		\item $H_i = (X_i \cup Y_i, E_i)$ the \emph{cut-graph at} $v_i$ (also called the \emph{$i$th cut-graph}; note that it is bipartite with the \emph{left side} $X_i$ and the \emph{right side} $Y_i$), 
		\item and if $i \geq 2$, then $Z_i = X_{i-1} \cup \{v_i\}$.
	\end{itemize}
	Observe that every vertex $v \neq v_i \in X_i$ has an incident edge in the $i$th cut whose other end-vertex belongs
	to $Y_i$. Therefore, the size of $X_i$ is at most $\ctw+1$. Further, for $i \in [n - 1]$, we have $X_{i+1} \subseteq
	X_i \cup \{v_{i+1}\}$ (for any vertex $v \neq v_{i+1} \in X_{i+1}$, an incident edge crossing the $(i+1)$th cut,
	also crosses the $i$th cut). Hence, $X_{i+1} \cap X_i = X_{i+1} \setminus \{v_{i+1}\}$, and $X_{i+1} \setminus X_i =
	\{v_{i+1}\}$.
To prove the tightness of our bounds, we will sometimes rely on results for problems parameterized by treewidth or pathwidth.

	A path decomposition of a graph $G$ is a sequence $\mathcal{B} = B_1, \dots, B_r$ of the so-called \emph{bags} such that:
    \begin{itemize}
        \item It holds that $B_1 \cup \dots \cup B_r = V$.
        \item For each $\{u, v\}\in E$, there exists an index $i \in [r]$ such that $u, v \in B_i$.
        \item And for every $i < j \in [r]$, the property $v \in B_i \cap B_j$ implies $v \in B_k$ for every $i \leq k \leq j$.
    \end{itemize}
    The \emph{pathwidth} of $\mathcal{B}$ is defined as $\pw = \max_{i \in [r]}|B_i| - 1$.
    The pathwidth of $G$ is the smallest $\pw(\mathcal{B})$ over all path decompositions $\mathcal{B}$ of $G$.
    For $i \in [r]_0$, we denote with $V_i = \cup_{j\in\{1, \dots i\}}B_j$ the union of the first $i$ bags and with $G_i = G[V_i]$ the subgraph induced by these vertices, we also denote the set of edges of this subgraph with $E_i$. 
    The \emph{pathwidth} of $\mathcal{B}$ is defined as $\pw = \max_{i \in [r]}|B_i| - 1$.
    The pathwidth of a graph $G$ is the smallest $\pw(\mathcal{B})$ over all path decompositions $\mathcal{B}$ of $G$.	
	
    A path decomposition $\mathcal{B} = B_1, \dots, B_r$ of a graph $G$ is \emph{nice} if $B_1 = B_r = \emptyset$ holds, and for every $i \in [r-1]$, we have $|B_i \triangle B_{i+1}| = 1$ where $\triangle$ denotes the symmetric difference. 
	In this case, for $2 \leq i \leq r$, we call~$B_i$ an \emph{introduce} bag if $B_{i-1} \subseteq B_i$ holds and we call it a \emph{forget} bag otherwise.

We skip the definition of treewidth since we do not work with it explicitly. We will mainly use the following result:
\begin{lemma}[\cite{Kinnersley92}]\label{lem:ctw-ub}
	The cutwidth of a graph is an upper bound for its tree- and pathwidth.
\end{lemma}

All lower bounds in this work assume the Strong Exponential Time Hypothesis (SETH). We use the following equivalent formulation of SETH:

\begin{conjecture}[\cite{ImpagliazzoP01,ImpagliazzoPZ01}]
	For any positive value $\delta$ there exists an integer $d$ such that $d$-\Psat\ cannot be solved in time
	$\ostar\bigl((2-\delta)^n\bigr)$ where $n$ denotes the number of variables.
\end{conjecture}

As mentioned in the introduction, there is a class of the so-called connectivity problems. In their seminal paper Cygan
et al.\ provided a breakthrough approach called \emph{\cnc{}} to solve connectivity problems in single-exponential time
based on the so-called consistent cuts~\cite{CyganNPPRW11}. Let us sketch the idea on an example of \Pcvc{}. 
A \emph{vertex cover} of a graph is a set of vertices such that every edge of the graph is incident to some vertex in this set. So the \Pcvc{} is defined as follows.

\begin{quote}
	\textbf{Input}: A graph $G = (V, E)$ and an integer $k$.
	
	\textbf{Question}: Is there a vertex cover $S \subseteq V$ of cardinality at most $k$ such that $G[S]$ is connected.
\end{quote}

Let $v \in V$ be a fixed vertex contained in some fixed vertex cover
$S$. A \emph{consistent cut} of $S$ is a partition $L \cup R = S$ such that $v \in L$ and there is no edge between
$L$ and $R$ in $G$. Let $\cc(S)$ denote the number of connected components of $G[S]$. By definition, every component is
completely contained either in $L$ or in $R$ and the component containing $v$ is in $L$. Therefore, the number of
consistent cuts of $S$ is $2^{\cc(S) - 1}$. Crucially, this number is odd if and only if $G[S]$ is connected. So if we
could assume that the solution is unique, then counting the number of consistent cuts modulo 2 would suffice to solve
the problem. However, a graph might contain several solutions so that the corresponding cuts cancel out. To overcome
this issue Cygan et al.\ assign weights to vertices and employ the Isolation Lemma of Mulmuley et al.\ \cite{MulmuleyVV87}
to ensure that with high probability the minimum-weight solution is unique. The following theorem summarizes this result
in a suitable for us form applied to the parameterization by cutwidth:
\begin{theorem}[\cite{CyganNPPRW11}]\label{thm:cut-and-count-short} Let $G = (V, E)$ be a graph, $\ell$ a linear
	arrangement of $G$ of cutwidth $\ctw$, $v \in V$ a fixed vertex, $\omega\colon V \to \bigl[2|V|\bigr]$ a weight function,
	and $\order \in [n]_0, \weight \in \bigl[2n|V|\bigr]_0$ integers. With $\CCC^\order_\weight$ (resp.\ $\DDD^\order_\weight$) we
	denote the family of pairs $\bigl(S, (L, R)\bigr)$ such that $S$ is a vertex cover (resp.\ dominating set) of $G$, $|S|
	= \order$, $\omega(S) = \weight$, $v \in L$, and $(L, R)$ is a consistent cut of $S$.
	If there exists an algorithm $\AAA$ that given the above input computes the size of $\CCC^\order_\weight$ (resp.\ $\DDD^\order_\weight$) modulo 2 in time $\OO^*\bigl(\beta(\ctw)\bigr)$ for some computable function $\beta$, then
	there exists a randomized algorithm that given a graph $G$ and its linear arrangement of cutwidth $\ctw$ solves the
	\pcvc{} (resp.\ \pcds{}) problem in time $\OO^*\bigl(\beta(\ctw)\bigr)$. The algorithm cannot give false positives and may
	give false negatives with probability at most $1/2$.	
\end{theorem}

\section{Coloring-like Problems}\label{app:general-approach}
In this section we define a coloring-like problem, present an algorithm solving such a problem in time depending on the rank of a certain consistency matrix, and then apply this approach to solve \Pcvc{} (\cref{subsec:cvc-upper-bound}) and \Pcds{} (\cref{app:cds}).
The latter problem requires combining this idea with inclusion-exclusion to obtain an optimal algorithm.

\subsection{Problem Setting}

We consider the so-called \emph{coloring-like} problems. Such a problem $P$ is defined by a finite set $C = \bigl\{1, \dots, |C|\bigr\}$ of colors, a set of \emph{special} colors $Q \subseteq C$, and a symmetric so-called \emph{consistency
matrix} $M \in \{0, 1\}^{|C| \times |C|}$. Additionally, the problem might contain a prime number $p$. An instance of $P$ consists of a graph $G = (V, E)$ with a linear arrangement $\ell: \bigl[|V|\bigr] \to V$ of cutwidth $\ctw$, a list function $a: V \to 2^C$ of allowed colors, 
numbers $N, \totalorder, \totalweight \in \NNN$, and a weight function $\omega: V \to [N]$.
The goal is to count the number of proper list-colorings (with respect to $a$) of $G$ such that the end-vertices of every edge are colored with consistent colors (with respect to $M$), exactly $\totalorder$ vertices are colored with special colors (i.e., the coloring has \emph{order} $\totalorder$), and these vertices have total weight of $\totalweight$ in $\omega$ (i.e., the coloring has \emph{weight} $\totalweight$). So we are asked to compute the number
\begin{align*}
	\biggl|\Bigl\{ c \in C^{V(G)} \Bigm\vert & \bigl|c^{-1}(Q)\bigr| = \totalorder, \omega\bigl(c^{-1}(Q)\bigr) = \totalweight, \\
	 &\forall v \in V: c(v) \in a(v), \forall \{u, v\} \in E(G): M\bigl[c(u), c(v)\bigr] = 1 \Bigr\}\biggr|.
\end{align*}
If a prime $p$ is given, the goal is to compute the above value modulo $p$. 

Let us mention the following notions related to coloring-like problems. Let $H$ be a graph (possibly with loops) with vertex set $C$ and adjacency matrix $M$. 
Then a coloring $c$ satisfying 
\[
	\forall v \in V: c(v) \in a(v), \forall \{u, v\} \in E(G): M\bigl[c(u), c(v)\bigr] = 1
\]
is a \emph{list $H$-coloring} of $G$ and a \emph{list homomorphism} from $G$ to $H$. 
The complexity of the \textsc{$H$-Coloring} problem parameterized by cutwidth has been studied by Piecyk and Rzazewski~\cite{PiecykR21}.
They provide lower and upper bounds however they are not matching so their results are not tight yet.
Moreover, we emphasize that a coloring-like problem differs from \textsc{$H$-coloring} and refines it. 
The differences are the following:
\begin{enumerate}
	\item We are interested in counting (modulo $p$ if it is part of the input) the colorings.
	\item We restrict the cardinality as well as the weight of the set of vertices mapped to $Q$.
\end{enumerate}
So in essence, a coloring-like problem is a counting version of ``twice-weighted'' \textsc{$H$-coloring}.

Observe that for $q = |C|$, $Q = \emptyset$, and the consistency matrix $M$ defined by
\[
	M[i, j] = 
	\begin{cases}
		1 & \mbox{if $i \neq j$,} \\
		0 & \mbox{if $i = j$,}
	\end{cases}
\]
for all $i, j \in [q]$, the coloring-like problem $P$ is exactly the problem of counting the number of proper list $q$-colorings of a graph (modulo $p$).
Groenland et al.\ showed that for a prime number $p$, this matrix has (full) rank of $q$ over $\FF_p$ if $p$ does not divide $q-1$ and rank of $q-1$ otherwise \cite{GroenlandMNS22}. Based on this property, they developed an $\mathcal{O}^*\bigl((q-1)^{\ctw}\bigr)$ (if $p$ divides $q-1$) resp.\ $\mathcal{O}^*(q^{\ctw})$ (otherwise) algorithm to solve the problem and also proved matching lower bounds in both cases. We generalize their idea to our notion of a coloring-like problem (i.e., defined by an arbitrary consistency matrix) and apply it to \Pcvc\ and (with additional tricks) to \Pcds.

\subsection{Dynamic Programming}\label{subsection:DP-for-CVC-short}

A \emph{coloring} of a vertex set $X \subseteq V$ is an assignment $x: X \to C$, i.e., $x \in C^X$.
A coloring $x\in C^X$ is \emph{valid} if for every vertex $v \in V$, it holds that $x(v) \in a(v)$ and for every edge $\{v, w\}$ of $G[X]$, it holds that $M\bigl[x(v), x(w)\bigr] = 1$. 
For sets $X \subseteq Y \subseteq V$ and colorings $x \in C^X$, $y\in C^Y$, we say that $y$ \emph{extends} $x$ if $y_{|_X} = x$. 
For sets $X, Y \subseteq V$ and colorings $x \in C^X$ and $y \in C^Y$ we say that $x$ is \emph{compatible} with $y$ and write $x \sim y$ if $x_{|_{X \cap Y}} = y_{|_{X \cap Y}}$ holds and for every edge $\{u, v\} \in E$ with $u \in X, v \in Y \setminus X$ it holds that $M\bigl[x(u), y(v)\bigr] = 1$. We emphasize that compatibility is not symmetrical in general.
The definition of compatibility is indeed not completely intuitive but it is motivated by the technical details of the correctness proof of our algorithm provided later: we allow $x$ and $y$ to be compatible even if $x$ has conflicts along the edges of $G[X]$.
Also note that if $X$ and $Y$ are disjoint, then it suffices to check the edges between $X$ and $Y$ for compatibility. For $2 \leq i \leq n$ recall that $X_i \cap X_{i-1} = X_i \setminus \{v_i\}$ and $X_i \setminus X_{i-1} = \{v_i\}$ hold. Thereby, we immediately obtain the following observation:

\begin{observation}\label{obs:compatibility}
	Let $2 \leq i \leq n$, $z \in C^{X_{i-1}}$, and $c \in C^{X_i}$. Then it holds that $z \sim c$ if and only if $z_{|_{X_i \setminus \{v_i\}}} = c_{|_{X_i \setminus \{v_i\}}}$ and $z \sim \bigl[v_i \mapsto c(v_i)\bigr]$ hold. 
\end{observation}

The following technical lemma will be used several times later. A similar version of this lemma was already proven by Groenland et al., we only adapt it for our notion of compatibility.
\begin{lemma}\label{lem:compatibility-y}
	Let $2 \leq i \leq n$, $z \in C^{X_{i-1}}$, $c \in C^{X_i}$, and $y \in C^{Y_i}$. Then $c \sim y$ and $z \sim c$ hold if and only if $z_{|_{X_i \setminus \{v_i\}}} = c_{|_{X_i \setminus \{v_i\}}}$, $z \sim \bigl[v_i \mapsto c(v_i)\bigr]$, $\bigl[v_i \mapsto c(v_i)\bigr] \sim y$, and $z \sim y$ hold.	
\end{lemma}

\begin{proof}

	Note that $X_i \cap Y_i = \emptyset = X_{i-1} \cap Y_i$ holds.

	``$\Rightarrow$'': Suppose $c \sim y$ and $z \sim c$ hold. The properties $z_{|_{X_i \setminus \{v_i\}}} = c_{|_{X_i \setminus \{v_i\}}}$ and $z \sim \bigl[v_i \mapsto c(v_i)\bigr]$ hold by the definition of $z \sim c$. Then, we obtain that 
	\[
		c = z_{|_{X_i \setminus \{v_i\}}} \cup \bigl[v_i \mapsto c(v_i)\bigr] = z_{|_{X_i \setminus \{v_i\}}}\bigl[v_i \mapsto c(v_i)\bigr].
	\]
	This together with $c \sim y$ implies that $\bigl[v_i \mapsto c(v_i)\bigr] \sim y$ and $z_{|_{X_i \setminus \{v_i\}}} \sim y$ hold. Finally, for every edge $\{uv\} \in E$ with $u \in X_{i-1}$ and $v \in Y_i$ it holds that $u \in X_i$ since $uv$ crosses the $i$-th cut. Due to $u \neq v_i$, the property $z_{|_{X_i \setminus \{v_i\}}} \sim y$ implies that $z \sim y$ holds. So all four claimed properties are implied.

	``$\Leftarrow$'': Suppose $z_{|_{X_i \setminus \{v_i\}}} = c_{|_{X_i \setminus \{v_i\}}}$, $z \sim \bigl[v_i \mapsto c(v_i)\bigr]$, $\bigl[v_i \mapsto c(v_i)\bigr] \sim y$, and $z \sim y$ hold. The first two properties imply that $z \sim c$ holds. To prove that $c \sim y$ holds as well, consider an edge $\{uv\} \in E$ with $u \in X_i$ and $v \in Y_i$. If $u = v_i$, then we obtain
	\[
		M\bigl[c(u), y(v)\bigr] = M\bigl[c(v_i), y(v)\bigr] = M\Bigl[\bigl[v_i \mapsto c(v_i)\bigr](v_i), y(v)\Bigr] \stackrel{\bigl[v_i \mapsto c(v_i)\bigr] \sim y}{=} 1.
	\]
	Otherwise, it holds that $u \in X_i \setminus \{v_i\}$. Then 
	\[
		M\bigl[c(u), y(v)\bigr] \stackrel{c_{|_{X_i \setminus \{v_i\}}} = z_{|_{X_i \setminus \{v_i\}}}}{=} M\bigl[z(u), y(v)\bigr] \stackrel{z \sim y}{=} 1.
	\]
	So $c \sim y$ indeed holds as well.
\end{proof}

For $i \in [n]$, $c \in C^{X_i}$, and $\order, \weight \in \ZZ$, let the value $T_i[c, \order, \weight]$ be defined as
\begin{align*}
	T_i[c, \order, \weight] = \biggl| \Bigl\{ & \phi \in C^{V_i} \Bigm\vert \phi_{|_{X_i}} = c, \bigl|\phi^{-1}({Q})\bigr| = \order, \omega\bigl(\phi^{-1}(Q)\bigr) = \weight, \\
	&\forall v \in V_i: \phi(v) \in a(v), \forall \{u, v\} \in E\bigl(G[V_i]\bigr): M\bigl[\phi(u), \phi(v)\bigr] = 1\Bigr\} \biggr|.
\end{align*}
This is the number of possibilities to extend a coloring of $X_i$ to a valid coloring of $V_i$ of order $\order$ and weight $\weight$.
Note that this value is zero whenever $\order < 0$ or $\weight < 0$ holds.
Recall that $X_n = \{v_n\}$ and $V_n = V$. Therefore, the number of list colorings that we are looking for as a final result is given by $\sum_{s \in C} T_n\bigl[[v_n \mapsto s], \totalorder, \totalweight\bigr]$.

The following lemma generalizes an analogous result of Groenland et al.\ 
\begin{lemma}\label{lem:dp-equality}
	For $2 \leq i \leq n$, $c \in C^{X_i}$, and $\order, \weight \in \ZZ$, the following holds
		\[
			T_i[c, \order, \weight] = \bigl[c(v_i) \in a(v_i)\bigr] \sum\limits_{\substack{z \in C^{X_{i-1}} \\ z \sim c}} T_{i-1} \bigl[z, \order - [c(v_i) \in Q], \weight - \omega(v_i)\cdot[c(v_i) \in Q]\bigr].
		\]
\end{lemma}

\begin{proof}
	Let $2 \leq i \leq n$, $c \in C^{X_i}$, and $\weight, \order \in \ZZ$ be arbitrary but fixed. Recall that every extension of $c$ assigns the color $c(v_i)$ to $v_i$. So if $c(v_i) \notin a(v_i)$, then there is no valid extension of $c$ to a coloring of $V_i$ and the claim holds. So we may now assume that $c(v_i) \in a(v_i)$ holds, i.e., $\bigl[c(v_i) \in a(v_i)\bigr] = 1$.
	
	Further, we may assume that $\order \geq 0$ and $\weight \geq 0$ hold: otherwise both sides of the equality are zero and the claim trivially holds. Similarly, if $c(v_i) \in Q$ and either $\order = 0$ or $\weight < \omega(v_i)$ holds, then both sides of the equality are zero as well: The left side is zero since every extension of $c$ would color $v_i$ with a color from $Q$, the right side is zero since $\order - 1 < 0$ or $\weight - \omega(v_i) < 0$ holds. From now on, we assume that whenever $c(v_i)$ belongs to $Q$, we have $\order > 0$ and $\weight \geq \omega(v_i)$.
        
		``$\leq$'': We consider a mapping $\xi: C^{V_i} \to C^{X_{i-1}} \times C^{V_{i-1}}$ given by $\phi \mapsto (\phi_{|_{X_{i-1}}}, \phi_{|_{V_{i-1}}})$. First of all, note that for two different colorings $\psi^1 \neq \psi^2 \in C^{V_i}$ extending $c$, it holds that 
		\[
			\psi^1(v_i) = c(v_i) = \psi^2(v_i) 
		\]
		and hence $\psi^1_{|_{V_{i-1}}} \neq \psi^2_{|_{V_{i-1}}}$. So the mapping $\xi$ restricted to colorings extending $c$ is injective.
		
		Now let $\phi \in C^{V_i}$ be a valid coloring of $V_i$ extending $c$ such that 
		\begin{equation}\label{eq:1}
			\bigl|\phi^{-1}(Q)\bigr| = \order \text{ and } \omega\bigl(\phi^{-1}(Q)\bigr) = \weight. 
		\end{equation}
		Then $\phi_{|_{V_{i-1}}}$ is a valid coloring of $V_{i-1}$ extending $\phi_{|_{X_{i-1}}}$ such that we have
		\begin{equation}\label{eq:a}
			\bigl|\phi_{|_{V_{i-1}}}^{-1} (Q)\bigr| = \bigl|\phi^{-1} (Q)\bigr| - \bigl[\phi(v_i) \in Q\bigr] \stackrel{\eqref{eq:1}}{=} \order - \bigl[c(v_i) \in Q\bigr]
		\end{equation}
		and
		\begin{equation}\label{eq:b}
			\omega\bigl(\phi_{|_{V_{i-1}}}^{-1} (Q)\bigr) = \omega\bigr(\phi^{-1} (Q)\bigr) - \omega(v_i)\cdot\bigl[\phi(v_i) \in Q\bigr] \stackrel{\eqref{eq:1}}{=} \weight - \omega(v_i)\cdot\bigl[c(v_i) \in Q\bigr]
		\end{equation}
		Since $\phi$ extends $c$, we have
		\[
			(\phi_{|_{X_{i-1}}})_{|_{X_i \setminus \{v_i\}}} = \phi_{|_{X_i \setminus \{v_i\}}} = c_{|_{X_i \setminus \{v_i\}}}.
		\]		
		So to prove that $\phi_{|_{X_{i-1}}} \sim c$ holds, by \cref{obs:compatibility} it suffices to check the edges $\{u, v_i\} \in E$ with $u \in X_{i-1}$. Indeed, 
		\[
			M\bigl[\phi_{|_{X_{i-1}}}(u), c(v_i)\bigr] = M\bigl[\phi(u), \phi(v_i)\bigr] = 1
		\]
		holds because $\phi$ is a valid coloring of $V_i$. 
		
		Altogether, $\xi$ injectively maps a coloring $\phi$ extending $c$ such that \eqref{eq:1} holds to a pair 
		\[
			\left(z = \phi_{|_{X_{i-1}}}, \phi' = \phi_{|_{V_{i-1}}}\right) 
		\]
		such that $z \in C^{X_{i-1}}$, $z \sim c$, and $\phi'$ is a valid coloring of $V_{i-1}$ extending $z$ for which~\eqref{eq:a} and~\eqref{eq:b} hold. 
		Hence,
		\[
            		T_i[c, \order, \weight] \leq \sum\limits_{\substack{z \in S^{X_{i-1}} \\ z \sim c}} T_{i-1} \Bigl[z, \order - \bigl[c(v_i) \in Q\bigr], \weight - \omega(v_i)\cdot\bigl[c(v_i) \in Q\bigr]\Bigr]
		\]
		holds. 
		
		``$\geq$'': Here we consider a mapping $\chi: C^{X_{i-1}} \times C^{V_{i-1}} \mapsto C^{V_i}$ defined by $(z, \phi) \mapsto \phi\bigl[v_i \mapsto c(v_i)\bigr]$. Consider two pairs $(z^1, \phi^1) \neq (z^2, \phi^2) \in C^{X_{i-1}} \times C^{V_{i-1}}$ such that $\phi^1_{|_{X_{i-1}}} = z^1$ and $\phi^2_{|_{X_{i-1}}} = z^2$. Distinctness implies that $\phi^1 \neq \phi^2$ holds and therefore $\phi^1\bigl[v_i \mapsto c(v_i)\bigr] \neq \phi^2\bigl[v_i \mapsto c(v_i)\bigr]$ holds as well. So $\chi$, restricted to pairs $(z, \phi)$ such that $\phi$ is an extension of $z$, is injective.
		
		Now let $(z, \phi) \in C^{X_{i-1}} \times C^{V_{i-1}}$ be such that $\phi$ is a valid coloring of $V_{i-1}$, $\phi$ is an extension of $z$, and $z \sim c$, 
		\begin{equation}\label{eq:2}
			\bigl|\phi^{-1}(Q)\bigr| = \order - \bigl[c(v_i) \in Q\bigr], \text{ and } 			\omega\bigl(\phi^{-1}(Q)\bigr) = \weight - \omega(v_i) \cdot \bigl[c(v_i) \in Q\bigr]
		\end{equation}
		hold.
		
		Since $\phi$ is an extension of $z$ and $z \sim c$ holds, we have 
		\[
			c_{|_{X_i \setminus \{v_i\}}} = z_{|_{X_i \setminus \{v_i\}}} = \phi_{|_{X_i \setminus \{v_i\}}}. 
		\]
		Thereby, the coloring $\phi\bigl[v_i \mapsto c(v_i)\bigr]$ is an extension of $c$. Recall that $\phi$ is a valid coloring of $V_{i-1}$. To see that $\phi\bigl[v_i \mapsto c(v_i)\bigr]$ is a valid coloring of $V_i$ as well, consider an edge $\{u, v_i\} \in E$ with $u \in V_{i-1}$. Note that this edge crosses the $(i-1)th$ cut, i.e., $u \in X_{i-1}$. So
		\[
			M\Bigl[\phi\big[v_i \mapsto c(v_i)\bigr](u), \phi\bigl[v_i \mapsto c(v_i)\bigr](v_i)\Bigr] = M\bigl[z(u), c(v_i)\bigr] \stackrel{z \sim c}{=} 1.
		\]
		Hence, $\phi\bigl[v_i \mapsto c(v_i)\bigr]$ is a valid coloring of $V_i$ extending $c$. For the sake of simplicity, we denote $\phi\bigl[v_i \mapsto c(v_i)\bigr]$ with $\tilde\phi$ for a moment. Then it holds that
		\begin{equation}\label{eq:c}
			\bigl|\tilde\phi^{-1}(Q)\bigr| = \bigl|\phi^{-1}(Q)\bigr| + \bigl[c(v_i) \in Q\bigr] \stackrel{\eqref{eq:2}}{=} \Big(\order - \bigl[c(v_i) \in Q\bigr]\Bigr) + \bigl[c(v_i) \in Q\bigr] = \order
		\end{equation}
		and
		\begin{align}
			\omega\bigl(\tilde\phi^{-1}(Q)\bigr) &= \omega\bigl(\phi^{-1}(Q)\bigr) + \omega(v_i) \cdot \bigl[c(v_i) \in Q\bigr] \nonumber \\
			&\stackrel{\eqref{eq:2}}{=} \Bigl(\weight - \omega(v_i) \cdot \bigl[c(v_i) \in Q\bigr]\Bigr) + \omega(v_i) \cdot \bigl[c(v_i) \in Q\bigr] = \weight \label{eq:d}.
		\end{align}
				
		Altogether, $\chi$ injectively maps pairs $(z, \phi) \in C^{X_{i-1}} \times C^{V_{i-1}}$ such that $\phi$ is a valid coloring of $V_{i-1}$ extending $z$ and $z \sim c$ and \eqref{eq:2} hold to a valid coloring $\tilde\phi$ of $V_i$ extending $c$ for which \eqref{eq:c} and \eqref{eq:d} hold. Hence,
		\[
			 T_i[c, \order, \weight] \geq \sum\limits_{\substack{z \in C^{X_{i-1}} \\ z \sim c}} T_{i-1} \Bigl[z, \order - \bigl[c(v_i) \in Q\bigr], \weight - \omega(v_i) \cdot \bigl[c(v_i) \in Q\bigr]\Bigr]
		\]
		holds and this concludes the proof.
\end{proof}

With that we can already obtain an $\OO^* \bigl(|C|^{\ctw} N\bigr)$-algorithm for $P$ as follows. 

\begin{lemma}\label{lemma:direct-DP-long}
	A coloring-like problem $P$ can be solved in $\OO^* \bigl(|C|^{\ctw} N\bigr)$.
\end{lemma}

\begin{proof}
	To solve the problem, we utilize \cref{lem:dp-equality} to compute the entries of the tables $T_i$ in the order of increasing $i$. We initialize the table $T_1$ via brute-force: we go through all $s \in a(v_1)$ and all $\order \in \{0, 1\}, \weight \in \bigl\{0, \omega(v_1)\bigr\}$ and compute the corresponding value $T_1\bigl[[v_1 \mapsto s], \order, \weight\bigr] \in \{0, 1\}$. All other entries of $T_1$ are implicitly set to $0$ and there is no need to store them. After that, for $2 \leq i \leq n$, we proceed as follows. Generally, we just want to apply \cref{lem:dp-equality}. Recall that for every $i \in [n]$, it holds that $|X_i| \leq \ctw + 1$. Since every vertex is allowed to be colored one of $|C|$ colors, there are at most $|C|^{\ctw + 1}$ colorings of $X_i$. However, if we proceed in a naive way, for every coloring of $X_i$, we would consider every coloring of $X_{i-1}$ resulting in $\OO^* \bigl(|C|^{2\cdot\ctw}\bigr)$ running time. So we want to avoid considering the same coloring of $X_{i-1}$ several times.
The main idea of this approach was introduced by Groenland et al.\ as well.

	Recall that by \cref{obs:compatibility}, a coloring $z \in C^{X_{i-1}}$ is compatible with a coloring $c \in C^{X_i}$ if and only if $z_{|_{X_i \setminus \{v_i\}}} = c_{|_{X_i \setminus \{v_i\}}}$ and $z \sim \bigl[v_i \mapsto c(v_i)\bigr]$ hold. This results in the following procedure. 
	Let $S_i: \dom(T_i) \times \ZZ^2 \to \NNN$ denote a table which will be equal to $T_i$ by the end of the following process. Initially, all values of this table are implicitly treated as zero. We iterate over all colorings $z \in C^{X_{i-1}}$ and all $s \in a(v_i)$. We check if $z \sim \bigl[v_i \mapsto c(v_i)\bigr]$ holds by looking at all edges having one end-vertex in $X_{i-1}$ and the other being $v_i$ and verifying that their colors are consistent with respect to $M$. If this holds, for all $0 \leq \order \leq i - 1$ and all $0 \leq \weight \leq (i-1)N$, we increase the value
	\[
		S_i\bigl[z_{|_{X_i \setminus \{v_i\}}}[v_i \mapsto s], \order + [s \in Q], \weight + \omega(v_i)\cdot[s \in Q]\bigr]
	\]  
	(if it was already initialized) by the value $T_{i-1}[z, \order, \weight]$ or initialize it with this value (otherwise). Finally, we interpret the non-initialized values as zero. Note that in particular, for every coloring $c \in C^{X_i}$ such that $c(v_i) \notin a(v_i)$ all entries $T_i[c, \cdot, \cdot]$ are zero. By the end of the iteration, the table $S_i$ contains exactly the values of $T_i$.

	Since the number of colorings of $G$ is at most $|C|^n$, each table entry is bounded by this value at most.
	Hence, a table entry has bit size at most $\poly(n)$ and in time $\poly(n)$ we can apply basic arithmetic
	on these entries. In total, the running time of the computation of $T_i$ given $T_{i-1}$ is bounded by
	\[
		\OO\Bigl(|C|^{\ctw + 1} \bigl|a(v_i)\bigr| i (iN) \poly(n)\Bigr) = \OO^* \Bigl(|C|^{\ctw} N\Bigr).
	\]
	Hence, in time $\OO^* \bigl(|C|^{\ctw}N\bigr)$ the entries of $T_n$ can be computed. Note that $X_n = \{v_n\}$ so for $s \in C$, the value $T_n\bigl[[v_n \mapsto s], \totalorder, \totalweight\bigr]$ is by definition the number of valid colorings $\phi$ of $V_n = V$ of order $\totalorder$, weight $\totalweight$ assigning the color $s$ to $v_n$.
	Thereby, we output the value $\sum_{s \in C} T_n\bigl[[v_n \mapsto s], \totalorder, \totalweight\bigr]$.
	If the input contains the prime number $p$, we return the above value modulo $p$.
\end{proof}

Now based on the straight-forward idea of this algorithm, we provide a more efficient one that profits from reducing the number of considered colorings.

\subsection{Computing Reduced Representative}

In this section, we recall that the input might contain a prime number $p$ and if so, we want to count the number of certain colorings modulo $p$. 
We will utilize this fact and accelerate our algorithm depending on this number.
For this reason, let $\FF$ denote the field $\FF = \FF_p$ if $p$ is a part of the input and the ring $\FF = \QQ$ otherwise. With ``$\equiv$'' we denote equality over $\FF$.

Let $X, Y \subseteq V$ be disjoint subsets of the vertex set and let $\widehat{f}, f: C^X \times \ZZ^2 \to \FF$ be tables. We say that $\widehat{f}$ \emph{$(X, Y)$-represents} $f$ if for every coloring $y \in C^Y$ and every $\order, \weight \in \ZZ$, it holds that
\[
	 \sum\limits_{\substack{x \in C^X \\ x \sim y}} \widehat{f}(x, \order, \weight) \equiv \sum\limits_{\substack{x \in C^X \\ x \sim y}} f(x,  \order, \weight).
\]

For $i \in [n]$ and table $f, \widehat{f}: C^{X_i} \times \ZZ^2 \to \FF$, we say that $\widehat{f}$ \emph{$i$-represents} $f$ if $\widehat{f}$ $(X_i, Y_i)$-represents $f$. Note that $(X, Y)$-representation is an equivalence relation and in particular, it is transitive.

Let $\rank(M)$ be the rank of the consistency matrix $M$ over $\FF$.
We may assume that the first $\rank(M)$ rows of the consistency matrix $M$ form its row basis over $\FF$ (otherwise we may permute the colors in $C$ and consider the corresponding consistency matrix). We call the colors $\bigl\{\rank(M)+1, \dots, |C|\bigr\}$ \emph{reduced}. 

For a set $X \subseteq V$ of vertices and a table $f: C^X \times \ZZ^2 \to \FF$ we say that a vertex $v \in X$ is \emph{reduced} (in $f$) if for every coloring $x \in C^X$ with $x(v) > \rank(M)$ (i.e., $v$ has a reduced color), it holds that $f(x, \order, \weight) \equiv 0$ for all $ \order, \weight \in \ZZ$.  

\begin{lemma}\label{lemma:reduce-algorithm}
	Let $X, Y \subseteq V$ be disjoint and let $i \in [n]$.
	Further, let $f: C^X \times \ZZ^2 \to \FF$ be a table with reduced vertices $R \subseteq X$ such that $f^{-1}\bigl(\ZZ \setminus \{0\}\bigr) \subseteq C^X \times [i]_0 \times [iN]_0$.
	Finally, let $v \in X \setminus R$ be a vertex that has exactly one neighbor in $Y$. 
	Then there is an algorithm \textbf{Reduce} that, given this information as input, in time $\mathcal{O}^*\bigl(\rank(M)^{|R|} |C|^{|X| - |R|} N\bigr)$ outputs a function $\widehat{f}: C^X \times \ZZ^2 \to \FF$ $(X, Y)$-representing $f$ with reduced vertices $R \cup \{v\}$.
	Moreover, it holds that $\widehat{f}^{-1}\bigl(\ZZ \setminus \{0\}\bigr) \subseteq C^{X_i} \times [i]_0 \times [iN]_0$. 
\end{lemma}

\begin{proof}
	For $b \in \bigl[|C|\bigr]$, let $m_b$ denote the $b$th row of $M$ and let $k = \rank(M)$. Since the set $\{m_1, \dots, m_k\}$ forms a row basis of $M$, for every $b \in \bigl\{k+1, \dots, |C|\bigr\}$ and $j \in [k]$, there exists $d_{b, j}\in \FF$ such that
	\begin{equation}\label{eq:basis-representation}
		m_b \equiv \sum\limits_{j = 1}^{k} d_{b,j} m_j.
	\end{equation}

	We define the values of $\widehat{f}$ as
	\[
		\widehat{f}(x, \order, \weight) \equiv
		\begin{cases}
			0 & \mbox{if $x(v) \geq k+1$,} \\
			f(x, \order, \weight) + \sum\limits_{b = k+1}^{|C|} d_{b, x(v)} f\bigl(x[v \mapsto b], \order, \weight\bigr) & \mbox{otherwise,}  \\
		\end{cases}
	\] 
	for all $x \in C^X$, $\order, \weight \in \ZZ$. 
		
	We first prove that $R \cup \{v\}$ is a set of reduced vertices of $\widehat{f}$. Let $\order, \weight \in \ZZ$ and let $x \in C^X$. By construction, for every coloring $x \in C^X$ such that $x(v) > k$ (i.e., the color of $v$ is reduced), we have $\widehat{f}(x, \order, \weight) \equiv 0$. So $v$ is a reduced vertex. In the following, we may assume that $x(v) \leq k$ holds. If $\order \notin [i]_0$ or $\weight \notin [iN]_0$, then we have $\widehat{f}(x, \order, \weight) \equiv 0$ by assumption that $f^{-1}\bigl(\ZZ \setminus \{0\}\bigr) \subseteq C^X \times [i]_0 \times [iN]_0$ holds.
	Otherwise, if for some vertex $u \in R \subseteq X$, it holds that $x(u) > k$, then we have
	\[
		\widehat{f}(x, \order, \weight) \equiv f(x, \order, \weight) + \sum\limits_{b = k+1}^{|C|} d_{b, x(v)} f\bigl(x[v \mapsto b], \order, \weight\bigr) \equiv 0
	\]
	where the last equality holds since $x[v \mapsto b](u) = x(u) > k$ for every $k+1 \leq b \leq |C|$ and $u$ is a
	reduced vertex of $f$. Thereby, $R \cup \{v\}$ is a set of reduced vertices of $\widehat{f}$. In particular, this implies
	that to compute and output the function $\widehat{f}$, we only need to compute and explicitly store the values $\widehat{f}(x, \order, \weight)$ for $\order \in [i]_0$, $\weight \in [iN]_0$, and such $x \in C^X$ that for all $u \in R \cup \{v\}$ we have $x(u) \leq k = \rank(M)$. So this can be done in
	\[
		\mathcal{O}^*\bigl(\rank(M)^{|R|+1} |C|^{|X| - \bigl(|R| + 1\bigr)} i (iN) |C| \poly(n) \bigr) = \mathcal{O}^*\bigl(\rank(M)^{|R|} |C|^{|X| - |R|} N\bigr). 
	\]
	
	We now need to show that $\widehat{f}$ $(X, Y)$-represents $f$, i.e., for all $y \in C^Y$ and all $\order, \weight \in \ZZ$, it holds that
	\[
		\sum\limits_{\substack{x \in C^X \\ x \sim y}} f(x, \order, \weight) \equiv \sum\limits_{\substack{z \in C^X \\ z \sim y}} \widehat{f}(z, \order, \weight).
	\]
	We denote the colorings on the left and right side of this equality with $x$ and $z$, respectively, in order to be able to distinguish them.
	Clearly, this holds if $\order \notin [i]_0$ or $\weight \notin [iN]_0$: in this case both sides are equal to zero. So we may assume that $\order \in [i]_0$ and $\weight \in [iN]_0$ hold. By the definition of $\widehat{f}$, we need to show that 
	\[
		\sum\limits_{\substack{x \in C^X \\ x \sim y}} f(x, \order, \weight) \equiv \sum\limits_{\substack{z \in C^X \\ z \sim y \\ 1 \leq z(v) \leq k}} \left( f(z, \order, \weight) + \sum\limits_{b = k + 1}^{|C|} d_{b, z(v)} f\bigl(z[v \mapsto b], \order, \weight\bigr) \right),
	\]
	i.e.,
	\begin{equation}\label{eq:proof-of-representation-general-approach}
		\sum\limits_{\substack{x \in C^X \\ x \sim y \\  k + 1 \leq x(v) \leq |C|}} f(x, \order, \weight) \equiv \sum\limits_{\substack{z \in C^X \\ z \sim y \\ 1 \leq z(v) \leq k}} \sum\limits_{b = k + 1}^{|C|} d_{b, z(v)} f\bigl(z[v \mapsto b], \order, \weight\bigr).	
	\end{equation}
	
	We will now show that for every coloring $x \in C^X$, the addend $f(x, \order, \weight)$ occurs the same number of times on the left and on the right side of the above equality. Observe that on both sides only the summands $f(x, \order, \weight)$ with $x(v) \geq k+1$ occur. So from now on, we assume that $x(v) \geq k+1$ holds. 
	Let $w$ be the unique neighbor of $v$ in $Y$. 
		 	
	\textbf{Case 1}: $x \sim y$. Then $f(x, \order, \weight)$ appears once on the left side. On the right side, $f(x, \order, \weight)$ appears $d_{x(v), z(v)}$ times for every coloring $z \in C^X$ such that:
	\begin{itemize}
		\item $x$ and $z$ only differ on the color of $v$,
		\item for $j = z(v)$, it holds that $j \leq k$,
		\item and $z$ is still compatible with $y$ (i.e., $z \sim y$).
	\end{itemize}
	Since we only flip the color of $v$ and this vertex has exactly one neighbor of $w$ in $Y$, the last property holds if and only if there is no conflict on the edge $\{v, w\}$ after this flipping, i.e., if 
	\[
		M\bigl[j, y(w)\bigr] = 1
	\]
	holds.
	So there is 
	\[
		\sum\limits_{j = 1}^{k} d_{x(v), j} M\bigl[j, y(w)\bigr] \stackrel{\eqref{eq:basis-representation}}{\equiv} M\bigl[x(v), y(w)\bigr] \stackrel{x \sim y}{=} 1 
	\]
	occurrence of $f(x, \order, \weight)$ on the right side as well. 
	
	\textbf{Case 2}: $x \not\sim y$. In this case $f(x, \order, \weight)$ does not occur on the left side. First, suppose there is a conflict on some edge other than $\{v, w\}$, i.e., there exists an edge $\{v', w'\} \neq \{v, w\} \in E$ (with $v' \in X, w' \in Y$) such that $M\bigl[x(v'), y(w')\bigr] = 0$. Since $v$ only has one neighbor in $Y$, we have $v \neq v'$. Therefore, if we flip the color of $v$, the conflict on the edge $\{v', w'\}$ would remain and the arising coloring would still be incompatible with $y$. Hence, $f(x, \order, \weight)$ does not occur on the right side as well. 
	
	Now we may assume that the unique conflict is on the edge $\{v, w\}$, namely $M\bigl[x(v), y(w)\bigr] = 0$ but for every $v' \in X, w' \in Y$ such that $\{v', w'\} \neq \{v, w\} \in E$, we have $M\bigl[x(v'), y(w')\bigr] = 1$. Hence, flipping the color of $v$ to some $j \in [k]$ makes the coloring $x$ compatible with $y$ if and only if $j$ is compatible with the color of $w$, i.e., $M\bigl[j, y(w)\bigr] = 1$. Thereby, $f(x, \order, \weight)$ appears
	\[
		\sum\limits_{j = 1}^{k} d_{x(v),j} M\bigl[j, y(w)\bigr] \stackrel{\eqref{eq:basis-representation}}{\equiv} M\bigl[x(v), y(w)\bigr] = 0
	\]
	times on the right side as well (the last equality holds due to the conflict on $\{v, w\}$). 
	
	Therefore, every summand occurs the same number of times both sides of \eqref{eq:proof-of-representation-general-approach} and $\widehat{f}$ indeed $(X, Y)$-represents $f$.
\end{proof}

\subsection{Dynamic Programming on Reduced Tables}\label{app:subsec:reduced-dp}

We will later need to construct a sufficiently large set of reduced vertices for some tables. The following simple lemma proven by Groenland et al.\ \cite{GroenlandMNS22} will be useful for that.

\begin{lemma}\label{lem:sets-A-i-B-i-R-i}
	Let $2 \leq i \leq n$ and let $R_i = X_i \setminus \bigl(A_i \cup B_i \cup \{v_i\}\bigr)$ where
	\[
		A_i = \Bigl\{u \in X_i \setminus \{v_i\} \Bigm\vert \bigl| N(u) \cap Y_i \bigr| \geq 2 \Bigr\}
	\]
	and
	\[
		B_i = \Bigl\{u \in X_i \setminus \{v_i\} \Bigm\vert \bigl| N(u) \cap Y_i \bigr| = 1 \text{ and } \{u, v_i\} \in E \Bigr\}.
	\]
	Then it holds that 
	\[
		|X_i \setminus R_i| \leq \bigl(\ctw - |R_i|\bigr) / 2 + 1.
	\]
\end{lemma}

\begin{proof}
	It holds that $A_i \cap B_i = \emptyset$ and every vertex in $A_i \cup B_i$ has at least two incident edges in the $(i-1)$th cut while every vertex in $R_i \subseteq X_i \setminus \{v_i\} \subseteq X_{i-1}$ has at least one incident edge in the $(i-1)$th cut. Thereby,
	\[
		|R_i| + 2\bigl(|A_i| + |B_i|\bigr) \leq \ctw
	\]
	holds and hence we have
	\[
		|X_i \setminus R_i| = |A_i \dot\cup B_i \dot\cup \{v_i\}| \leq |A_i \dot\cup B_i| + 1 \leq \bigl(\ctw - |R_i|\bigr) / 2 + 1
	\]
	as desired.
\end{proof}

For $i \in [n]$, a function $f: C^{X_i} \times \ZZ^2 \to \FF$ is \emph{fully reduced} if for every $v \in X_i$ that has exactly one neighbor in $Y_i$ (i.e., $v$ has degree one in the subgraph $H_i$), the vertex $v$ is reduced.

\begin{lemma}\label{lemma:dp-reduced-version}
	Let $2 \leq i \leq n$. Suppose that a table $\widehat{T_{i-1}}$ is fully reduced, the table $\widehat{T_{i-1}}$ $(i-1)$-represents $T_{i-1}$, and it holds that $(\widehat{T_{i-1}})^{-1}\bigl(\ZZ \setminus \{0\}\bigr) \subseteq C^{X_{i-1}} \times [i-1]_0 \times \bigl[(i-1)N\bigr]_0$.
	Given $\widehat{T_{i-1}}$ and a set $R_{i-1}$ of reduced vertices for $\widehat{T_{i-1}}$, it is possible to compute a function $\widehat{T_i}$ that $i$-represents $T_i$ in time 
	\[
		\mathcal{O}^*\bigl(\rank(M)^{|R_{i-1}|} |C|^{|X_{i-1}| - |R_{i-1}|} N\bigr),
	\]
along with a set $R_i$ of reduced vertices for $\widehat{T_i}$ such that 
	\[
		|X_i \setminus R_i| \leq \bigl(\ctw - |R_i|\bigr) / 2 + 1
	\]
	and it holds that $(\widehat{T_i})^{-1}\bigl(\ZZ \setminus \{0\}\bigr) \subseteq C^{X_i} \times [i]_0 \times [iN]_0$.
\end{lemma}

\begin{proof} 
	For $c \in C^{X_i}$ and $\order, \weight \in \ZZ$, we define 
	\begin{equation}\label{eq:reduced-dp}
		\widehat{T_i}[c, \order, \weight] \equiv \bigl[c(v_i) \in a(v_i)\bigr] \sum\limits_{\substack{z \in C^{X_{i-1}} \\ z \sim c}} \widehat{T_{i-1}} \Bigl[z, \order - [c(v_i) \in Q], \weight - \omega(v_i)\cdot\bigl[c(v_i) \in Q\bigr]\Bigr].
	\end{equation}
	Note that this is the same recurrence as in \cref{lem:dp-equality} but applied to a fully reduced table $\widehat{T_{i-1}}$ $(i-1)$-representing $T_{i-1}$.
	Observe that properties $(\widehat{T_{i-1}})^{-1}\bigl(\ZZ \setminus \{0\}\bigr) \subseteq C^{X_{i-1}} \times [i-1]_0 \times \bigl[(i-1)N\bigr]_0$ and  $\omega(v_i) \in [N]$ imply that $\widehat{T_i}$ satisfies the property $\widehat{T_i}^{-1}\bigl(\ZZ \setminus \{0\}\bigr) \subseteq C^{X_i} \times [i]_0 \times [iN]_0$.

	We will now show how to compute $\widehat{T_i}$ from $\widehat{T_{i-1}}$ efficiently by proceeding similarly to the proof of \cref{lemma:direct-DP-long}. Recall that by \cref{obs:compatibility}, a coloring $z \in C^{X_{i-1}}$ is compatible with a coloring $c \in C^{X_i}$ if and only if
	\[
		c_{|_{X_i \setminus \{v_i\}}} = z_{|_{X_i \setminus \{v_i\}}} \text{ and } z \sim \bigl[v_i \mapsto c(v_i)\bigr]
	\]
	hold. First, note that we do not need to store the entries $\widehat{T_i}[c, \order, \weight] \equiv 0$ for $\order \notin [i]_0$ or $\weight \notin[iN]_0$ explicitly. So from now on, we may assume that $\order \in [i]_0$ and $\weight \in [iN]_0$ hold.   
	Since $\widehat{T_{i-1}}$ is fully reduced, to compute the sum in \eqref{eq:reduced-dp}, it suffices to only consider such colorings $z \in C^{X_{i-1}}$ that for all $v \in R_{i-1}$ we have $z(v) \leq \rank(M)$.
	There are $\rank(M)^{R_{i-1}} |C|^{|X_{i-1}| - |R_{i-1}|}$ colorings of $X_{i-1}$ with this property.
	To compute the table $\widehat{T_i}$, we iterate over all such colorings $z \in C^{X_{i-1}}$ and all $s \in a(v_i)$. Then we determine if $z \sim [v_i \mapsto s]$ holds by looking at the edges having one end-vertex in $X_{i-1}$ and the other being $v_i$ and checking if their colors are consistent with respect to $M$. If so, for all $\order \in [i-1]_0$ and all $\weight \in \bigl[(i-1)N\bigr]_0$, we increase the value
	\[
		\widehat{T_i}\bigl[z_{|_{X_i \setminus \{v_i\}}}[v_i \mapsto s], \order + [s \in Q], \weight + \omega(v_i)\cdot[s \in Q]\bigr]
	\]
	by the value
	\[
		\widehat{T_{i-1}}[z, \order, \weight] 
	\]
	(if it was already initialized) or we initialize it with this value (otherwise). The non-initialized entries are implicitly treated as zero. Note that in particular, for any $c \in C^{X_i}$ with $c(v_i) \notin a(v_i)$, the value $\widehat{T_i}[c, \order, \weight]$ is zero. So by the end of this iteration, the table $\widehat{T_i}$ contains exactly the values as defined in \eqref{eq:reduced-dp}. The running time of the computation of $\widehat{T_i}$ is then bounded by
	\[
		\mathcal{O}\Bigl(\rank(M)^{|R_{i-1}|} |C|^{|X_{i-1}| - |R_{i-1}|} \bigl|a(v_i)\bigr| n (nN) \poly(n)\Bigr)
	\]
	and hence, the table can be computed within the claimed bound.
	
	Next we prove that $\widehat{T_i}$ indeed $i$-represents $T_i$. Let $y \in C^{Y_i}$ and $\order, \weight \in \ZZ$ be arbitrary but fixed. We need to show that 
	\[
		\sum\limits_{\substack{c \in C^{X_i} \\ c \sim y}} \widehat{T_i}[c, \order, \weight] \equiv \sum\limits_{\substack{c \in C^{X_i} \\ c \sim y}} T_i[c, \order, \weight]
	\]
	holds. 
	Recall that for any coloring $c \in C^{X_i}$ with $c(v_i) \notin a(v_i)$, we have
	\[
		\widehat{T_i}[c, \order, \weight] \equiv 0 \equiv T_i[c, \order, \weight].
	\]
	So it suffices to prove that
	\begin{equation}\label{eq:representation-for-gen-framework}
		\sum\limits_{\substack{c \in C^{X_i} \\ c(v_i) \in a(v_i) \\c \sim y}} \widehat{T_i}[c, \order, \weight] \equiv \sum\limits_{\substack{c \in C^{X_i} \\ c(v_i) \in a(v_i) \\ c \sim y}} T_i[c, \order, \weight]
	\end{equation}
	holds.
	
	For $\order \notin [i]_0$ or $\weight \notin [iN]_0$, the equality holds since both sides are equal to zero. So assume $\order \in [i]_0$ and $\weight \in [iN]_0$. For shortness of notation, we define a function $g: \ZZ^2 \times C \to \ZZ^2$ as
	\[
		g(\order', \weight', s) = \bigl(\order' - [s \in Q], \weight' - \omega(v_i) \cdot [s \in Q]\bigr).
	\]
	for all $\order', \weight' \in \ZZ$ and $s \in C$. 
	Then it holds that
	\begin{multline*}
		\sum\limits_{\substack{c \in C^{X_i} \\ c(v_i) \in a(v_i) \\ c \sim y}} T_i[c, \order, \weight] \\ 
		\stackrel{\eqref{eq:reduced-dp}}{\equiv} \sum\limits_{\substack{c \in C^{X_i} \\ c(v_i) \in a(v_i) \\ c \sim y}} \sum\limits_{\substack{z \in C^{X_{i-1}} \\ z \sim c}} T_{i-1}\Bigl[z, g\bigl(\order, \weight, c(v_i)\bigr)\Bigr] \\
		\equiv \sum\limits_{z \in C^{X_{i-1}}}\sum\limits_{\substack{c \in C^{X_i} \\ c(v_i) \in a(v_i) \\ c \sim y \\ z \sim c}} T_{i-1}\Bigl[z, g\bigl(\order, \weight, c(v_i)\bigr)\Bigr] \\
		\stackrel{\text{\cref{lem:compatibility-y}}}{\equiv} \sum\limits_{z \in C^{X_{i-1}}} \sum_{s \in a(v_i)}  T_{i-1}\Big[z, g\bigl(\order, \weight, z_{|_{X_i \setminus \{v_i\}}}[v_i \mapsto s](v_i)\bigr)\Bigr] 1_{z \sim [v_i \mapsto s]} 1_{[v_i \mapsto s] \sim y} 1_{z \sim y} \\
		\equiv \sum\limits_{z \in C^{X_{i-1}}} \sum_{s \in a(v_i)} T_{i-1}\bigl(z, g(\order, \weight, s)\bigr) 1_{z \sim [v_i \mapsto s]} 1_{[v_i \mapsto s] \sim y} 1_{z \sim y} \\
		\equiv \sum\limits_{s \in a(v_i)} 1_{[v_i \mapsto s] \sim y} \sum\limits_{\substack{z \in C^{X_{i-1}} \\ z \sim \bigl(y[v_i \mapsto s]\bigr)}} T_{i-1}\bigl[z, g(\order, \weight, s)\bigr] \\
		\stackrel{N(X_{i-1}) \cap \bigl(\{v_i\} \cup Y_i\bigr) \subseteq Y_{i-1}}{\equiv} \sum\limits_{s \in a(v_i)} 1_{[v_i \mapsto s] \sim y} \sum\limits_{\substack{z \in C^{X_{i-1}} \\ z \sim \bigl(y[v_i \mapsto s]\bigr)_{|_{Y_{i-1}}}}} T_{i-1}\bigl[z, g(\order, \weight, s)\bigr] \\
		\stackrel{(i-1)-\text{representation}}{\equiv} \sum\limits_{s \in a(v_i)} 1_{[v_i \mapsto s] \sim y} \sum\limits_{\substack{z \in C^{X_{i-1}} \\ z \sim \bigl(y[v_i \mapsto s]\bigr)_{|_{Y_{i-1}}}}} \widehat{T_{i-1}}\bigl[z, g(\order, \weight, s)\bigr]. \\
	\end{multline*}
	Observe that all of the above congruencies still hold if we replace $T_i$ with $\widehat{T_i}$ and $T_{i-1}$ with $\widehat{T_{i-1}}$. Therefore, the equality~\eqref{eq:representation-for-gen-framework} holds and $\widehat{T_i}$ indeed $i$-represents $T_i$.

	Now we need to provide a sufficiently large set of vertices reduced in $\widehat{T_i}$. 
	Let $A_i$, $B_i$, and $R_i$ be as in \cref{lem:sets-A-i-B-i-R-i}.
	By this lemma, we know that $|X_i \setminus R_i| \leq \bigl(\ctw - |R_i|\bigr) / 2 + 1$ holds.	
	It remains to prove that the vertices in $R_i$ are indeed reduced in $\widehat{T_i}$. 
	The idea of the following proof is the same as in \cite{GroenlandMNS22}.
	Suppose there is a not reduced vertex $u \in R_i$, i.e., there exist numbers $\order, \weight \in \ZZ$, and a coloring $c \in C^{X_i}$ such that $c(u) > \rank(M)$ and $\widehat{T_i}[c, \order, \weight] \not\equiv 0$. Due to \eqref{eq:reduced-dp}, there exists a coloring $z \in C^{X_{i-1}}$ such that $z \sim c$ and 
	\begin{equation}\label{eq:proof-reduced-vtx}
		\widehat{T_{i-1}}\Bigl[z, \order - \bigl[c(v_i) \in Q\bigr], \weight - \omega(v_i)\bigl[c(v_i) \in Q\bigr]\Bigr] \not\equiv 0
	\end{equation}
	hold.
	Recall that $u \in R_i \subseteq X_i \setminus \{v_i\} \subseteq X_{i-1}$ holds. Due to $z \sim c$, we have $z(u) = c(u) > \rank(M)$. Since $u \notin A_i \cup B_i$, the vertex $u$ has at most one incident edge going across the $(i-1)$th cut. Moreover, $u \in X_{i-1}$ so it has at least one edge going across the $(i-1)$th cut. Thus, the vertex $u$ has exactly one neighbor in $Y_{i-1}$. But then \eqref{eq:proof-reduced-vtx} implies that $\widehat{T_{i-1}}$ is not fully reduced -- a contradiction. Therefore $R_i$ is a set of reduced vertices of $\widehat{T_i}$. Clearly, $R_i$ can be computed in polynomial time.
\end{proof}

Now we are ready to assemble a faster algorithm for the problem $P$.

\begin{theorem}\label{thm:whole-algorithm-reduced}
	Let $P$ be a coloring-like problem with $\sqrt{|C|} \leq \rank(M)$ holds. Then $P$ can be solved in $\mathcal{O}^*\bigl(\rank(M)^{\ctw} N\bigr)$ time.
\end{theorem}

\begin{proof}
	We initialize the table $T_1$ via brute-force: we go through all $s \in C$, all $\order \in \{0, 1\}$, and all $\weight \in \bigl\{0, \omega(v_1)\bigr\}$ and compute the corresponding value $T_1\bigl[[v_1 \mapsto s], \order, \weight\bigr]$. All other entries of $T_1$ are implicitly set to $0$. If $v_1$ has exactly one neighbor in $G$, apply the \textbf{Reduce} algorithm from \cref{lemma:reduce-algorithm} or keep $T_1$ unchanged otherwise. As a result, obtain a fully reduced function $\widehat{T_1}$ that $1$-represents $T_1$ with some set of reduced vertices $R_1$. For $i = 2, \dots, n$ we repeat the following two steps.
	\begin{enumerate}
		\item Apply \cref{lemma:dp-reduced-version} with inputs $(\widehat{T_{i-1}}, R_{i-1})$ in order to obtain the table $\widehat{T_i}$ that $i$-represents $T_i$ and a set of reduced vertices $R_i$ for $\widehat{T_i}$.
		\item While $X_i \setminus R_i$ has a vertex that has exactly one neighbor in $Y_i$, apply the \textbf{Reduce} algorithm from \cref{lemma:reduce-algorithm} to $i$ and $(\widehat{T_i}, R_i)$ to obtain a new table $\widehat{\widehat{T_i}}$ that $i$-represents $\widehat{T_i}$ and whose reduced vertices are $R_i \cup \{v\}$ and set $\widehat{T_i} := \widehat{\widehat{T_i}}$ and $R_i := R_i \cup \{v\}$.
	\end{enumerate}

	At the end of step 2, we obtain a fully reduced table $\widehat{T_i}$ that $i$-represents $T_i$. Moreover, the set $R_i$ of reduced vertices has only increased in size compared to the set we obtained in step 1. Hence, there are at most $|X_i|$ repetitions of step 2. 

	As a result of this process, we have computed the entries of a table $\widehat{T_n}$ that is fully reduced and $n$-represents $T_n$ and a set $R_n$ of reduced vertices of $\widehat{T_n}$. We output the value 
	\[
		\sum\limits_{c \in C^{X_n}} \widehat{T_n}[c, \totalorder, \totalweight].
	\]  
	Note that $X_n = \{v_n\}$ and hence, this sum contains only $|C|$ addends. Next observe that $Y_n = \emptyset$ holds, so there is a unique (empty) coloring $y^*: Y_n \to C$ of $Y_n$ and every coloring of $X_n$ is compatible with $y^*$. Thereby, the following holds:
	\begin{align*}
		&\sum\limits_{c \in C^{X_n}} \widehat{T_n}[c, \totalorder, \totalweight] \\
		=& \sum\limits_{\substack{c \in C^{X_n} \\ c \sim y^*}} \widehat{T_n}[c, \totalorder, \totalweight] \\
		\stackrel{n-\text{representation}}{\equiv}& \sum\limits_{\substack{c \in C^{X_n} \\ c \sim y^*}} T_n[c, \totalorder, \totalweight] \\
		=& \sum\limits_{c \in C^{X_n}} T_n[c, \totalorder, \totalweight].
	\end{align*}
	So the output value is the number of valid colorings of $G$ 
 	of order $\totalorder$ and weight $\totalweight$ as desired.

	The total running time is now bounded by 
	\[
		\sum\limits_{i=1}^{n} \mathcal{O}^*\bigl(|X_i| \rank(M)^{|R_i|} |C|^{|X_i| - |R_i|} N\bigr) = \sum\limits_{i=1}^{n} \mathcal{O}^*\bigl(\rank(M)^{|R_i|} |C|^{|X_i| - |R_i|} N\bigr).
	\]
	Recall that by \cref{lemma:dp-reduced-version}, for every $2 \leq i \leq n$, it holds that $|X_i \setminus R_i| \leq \bigl(\ctw - |R_i|\bigr) / 2 + 1$. 
	Then we have:
	\begin{align*}
		& \rank(M)^{|R_i|} |C|^{|X_i| - |R_i|} \\
		\leq &\rank(M)^{|R_i|} |C|^{\bigl(\ctw - |R_i|\bigr) / 2 + 1} \\
		= &\rank(M)^{|R_i|} \sqrt{|C|}^{\ctw - |R_i|} |C| \\
		\stackrel{\sqrt{|C|} \leq \rank(M)}{\leq} &\rank(M)^{|R_i|} \rank(M)^{\ctw - |R_i|} |C| \\
		= &\rank(M)^{\ctw} |C|.
	\end{align*}
	Hence, the running time is bounded by $\mathcal{O}^*\bigl(\rank(M)^{\ctw} N\bigr)$.
\end{proof}

\subsection{\textsc{Connected Vertex Cover}} \label{subsec:cvc-upper-bound}

A \emph{vertex cover} of a graph is a set of vertices such that every edge of the graph is incident to some vertex in this set. 
So the \Pcvc{} is defined as follows.

\begin{quote}
	\textsc{Connected Vertex Cover}

	\textbf{Input}: A graph $G = (V, E)$ and an integer $k$.
	
	\textbf{Question}: Is there a vertex cover $S \subseteq V$ of size $k$ such that $G[S]$ is connected.
\end{quote}

\begin{theorem}\label{thm:cvc-ub}
	There exists an algorithm that given a graph $G$ and its linear arrangement of cutwidth $\ctw$ in time $\OO^*(2^{\ctw})$ solves the \textsc{Connected Vertex Cover} problem. The algorithm cannot give false positives and may give false negatives with probability at most~$1/2$.
\end{theorem}

\begin{proof}
	Recall that by \cref{thm:cut-and-count-short}, to prove the claim, it suffices to provide an algorithm $\AAA$ that given $G$, $\ell$, a fixed vertex $v \in V$, a weight function $\omega: V \to \bigl[2|V|\bigr]$, and numbers $\totalorder \in [n]_0, \totalweight \in \bigl[2n|V|\bigr]_0$ computes the size of $\CCC^\totalorder_\totalweight$ modulo 2 in time $\OO^*(2^{\ctw})$.
	
	We fit this task into a coloring-like problem as follows. Let $C = \{X, L, R\}$ ($X$ stands for not belonging to a vertex cover while $L$ and $R$ are the sides of a consistent cut), let $Q = \{L, R\}$ (to count the number of vertices in a vertex cover), and let the consistency matrix $M$ have zero-entries $M[X, X] = M[L, R] = M[R, L] = 0$ whilst the remaining values are equal to one, i.e.,
	\[
		M = \kbordermatrix{
			& X & L & R  \\
			X & 0 & 1 & 1 \\
 			L & 1 & 1 & 0  \\
			R & 1 & 0 & 1  
 		}
 		.
	\]
	The consistency matrix reflects that for every edge, there must be an end-vertex in a vertex cover and there is no edge between $L$ and $R$. 
	Further, we define the list function $a$ as follows. For every $w \in V$, we set
	\[
		a(w) =
		\begin{cases}
			S & \mbox{if $w \neq v$,} \\
			\{L\} & \mbox{otherwise.}
		\end{cases}
	\]

	The instance of a coloring-like problem defined by $C$, $Q$, $M$, $p=2$, $a$, $G$, $\ell$, $\totalorder$, and $\totalweight$ then corresponds to the problem of counting the number of consistent cuts of $G$ of order $\totalorder$ and weight $\totalweight$. Observe that the consistency matrix $M$ has rank of two over $\FF_2$: the sum of the first and the second row is equal to the third row (modulo 2). Thereby, by \cref{thm:whole-algorithm-reduced}, the size of $\CCC^\order_\weight$ modulo 2 can be determined in time $O^*\bigl(2^{\ctw} 2|V|\bigr) = \OO^*(2^{\ctw})$ and this concludes the proof.
\end{proof}

\subsection{Connected Dominating Set} \label{app:cds}

A \emph{dominating set} of a graph is a set of vertices such that for every vertex of the graph not in this set, at least one of its neighbors belongs to this set. \Pcds{} is defined as follows.

\begin{quote}
	\Pcds{}

	\textbf{Input}: A graph $G = (V, E)$ and an integer $k$.

	\textbf{Question}: Is there a dominating set $S \subseteq V$ of $G$ of cardinality at most $k$ such that~$G[S]$ is connected.
\end{quote}

By \cref{thm:cut-and-count-short}, the \pcds\ problem again reduces to the problem of counting the number of certain consistent cuts (i.e., vertex-colorings). Unfortunately, the conditions on these colorings do not yet fit in our framework.
The conditions defining a general coloring-like problem are \emph{universal properties} determined by a consistency matrix $M$: \textbf{all} neighbors of a vertex of a color $a$ are only allowed to get one of colors from some set $B$. The satisfaction of these conditions can be verified by checking every edge of the graph independently. The trick from \cref{lemma:reduce-algorithm} that allowed to reduce the vertices with only one incident edge in the $i$th cut (for $i \in [n]$) utilized this property. 
However, \pcds\ contains an \emph{existential property}: for every vertex not in the dominating set, there \textbf{exists} a neighbor in the dominating set, i.e., on the left or right side of a consistent cut.
Therefore, the edges cannot be checked independently. 

We will overcome this issue with the help of the inclusion-exclusion approach.
In addition to the colors $L$, $R$ (for the sides of a consistent cut), and $D$ (for dominated vertices), to handle partial solutions we introduce two further colors $A$ and $F$. 
The color $A$ (stands for ``allowed'') denotes that a vertex is not in the partial solution but it is possibly dominated by it. The color $F$ (stands for ``forbidden'') denotes that a vertex is not in a partial solution and it is not dominated by it. 
Then the number of partial solutions to which a vertex does not belong but in which it is dominated is then equal to the number of solutions in which it is possibly dominated minus the number of solutions where it is not dominated (informally: $D = A - F$).
This is an application of the well-known inclusion-exclusion paradigm. In this form, it has been used by Pilipczuk and Wrochna to solve the \Pds\ problem parameterized by treedepth \cite{PilipczukW18}.

In this way, the neighbors of a vertex colored with $A$ are allowed to have any color and the neighbors of a vertex colored with $F$ are not allowed to have the color $L$ or $R$.
Thus, we obtain the desired universal property and can solve the problem with almost the same techniques as we used for a general coloring-like problem. 
However, we can not yet directly apply an algorithm solving a general coloring-like problem. In the end, we aim to obtain the number of certain colorings with colors $L$, $R$, and $D$. So for every vertex, we need to carry out the conversion from colors $A$ and $F$ into the color $D$ at some point. But we can only ``access'' the color of a certain vertex only as long as it belongs to the left side $X_i$ of the current cut $E_i$. Therefore, we cannot first solve the above-mentioned coloring-like problem with colors $L$, $R$, $A$, and $F$ and then convert the colors $A$ and $F$ into the color $D$. Instead, we will proceed similarly to the previous subsections but every time a vertex appears in some $X_i$ for the last time, we will ``finalize'' its colors, namely we will carry out the aforementioned conversion for this vertex and from that moment only count the colorings using a color $L$, $R$, or $D$ on it.
Now we provide the technical details of this idea. 

Here let $G = (V, E)$ again be a fixed graph with a linear arrangement $\ell$ of cutwidth $\ctw$. Let $N = 2|V|$, let $\omega: V \to [N]$ be a weight function and let $\totalorder \in [n]_0$ and $\totalweight \in [nN]_0$. Finally, let $v^* \in V$ be a distinguished vertex.
Let $\finalcolors = \{L, R, D\}$ be the so-called \emph{final} colors since they are used in the colorings we want to count, let $\colors = \{L, R, A, F\}$ be the so-called \emph{intermediate} colors as these are the colors of partial solutions, and let $\allcolors = \{L, R, D, A, F\}$ denote the set of all colors. 
For a subset $X \subseteq V$ and a coloring $f \in \allcolors^X$, we say that $f$ has \emph{order} $\Bigl|f^{-1}\bigl(\{L, R\}\bigr)\Bigr|$ and \emph{weight} $\omega\Bigl(f^{-1}\bigr(\{L, R\}\bigl)\Bigr)$.
Let 
\[
	M = \kbordermatrix{
		& L & R & A & F  \\
		L & 1 & 0 & 1 & 0  \\
 		R & 0 & 1 & 1 & 0  \\
		A & 1 & 1 & 1 & 1 \\
		F & 0 & 0 & 1 & 1
 	}
\]
be the consistency matrix indexed by intermediate colors and capturing the interpretation sketched above.

For $i \in [n]$, we say that a coloring $\phi \in \allcolors^{V_i}$ is a \emph{valid extension} of a coloring $x \in \allcolors^{S}$ (for some $S \subseteq V_i$) if the following properties hold:
\begin{enumerate}
	\item $\phi_{|_S} = x$, in this case we say that $\phi$ is an \emph{extension} of $x$,
	\item for all $v \in V_i \setminus S\colon \phi(v) \in \finalcolors$,
	\item for every edge $\{u, v\} \in E$ such that $u, v \in V_i$ and $\phi(u), \phi(v) \in \colors$, it holds that $M\bigl[\phi(u), \phi(v)\bigr] = 1$,
	\item for every vertex $u \in V_i$ such that $\phi(u) = D$, there exists a vertex $v \in V_i$ such that $\{u, v\} \in E$ and $\phi(v) \in \{L, R\}$,
	\item and if $v^* \in V_i$, then it holds that $\phi(v^*) = L$.
\end{enumerate}
The main idea behind it is that at the $i$th cut we will be interested in colorings of $V_i$ such that the vertices from $V_i \setminus X_i$ have already been ``finalized'', namely they only have colors $L$, $R$, and $D$; whilst the vertices of $X_i$ might still have intermediate colors $L$, $R$, $A$, and $F$.

Recall that for $2 \leq i \leq n$, the set $Z_i$ is defined as $Z_i = X_{i-1} \cup \{v_i\}$.
For $1 \leq i \leq n$ (resp.\ $2 \leq i \leq n$) and $x \in \colors^{X_i}$ (resp.\ $z \in \colors^{Z_i}$), we define $T_i[x, K, W]$ (resp.\ $P_i[z, K, W]$) as the number of valid extensions $\phi \in \allcolors^{V_i}$ of $x$ (resp.\ $z$) of order $K$ and weight $W$.

For $2 \leq i \leq n$, let $\{w_i^1, \dots, w_i^{t_i}\} = Z_i \setminus X_i = X_{i-1} \setminus X_i$ and let us fix this ordering. These vertices do not have any neighbors in $V \setminus V_i$. So for any dominating set of $G$, any of these vertices either belongs to the dominating set itself or it is dominated by a vertex in $V_i$. 
Now let $0 \leq l \leq t_i$ and let 
\[
	z \in \colors^{X_i} \times \finalcolors^{\{w_i^1, \dots, w_i^l\}} \times \colors^{\{w_i^{l+1}, \dots, w_i^{t_i}\}}.
\]
Then the value $Q_i^l[z, K, W]$ is defined as the number of valid extensions of $z$ of order $K$ and weight $W$. 
Finally, we set $Q_i = Q_i^{t_i}$ for simplicity.

By the definition of these tables, we obtain the following observation.
\begin{observation}
	For $i \in [n]$, it holds that the functions $P_i$ and $Q^0_i$ have the same domain and $P_i = Q^0_i$.
\end{observation}

In the next three lemmas, we provide the dynamic programming equalities for the straight-forward algorithm solving the problem in $\mathcal{O}^*(4^{\ctw})$. In the next section we accelerate it to achieve the desired running time of $\mathcal{O}^*(3^{\ctw})$. 
Given the table $T_{i-1}$ (for some $2 \leq i \leq n$), the algorithm first computes the table $P_i = Q_i^0$ from it, then the tables $Q_i^1$, $Q_i^2$, $\dots$, $Q_i^{t_i} = Q_i$ in this order, and finally it computes the table $T_i$ using the equalities from the following three lemmas.
To prove the correctness of each of these equalities, we will provide a bijection between the objects counted on the left and on the right side for each case.

\begin{lemma}\label{lem:CDS-non-reduced-DP-a}
	Let $2 \leq i \leq n$, $\order, \weight \in \ZZ$, and $z \in \dom(P_i)$. 
	Let $\order^z$ and $\widetilde{\weight}^z$ denote the values $\order^z = \order - \bigl[z(v_i) \in \{L, R\}\bigr]$ and $\widetilde{\weight}^z = \weight - \omega(v_i)\cdot \bigl[z(v_i) \in \{L, R\}\bigr]$.
	Then it holds that 
	\begin{align*}
		&P_i[z, \order, \weight] \\
		=&1_{v^* \neq v_i \lor z(v_i) = L} \cdot
		\begin{cases}
			T_{i-1}[z_{|_{X_{i-1}}}, \order^z, \widetilde{\weight}^z] & \mbox{if $z_{|_{X_{i-1}}} \sim z$,} \\
			0 & \mbox{otherwise,}
		\end{cases}
		\\
		= &1_{v^* \neq v_i \lor z(v_i) = L} \sum\limits_{\substack{x \in \dom(T_{i-1}) \\ x \sim z}} T_{i-1}[x, \order^z, \widetilde{\weight}^z]
			.
	\end{align*} 
\end{lemma}

\begin{proof}
	First of all, recall that $X_{i-1} \subsetneq Z_i$ holds and therefore, a coloring $x \in \dom(T_{i-1}) = \colors^{X_{i-1}}$ of $X_{i-1}$ can only be compatible with $z$ if $x = z_{|_{X_{i-1}}}$ holds so the second equality is true. Now we prove that the first equality holds as well.
	Note that if $v_i = v^*$ and $z(v_i) \neq L$, then no extension of $z$ is valid so the equality holds. So from now on we assume that either $v_i \neq v^*$ or $z(v_i) = L$ holds.  
	Recall that $x$ only uses colors from $\colors$.
	So if $z_{|_{X_{i-1}}} \not\sim z$ holds (i.e., there is one edge the colors of whose end-vertices violate the consistency matrix), then there is no valid extension of $z$ and the first equality holds. 
	Further, if $z(v_i) \in \{L, R\}$ and $\weight < \omega(v_i)$ or $\order \leq 0$, then both sides of the equality are trivially zero and the first equality is true. So we may further assume that this does not apply.
	
	Let $\phi \in \allcolors^{V_i}$ be a valid extension of $z$ counted on the left side in $P_i[z, K, W]$. 
	We claim that $\phi_{|_{V_{i-1}}}$ is a valid extension of $z_{|_{X_{i-1}}}$ counted on the right side in $T_{i-1}[z_{|_{X_{i-1}}}, \order^z, \widetilde{\weight}^z]$. 
	The order and weight conditions are clearly satisfied and the consistency with respect to $M$ still holds.
	The only non-trivial condition to prove is that every vertex $v \in V_{i-1}$ of color $D$ in $\phi$ still has a neighbor in color $L$ or $R$ when $\phi$ is restricted to $V_{i-1}$. By the definition of $\dom(P_i)$, such a vertex $v$ belongs to $V_i \setminus Z_i$, i.e., $v \notin X_{i-1}$ holds. 
	So $v$ does not have an incident edge crossing the $(i-1)$th cut and in particular, $v$ is not adjacent to $v_i$. So this condition is still satisfied on $V_{i-1}$.
		
	Now let $\phi \in \allcolors^{V_{i-1}}$ be a valid extension of $z_{|_{X_{i-1}}}$ counted in $T_{i-1}[z_{|_{X_{i-1}}}, \order^z, \widetilde{\weight}^z]$. We claim that $\phi' = \phi\bigl[v_i \mapsto z(v_i)\bigr] \in \allcolors^{V_i}$ is a valid extension of $z$ counted in $P_i[z, \order, \weight]$. 
	The order and weight conditions are clearly satisfied. 
	The consistency with respect to $M$ still holds for two reasons. First, it holds for edges induced by $V_{i-1}$ since $\phi$ is a valid extension. Second, for edges incident to $v_i$ and having their other end-vertex $v$ in $V_{i-1}$, it holds that $v \in X_{i-1}$ so for such edges the property $z_{|_{X_{i-1}}} \sim z$ implies the consistency in $\phi'$ as well.
	Finally, recall that $\dom(P_i) = \colors^{Z_i}$ so every vertex $v$ in $V_i$ colored with $D$ in $\phi'$ is not equal to $v_i$ and therefore it has a neighbor in color $L$ or $R$ in $\phi'$ since it already had one in $\phi$.
\end{proof}

\begin{lemma}\label{lem:CDS-non-reduced-DP-b}
	Let $2 \leq i \leq n$, $0 \leq l \leq t_i-1$, $\order, \weight \in \ZZ$, and $z \in \dom(Q_i^{l+1})$.
	Then it holds that
	\[
		Q_i^{l+1}[z, \order, \weight] = 
		\begin{cases}
			Q_i^l[z, \order, \weight] & \mbox{if $z(w_i^{l+1}) \in \{L, R\}$,} \\
			Q_i^l\bigl[z[w_i^{l+1} \mapsto A], \order, \weight\bigr] - Q_i^l\bigl[z[w_i^{l+1} \mapsto F], \order, \weight\bigr] & \mbox{if $z(w_i^{l+1}) = D$.}
		\end{cases}
	\]
\end{lemma}

\begin{proof}
	For the simplicity of notation, we use $v = w_i^{l+1}$.
	First of all, note that by the definition of $Q_i^{l+1}$, we have $z(v) \in \{L, R, D\}$. 
	If we have $z(v) \in \{L, R\}$, then $z$ belongs to the domain of $Q^l_i$ as well and the claim is trivially true. 
	So we may assume $z(v) = D$.
	We prove that 
	\[
		Q_i^{l+1}[z, \order, \weight] + Q_i^l\bigl[z[v \mapsto F], \order, \weight\bigr] = Q_i^l\bigl[z[v \mapsto A], \order, \weight\bigr]
	\]
	holds. 
	
	Let $\phi \in \allcolors^{V_i}$ be a valid extension of $z[v \mapsto A]$ counted on the right side in $Q_i^l\bigl[z[v \mapsto A], \order, \weight\bigr]$. Note that $\bigl(z[v \mapsto A]\bigr)[v \mapsto F] = z[v \mapsto F]$ and $\bigl(z[v \mapsto A]\bigr)[v \mapsto D] = z$ holds.
	If there is a neighbor $u \in V_i$ of $v$ such that $\phi(u) \in \{L, R\}$, then $\phi[v \mapsto D]$ is a valid extension of $z$ counted in $Q_i^{l+1}[z, \order, \weight]$. Otherwise, $\phi$ is a valid extension of $z[v \mapsto F]$ counted in $Q_i^l\bigl[z[v \mapsto F]\bigr]$.
		
	Now we consider the colorings counted on the left side.
	First, let $\phi \in \allcolors^{V_i}$ be a valid extension of $z$ counted in $Q_i^{l+1}[z, \order, \weight]$. 
	Then we claim that $\phi[v \mapsto A]$ is a valid extension of $z$ counted in $Q_i^l\bigl[z[v \mapsto A], \order, \weight\bigr]$. First, a vertex of color $A$ is allowed to have neighbors of any color. Second, consider a vertex $u \neq v$ such that $\phi(u) = \phi[v \mapsto A](u) = D$. Since $\phi$ is valid, the vertex $u$ had a neighbor $w \in V_i$ in color $L$ or $R$ in $\phi$. Then it holds that $w \neq v$ and hence, the vertex $w$ still has the same color $L$ or $R$ in $\phi[v \mapsto A]$. Clearly, the order and weight conditions hold as well.
	So $\phi[v \mapsto A]$ is indeed counted in $Q_i^l\bigl[z[v \mapsto A], \order, \weight\bigr]$.
		
	Second, let $\phi'$ be a valid extension of $z[v \mapsto F]$ counted in $Q_i^l\bigl[z[v \mapsto F], \order, \weight\bigr]$. 
	Then $\phi'$ is also a valid extension of $z$ counted in $Q_i^l\bigl[z[v \mapsto A], \order, \weight\bigr]$: again, since we only flipped the color of $v$ from $F$ to $A$, no vertex could have lost a neighbor in $L$ or $R$.
		
	Let $\phi$ resp.\ $\phi'$ be as before.
	Note that because of $\phi(v) = z(v) = D$, there is a neighbor $u \in V_i$ of $v$ such that $\phi(u) \in \{L, R\}$. On the other hand, because of $\phi'(v) = F$, there is no neighbor $u \in V_i$ of $v$ such that $\phi'(u) \in \{L, R\}$. Therefore, for any such $\phi$ and $\phi'$, it holds $\phi[v \mapsto A] \neq \phi'[v \mapsto A]$ and hence, every coloring counted on the left side is injectively mapped to a coloring counted on the right side.
\end{proof}

\begin{lemma}\label{lem:CDS-non-reduced-DP-c}
	Let $2 \leq i \leq n$, $\order, \weight \in \ZZ$, and $x \in \dom(T_i)$.
	Then it holds that
	\[
		T_i[x, \order, \weight] = \sum\limits_{\substack{z \in \dom(Q_i) \\ z_{|_{X_i}} = x}} Q_i[z, \order, \weight].
	\]
\end{lemma}

\begin{proof}
	We recall that by definition it holds that $Q_i = Q_i^{t_i}$.
	Let $\phi \in \allcolors^{V_i}$ be counted on the left side in $T_i[x, \order, \weight]$ and let $z = \phi_{|_{Z_i}}$.
	First, it holds that $z \in \dom(Q_i) = \colors^{X_i} \times \finalcolors^{Z_i \setminus X_i}$ holds: since $\phi$ is a valid extension of $x \in \dom(T_i) = \colors^{X_i}$, all vertices of $X_i$ are colored with $\colors$ in $\phi$ and all vertices in $Z_i \setminus X_i \subseteq V_i \setminus X_i$ are colored with $\finalcolors$ in $\phi$.  
	Second, since $\phi$ is a valid extension of $x$ and $z$ is an extension of $x$, $\phi$ is a valid extension of $z$ and it is counted on the right side in $Q_i[z, \order, \weight]$.
				
	Let $\phi \in \allcolors^{V_i}$ be counted on the right side in $Q_i[z, \order, \weight]$ for some $z \in \dom(Q_i)$ such that $z_{|_{X_i}} = x$. 
	Recall that by the definition of the domain of $Q_i$, for every vertex $v \in Z_i \setminus X_i$, it holds that $z(v) \in \finalcolors$. 
	Since $\phi$ is a valid extension of $z$, it also holds that $\phi(V_i \setminus Z_i) \subseteq \finalcolors$. Therefore, it holds that $\phi(V_i \setminus X_i) \subseteq \finalcolors$ and $\phi$ is also a valid extension of $x$ counted in $T_i[x, \order, \weight]$.
\end{proof}

Note that for the domain of every table considered in \cref{lem:CDS-non-reduced-DP-a,lem:CDS-non-reduced-DP-b,lem:CDS-non-reduced-DP-c}, the following holds. For every vertex $v$, the domain contains only colorings assigning $v$ a color from $\{L, R, A, F\}$ or only colorings assigning $v$ a color from $\{L, R, D\}$. Therefore, the size of the domain is always upper bounded by $4^{|X_i|}$ resp.\ $4^{|Z_i|}$ (for $i \in [n]$), i.e., it is bounded by $\mathcal{O}(4^{\ctw})$. 
This almost immediately results in an $\mathcal{O}^*(4^{\ctw})$ algorithm solving the \textsc{Connected Dominating Set} problem. We omit the details here since our main result is an $\mathcal{O}^*(3^{\ctw})$ algorithm provided next.

Now we proceed analogously to our algorithm solving a general coloring-like problem to compute smaller tables representing the partial solutions. 
Here with $\equiv$ we denote the congruency modulo the prime number $p = 2$.

As before, we need the notion of compatibility with respect to the consistency matrix $M$. However, in our approach the color $D \notin \colors$ will sometimes also appear in partial solutions so we have to adapt the definition of compatibility to make sure it is well-defined. At the same time, we would like to apply the results from \cref{app:subsec:reduced-dp}, therefore we want the new definition to be consistent with the old one.

\begin{definition}
	Let $X, Y, Z \subseteq V$ be vertex sets such that $X \subseteq Z$ and $(Z \setminus X) \cap Y = \emptyset$ hold and there is no edge between $Z \setminus X$ and $Y$ in $G$. Further, let $x \in \finalcolors^{Z \setminus X} \times \colors^X$ and $y \in \colors^Y$ be colorings. We say that $x$ is \emph{compatible} with $y$ and write $x \sim y$ if $x_{|_X} \sim y$ holds. Namely, we have $x_{|_{X \cap Y}} = y_{|_{X \cap Y}}$ and for every edge $\{u, v\} \in E$ with $u \in Z$ and $v \in Y \setminus X$ (implies $u \in X$), it holds that $M\bigl[x(u), y(v)\bigr] = 1$. (Note that according to this definition, the colors of vertices in $Z \setminus X$ in $x$ do not influence the compatibility with $y$.)
	
	Further, let $A \cup B = Z \setminus X$ be a partition of $Z \setminus X$ and let $f, \widehat{f}: (\finalcolors^A \times \colors^B \times \colors^X) \times \ZZ^2 \to \ZZ$. We say that $f$ $(X, Y, Z)$-represents $\widehat{f}$ if for every coloring $y \in \colors^Y$ and every $K, W \in \ZZ$ it holds that
	\[
		\sum\limits_{\substack{c \in \dom(f) \\ c \sim y}} f(c, K, W) \equiv \sum\limits_{\substack{c \in \dom(f) \\ c \sim y}} \widehat{f}(c, K, W),
	\]
	i.e., it holds that
	\[
		\sum\limits_{\substack{c \in \dom(f) \\ c_{|_X} \sim y}} f(c, K, W) \equiv \sum\limits_{\substack{c \in \dom(f) \\ c_{|_X} \sim y}} \widehat{f}(c, K, W).
	\]
\end{definition}
It is crucial, that by setting $Z = X$ we obtain exactly the notions of compatibility and representation as for a general coloring-like problem. Clearly, this new notion of the representation is transitive as well.
Now the key observation for the construction of a faster algorithm is that the consistency matrix 
\[
	M = \kbordermatrix{
		& L & R & A & F  \\
		L & 1 & 0 & 1 & 0  \\
 		R & 0 & 1 & 1 & 0  \\
		A & 1 & 1 & 1 & 1 \\
		F & 0 & 0 & 1 & 1
 	}
\]
has rank $3$ in $\mathbb{F}_2$: the first three rows sum up to the fourth one and they build a row basis. So let $F$ be the (only) \emph{reduced} color. 
Let $f: \colors^{X_i} \times \ZZ^2 \to \ZZ$ with $i \in [n]$ be a table.
As before, we say that a vertex $v$ is \emph{reduced} in $f$ if for every coloring $c \in \colors^{X_i}$ such that $c(v) = F$ and every $K, W \in \ZZ$ it holds that $f(c, K, W) = 0$.
We say that $f: \colors^{X_i} \times \ZZ^2 \to \ZZ$ is \emph{fully reduced} if every vertex $v \in X_i$ that has one neighbor in $Y_i$ is reduced in $f$.
Since all of these definitions are consistent with the previous section, \cref{lemma:reduce-algorithm} can be applied to reduce vertices that have only one incident edge in some cut. The following lemma is therefore just a special case of it.
\begin{lemma}\label{lem:degree-1-vts-CDS}
	Let $X, Y \subseteq V$ be disjoint and let $i \in [n]$. 
	Further, let $g: \colors^X \times \ZZ^2 \to \ZZ$ be a table with reduced vertices $R \subseteq X$ such that $g^{-1}\bigl(\ZZ \setminus \{0\}\bigr) \subseteq \colors^X \times [i]_0 \times [iN]_0$.
	Finally, let $v \in X \setminus R$ be a vertex such that $v$ has exactly one neighbor in $Y$. 
	Then there is an algorithm \textbf{Reduce} that, given this information as input, in time $\mathcal{O}^*(3^{|R|} |S|^{|X| - |R|} N)$ outputs a function $\widehat{g}: \colors^X \times \ZZ^2 \to \ZZ$ $(X, Y, X)$-representing $g$ with reduced vertices $R \cup \{v\}$.
	Moreover, it holds that $\widehat{g}^{-1}\bigl(\ZZ \setminus \{0\}\bigr) \subseteq \colors^X \times [i]_0 \times [iN]_0$.
\end{lemma}
This lemma is crucial to achieve the desired running time. 
Our algorithm will run the dynamic programming defined by equalities in \cref{lem:CDS-non-reduced-DP-a,lem:CDS-non-reduced-DP-b,lem:CDS-non-reduced-DP-c} and then apply \cref{lem:degree-1-vts-CDS} to reduce in $T_i$ all vertices in $X_i$ having exactly one neighbor in $Y_i$. Now we prove that this approach is correct and after that, we provide more details to achieve the desired running time.
In the next three lemmas, for $2 \leq i \leq n$, we will write about $(X_i, Y_i, X_i)$- and $(X_i, Y_i, Z_i)$-representation. Recall that $X_i \cap Y_i = \emptyset$ and there is no edge $\{u, v\} \in E$ with $u \in Z_i \setminus X_i$ and $v \in Y_i$ since such an edge would cross the $i$th cut and $u$ would then belong to $X_i$.
So the representation is well-defined in these cases.
	
\begin{lemma}\label{lem:reduced-dp-CDS-a}
	Let $2 \leq i \leq n$ and let a table $\widehat{T_{i-1}}$ $(X_{i-1}, Y_{i-1}, X_{i-1})$-represent $T_{i-1}$.
	For all $z \in \dom(P_i)$ and $\order, \weight \in \ZZ$ let $\order^z = \order - \bigl[z(v_i) \in \{L, R\}\bigr]$, $\widetilde{\weight}^z = \weight - \omega(v_i)\cdot \bigl[z(v_i) \in \{L, R\}\bigr]$, and 
	\[
		\widehat{P_i}[z, \order, \weight] \equiv 1_{v^* \neq v_i \lor z(v_i) = L} \sum\limits_{\substack{x \in \dom(T_{i-1}) \\ x \sim z}} \widehat{T_{i-1}}[x, \order^z, \widetilde{\weight}^z].
	\]
	Then $\widehat{P_i}$ $(X_i, Y_i, Z_i)$-represents $P_i$.
\end{lemma}	
	
\begin{proof}
	Let $y \in \colors^{Y_i}$. 
	Recall that we have $\dom(P_i) = \colors^{Z_i}$ and $\dom(T_{i-1}) = \colors^{X_{i-1}}$.
	Let $a(v_i) = \colors$ if $v_i \neq v^*$ and $a(v_i) = a(v^*) = \{L\}$ 
	denote the set of colors $v_i$ can have in a coloring extendable to a valid coloring.
	Then we have
	\begin{multline*}
		\sum\limits_{\substack{z \in \dom(P_i) \\ z \sim y}} P_i[z, \order, \weight] \\
		\equiv \sum\limits_{\substack{z \in \dom(P_i) \\ z \sim y}} 1_{v^* \neq v_i \lor z(v_i) = L} \sum\limits_{\substack{x \in \dom(T_{i-1}) \\ x \sim z}} T_{i-1}[x, \order^z, \widetilde{\weight}^z] \\
		= \sum\limits_{x \in \dom(T_{i-1})} \sum\limits_{\substack{z \in \dom(P_i) \\ x \sim z, z \sim y}} 1_{v^* \neq v_i \lor z(v_i) = L} T_{i-1}[x, \order^z, \widetilde{\weight}^z] \\			
		= \sum\limits_{x \in \dom(T_{i-1})} \sum\limits_{z \in \dom(P_i)} 1_{v^* \neq v_i \lor z(v_i) = L} 1_{x \sim z} 1_{z \sim y} T_{i-1}[x, \order^z, \widetilde{\weight}^z] \\	
		\stackrel{Z_i = X_{i-1} \dot\cup \{v_i\}}{=} \sum\limits_{x \in \dom(T_{i-1})} \sum\limits_{s \in \colors} 1_{v^* \neq v_i \lor s = L} 1_{x \sim [v_i \mapsto s]} 1_{[v_i \mapsto s] \sim y} 1_{x \sim y} T_{i-1}[x, \order^z, \widetilde{\weight}^z] \\
		= \sum\limits_{x \in \dom(T_{i-1})} \sum\limits_{s \in a(v_i)} 1_{x \sim [v_i \mapsto s]} 1_{[v_i \mapsto s] \sim y} 1_{x \sim y} T_{i-1}[x, \order^z, \widetilde{\weight}^z] \\
		= \sum\limits_{s \in a(v_i)} 1_{[v_i \mapsto s] \sim y} \sum\limits_{x \in \dom(T_{i-1})} 1_{x \sim [v_i \mapsto s]}  1_{x \sim y} T_{i-1}[x, \order^z, \widetilde{\weight}^z] \\
		= \sum\limits_{s \in a(v_i)} 1_{[v_i \mapsto s] \sim y} \sum\limits_{\substack{x \in \dom(T_{i-1}) \\ x \sim \bigl(y[v_i \mapsto s]\bigr)}} T_{i-1}[x, \order^z, \widetilde{\weight}^z] \\
		\stackrel{N(X_{i-1}) \cap \bigl(\{v_i\} \cup Y_i\bigr) \subseteq Y_{i-1}}{=} \sum\limits_{s \in a(v_i)} 1_{[v_i \mapsto s] \sim y} \sum\limits_{\substack{x \in \dom(T_{i-1}) \\ x \sim \bigl(y[v_i \mapsto s]\bigr)_{|_{Y_{i-1}}}}} T_{i-1}[x, \order^z, \widetilde{\weight}^z] \\
		\stackrel{(X_{i-1}, Y_{i-1}, X_{i-1})-\text{repr.}}{\equiv} \sum\limits_{s \in a(v_i)} 1_{[v_i \mapsto s] \sim y} \sum\limits_{\substack{x \in \dom(T_{i-1}) \\ x \sim \bigl(y[v_i \mapsto s]\bigr)_{|_{Y_{i-1}}}}} \widehat{T_{i-1}}[x, \order^z, \widetilde{\weight}^z]. \\
	\end{multline*}
	Observe that the above equalities and congruencies still hold if we replace $P_i$ with $\widehat{P_i}$ and $T_{i-1}$ with $\widehat{T_{i-1}}$. The first congruency is true by \cref{lem:CDS-non-reduced-DP-a} for $T_i$ and by the definition of $\widehat{P_i}$ for $\widehat{T_i}$. The remaining congruencies hold since they only use the definition of the compatibility and not the table whose entries are summed up.
	Therefore, we obtain that 
	\[
		\sum\limits_{\substack{z \in \dom(P_i) \\ z \sim y}} P_i[z, \order, \weight] \equiv \sum\limits_{\substack{z \in \dom(P_i) \\ z \sim y}} \widehat{P_i}[z, \order, \weight]
	\]
	holds.
	So $\widehat{P_i}$ $(X_i, Y_i, Z_i)$-represents $P_i$.
\end{proof}	
	
\begin{lemma}\label{lem:reduced-dp-CDS-b}
		Let $2 \leq i \leq n$ and let a table $\widehat{P_i}$ $(X_i, Y_i, Z_i)$-represent $P_i$.
		Let $\widehat{Q_i^0} \equiv \widehat{P_i}$ and for all $0 \leq l \leq t_i-1$ and $z \in \dom(Q_i^{l+1})$, let
		\[
			\widehat{Q_i^{l+1}}[z, \order, \weight] \equiv 
			\begin{cases}
				\widehat{Q_i^l}[z, \order, \weight] & \mbox{if $z(w_i^{l+1}) \in \{L, R\}$,} \\
				\widehat{Q_i^l}\bigl[z[w_i^{l+1} \mapsto A], \order, \weight\bigr] - \widehat{Q_i^l}\bigl[z[w_i^{l+1} \mapsto F], \order, \weight\bigr] & \mbox{if $z(w_i^{l+1}) = D$.}
			\end{cases}
		\]
		Then for all $0 \leq l \leq t_i$, the table $\widehat{Q_i^{l+1}}$ $(X_i, Y_i, Z_i)$-represents $Q_i^{l+1}$.
\end{lemma}	
	
\begin{proof}
	We prove the statement by induction over $l$. Base case $l = 0$ is true by the previous item because of $\widehat{Q_i^0} \equiv \widehat{P_i}$. So suppose the statement holds for some $0 \leq l < t_i$, we will show that this also holds for $l+1$.
	Let $y \in \colors^{Y_i}$. For the simplicity of notation, we denote $v = w_i^{l+1}$. Then we have
	\begin{multline*}
		\sum\limits_{\substack{x \in \dom(Q_i^{l+1}) \\ x \sim y}} Q_i^{l+1}[x, K, W] \\
		\equiv \sum\limits_{\substack{x \in \dom(Q_i^{l+1}) \\ x_{|_{X_i}} \sim y}} Q_i^{l+1}[x, \order, \weight] \\
		= \sum\limits_{\substack{x \in \dom(Q_i^{l+1}) \\ x_{|_{X_i}} \sim y, x(v) = D}} Q_i^{l+1}[x, \order, \weight] + \sum\limits_{\substack{x \in \dom(Q_i^{l+1}) \\ x_{|_{X_i}} \sim y, x(v) \in \{L, R\}}} Q_i^{l+1}[x, \order, \weight] \\
		= \sum\limits_{\substack{x \in \dom(Q_i^{l+1}) \\ x_{|_{X_i}} \sim y, x(v) = D}} \Bigl(Q_i^l\bigl[x[v \mapsto A], K, W\bigr] - Q_i^l\bigl[x[v \mapsto F], \order, \weight\bigr]\Bigr)\\
		\hspace{.6\linewidth}+ \sum\limits_{\substack{x \in \dom(Q_i^{l+1}) \\ x_{|_{X_i}} \sim y, x(v) \in \{L, R\}}} Q_i^l[x, \order, \weight] \\
		= \sum\limits_{\substack{x \in \dom(Q_i^{l+1}) \\ x_{|_{X_i}} \sim y, x(v) = D}} Q_i^l\bigl[x[v \mapsto A], \order, \weight\bigr] - \sum\limits_{\substack{x \in \dom(Q_i^{l+1}) \\ x_{|_{X_i}} \sim y, x(v) = D}} Q_i^l\bigl[x[v \mapsto F], \order, \weight\bigr] \\
		\hspace{.6\linewidth} + \sum\limits_{\substack{x \in \dom(Q_i^{l+1}) \\ x_{|_{X_i}} \sim y, x(v) \in \{L, R\}}} Q_i^l[x, \order, \weight] \\
		\stackrel{v \notin X_i}{=} \sum\limits_{\substack{x \in \dom(Q_i^l) \\ x_{|_{X_i}} \sim y, x(v) = A}} Q_i^l[x, \order, \weight] - \sum\limits_{\substack{x \in \dom(Q_i^l) \\ x_{|_{X_i}} \sim y, x(v) = F}} Q_i^l[x, \order, \weight]\\
		\hspace{.6\linewidth} + \sum\limits_{\substack{x \in \dom(Q_i^l) \\ x_{|_{X_i}} \sim y, x(v) \in \{L, R\}}} Q_i^l[x, \order, \weight] \\
		\stackrel{\text{in } \FF_2}{\equiv} \sum\limits_{\substack{x \in \dom(Q_i^l) \\ x_{|_{X_i}} \sim y, x(v) = A}} Q_i^l[x, \order, \weight] + \sum\limits_{\substack{x \in \dom(Q_i^l) \\ x_{|_{X_i}} \sim y, x(v) = F}} Q_i^l[x, \order, \weight]\\
		\hspace{.6\linewidth} + \sum\limits_{\substack{x \in \dom(Q_i^l) \\ x_{|_{X_i}} \sim y, x(v) \in \{L, R\}}} Q_i^l[x, \order, \weight] \\
		= \sum\limits_{\substack{x \in \dom(Q_i^l) \\ x_{|_{X_i}} \sim y}} Q_i^l[x, K, W] = \sum\limits_{\substack{x \in \dom(Q_i^l) \\ x \sim y}} Q_i^l[x, \order, \weight]\\
		\stackrel{(X_i, Y_i, Z_i)-\text{repr.}}{\equiv} \sum\limits_{\substack{x \in \dom(Q_i^l) \\ x \sim y}} \widehat{Q_i^l}[x, \order, \weight].
	\end{multline*}
	Again, observe that all of the above equalities and congruencies again still hold if we replace $Q_i^l$ with $\widehat{Q_i^l}$ and $Q_i^{l+1}$ with $\widehat{Q_i^{l+1}}$. 
	Therefore, we obtain
	\[
		\sum\limits_{\substack{x \in \dom(Q_i^{l+1}) \\ x \sim y}} Q_i^{l+1}[x, \order, \weight] \equiv \sum\limits_{\substack{x \in \dom(Q_i^{l+1}) \\ x \sim y}} \widehat{Q_i^{l+1}}[x, \order, \weight]
	\]
	and this proves the claim for $l + 1$.
\end{proof}	

\begin{lemma}\label{lem:reduced-dp-CDS-c}
	Let $2 \leq i \leq n$ and let a table $\widehat{Q_i}$ $(X_i, Y_i, Z_i)$-represent $Q_i$.
	For all $x \in \dom(T_i)$, let
	\[
		\widehat{T_i}[x, \order, \weight] \equiv \sum\limits_{\substack{z \in \dom(Q_i) \\ z_{|_{X_i}} = x}} \widehat{Q_i}[z, \order, \weight]
	\]
	Then the table $\widehat{T_i}$ $(X_i, Y_i, X_i)$-represents $T_i$.
\end{lemma}	
	
\begin{proof}
	Let $y \in \colors^{Y_i}$. Then we have
	\begin{align*}
		&\sum\limits_{\substack{x \in \dom(T_i) \\ x \sim y}} T_i[x, \order, \weight] \\ 
		\equiv& \sum\limits_{\substack{x \in \dom(T_i) \\ x \sim y}} \sum\limits_{\substack{z \in \dom(Q_i) \\ z_{|_{X_i}} = x}} Q_i[z, \order, \weight] \\
		=& \sum\limits_{z \in \dom(Q_i)} \sum\limits_{\substack{x \in \dom(T_i) \\ x \sim y, z_{|_{X_i}} = x}} Q_i[z, \order, \weight] \\
		\stackrel{\dom(T_i) = X_i \subset Z_i = \dom(Q_i)}{=}& \sum\limits_{z \in \dom(Q_i)} Q_i[z, \order, \weight] 1_{z_{|_{X_i}} \sim y} \\
		=& \sum\limits_{\substack{z \in \dom(Q_i) \\ z_{|_{X_i}} \sim y}} Q_i[z, \order, \weight] \\
		=& \sum\limits_{\substack{z \in \dom(Q_i) \\ z \sim y}} Q_i[z, \order, \weight] \\
		\stackrel{(X_i, Y_i, Z_i)-\text{repr.}}{\equiv}& \sum\limits_{\substack{z \in \dom(Q_i) \\ z \sim y}} \widehat{Q_i}[z, \order, \weight]. \\
	\end{align*}
	Again, observe that all of the above equalities and congruencies still hold if we replace $Q_i$ with $\widehat{Q_i}$ and $T_i$ with $\widehat{T_i}$. This proves that $\widehat{T_i}$ $(X_i, Y_i, Z_i)$-represents $T_i$.
	Therefore, it holds that:
	\[
		\sum\limits_{\substack{x \in \dom(T_i) \\ x \sim y}} T_i[x, \order, \weight] \equiv \sum\limits_{\substack{x \in \dom(T_i) \\ x \sim y}} \widehat{T_i}[x, \order, \weight].
	\]
	This proves that $\widehat{T_i}$ $(X_i, Y_i, Z_i)$-represents $T_i$.
\end{proof}	

\begin{lemma}\label{lem:reduced-dp-CDS}
	Let $2 \leq i \leq n$ and let a table $\widehat{T_{i-1}}$ and a set $R_{i-1} \subseteq X_{i-1}$ satisfying the following properties be given: 
	\begin{enumerate}[(i)]
		\item the table $\widehat{T_{i-1}}$ $(X_{i-1}, Y_{i-1}, X_{i-1})$-represents $T_{i-1}$,
		\item the table $\widehat{T_{i-1}}$ is fully reduced, 
		\item $\widehat{T_{i-1}}^{-1}\bigl(\ZZ \setminus \{0\}\bigr) \subseteq C^{X_{i-1}} \times [i-1]_0 \times \bigl[(i-1)N\bigr]_0$ holds,
		\item and the vertices in $R_{i-1}$ are reduced in $\widehat{T_{i-1}}$.
	\end{enumerate}
	Then a table $\widehat{T_i}$ that $(X_i, Y_i, Z_i)$-represents $T_i$ along with the set $R_i \subset X_i$ of reduced vertices of $\widehat{T_i}$ can be computed in time $\mathcal{O}^*(3^{|R_{i-1}|} 4^{|X_{i-1}| - |R_{i-1}|})$ such that $\widehat{T_i}^{-1} \bigl(\ZZ \setminus \{0\}\bigr) \subseteq C^{X_i} \times [i]_0 \times [iN]_0$ and $|X_i \setminus R_i| \leq \bigl(\ctw - |R_i|\bigr) / 2 + 1$ hold.
\end{lemma}

\begin{proof}
	First, let the table $\widehat{P_i}$ be as in \cref{lem:reduced-dp-CDS-a}, then it $(X_i, Y_i, Z_i)$-represents $P_i$. We let $\hat{Q_i^0} = \widehat{P_i}$ so that $\hat{Q_i^0}$ $(X_i, Y_i, Z_i)$-represents $Q_i^0 = P_i$.
	Next, in the order of increasing $l = 1, \dots, t_i$, let $\widehat{Q_i^l}$ be as in \cref{lem:reduced-dp-CDS-b}, then it $(X_i, Y_i, Z_i)$-represents $Q_i^l$.
	Finally, let $\hat{T_i}$ be as in \cref{lem:reduced-dp-CDS-c}, then we know that it $(X_i, Y_i, Z_i)$-represents $T_i$. We will now show that this table has the desired properties and describe how it and a set $R_i$ can be obtained in the claimed running time.
	
	In each of the equalities in \cref{lem:reduced-dp-CDS-a,lem:reduced-dp-CDS-b,lem:reduced-dp-CDS-c}, it suffices to only sum over the support (i.e., non-zero entries) of the table on the right side of the equality. We will now bound the size of the support of each table in terms of $|R_{i-1}|$ and $|X_{i-1}|$. 
	Recall that we have $\widehat{T_{i-1}}[\cdot, \order, \weight] \equiv 0$ whenever $\order \notin \bigl\{0, \dots, (i-1)\bigr\}$ or $\weight \notin \bigl\{0, \dots, (i-1)N\bigr\}$ holds. Therefore we also have $\widehat{P_i}[\cdot, \order, \weight] \equiv 0$, $\widehat{Q_i^l}[\cdot, \order, \weight] \equiv 0$, and $\widehat{T_i}[\cdot, \order, \weight] \equiv 0$ whenever $\order \notin \{0, \dots, i\}$ or $\weight \notin \{0, \dots, iN\}$ holds and we do not need to store these values explicitly. 
	From now we assume that $\order \in \{0, \dots, i\}$ and $\weight \in \{0, \dots, iN\}$ holds. Note that iterating through all such $\order, \weight$ is polynomial in the size of the graph due to $N = 2|V|$. This part of the running time will be hidden behind $\mathcal{O}^*$ notation so we may forget about it.
	
	First, we consider the computation of $\widehat{P_i}$. 
	Recall that by \cref{lem:CDS-non-reduced-DP-a}, for every $z \in \dom(T_i)$ there is at most one coloring $x$ in the summation, namely $x = z_{|_{X_{i-1}}}$. Therefore $\widehat{P_i}$ can be computed by iterating through $x \in \supp(\widehat{T_{i-1}})$, $s \in a(v_i)$, and $\order, \weight$: check if $x \sim x[v_i \mapsto s]$ holds and if so, set $\widehat{P_i}\bigl[x[v_i \mapsto s], \order, \weight\bigr] \equiv \widehat{T_{i-1}}[z, \order, \weight]$. The non-initialized values of $\widehat{P_i}$ are interpreted as zero.	
	For every coloring $c \in \supp(\widehat{T_{i-1}})$ and every vertex $r \in R_{i-1}$ it holds that $c(r) \neq F$, therefore we have
	\[
		\bigl|\supp(T_{i-1})\bigr| \leq \Bigl|\bigl(\colors \setminus \{F\}\bigr)^{R_{i-1}} \times \colors^{X_{i-1} \setminus R_{i-1}} \Bigr| = 3^{|R_{i-1}|} 4^{|X_{i-1}| - |R_{i-1}|}.
	\]	
	So it is enough to only iterate through $x \in \bigl(\colors \setminus \{F\}\bigr)^{R_{i-1}} \times \colors^{X_{i-1} \setminus R_{i-1}}$ and it will contain the whole support of $\widehat{T_{i-1}}$. Hence, $\widehat{P_i}$ can be computed within the desired running time. Moreover, by the construction of $\widehat{P_i}$, all vertices in $R_{i-1}$ are still reduced in $\widehat{P_i}$.
	
	Similarly, for $l \in \{0, \dots, t_i\}$, we claim that we can compute the table $\widehat{Q_i^l}$ in time $\mathcal{O}^*(3^{|R_{i-1}|} 4^{|X_{i-1}| - |R_{i-1}|})$ and the set $R_{i-1} \setminus \{w_i^1, w_i^2, \dots, w_i^l\}$ is a set of reduced vertices of $\widehat{Q_i^l}$. Base case $l = 0$ holds due to $\widehat{Q_i^0} = \widehat{P_i}$. So suppose the claim is true for some $0 \leq l < t_i$. To compute $Q_i^{l+1}$, it suffices to iterate through $z' \in \supp\bigl(\widehat{Q_i^l}\bigr)$ and $K, W$: the corresponding table entry of $\widehat{Q_i^l}$ is then added to or subtracted from the corresponding entry of $\widehat{Q_i^{l+1}}$ depending on the color of $w_i^{l+1}$ in $z'$. 
	For shortness, let $S_i^l$  denote the set $\{w_i^j, w_i^2, \dots, w_i^l\}$.
	Recall that we have $\dom(Q_i^l) = \dom\bigl(\widehat{Q_i^l}\bigr) = \finalcolors^{S_i^l} \times \colors^{Z_i \setminus S_i^l}$.	
	By induction hypothesis, the vertices of $R_{i-1} \setminus S_i^l$ are reduced in $\widehat{Q_i^l}$ and hence 
	\[
		\Bigl|\supp\bigl(\widehat{Q_i^l}\bigr)\Bigr| \leq \Bigl|\colors^{Z_i \setminus (R_{i-1} \cup S_i^l)} \times \finalcolors^{S_i^l} \times \bigl(\colors \setminus \{F\}\bigr)^{R_{i-1} \setminus S_i^l}\Bigr| \leq 4^{|Z_i| - |R_{i-1}|} 3^{|R_{i-1}|}
	\]
	Recall that $Z_i = X_{i-1} \cup \{v_i\}$ and hence $|Z_i| \leq \ctw + 2$. So $\widehat{Q_i^{l+1}}$ can be computed within the claimed bound. Finally, observe that by the construction of $\widehat{Q_i^{l+1}}$, the vertices in $(R_{i-1} \setminus S_i^l) \setminus \{v_i^{l+1}\}$ are still reduced in $\widehat{Q_i^{l+1}}$. So the claim holds for $l+1$. 
	
	Finally, to compute $\widehat{T_i}$ we can iterate through $z \in \supp\bigl(\widehat{Q_i} = \widehat{Q_i^{t_i}}\bigr)$ and $\order, \weight$ and increase the entry $\widehat{T_i}[z_{|_{X_i}}, \order, \weight]$ by $\widehat{Q_i}[z, \order, \weight]$. 
	Since the vertices in $R_{i-1} \setminus S_i^{t_i}$ are reduced in $\widehat{Q_i}$, it holds that 
	\[
		\Bigl|\supp\bigl(\widehat{Q_i^{t_i}}\bigr)\Bigr| \leq \Bigl|\colors^{Z_i \setminus (R_{i-1} \cup S_i^{t_i})} \times \finalcolors^{S_i^{t_i}} \times \bigl(\colors \setminus \{F\}\bigr)^{R_{i-1} \setminus S_i^{t_i}}\Bigr| \leq 4^{|Z_i| - |R_{i-1}|} 3^{|R_{i-1}|}
	\]
	So $\widehat{T_i}$ can be computed within the claimed bound as well.
	
	Now we provide a set $R_i$ of reduced vertices of $\widehat{T_i}$. 
	Let again $A_i$, $B_i$, and $R_i$ be as in \cref{lem:sets-A-i-B-i-R-i}.
	We recall that these sets are defined as 
	$R_i = X_i \setminus \bigl(A_i \cup B_i \cup \{v_i\}\bigr)$ where
	\[
		A_i = \Bigl\{u \in X_i \setminus \{v_i\} \Bigm\vert \bigl| N(u) \cap Y_i \bigr| \geq 2 \Bigr\}
	\]
	and 
	\[
		B_i = \Bigl\{u \in X_i \setminus \{v_i\} \Bigm\vert \bigl| N(u) \cap Y_i \bigr| = 1 \text{ and } \{u, v_i\} \in E \Bigr\}.
	\]
	By that lemma, we know that
	\[
		|X_i \setminus R_i| \leq \bigl(\ctw - |R_i|\bigr) / 2 + 1
	\]
	holds.
 
	It remains to prove that the vertices in $R_i$ are indeed reduced in $\widehat{T_i}$. 
	Here we again follow the argument similar to \cite{GroenlandMNS22} back-propagated over the tables of our dynamic programming.
	Suppose there is a not reduced vertex $u \in R_i$, i.e., there exist numbers $K, W$, and a coloring $x \in \dom(T_i)$ such that $x(u) = F$ and $\widehat{T_i}[c, K, W] \not\equiv 0$ hold. 
	Then by the construction of $\widehat{T_i}$, there exists a coloring $z \in \dom(Q_i)$ such that $z_{|_{X_i}} = x$ and $\widehat{Q_i}[z, K, W] \equiv \widehat{Q_i^{t_i}}[z, K, W] \not\equiv 0$. In particular, we have $z(u) = F$.
	Similarly, by the construction of $\widehat{Q_i^{t_i}}$, there is a coloring $z'$ such that $z'_{|_{Z_i \setminus \{w_i^{t_i}\}}} = z_{|_{Z_i \setminus \{w_i^{t_i}\}}}$ and $(Q_i^{t_i - 1})[z', K, W] \not\equiv 0$. Note that since $R_i \subset X_i$, it holds that $\{w_i^1, \dots, w_i^{t_i}\} \cap R_i = \emptyset$ so $z'(u) = z(u) = F$.
	So we can back-propagate this observation to $\widehat{Q_i^{t_i-1}}, \dots, \widehat{Q_i^1}, \widehat{Q_i^0} = \widehat{P_i}$ to obtain a coloring $z^* \in \dom(Q_i^0) = \dom(P_i)$ such that $z^*(u) = F$ and $\widehat{P_i}[z^*, K, W] \not \equiv 0$. 
	This, in turn, implies that we have $z^*_{|_{X_{i-1}}} \sim z^*$ and $\widehat{T_{i-1}}[z^*_{|_{X_{i-1}}}, \order^{z^*}, \widetilde{\weight}^{z^*}] \not\equiv 0$ and it holds that $z^*_{|_{X_{i-1}}}(u) = F$. 
	We now prove that $u$ has exactly one neighbor in $Y_{i-1}$ contradicting the fact that $\widehat{T_{i-1}}$ is fully reduced.  
	Since $u \notin A_i \cup B_i$, the vertex $u$ has at most one incident edge crossing the $(i-1)$th cut. On the other hand, we have $u \in X_{i-1}$ so $u$ has at least one incident edge going across the $(i-1)$th cut. Thus, the vertex $u$ indeed has exactly one neighbor in $Y_{i-1}$ so $\widehat{T_{i-1}}$ is not fully reduced -- a contradiction. So $R_i$ is a set of reduced vertices of $\widehat{T_i}$. Clearly, $R_i$ can be computed in polynomial time.
\end{proof}

Now we are able to provide an algorithm solving \pcds\ in $\ostar(3^{\ctw})$.

\begin{theorem}
	There exists an algorithm that given a graph $G = (V, E)$ and a linear arrangement $\ell$ of $G$ of cutwidth $\ctw$ in time $\ostar(3^{\ctw})$ solves the \Pcds\ problem. The algorithm cannot give false positives and may give false negatives with probability at most~$1/2$.
\end{theorem}

\begin{proof} 
	By \cref{thm:cut-and-count-short}, it suffices to provide an algorithm that additionally to a graph and its linear arrangements gets a weight function $\omega: V \mapsto \bigl[2|V|\bigr]$ and numbers $\totalorder \in [n]_0, \totalweight \in \bigl[2n|V|\bigr]$ as input and computes the size (modulo 2) of the set $\DDD_{\totalweight}^{\totalorder}$.
	
	Observe that there is a bijection between the elements of $\DDD_\totalweight^\totalorder$ and the colorings $c: V \to \{L, R, D\}$ such that every vertex in $c^{-1}(D)$ has at least one neighbor in $c^{-1}\bigl(\{L, R\}\bigr)$ and there is no edge between $c^{-1}(L)$ and $c^{-1}(R)$. Now we provide an algorithm counting the number of such colorings modulo 2 and argue its correctness. So we will count such colorings.
	
	Recall that $X_1 = \{v_1\}$.
	We initialize the table $T_1$ via brute-force: iterate over all $s \in \colors$, all $\order \in \{0, 1\}$, and all $\weight \in \bigl\{0, \omega(v_1)\bigr\}$ and compute the corresponding entry $T_1\bigl[[v_1 \mapsto s], \order, \weight\bigr]$. All other entries of $T_1$ are implicitly set to $0$ and there is no need to store them. If $v_1$ has exactly one neighbor in $G$, apply the \textbf{Reduce} algorithm from \cref{lem:degree-1-vts-CDS} to $X = X_1$, $Y = Y_1$, and $T_1$, or $T_1$ remains unchanged otherwise. As a result, obtain a fully reduced function $\widehat{T_1}$ that $(X_1, Y_1, X_1)$-represents $T_1$ with some set of reduced vertices $R_1$. For $i = 2, \dots, n$, we repeat the following two steps.
	\begin{enumerate}
		\item Apply \cref{lem:reduced-dp-CDS} with inputs $(\widehat{T_{i-1}}, R_{i-1})$ in order to obtain a table $\widehat{T_i}$ that $(X_i, Y_i, X_i)$-represents $T_i$ and a set of reduced vertices $R_i$ for $\widehat{T_i}$.
		\item While $X_i \setminus R_i$ has a vertex with exactly one neighbor in $Y_i$, apply the \textbf{Reduce} algorithm from \cref{lem:degree-1-vts-CDS} to $i$ and $(\widehat{T_i}, R_i)$ to obtain a new table $\widehat{\widehat{T_i}}$ that $(X_i, Y_i, X_i)$-represents $\widehat{T_i}$ and whose reduced vertices are $R_i \cup \{v\}$ and set $\widehat{T_i} := \widehat{\widehat{T_i}}$ and $R_i := R_i \cup \{v\}$.
	\end{enumerate}

	At the end of step 2, we obtain a fully reduced table $\widehat{T_i}$ that $(X_i, Y_i, X_i)$-represents $T_i$. Moreover, the set $R_i$ of reduced vertices has only increased in size compared to the set we obtained in step 1. Hence, there are at most $|X_i|$ repetitions of step 2. 

	As a result of this process, we have computed the entries of the table $\widehat{T_n}$ that is fully reduced and $(X_n, Y_n, X_n)$-represents $T_n$ with a set $R_n$ of reduced vertices of $\widehat{T_{n}}$. 
	
	Finally, we output the value 
	\[
		\sum\limits_{x \in \dom(T_n)} \widehat{T_n}[x, \totalorder, \totalweight].
	\]
	Recall that $X_n = \{v_n\}$ and $Y_n = \emptyset$. 
	Let $y^*$ be the unique (i.e., empty) coloring of $Y_n$. Then every coloring of $X_n$ with an intermediate color is compatible with $y^*$. Therefore, it holds that
	\begin{multline*}
		 \sum\limits_{x \in \dom(T_n)}\widehat{T_n}[x, \totalorder, \totalweight] \\
		= \sum\limits_{\substack{x \in \dom(T_n) \\ x \sim y^*}} T_n[x, \totalorder, \totalweight] \\
		\stackrel{(X_n, Y_n, X_n)-\text{repr.}}{\equiv} \sum\limits_{\substack{x \in \dom(T_n) \\ x \sim y^*}} T_n[x, \totalorder, \totalweight] \\
		= \sum\limits_{x \in \dom(T_n)} T_n[x, \totalorder, \totalweight] \\
		= T_n\bigl[[v_n \mapsto L], \totalorder, \totalweight\bigr] + T_n\bigl[[v_n \mapsto R], \totalorder, \totalweight\bigr] + T_n\bigl[[v_n \mapsto A], \totalorder, \totalweight\bigr] \\
		 + T_n\bigl[[v_n \mapsto F], \totalorder, \totalweight\bigr] \\
		\stackrel{\text{in }\FF_2}{\equiv} T_n\bigl[[v_n \mapsto L], \totalorder, \totalweight\bigr] + T_n\bigl[[v_n \mapsto R], \totalorder, \totalweight\bigr] + \Bigl(T_n\bigl[[v_n \mapsto A], \totalorder, \totalweight\bigr] \\ 
		- T_n\bigl[[v_n \mapsto F], \totalorder, \totalweight\bigr]\Bigr).
	\end{multline*}
	We claim that the last value is the size of $\DDD_{\totalweight}^{\totalorder}$ as desired. 
	Let 
	\[
		\DDD_L = \Bigl\{\bigl(S, (L, R)\bigr) \in \DDD_{\totalweight}^{\totalorder} \Bigm\vert v_n \in L\Bigr\} \text{ and } \DDD_R = \Bigl\{\bigl(S, (L, R)\bigr) \in \DDD_{\totalweight}^{\totalorder} \mid v_n \in R\Bigr\}
	\]
	and let 
	\[
		\DDD_D = \Bigl\{\bigl(S, (L, R)\bigr) \in \DDD_{\totalweight}^{\totalorder} \Bigm\vert v_n \notin S\Bigr\}.
	\]
	Then $\DDD_L, \DDD_R, \DDD_D$ is a partition of $\DDD_{\totalweight}^{\totalorder}$.
	The value $T_n\bigl[[v_n \mapsto L], \totalorder, \totalweight\bigr]$ (resp.\ $T_n\bigl[[v_n \mapsto R], \totalorder, \totalweight\bigr]$) is the number of valid extensions of $[v_n \mapsto L]$ (resp.\ $[v_n \mapsto R]$) of order $\totalorder$ and weight $\totalweight$. Recall that in any valid extension $\phi \in \allcolors^{V_n = V}$ of a coloring $c \in \colors^{X_n = \{v_n\}}$, all vertices of $V_n \setminus \{v_n\} = V \setminus \{v_n\}$ are colored with $L$, $R$, or $D$. So by the definitions of a valid extension and a consistent cut, we obtain 
	\[
		|\DDD_L| = T_n\bigl[[v_n \mapsto L], \totalorder, \totalweight\bigr] \text{ and } |\DDD_R| = T_n\bigl[[v_n \mapsto R], \totalorder, \totalweight\bigr].
	\]	
	Similarly, by definition, the entry $T_n\bigl[[v_n \mapsto A], \totalorder, \totalweight\bigr]$ (resp.\ $T_n\bigl[[v_n \mapsto F], \totalorder, \totalweight\bigr]$) is the number of pairs $\bigl(S, (L, R)\bigr)$ such that $S$ is a dominating set of $G[V_{n-1}]$ of order $\totalorder$ and weight $\totalweight$, $(L, R)$ is a partition of $S$, and for the vertex $v_n$ it holds that $v_n \notin S$ and $v_n$ is not necessarily (but possibly) dominated by $S$ (resp.\ $v_n$ is not dominated by $S$). Therefore, we obtain 
	\[
		|\DDD_D| = T_n\bigl[[v_n \mapsto A], \totalorder, \totalweight\bigr] - T_n\bigl[[v_n \mapsto F], \totalorder, \totalweight\bigr].
	\]
	Altogether, it holds that
	\[
		\sum\limits_{x \in \dom(T_n)} \widehat{T_n}[x, \totalorder, \totalweight] \equiv |\DDD_L| + |\DDD_R| + |\DDD_D| = |\DDD_{\totalweight}^{\totalorder}|
	\]
	and the output value is correct.
	
	It remains to prove the claimed bound on the running time.
	By \cref{lem:reduced-dp-CDS} and \cref{lem:degree-1-vts-CDS}, the running time is bounded by 
	\[
		\sum\limits_{i=1}^{n} \mathcal{O}^*\bigl(|X_i| 3^{|R_i|} 4^{|X_i| - |R_i|}\bigr) = \sum\limits_{i=1}^{n} \mathcal{O}^*(3^{|R_i|} 4^{|X_i| - |R_i|}).
	\]
	Recall that by \cref{lem:reduced-dp-CDS}, it holds that $|X_i \setminus R_i| \leq (\ctw - |R_i|) / 2 + 1$ for any $i \in [n]$ so 
	\[
		3^{|R_i|} 4^{|X_i| - |R_i|} \leq 3^{|R_i|} 4^{\bigl(\ctw - |R_i|\bigr) / 2 + 1} =
		4 \cdot 3^{|R_i|} 2^{\ctw - |R_i|} \leq 4 \cdot 3^{|R_i|} 3^{\ctw - |R_i|} = 4 \cdot 3^{\ctw}.
	\]
	So the claimed bound on the running time holds.
\end{proof}

\section{Upper Bounds via Edge-Subdivision}\label{app:ub-edge-subdiv}
In this section, we solve the \Poct{} and \Pfvs{} problems in $\mathcal{O}(2^\ctw)$. 

Later, in \cref{app:lb}, we show that this bound is essentially tight under SETH in both cases.
For these problems, we employ edge-subdivision. 
Each edge is subdivided twice. 
We show that this operation does not increase cutwidth. 
Moreover, there exists a (natural) linear arrangement such that every cut contains only a constant number of vertices of the original graph. 
The key trick is reduce number of states of the dynamic programming for subdivision vertices so that the total number of states of a cut of a linear arrangement is small.
For the simplicity of the writeup, we reduce this problem to well-known algorithms for \Poct{} and \Pfvs{} parameterized by pathwidth where we additionally employ some special property of the constructed path decomposition.
 
\subsection{\Oct}\label{app:oct} 
In this section we approach the \Poct{} problem. 

\begin{quote}
	\Poct{}
	
	\textbf{Input}: A graph $G = (V, E)$ and an integer $k$.
	
	\textbf{Question}: Is there a subset $S \subseteq V$ of size $k$ such that $G - S$ is bipartite.
\end{quote}

We assume that a graph $G$ is provided with a linear arrangement $\ell$ of cutwidth $\ctw$.
In order to solve it in time $\ostar(2^{\ctw})$, we first subdivide each edge of the input graph $G$ twice and obtain a graph $\hat{G}$. 
The graph $\hat G$ has at most the same \ctwdth{} as $G$ since \ctwdth{}
is closed under edge subdivision: To see this, consider an optimal linear arrangement of $G$ and for every edge~$\{u, v\}$ of $G$, insert the subdivision-vertices of $\{u, v\}$ between $u$ and $v$ in such a way that the edges subdividing~$\{u, v\}$ do not overlap. 

Let $\ctw$ denote the cutwidth of $G$.
We show that $\hat{G}$ admits a linear arrangement of \ctwdth{} at most $\ctw$ such that each cut contains at most one vertex of $G$. We also show that this implies the existence of a path decomposition of $\hat G$ of pathwidth at most $\ctw$ such that each bag contains at most one vertex of $G$. 
We will show that solving \Poct{} on $G$ can be reduced to solving a similar problem \Poctp{} (defined later) on $\hat{G}$. 
We then use the derived path decomposition to provide a dynamic programming that solves this equivalent problem. 
First, we repeat the definitions of a path decomposition and a nice path decomposition.

\begin{definition}\label{append:pathdecomp}
    Let $G = (V, E)$ be an undirected graph. A path decomposition of $G$ is a sequence $\mathcal{B} = B_1, \dots, B_r$ such that:
    \begin{itemize}
        \item It holds that $B_1 \cup \dots \cup B_r = V$.
        \item For each $\{u, v\}\in E$, there exists an index $i \in [r]$ such that $u, v \in B_i$.
        \item For every $i < j \in [r]$, the property $v \in B_i \cap B_j$ implies $v \in B_k$ for every $i \leq k \leq j$.
    \end{itemize}
    The sets $B_1, \dots, B_r$ are called \emph{bags}.
    For $i \in [r]_0$, we denote with $V_i = \cup_{j\in\{1, \dots i\}}B_j$ the union of the first $i$ bags and with $G_i = G[V_i]$ the subgraph induced by these vertices, we also denote the set of edges of this subgraph with $E_i$. 
    The \emph{pathwidth} of $\mathcal{B}$ is defined as $\pw = \max\limits_{i \in [r]}|B_i| - 1$.
    The pathwidth of a graph $G$ is the smallest $\pw(\mathcal{B})$ over all path decompositions $\mathcal{B}$ of $G$.
\end{definition}

\begin{definition}
	A path decomposition $\mathcal{B} = B_1, \dots, B_r$ of a graph $G$ is \emph{nice} if $B_1 = B_r = \emptyset$ holds, and for every $i \in [r-1]$, we have $|B_i \triangle B_{i+1}| = 1$ where $\triangle$ denotes the symmetric difference.  
	In this case, for $2 \leq i \leq r$, we call~$B_i$ an \emph{introduce} bag if $B_{i-1} \subseteq B_i$ holds and we call it a \emph{forget} bag otherwise.
\end{definition}

We state the following well-known lemma without proof. 
We refer to \cite{CyganFKLMPPS15} for more details.
\begin{lemma}[Folklore]\label{lem:appned-nice-decomps-exists}
	Given a path decomposition $\mathcal{B} = B_1, \dots, B_r$ of a graph $G = (V, E)$, in polynomial time a nice path decomposition~$\mathcal{B}'$ of $G$ of the same width can be computed such that the following properties hold.
	Each bag of $\mathcal{B}'$ is a subset of some bag of $\mathcal{B}$, and the number of bags of $\mathcal{B'}$ is polynomial in $r + |V|$. 
\end{lemma}

Let $G= (V, E)$ be a graph. With $\hat{G} = (\hat{V}, \hat{E})$ we denote the graph arising from $G$ after subdividing each edge twice. 
Let $V_0$ be the set of all subdivision vertices, i.e.\ $V_ 0 = \hat{V} \setminus V$.
For an edge $e$ of $\hat{G}$, let $o(e)$ denote the edge of $G$ subdivided by $e$.
The graph $\hat G$ results from $G$ by replacing each edge of $G$ by a simple path $P_4$ on four vertices. 
For a vertex $v \in V_0$, let $P(v)$ denote the unique path $P_4$ in $\hat{G}$ that contains $v$ and whose both end-vertices belong to $V$. We also denote the unique edge in $G$ subdivided by $v$ with $o(v)$.
We also denote the number of vertices (resp.\ edges) of $\hat G$ with $\hat n$ (resp.\ $\hat m$).
\begin{lemma} \label{lem:append-subdiv-arrangement} 
	Let $G = (V, E)$ be a graph and let $\ell$ be a linear arrangement of $G$ of \ctwdth{} at most $\ctw$. 
	Then $\hat G$ admits a linear arrangement $\hat{\ell}$ of width at most $\ctw$. 
	Moreover, $\hat G$ and $\hat \ell$ can be computed from $G$ and $\ell$ in polynomial time.
\end{lemma}
\begin{proof}
	Let $v_1, \dots, v_n$ be the order in which the vertices of $G$ appear in $\ell$.
	We describe the construction of $\hat \ell$.
	For each vertex $v \in V$, let $S_v$ be the set of vertices incident to $v$ in $\hat G$. Let $S^1_v$ be the set of vertices $w \in S_v$ such that the end-vertex $u \neq v$ of $P(w)$ 
	lies before $v$ on $\ell$. Let $S^2_v = S_v \setminus S^1_v$. We form $\hat \ell$ from $\ell$ by inserting the vertices in $S^1_v$ directly before $v$ into $\ell$ and the vertices in $S^2_v$ directly after $v$ into $\ell$. 
	The relative order of vertices in $S^1_v$ (resp.\ $S^2_v$) is irrelevant so we just fix an arbitrary one. Clearly, $\hat G$ and $\ell$ can be constructed in polynomial time.
	
	By construction, for each vertex $w \in V_0$, the edges of $P(w)$ do not overlap in $\ell$ (see \cref{fig:subdiv} for illustration).
	Let $w_1, \dots, w_{\hat n}$ denote the order of vertices in $\ell$.
	Let $j \in [\hat n]$ be arbitrary, we are now interested in the cut graph at $w_j$. 
	Let $i$ be the largest index with $i \leq j$ such that $w_i \in V$ holds. Then the cut at $w_j$ in $\hat \ell$ results from the cut at $w_i$ in $\ell$ by replacing each edge $e$ by one of the three edges subdividing this edge (see \cref{fig:subdiv}). Hence, both cuts have the same size. Since $j$ was chosen arbitrarily, we obtain that $\ell$ and $\hat \ell$ have the same cutwidth. 
\end{proof}
\begin{figure}
	\center
	\includegraphics[width = .5\linewidth]{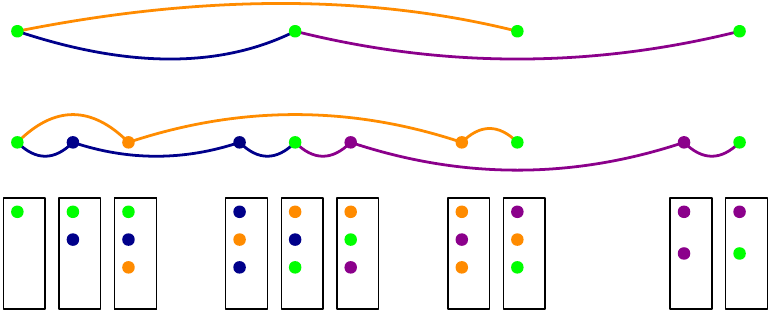}
	\caption{A linear arrangement $\ell$ of a graph $G$ (top) and the linear arrangement $\hat{\ell}$ of the graph $\hat G$ (middle). At the bottom we depict a path decomposition of $\hat{G}$ of width at most $\ctw(\ell)$ as in \cref{lem:append-linear-pathdecomp}. The subdivision vertices and edges are depicted in the same color as the edge of $G$ they subdivide. Note that each bag of the path decomposition contains at most one
		green vertex as proven by \cref{lem:append-linear-pathdecomp-bags}.
	}
	\label{fig:subdiv}
\end{figure}

\begin{lemma}\label{lem:append-linear-pathdecomp} 
	Let $G = (V, E)$ be a graph with at least one edge and let $\ell$ be a linear arrangement of $G$ of width at most $\ctw$. 
	Then $G$ admits a path decomposition $\mathcal{B}$ of width at most $\ctw$ that can be computed from $G$ and $\ell$ in polynomial time.
\end{lemma}
\begin{proof}
	Let $v_1, \dots, v_n$ denote the order of vertices in $\ell$.
	We define a sequence $\mathcal{B} = B_1, \dots B_n$ by setting $B_1 = \{v_1\}$ and 
	\[
	B_i = \{v_i\}\cup\bigl\{v_j \bigm\vert j \leq i - 1, \exists k \geq i \colon \{v_j, v_k\} \in E\bigr\}
	\]
	for every $2 \leq i \leq n$. Clearly, the sequence $\mathcal{B}$ can be computed in polynomial time.
	We show that $\mathcal{B}$ is a path decomposition of pathwidth at most $k$.
	
	First, we show that $\mathcal{B}$ is a path decomposition. 
	Each vertex $v_i$ is contained in the bag $B_i$. Next, let $\{v_i, v_j\}$ be an edge of $G$ and without loss of generality assume that $i < j$ holds. Then both $v_i$ and $v_j$ are contained in $B_j$. 
	Let $i \in [n]$ be arbitrary but fixed. Observe that 
	\[
		i = \min \bigl\{j \in [n] \bigm\vert v_i \in B_j\bigr\}
	\]
	holds.
	Let $j$ be the largest index such that $v_i \in B_j$. We show that for every $i \leq k \leq j$, we also have $v_i \in B_k$. If $v_i$ is an isolated vertex, then we have $i = j$ and the claim is trivially true. Otherwise, $j$ is the largest index such that $\{v_i, v_j\}$ is an edge of $G$. 
	But then for every $i < k < j$, we have that the edge $\{v_i, v_j\}$ has the properties $i \leq k-1$ and $j \geq k$ so by definition, the vertex $v_i$ also belongs to $B_k$. It is clear Hence, $\mathcal{B}$ is a path decomposition.
	
	Observe that for $2 \leq i \leq n$, the set $\bigl\{v_j \bigm\vert j \leq i - 1, \exists k \geq i \colon \{v_j, v_k\} \in E_i\bigr\}$ has the size of at most $\ctw$ since these vertices are exactly the left-ends of edges crossing the $(i-1)$th cut.
	Therefore, the size of $B_i$ is at most $\ctw + 1$. Since the graph contains at least one edge, its cutwidth is at least one and this covers the case of the bag $B_1$ as well.
	So the pathwidth~of $\mathcal{B}$ is upper bounded by $\ctw$.
\end{proof}
\begin{lemma}\label{lem:append-linear-pathdecomp-bags}
	Let $G = (V, E)$ be a graph and let $\ell$ be a linear arrangement of $G$.
	Let $\hat G$ be the graph resulting from $G$ by subdividing each edge twice and let $\hat{\ell}$ be the linear arrangement of $\hat G$ satisfying the properties from \cref{lem:append-subdiv-arrangement}. Finally, let $\hat {\mathcal{B}}$ be the path decomposition of $\hat G$ constructed from $\hat \ell$ as in the proof of \cref{lem:append-linear-pathdecomp}. 
	Then each bag of $\hat{\mathcal{B}}$ contains at most one vertex of $V$.
\end{lemma}
\begin{proof}
	Let $w_1, \dots, w_{\hat n}$ be the ordering of the vertices in $\hat \ell$ and $v_1, \dots, v_n$ be the ordering of the vertices in $\ell$. 
	Let $B_1, \dots, B_{\hat{n}}$ be the bags of $\mathcal{B}$ as defined in the proof of \cref{lem:append-linear-pathdecomp}. 
	Recall that the bag $B_i$ contains the vertex $w_i$, and all vertices $w_{i'}$, such that $i' < i$ and there exists $i'' \geq i$ such that $\{w_{i'}, w_{i''}\} \in \hat E$. 
	Also recall that we constructed $\hat{\ell}$ from $\ell$ by adding all vertices in $S^1_v$ directly before $v$, and all vertices in $S^2_v$ directly after $v$ for each vertex $v \in V$. 
	In particular, the ordering of the vertices of $V$ agree in both $\ell$ and $\hat{\ell}$.
	Lastly, observe that for every vertex $v \in V$, it holds that $N_{\hat{G}}[v] = S^1_v \cup S^2_v \cup \{v\}$.

	For $j \in [n]$, the set $L_j$ (also called the \emph{$i$th segment}) consists of the vertices $\{v_j\} \cup S_{v_j}$. 
	By the construction of $\ell$, for every $j \in [n]$, the set $L_i$ consists of consecutive vertices on $\ell$. 
	The segments are pairwise disjoint and they partition the set $\hat V$. 
	
	Now consider an arbitrary $j \in [n]$. 
	We are interested in the set of bags in which the vertex~$v_j \in V$ might occur.
	Let $i \in [\hat n]$ be such that $v_j = w_i$.
	Consider an arbitrary $i' \in [\hat n]$ such that the bag $B_{i'}$ contains $v_j$, i.e., $v_j = w_i \in B_{i'}$.
	By the definition of $B_{i'}$, it can occur in one of the following two cases. 
	First, if $v_j = w_i = w_{i'}$, then in particular, we have $w_{i'} = v_j \in L_j$.
	Otherwise, we have $i < i'$ and there exists some $i'' \geq i'$ with $\{w_{i}, w_{i''}\} \in \hat E$. 
	Therefore, it holds that $w_{i''} \in N_{\hat G}[w_i] = N_{\hat G}[v_j] = L_j$. 
	Recall that the vertex $w_{i'}$ lies between~$w_{i}$ and~$w_{i''}$ in~$\ell$ (i.e.,~$i < i' \leq i''$).
	Since both $w_{i} = v_j$ and $w_{i''}$ belong to the consecutive segment~$L_j$ of $\ell$, the vertex $w_{i'}$ also does.
	Altogether, in both cases, we obtain $w_{i'} \in L_j$. 
	Hence, a vertex $v_j$ can only occur in bags $B_{i'}$ such that $w_{i'} \in L_j$. 
	Since $L_1, \dots, L_n$ is a partition of $\hat V$, at most one vertex from $V$ occurs in a bag $B_{i'}$ for every $i' \in [\hat n]$.
\end{proof}

From now on, we always assume that the path decomposition $\mathcal{\hat B}$ of $\hat G$ satisfies \cref{lem:append-linear-pathdecomp-bags}.
\begin{remark}\label{rem:append-linear-pathdec-nice}
	Note that if we apply the standard construction (see \cite{CyganFKLMPPS15} for more details) transforming $\mathcal{\hat B}$ into a nice path decomposition of the same width, then the arising path decomposition still satisfies the property of \cref{lem:append-linear-pathdecomp-bags}. Moreover, this path decomposition contains $\mathcal{O}\Bigl(\pw\bigl(\mathcal{\hat B}\bigr)\poly(n)\Bigr) = \poly(n)$ bags.
\end{remark}

We summarize the above statements in a form of one lemma.

\begin{lemma}\label{lem:special-path-decomposition}
	Let $G = (V, E)$ be a graph and let $\ell$ be a linear arrangement of $G$ of cutwidth $\ctw$. Then the graph $\hat G$ admits a nice path decomposition $\mathcal{\hat B}$ of pathwidth at most $\ctw$ such that every bag of $\mathcal{\hat B}$ contains at most one vertex of $V$. Moreover, the path decomposition $\mathcal{\hat B}$ contains a polynomial number of bags and it can be computed in polynomial time from $G$ and~$\ell$.
\end{lemma}

Now we show how to use $\hat G$ and $\mathcal{B}$ to solve the \Poct{} problem on $G$ more efficiently. We start by showing a bijection between the cycles of both graphs.
\begin{lemma}
	\label{lem:subdivprescycles} Let $G$ be a graph. 
	Each cycle in $\hat{G}$ results from a cycle in $G$ by subdividing each edge of this cycle twice and vice versa. This defines a bijection between the cycles of both graphs.
\end{lemma}
\begin{proof}
	A cycle in $\hat{G}$ containing a vertex $w \in V_0$ must contain the whole path $P(w)$ as a subpath, since $w$ is a degree-two vertex.
	Therefore, we can replace every such path $P(w)$ by the edge $o(w)$ and obtain a cycle in $G$.
	For the other direction, given a cycle in $G$, we can replace every edge by its subdivision and obtain a cycle in $\hat G$. Both mappings are injective.
\end{proof}

\begin{corollary}
	\label{cor:subdivoddcycles}
	There is a bijection between odd cycles in a graph $G$ and odd cycles in the graph $\hat{G}$.
\end{corollary}
\begin{proof}
	This follows from \cref{lem:subdivprescycles} by noticing that subdividing an edge of a cycle twice increases its length by two and therefore, does not change the parity.
\end{proof}

\begin{lemma}
	\label{lem:app-oct-subdiv}
	 Let $G = (V, E)$ be a graph. 
	 Then an odd cycle transversal of $G$ is an odd cycle transversal of $\hat G$. Also an odd cycle transversal of $\hat G$
	 that only uses vertices from $V$ is an odd cycle transversal of $G$ as well. Finally, the graph $\hat{G}$ admits a
	 minimum odd cycle transversal that only contains vertices from $V$. In particular, both graphs have the same odd-cycle-transversal number. 
\end{lemma}
\begin{proof}
	Let $S$ be an \oct{} of $G$. Assume that $\hat{H} = \hat{G} - S$ is not bipartite. Let~$C$ by an odd cycle of $\hat{H}$, then it is an odd cycle of $\hat{G}$ that avoids $S$. By \cref{cor:subdivoddcycles}, this
	cycle results from an odd cycle $C_0$ in $G$ by subdividing each of its edges twice. Hence, $C_0$ is disjoint from
	$S$ as well -- a contradiction. 
	
	Now let $\hat{S}$ be an odd cycle transversal $\hat{S}$ of $\hat{G}$ that only uses
	vertices in $V$. Then $\hat S$ is an odd cycle transversal of $G$ as well, since we could turn an odd cycle $G - \hat S$ into an odd cycle of $\hat{G} - \hat S$ by subdividing each of its edges twice. 
	
	Finally, consider a minimum odd cycle transversal $S$ of $\hat G$.
	If it does not contain any vertex of $V_0$, we are done. Otherwise, consider a vertex $w \in S \cap V_0$ and let $v \in V$ be the unique neighbor of $w$ in $V$. 
	Since $w$ has the degree of two in $\hat G$, every cycle containing $w$ also contains $v$. Therefore, the set $\bigl(S \setminus \{w\}\bigr) \cup \{v\}$ is an odd cycle transversal of $\hat G$ of at most the same size as $S$. We can repeat this process until we obtain an odd cycle transversal $S'$ of $\hat G$ that has at most the same size as $S$ (and hence, it is also a minimum odd cycle transversal) and contains no vertices from $V_0$ as claimed.
\end{proof}

We define the \Poctp{} problem as follows:

\begin{quote}
	\Poctp{}

	\textbf{Input}: An undirected graph $G = (V, E)$, subsets $Y, Y' \subseteq V$, and an integer $k$.

	\textbf{Question}: Is there an odd cycle transversal $S$ of $G$ of size at most $k$ such that $Y \subseteq S$ and $Y' \cap S = \emptyset$.
	
\end{quote}

\begin{lemma} \label{lem:oct-alg}
	Given a graph $G = (V, E)$, two disjoint sets $Y, Y' \subseteq V$, a nice path decomposition $\mathcal{B} = B_1, \dots, B_r$ of $G$ of width at most $\pw$, and a positive integer $k$, the \Poctp{} problem on $(G, Y, Y', k)$ can be solved in time $\ostar(2^{\pw}3^{h} r)$ where 
	\[
		h = \max\limits_{i \in [r]} |B_i \setminus Y'|
	\]
	is the maximum number of vertices in a bag that do not belong to $Y'$.
\end{lemma}
\begin{proof}
	We rely on a well-known dynamic programming scheme for \Poct{} and adapt it for the \Poctp{} problem. 
	Recall that a vertex set $S$ is an \oct{} of a graph $G$ if and only if $G - S$ admits a proper 2-coloring. 
	Therefore, the natural choice of states of the vertices in a bag for the dynamic programming is $\mathcal{C} = \{0, 1, 2\}$ where $0$ represents being in an \oct{} and $1$ and $2$ are colors in a proper 2-coloring.
	For $i \in [r]$, the table $T_i$ is indexed by assignments $\delta \colon B_i \to \mathcal{C}$ of states to vertices in $B_i$ and integers $K \in [k]_0$. The value $T_i[\delta, K]$ is the number of assignments $\sigma \colon V_i \to \mathcal{C}$ with the following properties:
	\begin{enumerate}
		\item $\sigma_{|_{_{B_i}}} = \delta$, i.e., $\sigma$ extends $\delta$,
		\item $\bigl|\sigma^{-1}(0)\bigr| = K$,
		\item $\sigma_{|_{V_i\setminus \sigma^{-1}(0)}}$ is a proper coloring of $G_i - \sigma^{-1}(0)$,
		\item $\sigma(v) = 0$ for all $v \in V_i \cap Y$,\label{it:octmustdelete}
		\item $\sigma(v) \neq 0$ for all $v \in V_i \cap Y'$.\label{it:octforbidden}
	\end{enumerate} 
	In other words, $T_i[\delta, K]$ is the number of partitions $(X, A, B)$ of $V_i$ such that $X$ is an \oct{} of $G_i$ of size $K$ with $V_i \cap Y \subseteq X$ and $Y \cap X = \emptyset$, $(A, B)$ is an (ordered) bipartition of $G_i - X$, and $(X, A, B)$ agrees with $\sigma$ on the assignment of $B_i$.
	Recall that since $\mathcal{B}$ is nice, we have $B_r = \emptyset$. Therefore, there exists exactly one assignment~$\sigma_r \colon B_r \to \mathcal{C}$, namely the empty assignment.
	For $K \in [k]_0$, the value $T_r[\sigma_r, K]$ is then exactly the number of tuples $(X, A, B)$ where $X$ is an odd cycle transversal of $G_r = G$ of size $K$ with $Y \subseteq X$ and $Y' \cap X = \emptyset$ and $(A, B)$ is a bipartition of $G_n - X = G - X$. Therefore, the tuple $(G, Y, Y', k)$ is a yes-instance of \Poctp{} if and only if $\sum_{0 \leq K \leq k}T_r[\sigma_r, k] > 0$ holds. 
	
	For $i \in [r]$, we call a state assignment $\delta \colon B_i \to \mathcal{C}$ over $B_i$ \emph{valid} if for all $v \in B_i \cap Y$, it holds that $\delta(v) = 0$ and for all $v \in B_i \cap Y'$, it holds that $\delta(v) \neq 0$. It suffices to explicitly compute the table entries indexed by valid state assignments, since $T_i[\delta, K] = 0$ holds whenever $\delta$ is not valid.

	The dynamic programming routine iterates over the bags $B_i$ of $\mathcal{B}$ in the ascending order of~$i \in [r]$ as follows.
	Recall that $B_1 = \emptyset$ holds, so there exists exactly one assignment $\sigma_1 \colon B_1 \to \mathcal{S}$, namely the empty assignment.
	Then it holds that $T_1[\sigma_1, 0] = 1$ and $T_1[\sigma_1, K] = 0$ for every~$K \in [k]$.
	Now suppose that for some $2 \leq i \leq r$ the table $T_{i-1}$ is already computed. 
	Then the table $T_i$ can be computed as follows. 
	As mentioned above, it suffices to compute the entries $T_i[\delta, K]$ for valid values of $\delta$.
	First, suppose that $B_i$ is an introduce bag and let $\{v\} = B_i \setminus B_{i-1}$. 
	For each valid state assignment $\delta \colon B_{i-1} \to \mathcal{S}$, each $x \in \{1,2\}$, and each $K \in [k]_0$, it holds that
	\[
		T_i\bigl[\delta[v\mapsto x], K\bigr] = T_{i-1}[\delta, K] \cdot [\delta \sim [v\mapsto x]\bigr], 
	\]
	where $\delta \sim [v\mapsto x]$ denotes that $\sigma(u) \neq x$ for all $u \in B_{i-1} \cap N(v)$. 
	It is well-known that for every neighbor $w \in V_i$ of $v$, it holds that $w \in B_{i-1}$. Therefore, if $\delta \sim [v\mapsto x]$ holds, then for every extension $\sigma$ counted in $T_{i-1}[\delta, K]$ the assignment $\sigma[v \mapsto x]$ is counted in $T_i\bigl[\delta[v \mapsto x], K\bigr]$. And vice versa: if $\delta[v \mapsto x]$ is valid and an assignment $\sigma[v \mapsto x]$ is counted in $T_i\bigl[\delta[v \mapsto x], K\bigr]$, then $\delta \sim [v\mapsto x]$ holds and $\sigma$ is counted in $T_{i-1}[\delta, K]$. 
	
	Also if $v \notin Y'$ holds, for every valid state assignment $\delta \colon B_{i-1} \to \mathcal{S}$ and each $K \in [k]$, it holds that
	\[
		T_i\bigl[\delta[v\mapsto 0], K\bigr] = T_{i-1}[\sigma, K - 1]
	\]
	and
	\[
		T_i\bigl[\delta[v\mapsto 0], 0\bigr] = 0.
	\]
	(Note that if $v \in Y'$, then $\delta[v \mapsto 0]$ is not valid.)
	Recall that the vertex $v$ does not belong to $G_{i-1}$.
	The correctness of the above equation follows by the observation that if $X \cup \{v\}$ is an \oct{} of $G_i$ such that $(A, B)$ is a bipartition of $G_i - \bigl(X \cup \{v\}\bigr)$, then $X$ is an \oct{} of $G_{i-1}$ such that $(A, B)$ is a bipartition of $G_{i-1} - X$.
	
	If $B_i$ is not an introduce bag, then it is a forget bag.
	So let $\{v\} = B_{i-1} \setminus B_i$.
	For each valid state assignment $\delta \colon B_i \to \mathcal{S}$ and each $k \in [K]_0$, we define 
	\[
		T_i[\delta, K] = \sum\limits_{x \in \mathcal{S}} T_{i-1}\bigl[\delta[v\mapsto x], K\bigr].
	\]
	To see that this is correct, observe that we have $G_{i-1} = G_i$. So the triples $(X, A, B)$ where $X$ is an \oct{} of $G_i$ (resp.\ $G_{i-1}$) and $(A, B)$ is a bipartition of the remainder are exactly the same for $G_i$ and $G_{i-1}$. 

	Finally, note that an assignment $\delta \colon B_i \to \mathcal{S}$ can only be valid if $\delta_{|_{B_{i-1}}}$ is valid. Therefore, the above equalities indeed define the table entries for all valid partial assignments.

	Now we bound the running time of this algorithm.
	The table $T_1$ can be computed in constant time. Then for every $2 \leq i \leq r$, given the table $T_{i-1}$, we compute the table $T_i$.
	We bound the number of valid assignments $\delta \colon B_i \to \mathcal{S}$. There are at most $1^{|B_i \cap Y|} 2^{|B_i \cap Y'|} 3^{|B_i \setminus (Y' \cup Y)|}$ of them: the vertices in $X_i \cap Y$ may only have state $0$, the vertices in $B_i \cap Y'$ are not allowed to get state $0$, and the remaining vertices might have any state.
	It holds that 
	\[
		1^{|B_i \cap Y|} 2^{|B_i \cap Y'|} 3^{|B_i \setminus (Y' \cup Y)|} \leq 2^{\ctw} 3^h.
	\]
	We also have that $K \in [k]_0$ and we may assume that $K < |V|$ holds since otherwise we have a trivial yes-instance of \Poctp{}.
	Therefore, the table $T_i$ has $\ostar(2^{\ctw} 3^h)$ states and each of them can be computed in polynomial time. Altogether, the running time of the algorithm is then bounded by $\ostar(2^{\ctw} 3^h r)$.
\end{proof}
\begin{theorem}
	Given a graph $G = (V, E)$ together with a linear arrangement $\ell$ of $G$ of \ctwdth{} at most $\ctw$, and a positive integer $k$, the \Poct{} problem can be solved in running time $\ostar(2^{\ctw})$.
\end{theorem}
\begin{proof}
	First, we compute the 2-subdivision $\hat G = (\hat V, \hat E)$ of the graph $G$ and a path decomposition $\mathcal{\hat B} = B_1, \dots, B_r$ of $\hat G$ satifying the properties of \cref{lem:special-path-decomposition} in polynomial time. 
	
	By \cref{lem:app-oct-subdiv}, the instance $(G, k)$ of \Poct{} is equivalent to 
	the instance $(\hat{G}, \emptyset, V_0, k)$ of \Poctp{}, i.e., it suffices to consider odd cycle transversals of $\hat G$ that do not contain subdivision vertices. 
	Recall that by the properties from \cref{lem:special-path-decomposition}, each bag of $\mathcal{\hat B}$ contains at most one vertex of $V$, i.e., at most one vertex outside $V_0$.
	Also, the path decomposition $\mathcal{\hat B}$ contains a polynomial number $r$ of bags.
	We now apply \cref{lem:oct-alg} with $h \leq 1$ to solve this instance of \Poctp{} in time $\ostar(2^{\ctw} 3 r) = \ostar(2^{\ctw})$.
\end{proof}

In \cref{app:oct-lb} we provide a matching lower bound and show that this algorithm is essentially optimal under SETH.
\subsection{Feedback Vertex Set}\label{app:fvs-ub}

In this section we present an algorithm to solve the \Pfvs{} problem.

\begin{quote}
	\Pfvs

	\textbf{Input}: An undirected graph $G = (V, E)$ and an integer $k$.

	\textbf{Question}: Is there a set $Y \subseteq V$ of cardinality at most $k$ such that $G - Y$ is a forest.

\end{quote}
A set $Y$ of vertices such that $G - Y$ is a forest is called a \emph{feedback vertex set} of $G$.

We provide a randomized algorithm that employs the \cnc{} approach of Cygan et al.\ \cite{CyganNPPRW11} to solve the \Pfvs{} problem in $\ostar(2^{\ctw})$. 
Similarly to the previous section, we employ edge-subdivision to restrict the set of vertices that belong to a feedback vertex set thus reducing the number of states of the dynamic programming and accelerating the algorithm. 
In the dynamic programming they use the states $0, 1_L, 1_R$ whose precise interpretation will be explained later. 
But as before, the subscripts $L$ and $R$ denote the sides of a consistent cut.
If similarly to \cref{app:general-approach}, we define the consistency matrix
$M$ indexed by these states relying on the ideas of the original algorithm, then it looks as follows: 
\begin{equation}\label{equ:fvs-consist}
	M = \kbordermatrix{
		& 0 & 1_{L} & 1_{R} \\
		0   & 1 & 1     & 1     \\
 		1_L & 1 & 1     & 0     \\
		1_R & 1 & 0     & 1     \\
 	}
\end{equation}
It simply reflects the fact that there are no edges between different sides of a consistent cut.	
This matrix 
has the full rank of $3$ over $\mathbb{F}_2$. Hence, our framework for coloring-like problems from \cref{app:general-approach} would not yield an algorithm faster than $\ostar(3^\ctw)$. However, we show that similarly to the \Poct{} problem, the approach relying on two-subdivision of a graph results in an $\ostar(2^\ctw)$ algorithm.
The dynamic programming will be similar to the one defined in the \cnc{} paper \cite{CyganNPPRW11}. However, we will consider 
a slightly modified version of the \pfvs{} problem. 
Then we adapt the original algorithm to solve this new problem. 
This allows us to employ some special properties of the subdivided graph. 
We define the \Pfvsp{} problem:

\begin{quote}
	\Pfvsp

	\textbf{Input}: A graph $G = (V, E)$, a subset $S \subseteq V$, and an integer $k$.

	\textbf{Question}: Is there a set $Y \subseteq V \setminus S$ of cardinality at most $k$ such that $G - Y$ is a forest.
\end{quote}

So here we are additionally given a set $S \subseteq V$ of ``forbidden'' vertices and we seek a feedback vertex set that does not use any of these vertices.
The key observation that a vertex set $Y$ is a feedback vertex set if and only if the subgraph $G - Y$ is a forest. Therefore, the \Pfvs{} problem is equivalent to looking for an induced forest of a certain size in the the graph.
For a graph $G$, with $cc(G)$ we denote the number of connected components of $G$.
We use the following lemma from the \cnc{} paper:
\begin{lemma}[\cite{CyganNPPRW11}] \label{lem:forest-components}
	A graph $G = (V, E)$ is a forest if and only if $\cc(G) \leq |V| - |E|$.
\end{lemma}

First of all, we repeat the used notation from \cref{app:oct}. From now on, let $G = (V, E)$ denote some fixed graph and let $\ell$ be a linear arrangement of cutwidth $\ctw$ of this graph. 
With $\hat{G} = (\hat{V}, \hat{E})$ we denote the graph arising from $G$ after subdividing each edge twice. 
Let $V_0$ be the set of all subdivision vertices, i.e., $V_0 = \hat{V} \setminus V$. 
We denote the number of vertices (resp.\ edges) of $\hat G$ with $\hat n$ (resp.\ $\hat m$).
Note that $\hat n \in \mathcal{O}(n)$ and $\hat m \in \mathcal{O}(m)$.

To make this section self-contained, we restate the results from \cref{app:oct} about the existence of a path decomposition of $\hat G$ with certain nice properties. The reader unfamiliar with path decompositions is referred to \cref{app:oct} for definitions.

\begin{lemma}\label{lem:special-path-decomposition-copy}
	Let $G = (V, E)$ be a graph and let $\ell$ be a linear arrangement of $G$ of cutwidth $\ctw$. Then the graph $\hat G$ admits a nice path decomposition $\mathcal{\hat B}$ of pathwidth at most $\ctw$ such that every bag of $\mathcal{\hat B}$ contains at most two vertices from $V$. Moreover, the path decomposition~$\mathcal{\hat B}$ contains $\mathcal{O}\bigl(\ctw \poly(n)\bigr)$ bags and it can be computed in polynomial time from $G$ and $\ell$.
\end{lemma}

Similarly to odd cycles from \cref{app:oct}, the following lemma holds for general cycles:

\begin{lemma}
	\label{lem:app-fvs-subdiv}
	Let $G = (V, E)$ be a graph.
	Then a \fvs{} of $G$ is a \fvs{} of $\hat G$, and a \fvs{}
	of $\hat G$ that only uses vertices from $V$ is a \fvs{} of $G$ as well. Moreover, $\hat G$ admits a minimum \fvs{} that only contains vertices from $V$.
	In particular, both graphs have the same feedback-vertex-set number.
\end{lemma}
\begin{proof}
	Let $S$ be a \fvs{} of $G$. Suppose $\hat{H} = \hat{G} - S$ is not a forest. Let $C$ be a cycle in $\hat{H}$. Then $C$ is a cycle in $\hat{G}$ that avoids $S$. By \cref{lem:subdivprescycles}, this cycle results from some cycle $C_0$ in $G$ by subdividing each of its edges twice. Hence, the cycle $C_0$ is disjoint from $S$ as well so it is a cycle in $G - S$ -- a contradiction. 
	
	Now let $\hat S$ be a \fvs{} of $\hat{G}$ that only uses vertices from $V$, then $\hat S$ is a \fvs{} of $G$ as well, since we can turn a cycle of $G - \hat S$ into a cycle of $\hat{G} - \hat S$ by subdividing each of its edges twice. 
	
	Finally, consider a minimum feedback vertex set $S$ of $\hat G$.
	If it does not contain any vertex of $V_0$, we are done. Otherwise, consider a vertex $w \in S \cap V_0$ and let $v \in V$ be the unique neighbor of $w$ in $V$. 
	Since $w$ has the degree of two in $\hat G$, every cycle containing $w$ also contains $v$. Therefore, the set $\bigl(S \setminus \{w\}\bigr) \cup \{v\}$ is a feedback vertex set of $\hat G$ of at most the same size as $S$. We can repeat this process until we obtain a feedback vertex set $S'$ of $\hat G$ that has at most the same size as $S$ (and hence, it is also a minimum feedback vertex set) and contains only vertices from $V$ as claimed.
\end{proof}

Now we show how to solve the \Pfvs{} problem for $G$ in time $\ostar(2^{\ctw})$. As we already mentioned before, Cygan et al.\ provided a \cnc{} formulation for this problem. In the following lemma, we summarize their result in a suitable for us form.
In this lemma, special elements $\Msymb$ and $\Fsymb$ occur. These are just special identifiers that allow every vertex to have two weights depending on how it occurs in a partial solution.
For a set~$U \subseteq \NNN$, a weight function $\omega: V \times \{ \Msymb, \Fsymb \} \to U$, and a subset $Y \subseteq V$, we use the notation $\omega(Y, \Msymb) = \sum_{v \in Y} \omega(v, \Msymb)$ and similarly, $\omega(Y, \Fsymb) = \sum_{v \in Y} \omega(v, \Fsymb)$.

\begin{lemma}[\cite{CyganNPPRW11}] \label{lem:app-cut-and-count-fvs} 
	Let $G = (V, E)$ be a graph, let $\mathcal B$ be a path decomposition of $G$, 
let $\totalweight = 4 n$, and let $S \subseteq V$. 
	Further, let $\omega\colon V\times\{ \Msymb, \Fsymb \}\to \bigl[W^*\bigr]$ be a weight function. 
	For numbers $A \in [n]_0, B \in [n]_0, C \in [n]_0, \weight \in \bigl[n \totalweight\bigr]_0$, with $\FFF^{A,B,C}_{\weight}$ we denote the family of all pairs $\bigl((X, M), (L, R)\bigr)$ with $X,M,L,R \subseteq V$ such that the following properties hold:
	\begin{itemize}
		\item $S \subseteq X$,
		\item $(L, R)$ is a consistent cut of $X$,
		\item $M \subseteq L$,		
		\item $|X| = A, \Bigl|E\bigl(G[X]\bigr)\Bigr| = B, |M| = C$,
		\item and $\omega(X, \Fsymb) + \omega(M, \Msymb) = \weight$.
	\end{itemize}
If there exists an algorithm $\AAA$ that given $G,S, \mathcal{B}, A, B, C, \weight$ computes the size of $\FFF^{A,B,C}_{\weight}$ modulo $2$ in time $\ostar\bigl(\beta(\mathcal{B})\bigr)$ for some computable function $\beta$, then there also exists a randomized algorithm that given a graph $G$, a vertex subset $S$ of $G$, a path decomposition $\mathcal{B}$, and an integer~$k$, solves the instance $(G, S, k)$ of the \Pfvsp{} problem in time $\ostar\bigl(\beta(\mathcal{B})\bigr)$. The algorithm does not give false positives and may give false negatives with probability at most $1/2$.	
\end{lemma}

\begin{remark}
	Observe that a vertex set $Y$ is a feedback vertex set of $G$ if and only if the subgraph $G - Y$ is an (induced) forest in $G$.
	In the above lemma, the set $X$ represents the set of vertices left after the removal of a potential feedback vertex set and therefore, should induce a forest. 
	The difference between our statement of this lemma and its original version from the \cnc{} paper is that they require the property $S \cap X = \emptyset$ (i.e., all vertices of $S$ must belong to the feedback vertex set) while we use $S \subseteq X$ (i.e., no vertex from $X$ belongs to the feedback vertex set). 
	Although these differences have an impact on the dynamic programming routine (which we provide explicitly next), it does not effect the correctness of the above lemma. Its proof relies merely on \cref{lem:forest-components} and the fact, that for fixed $(S, M)$, it holds that 
	\[
		\biggl|\Bigl\{\bigl((X, M),(L,R)\bigr) \in \FFF^{A,B,C}_{w}\Bigr\}\biggr|\equiv_2 2^{cc\bigl(M, G[X]\bigr)}
	\]
	where $cc\bigl(M, G[X]\bigr)$ denotes the number of connected components of $G[X]$ which do not contain any vertex from the set $M$. This number is odd if and only if all connected components contain an element of $M$ and in particular, there are at most $|M| = C$ connected components.
	We refer to \cite{CyganNPPRW11} for more details. 	
\end{remark}

	Our following lemma is crucial for the acceleration of the dynamic programming when parameterized by cutwidth. Although it works with a path decomposition, by \cref{lem:special-path-decomposition-copy} a linear arrangement of a graph yields a special path decomposition of its 2-subdivision, and we use it later.

\begin{lemma}\label{lem:algo-fvs-prime}
	\label{lem:app-fvs-alg} Let $G = (V, E)$ be a graph and let $\mathcal{B} = B_1, \dots, B_r$ be a nice path decomposition of~$G$. 
	Let $S \subseteq V$ and let $k$ be an integer.
	Let $W^* = 4n$ and let $\omega\colon V\times\{ \Msymb, \Fsymb \}\to \bigl[W^*\bigr]$ be a weight function. 
	Let $A \in [n]_0$, $B \in [n]_0$, $C \in [n]_0$, and $\weight \in [nW^*]_0$.
	Also let 
	\[
		h = \max\limits_{i \in [r]} |B_i \setminus S|.
	\]
	Then there exists an algorithm $\AAA$ that given the above input, computes the size of $\FFF^{A,B,C}_{\weight}$ modulo 2 in time $\ostar(r 2^{\pw(\mathcal{B})} 3^h)$.
\end{lemma}
\begin{proof}
	For $i \in [r]$, let $V_i = \bigcup_{j \in [i]} B_i$ and $G_i = G[V_i]$.
	For $i \in [r]$, $a,b,c \in [n]_0$, and $w \in [nW^*]_0$, we define 
	\begin{align*}
		\mathcal{R}_i (a,b,c,w) = \biggl\{ &(X, M) \Bigm \vert X \subseteq V_i, S \cap V_i \subseteq X, |X| = a,
			\Bigl|E\bigl(G[X]\bigr)\Bigr| = b,\\
	& |M| = c, M \subseteq X \setminus B_i, \omega(X \times {\Fsymb}) + \omega(M \times {\Msymb}) = w\biggr\},\\
	\mathcal{C}_i (a,b,c,w) = \Bigl\{&\bigl((X, M),(X_L, X_R)\bigr) \Bigm \vert (X, M) \in \mathcal{R}_i(a,b,c,w), M \subseteq X_L,\\
					 &(X_L, X_R) \text{ is a consistent cut of $X$}\Bigr\}
	\end{align*}
	It follows from the definition of both sets, that 
	\[
		\bigl|\FFF^{A,B,C}_{\totalweight}\bigr| = \bigl|\mathcal{C}_r(A,B,C,\totalweight)\bigr|.
	\]
	In order to compute this value, we define the following tables and provide a dynamic programming algorithm to compute their values. For $i\in[r]$, $a,b,c\in[n]_0$, $w
	\in [n\totalweight]_0$, and $s \in \{0, 1_L, 1_R\}^{B_i}$ we define:
	\begin{align*}
	A_i (a,b,c,w,s) = \biggl|\Bigl\{&\bigl((X, M),(X_L, X_R)\bigr) \in \mathcal{C}_i(a,b,c,w)\Bigm \vert\\
	&\bigl(s(v) = 1_x \implies v \in X_x \text{ for } x \in \{L, R\}\bigr), \bigl(s(v) = 0 \implies v \notin X \bigr)\Bigr\}\biggr|.
	\end{align*}
	Observe that if for some vertex $w \in S \cap B_i$, we have $s(w) = 0$, then $A_i (a,b,c,w,s) = 0$ holds. 
	For this reason, we will neither consider such entries nor compute them explicitly. We will later see that this is crucial for the running time of the algorithm.	
	Also note that we have 
	\[
		\bigl|\mathcal{C}_r(A,B,C,W)\bigr| = A_r(A,B,C,W,\emptyset)
	\]
	where $\emptyset$ denotes the unique mapping from $B_r = \emptyset$ to $\{0, 1_L, 1_R\}$.
	Therefore, to compute the desired size of $\FFF^{A,B,C}_{\weight}$, it suffices to compute the entries of the table $A_r$.
	We will compute the entries of tables $A_i$ in the order of increasing $i \in [r]$.
	All entries that we do not set explicitly are treated as zero. We also do not store them explicitly to achieve better running time.

	Recall that $B_1 = \emptyset$. 
	So we first set $A_1(0, 0, 0, \emptyset) = 1$, clearly, the remaining entries of $A_1$ are zeroes. 
	Now we iterate over $i = 2,3,\dots, r$ as follows. 
	For a fixed $i \geq 2$, the bag $B_i$ of a nice path decomposition $\mathcal{B}$ is either an introduce or a forget bag.

	If $B_i$ is an introduce bag, let $v \in B_i \setminus B_{i-1}$ be the vertex \emph{introduced} at $B_i$. 
	For each state assignment $s \in \{0, 1_L, 1_R\}^{B_{i-1}}$, let $d_s = \Bigl|N_{G_i}(v) \cap s^{-1}\bigl(\{1_L, 1_R\}\bigr)\Bigr|$.
	Informally speaking, this is the number of neighbors of $v$ in a partial solution represented by $s$. 
	Note that by a property of a path decomposition, all neighbors of $v$ in $G_i$ must belong to $B_i$ since $v$ occurs in $B_i$ for the first time.
	Therefore, by ``adding $d_s$ to the number of edges of a partial solution'' in the dynamic programming below, we indeed count exactly the edges between $v$ and vertices of the current partial solution.
	Then for each $a, b, c \in [n]_0$ and $w \in [n\totalweight]_0$, we set: 
	\begin{alignat*}{5}
		&A_i\bigl(a, b, c, w, s[v\mapsto 0]\bigr) && = A_{i-1}&&(a, b, c, w, s).\\
	\end{alignat*}
Let us recall that we only do this if $s \notin S$ holds. Otherwise, the entry is zero.	
And for each $a \in [n-1]_0$, $b \in [n - d_s]_0$, $c \in [n]_0$, and $w \in \bigl[n\totalweight - \omega(v, \Fsymb)\bigr]$ we set:
	\begin{alignat*}{5}
		&A_i\bigl(a + 1, b + d_s,c,w + \omega(v, \Fsymb),s[v\mapsto 1_L]\bigr)
		&& = A_{i-1}&&(a, b, c, w, s)\\
		& && &&\cdot\bigl[N_{G_i}(v)\cap s^{-1}(1_R) = \emptyset\bigr],\\
		&A_i\bigl(a + 1, b + d_s,c,w + \omega(v, \Fsymb),s[v\mapsto 1_R]\bigr)
		&& = A_{i-1}&&(a, b, c, w, s)\\
		& && &&\cdot\bigl[N_{G_i}(v) \cap s^{-1}(1_L) = \emptyset\bigr].
	\end{alignat*}

Otherwise, we have that $B_i$ is a forget bag. So let $v \in B_{i-1} \setminus B_i$ be the vertex \emph{forgotten} at~$B_i$. For $a,b,c \in [n]_0$, $w\in[n\totalweight]_0$, and $s \in \{0, 1_L, 1_R\}^{B_{i}}$, we set the value $A_i(a, b, c, w)$ as:
	\begin{align*}
		&A_i(a, b, c, w) = \\
		&A_{i-1}\bigl(a,b,c,w,s[v\mapsto 0]\bigr) \\
		+&A_{i-1}\bigl(a,b,c,w,s[v\mapsto 1_R]\bigr)\\
		+&A_{i-1}\bigl(a,b,c,w,s[v\mapsto 1_L]\bigr)\\
		+&A_{i-1}\bigl(a,b,c-1,w-\omega(v, \Msymb),s[v\mapsto 1_L]\bigr) \cdot [c \geq 1] \cdot \bigl[w \geq \omega(v, \Msymb)\bigr]
	\end{align*}
	Note that in the last summand we possibly try to access undefined table entry. But this only can happen if $c < 1$ or $w < \omega(v, \Msymb)$ holds and in this case the entry is multiplied by an Iverson bracket equal to zero. Therefore, the whole summand is just equal to zero in this case and hence, it is well-defined.
	This last summand handles the addition of $v$ to the set $M$ of a partial solution: this is only allowed if $v$ is not in the current bag anymore, and in this case $v$ must be placed on the left side of a consistent cut (recall the definition of $\mathcal{C}_i$).
	Also recall that the first summand can only be non-zero if we have $v \notin S$.

	The correctness of the algorithm follows by induction over $i \in [r]$. 
	Crucially observe that for a vertex $v$ introduced in some bag $B_i$, all neighbors of $v$ in $G_i$ belong to $B_{i-1}$. 
	Therefore, every edge of some partial solution is counted exactly once, namely when its later introduced end-vertex is introduced.
	We refer to \cite{CyganFKLMPPS15} for more details on dynamic programming over nice tree decompositions.
	
	The running time can be bounded as follows. 
	For each bag $B_i$, we consider all mappings $s \in \{0, 1_L, 1_R\}^{B_i}$ (such that $s$ does not assign the value $0$ any vertex from $S$) and all reasonable (as described above) values $a,b,c \in [n]_0$, $w \in [n\totalweight]_0$, and we carry out a constant number of arithmetic operations using the entry $A_i[a, b, c, w, s]$. 
	Then the time spent on the bag $B_i$ can be bounded by 
	\[
		\ostar(2^{|B_i \cap S|} 3^{|B_i \setminus S|}) = \ostar(2^{\pw(\mathcal{B})} 3^h)
	\]
	since the values of $a, b, c, w$ are at most quadratic in $n$.
	Altogether, the table $A_r$ and hence, also the size of $\FFF^{A,B,C}_{\weight}$ can be computed in $\ostar(r 2^{\pw(\mathcal{B})} 3^h)$ as claimed.
\end{proof}

\begin{theorem}
	Given a graph $G = (V, E)$ together with a linear arrangement $\ell$ of $G$ of \ctwdth{} at most $\ctw$, and an integer $k$. There exists a randomized algorithm that solves the \Pfvs{} problem in time $\ostar(2^{\ctw})$. The algorithm does not give false positives and may give false negatives with probability at most $1/2$.	
\end{theorem}
\begin{proof}
	We start by building the graph $\hat G$ together with a nice path decomposition $\hat{\mathcal{B}} = B_1, \dots, B_r$ of $\hat G$ that satisfies \cref{lem:special-path-decomposition-copy} in polynomial time.
	By \cref{lem:app-fvs-subdiv}, the instance~$I = (G, k)$ of the \Pfvs{} problem is equivalent to the instance $I' = (\hat G, V_0, k)$ of the \Pfvsp{} problem. 
	By \cref{lem:algo-fvs-prime}, using the path decomposition $\hat{\mathcal{B}}$, the instance $I'$ can be solved in time $\ostar(r 2^{\pw(\hat B)} 3^h)$ where $h$ denotes the largest number of vertices outside $V_0$ in some bag.
	Now recall that by the properties from \cref{lem:special-path-decomposition-copy}, we have~$r \in \mathcal{O}\bigl(\ctw \poly(n)\bigr)$, $h \leq 2$, and $\pw(\hat B) \leq \ctw$. Therefore, the instance~$I'$ (and therefore,~$I$ as well) can be solved in time $\ostar(2^\ctw)$.
\end{proof}

In \cref{app:fvs-lb} we provide a matching lower bound and show that this algorithm is essentially optimal under SETH.
\section{Lower Bounds}\label{app:lb}
\subsection{Steiner Tree} \label{append:st}

When parameterized by treewidth, the \Pst{} problem can be solved in time $\ostar(3^{\tw})$~\cite{CyganNPPRW11}. So by
\cref{lem:ctw-ub}, it can also be solved in time $\ostar(3^{\ctw})$ when parameterized by \ctwdth{}. We prove that under SETH, this
problem cannot be solved in time $\ostar((3-\varepsilon)^{\ctw})$ for any positive $\varepsilon$ and
hence, the upper bound of $\ostar(3^{\ctw})$ is essentially tight.

\begin{quote}
	\Pst{}
	
	\textbf{Input}: A graph $G = (V, E)$, a set $T \subseteq V$ of \emph{terminals}, and an integer $k$.

	\textbf{Question}: Is there a set $S \subseteq V$ of cardinality at most $k$ such that $T \subseteq S$ and $G[S]$ is connected.
\end{quote}

\begin{theorem}\label{app:theo-st}
	Assuming SETH, there is no algorithm solving the \Pst{} problem in time $\ostar((3-\varepsilon)^{\ctw})$ for
    any positive real $\varepsilon$, even when a linear arrangement $\ell$ of $G$ of this width is given.
\end{theorem}

\subparagraph*{Construction.}
    Let $d \in \NN$ be arbitrary but fixed and let $I$ be an instance of $d$-SAT. Let $n$ be the number of variables in $I$. Let $t_0 \in \NN$ be a constant depending on $\varepsilon$ only that we fix later, $s = \lceil n / t_0\rceil$, and $t = \lceil t_0 \log_3 2\rceil$. Let $n' = s \cdot t$. In the
    following, we will use the following indices: $i \in [s]$, $q \in [t]$, $r \in [2n'+1]$, and $j \in [m]$. These indices will always belong to these intervals so for shortness, we omit the domain when we use any of these indices.

    We construct an equivalent instance $\bigl(G = (V, E), T, k\bigr)$ of the \Pst\ problem as follows. 	
    Initially $G$ and $T$ are empty. We call a vertex in $T$ a terminal. When we say ``we add a 
    terminal'' or ``subdivide an edge by a terminal'', we mean that first we apply the
    operation and then add the new vertex to $T$. We denote the vertex that results from
    subdividing an edge $\{u, v\}$ by a terminal by $t(u, v)$.    
    First, we partition the variables of $I$ into $s$ groups of size at most $t_0$ each, and denote them
    by $U_1, \dots U_s$. For each $i$ (i.e., for each group of variables), we add a group of $t$ path-like
    structures $P_{i, 1}, \dots, P_{i, t}$. 
    Each path-like structure
    consists of a sequence of $(2n'+1)m$ copies of a gadget called ``path gadget'', which we describe later.
    For simplicity, we call the path-like structures \emph{paths}. For a fixed $i$, we call the set of paths
    $\{P_{i, 1}, \dots P_{i, t}\}$ the \emph{path group} corresponding to the variable group $U_i$. We partition every path $P_{i,q}$ into $2n'+1$ \emph{path segments} each consisting of $m$ consecutive path gadgets and we denote the $r$th path segment of $P_{i,q}$ with $W^r_{i,q}$.
    Let $X_{r, j}^{i, q}$
    denote the $j$th path gadget of the $r$th path segment on the path $P_{i,q}$. The \emph{$(r, j)$th column} $Q_{r,j}$ is the set of gadgets $X_{r, j}^{i, q}$ for all $i, q$. Similarly, for fixed $i, r, j$, the
    \emph{bundle} $B^i_{r,j}$ consists of all gadgets $X_{r, j}^{i, q}$ for all $q$. Finally, for fixed $r$, we call the set consisting of all $r$th path segments $W^r_{i,q}$ for all $i$ the \emph{$r$th segment}.

    In each path gadget, there will be two special vertices: the left-end vertex $v$ and the right-end vertex $v'$. From
    the right-end vertex of each path gadget (except for the last one on every path), we add an edge to the left-end
    vertex of the following path gadget on the same path. These edges give the path-like structure of these paths, since
    they can be seen as cuts of size one between the path gadgets. In other words, a path-like structure is a path in which each vertex is replaced by a path gadget.

    Next we add to $G$ a simple path $\rootpath$ consisting of new vertices, whose length and the order of vertices on it
    will be fixed later. We make each of its vertices adjacent to a new terminal. This way we ensure that any \Steinertree{} of $G$ contains
    all vertices of $\rootpath$. Now let $\rootpath'$ denote the set of vertices consisting of all vertices of the path $\rootpath$ together with all terminals attached to it. The path $\rootpath$ will ensure the connectivity of the desired \Steinertree{}. We will say that a vertex $u \notin \rootpath'$ of $G$ is \emph{root-connected}, if it is adjacent to a vertex in
    $\rootpath$ and we denote this vertex on $\rootpath$ by $\rootpar{u}$.  
    We will ensure that no two distinct root-connected vertices are adjacent to the same vertex in $\rootpath$ so that for a root-connected vertex $u$ the vertex $\rootpar{u}$ is unique and we call this vertex the 
    \emph{private} neighbor of $u$ in $\rootpath$. 
    For a linear arrangement $\ell$ of $G \setminus \rootpath'$, we can choose the order of the vertices on $\rootpath$ to match the order of their adjacent
    vertices on $\ell$. This allows us to extend $\ell$ to a linear arrangement of $G$ with only constant increase in the \ctwdth{} of the arrangement. We expand on that later.

    We also add two new root-connected vertices $g$ and $g'$ called \emph{guards} to $G$. We add a new terminal and make it adjacent only to $g$ and we add another new terminal and make it adjacent only to $g'$. Next, we add an edge between $g$ and the left-end vertex $v$ of the first path gadget $X^{i,q}_{1,1}$ of each path $P_{i,q}$. Similarly, we add an edge between $g'$ and the right-end vertex $v'$ of the last path gadget $X^{i,q}_{2n'+1,m}$ of each path $P_{i,q}$.

\newsavebox{\bigfigure}
    \begin{figure}[t]
        \centering
            \savebox{\bigfigure}{\centering \hbox{\includegraphics[width=.45\linewidth]{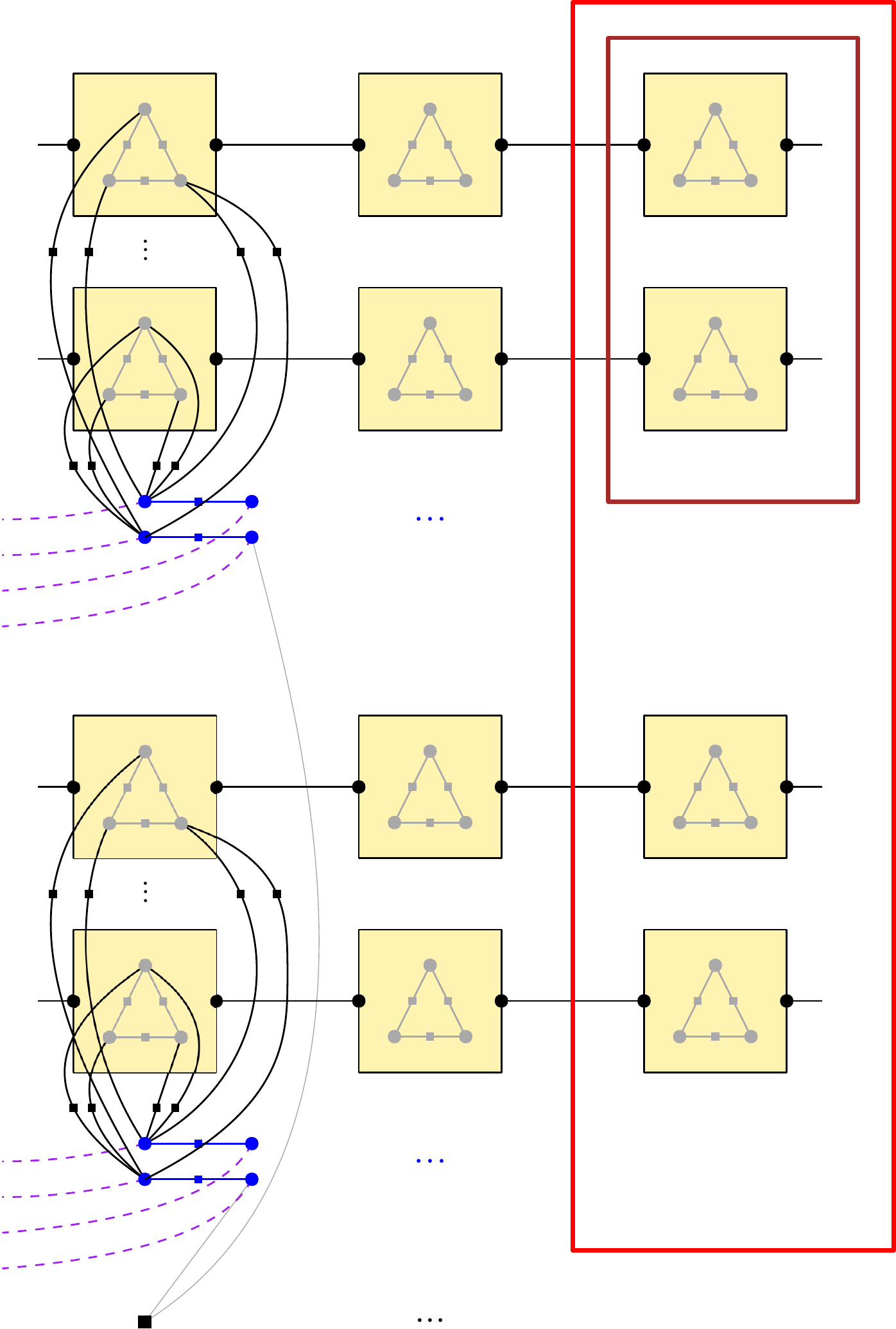}}}

        \begin{subfigure}[t]{.56\linewidth}
            \centering
            \usebox{\bigfigure}

            \caption{A sketch of the global structure of the graph $G$.            
            The yellow boxes correspond to path gadgets, where only the left-, right-end vertices and the clique vertices are drawn. A decoding gadget (depicted in blue) is attached to each bundle. Only the first one is drawn. It contains $3^t$ pairs of vertices corresponding to the different combinations of state assignments over a bundle. In brown, we surround one bundle, and in red a whole column.. A  terminal is attached to each column  and is only adjacent to vertices of the decoding gadgets of this column.\label{fig:pathgroup-app}}
        \end{subfigure}
    \hfill
        \begin{subfigure}[t]{.42\linewidth}
            {\raisebox{\dimexpr.6\ht\bigfigure-.6\height\relax}
                {\includegraphics[width=\linewidth]{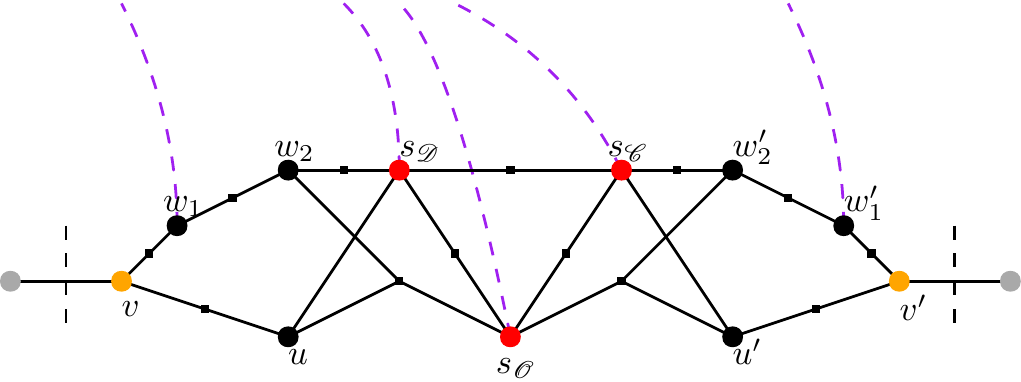}}}
            \caption{A path gadget with left- and right-end vertices are orange and clique vertices are red.
            The right-end vertex of the previous gadget and the left-end vertex of the following gadget are gray.
            \label{fig:pathgadget-app}}

        \end{subfigure}

        \caption{A path gadget (\subref{fig:pathgroup-app}), and a group of paths (\subref{fig:pathgadget-app})
        corresponding to a  variable group. Terminals depicted as small squares. Adjacencies to rootpath are depicted in purple.}
        \label{fig:main-st-app}
    \end{figure}

    Now we describe the structure of a single path gadget (see \cref{fig:main-st-app}~(b) for
    illustration). Each path gadget contains the left-end vertex $v$ and the right-end vertex $v'$ mentioned above. It also contains a clique on three vertices denoted $\verC, \verD,
    \verO$ and called the \emph{clique vertices}. All edges of this clique are subdivided by terminals. For simplicity, we still refer to it as the clique $K_3$. The clique vertices are root-connected. For a fixed path gadget, with $K'$, we denote the set consisting of clique vertices together with the terminals subdividing the edges between them.
    A path gadget also contains vertex $u$, an edge between $u$ and $v$, and an edge between $u$ and $\verD$. Similarly, it contains a
    vertex $u'$, an edge between $u'$ and $v'$, and an edge between $u'$ and $\verC$. It also contains a pair of adjacent vertices $w_1$ and $w_2$. 
    The vertex $w_1$ is root-connected. 
    There are also 
    an edge between $w_1$ and $v$ and an edge between $w_2$ and $\verD$. Similarly, there is a pair of adjacent vertices $w'_1$ and $w'_2$.
    The vertex $w'_1$ is root-connected. There are also
    an edge between $w'_1$ and $v'$ and an edge between $w'_2$ and $\verC$. 
    Further, it contains
    a new terminal adjacent to $u$, $w_2$, and $\verO$ and another new terminal adjacent to $u'$, $w'_2$, and $\verO$ by edges. Finally, the edges $\{v, u\}$, $\{v, w_1\}, \{w_1, w_2\}, \{w_2, \verD\}, \{v', u'\}, \{v', w'_1\}, \{w'_1, w'_2\}, \{w'_2, \verC\}$ are subdivided by terminals. If not clear from the context, we write the name of a path gadget after the name of a vertex in parentheses to indicate to which gadget it belongs, e.g., the vertex $v(X)$ is the left-end vertex of a path gadget $X$.
    If the gadget is clear from the context, we omit this clarification.

    Intuitively, the construction will have the following interpretation. 
    Consider an arbitrary \Steinertree{} of $(G, T)$.
    The intersection of a
    \Steinertree{} with a path gadget defines a \emph{state} of this gadget. 
    Ideally, we would like to (mostly) have the same states in consecutive path gadgets of the same
    path.
    Although this is not always the case, the sufficiently large number of copies of a path gadget on a path ensures that we can always find a sufficiently long subsequence of a path with this property.
    We will provide more details in proof of \cref{lem:st-lb-correctness-2}.
    By construction, at least two clique vertices of a path gadget belong to any \Steinertree{}: otherwise, there would be a terminal subdividing a clique edge that is isolated from the remainder of the \Steinertree{}.
    We will choose the value of $k$ (i.e., the size of a
    \Steinertree{}) in a tight way 
    to ensure that any \Steinertree{} picks exactly two vertices from each such clique. This allows us to identify a state of a path gadget in a \Steinertree{} with a clique vertex that does not belong to this \Steinertree{}. We call these states $\scrC, \scrD, \scrO$ respectively. 
    
    We attach a so-called decoding gadget (described next) to each bundle. Decoding gadgets connect path gadgets with clause
    gadgets (described later) and enforce certain state assignments to bundles, we
    expand on this later. Now we describe the construction of a decoding gadget $Y = Y^i_{r,
    j}$ attached to a bundle $B = B^i_{r,j}$. A decoding gadget $Y$ consists of an induced matching of size $3^t$ where all end-vertices of the matching are root-connected. 
    Now the matching contains one edge for each possible combination of states in the bundle.
    We subdivide each of these edges by a
    terminal. For simplicity, we still refer to subdivided edges as edges. With each state
    assignment $\sigma: \{X^{i, 1}_{r, j}, \dots X^{i, t}_{r,j}\} \to \{\scrC, \scrD, \scrO\}$ over the bundle $B$, we associate a unique edge $e_\sigma$ of $Y$. Let the
    end-vertices of $e_{\sigma}$ be $v_\sigma$ (called the \emph{left side}) and $u_\sigma$ (called
    the \emph{right side}). For each $q$, let $X = X^{i, q}_{r,j}$ be the $q$th path gadget of $B$ and let $x = \sigma(X)$ is the state assigned to this gadget by $\sigma$. For every clique vertex $s_{x'}, x' \in \{\scrC, \scrD, \scrO\} \setminus \{x\}$ of $X$, we add the edge $\{v_\sigma, s_{x'}\}$ and subdivide it by a terminal.
   After the description of a clause gadget, we will explain how a decoding gadget is connected to it. 
   
    Now we describe the structure of a clause gadget. Clause gadgets will represent the clauses of the formula $I$. For every $r, j$, a clause gadget $Z_{r, j}$ will be attached to the column $Q_{r,j}$. 
    By the choice of $t$, it holds that $2^{t_0} \leq 3^t$. So it is possible to ensure that each partial (truth-value) assignment of a variable group $U_i$ corresponds to a different combination of states over a bundle $B^i_{r, j}$. 
    To achieve this, we fix an arbitrary injective mapping $\Phi: 2^{[t_0]} \mapsto \{\scrC, \scrD, \scrO\}^{t}$ 
    that maps Boolean assignments of $t_0$ variables to state assignments of $t$ path gadgets. 
    Since every variable group $U_1, \dots, U_s$ contains at most $t_0$ variables, we can fix an ordering of variables in $U_i$ and interpret the truth-value assignments of $U_i$ as subsets of $[t_0]$ (namely, the indices of true variables). Therefore, for simplicity, with a slight abuse of notation, we will apply $\Phi$ to truth-value assignments of every $U_i$.
    Similarly, since every bundle contains exactly $t$ gadgets, an image of $\Phi$ is naturally interpreted as an assignment of states to this bundle.
    For any fixed $r, j$, the clause gadget $Z_{r, j}$ will represent the clause $C_j$. So the $r$th segment for some fixed $r$ contains a column $Q_{r, j}$ and a clause gadget $Z_{r, j}$ for each clause of the input. The clause gadget $Z_{r, j}$ will enforce that in any \Steinertree{} of the desired size, at least one bundle $B^i_{r, j}$ of this column is in a state
assignment corresponding to a partial assignment of $U_i$ satisfying $C_j$ (with respect to $\Phi$). 
A clause gadget $Z = Z_{r, j}$ consists of a new single terminal vertex $w$. For each $i$, let
    $\Pi^i_j$ be the set of all partial assignments of $U_{i}$ that satisfies $C_j$ (i.e., such assignments that at least one literal of $C_j$ corresponding to a variable from $U_i$ is true) and let $\Sigma^i_j = \Phi(\Pi^i_j)$.
    Informally speaking, $\Sigma^i_j$ is the set of state assignments to a bundle $B^i_{r, j}$ (for any $r$) that correspond to a partial assignment of $U_i$ satisfying $C_j$. Now for every edge $e_\sigma$ of the matching in the decoding gadget $Y^i_{r,j}$ with $\sigma \in \Sigma^i_j$, we add an edge between $w$ and $u_{\sigma}$ (the right side of $Y^i_{r,j}$). This concludes the construction of the graph $G$.

\subparagraph*{Budget.}
    First we fix the number of vertices on $\rootpath$ by counting the number of root-connected vertices (recall that every root-connected vertex has a private neighbor on $\rootpath$): 
    \begin{itemize}
        \item Every path gadget contains 5 root-connected vertices. There are $n'(2n'+1)m$ such gadgets, this gives rise to $5n'(2n'+1)m$ root-connected vertices in path gadgets.
    \item Every decoding gadget contains $2 \cdot 3^t$ root-connected vertices. There are $s(2n'+1)m$ decoding gadgets in the graph, so we get $2\cdot3^ts(2n'+1)m$ root-connected vertices in decoding gadgets in total.
    \item The guards $g$ and $g'$ are root-connected.
    \end{itemize}
    So the number of vertices on the path $\rootpath$ (equal to the number of root-connected vertices) is 
    \[
        k'' = 2 + (5n'+2 \cdot 3^ts)(2n'+1)m. 
    \]
Now we count the number of terminals in the graph:
    \begin{itemize}
        \item Every path gadget contains 13 terminals, so in total $13 n' (2n'+1)m$ over all path gadgets.
        \item Every decoding gadget contains $3^t$ terminals subdividing the matching edges. Further, each vertex on the left side of one of the $3^t$ edges of a decoding gadget is linked by edges subdivided by terminals to $2t$ vertices in path gadgets. 
        This gives rise to $3^t + 2t3^t = 3^t(2t+1)$ terminals per decoding
        gadget and hence, $3^t(2t+1)s(2n'+1)m$ in total. 
        \item Each clause gadget consists of a single terminal. Since there is one clause gadget per column, there are $(2n'+1)m$ such terminals. 
        \item The guards $g$ and $g'$ and all vertices on $\rootpath$ have private neighbors that are terminals. So there are $2+k''$ such terminals.
    \end{itemize}
    So totally, there are
    \begin{align*}
        k' &= 13n'(2n'+1)m + 3^t(2t+1)s(2n'+1)m+(2n'+1)m + 2 + k''\\
        &= (13n' + 3^t(2t+1)s+1)(2n'+1)m + 2 + k''        
    \end{align*}
    terminals in $T$. 

    Now we provide the budget (i.e., the size of the desired \Steinertree{}) $k$. We choose the budget $k$ matching a lower bound on the size of any \Steinertree{} of the graph $G$ spanning $T$. We define this lower bound using a so-called
    \emph{packing} $\mathcal{P}(G)$ of $G$. 
    A packing is a family of pairwise disjoint sets of vertices (also called \emph{components}) of $G$ such that each component $C$ is assigned a value $p(C)$ such that each \Steinertree{} contains at least $p(C)$ vertices from $C$. Since components are disjoint, the sum of values $p(C)$ over all components $C$ of a packing is a lower bound on the size of any \Steinertree{}.
    Moreover, if a \Steinertree{} of size $k$ exists, then it contains exactly $p(C)$ vertices from every component $C$.
    To construct $\PP(G)$, we will provide packings for certain parts of $G$ (e.g., path gadgets) and the build a union of them.
    
    The first component consists of all terminals in $G$. By definition, a \Steinertree{} must contain all these vertices so
    we set $p(T) = k'$. 
    Next we define a component $V(\rootpath) \cup \{g, g'\}$. Any \Steinertree{}
    of $(G, T)$ must contain all vertices in this component, since each of these vertices is
    adjacent to a private neighbor that is a degree-1 terminal. So we set $p\bigl(V(\rootpath) \cup \{g, g'\}\bigr) = k'' + 2$.
    
    For every path gadget $X = X_{r,j}^{i,q}$, we define a packing denoted by $\PP(X)$ as 
\[
	\mathcal{P}(X) = \bigl\{\{\verC, \verD, \verO\}, \{w_1, w_2\},\{u, v\}, \{w'_1,w'_2\}, \{u', v'\} \bigr\}.
\] 
    Note that this this indeed a packing since every component is either a clique $\{\verC, \verD, \verO\}$ subdivided by terminals or an edge subdivided by terminals. From each of these components, any \Steinertree{} contains at least all but one vertices: otherwise, there would exist an isolated terminal. Therefore, we set $p\bigl(\{\verC, \verD, \verO\}\bigr) = 2$ and $p(C) = 1$ for $C \in \mathcal{P}(X) \setminus \bigl\{\{\verC, \verD, \verO\}\bigr\}$. Hence, any \Steinertree{} must contain at least $6$ non-terminal vertices from each path gadget, i.e., at least $6n'(2n'+1)m$ such vertices in total.
   
   For every decoding gadget $Y = Y^i_{r, j}$, we define a packing denoted by $\PP(Y)$ as 
   \[ 
   	\mathcal{P}(Y) = \bigl\{\{u_{\sigma}, v_{\sigma}\} \bigm\vert \sigma \in \{\verC, \verD, \verO\}^{B^i_{r,j}}\bigr\}. 
   \]
   Simply speaking, we take the end-vertices of all $3^t$ matching edges of the decoding gadget.
We set
    $p\bigl(\mathcal{P}(Y)\bigr) = 3^t$: from each of the matching edges, a \Steinertree{} contains at least one end-vertex since otherwise the terminal subdividing this edge would be isolated.
    Totally, any \Steinertree{} contains at least $3^t s(2n'+1)m$ non-terminal vertices from all decoding gadgets. 
    
    Note that the packing provided above contains pairwise disjoint components. Therefore, their union $\PP(G)$ is a packing as well.
    Using the values of $p(C)$ for the components $C$ provided above, we obtain that any \Steinertree{} has the size of at least 
    \begin{align*}
        k := &k' + k'' + 2 + 6n'(2n'+1)m + 3^t s(2n'+1)m\\
        = &k' + k'' + 2 + (6n'+3^ts)(2n'+1)m. 
    \end{align*}
    In particular, every \Steinertree{} of size $k$ contains exactly $p(C)$ vertices from each component~$C$ of $\PP(G)$. 

	Finally, observe that since $t$ is a constant only depending on $\varepsilon$ (i.e., independent of $n$), the instance $(G, T, k)$ can be constructed in polynomial time.

    \begin{lemma}\label{lem:app-st-ctw}
    Let $(G, T, k)$ be the instance of \Pst\ arising from $I$ as described above.
    Then $G$ has the \ctwdth{} of at most $n' + \mathcal{O}(1)$.
    Moreover, a linear arrangement of $G$ of this cutwidth can be constructed in polynomial time from $I$.
	\end{lemma}
    \begin{proof}
        Recall that $t$ is a constant independent of $n$.
        We describe a linear arrangement $\ell$ of $G$ of \ctwdth{} at most at most $n' + (2t+d)3^t
    + 5$. 
    First we will build a linear arrangement of the graph $G' = G \setminus
    \rootpath'$ that arises from $G$ after removing $\rootpath$ and the private terminals attached to it. For every path gadget $X = X^{i, q}_{r, j}$, we will now construct an ordering $\widetilde{O}(X)$ of its vertices. 
    The vertices of $X$ will appear in the following order in the linear arrangement:
    We start with the order $v, u, w_1, w_2, \verD, \verO, \verC, w'_2, w'_1, u', v'$. Then we add the terminal adjacent to $u$, $w_2$, and $\verO$ directly before $\verD$, and the terminal connected to $u', w'_2$ and $\verO$ directly after $\verC$ $\widetilde{O}(X)$. Finally, we add each terminal subdividing an edge at an arbitrary position between the end-vertices of this edge. Observe that $\widetilde{O}(X)$ is a linear arrangement of $G[X]$ of cutwidth at most~$4$. See \cref{fig:append-st-lin-arr} for graphical representation of this arrangement.
    
    Now similarly, for every decoding gadget $Y =
    Y^{i}_{r, j}$ we fix the ordering $\widetilde O(Y)$, we describe how this order is constructed from left to right. 
    We simply iterate through the matching edges and for a fixed such edge,
     we first add its left end, then the subdivision-terminal, and finally, its right end. Then $\widetilde{O}(Y)$ is a linear arrangement of $G[Y]$ of cutwidth at most~$1$. Even though these two fixed arrangements serve as a tight bound 
     on the cutwidth of the graph induced by a path gadget and a decoding gadget respectively, putting the vertices in an 
     arbitrary order would have sufficed, since any constant bound on the cutwidth of both gadgets would suffice, and 
     this is trivially the case, since both path gadgets and decoding gadgets have constant size.

    \begin{figure}
        \includegraphics[width=.9\linewidth]{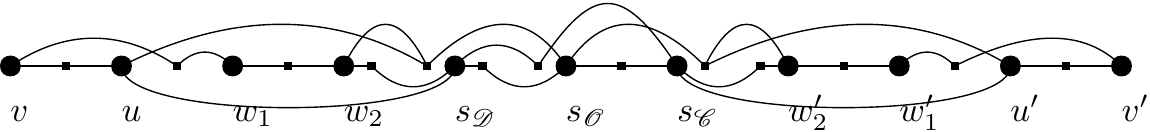}
        \caption{A graphical representation of the linear arrangement $\widetilde{O}(X)$.
        It can be seen from the figure, that the largest cut in the arrangement has size $4$.}
        \label{fig:append-st-lin-arr}
    \end{figure}

    Now we are ready to define a linear arrangement $\ell'$ of $G' = G \setminus \rootpath'$.
    We describe how it is created in left to right order.
    This linear arrangement is created ``column-wise'', i.e., we iterate through all pairs $(r, j)$ in lexicographic order. For fixed pair $(r, j)$, we proceed as follows. We iterate through all pairs $(i, q)$ in lexicographic order and add the vertices of the path gadget $X_{r,j}^{i, q}$ in the order given by $\widetilde O(X_{r,j}^{i, q})$. If $q = t$, then after adding the path gadget, we add the vertices of the decoding gadget $Y^i_{r,j}$ in the order given by $\widetilde O(Y^i_{r,j})$. If $(i, q) = (s, t)$, then after adding the decoding gadget, we add the single vertex of the clause gadget $Z_{r, j}$ and move to the next pair $(r, j)$.
    Observe that by construction:
    \begin{itemize}
    	\item The edges induced by different path gadgets do not overlap.
    	\item The edges induced by different decoding gadgets do not overlap.
    	\item No edge induced by a path gadget overlaps with an edge induced by a decoding gadget.
    \end{itemize}
    Therefore, the edges whose both end-vertices belong to the same gadget totally contribute at most four edges to any cut.
    
    After processing all pairs $(r, j)$, we add the guard $g$ at the beginning of the linear arrangement and then we add its private terminal to left of it.
     Similarly, we add the guard $g'$ at the end of the linear arrangement and then we add its private terminal to the right of it.
     This concludes the construction of a linear arrangement $\ell'$ of $G'$. 
     
     Now we will bound the \ctwdth{} of $\ell'$. 
     Here, writing about a cut, we always refer to the cuts of $\ell'$.
     We distinguish between the different types of cross-gadget edges (i.e., the edges whose end-vertices belong to different gadgets) in $G'$ and bound their number in any cut:
  \begin{itemize}
	\item There are no edges between the path gadgets in the same column. Each path gadget $X = X_{r,j}^{i,q}$ has one incident edge to the following path gadget on the same path (or to $g'$) and one incident edge to the path gadget preceding $X$ on this path (or to $g$). For a fixed path $P_{i,q}$, such edges do not overlap. There are $n'$ paths, so such edges totally contribute at most $n'$ to the size of any cut. 
	\item Each decoding gadget $Y^i_{r, j}$ contains $3^t$ vertices on the left side and each of them has two incident edges to each path gadget of the bundle it is attached to. So there are $2t3^t$ edges between the decoding gadget $Y^i_{r, j}$ and the bundle $B^i_{r, j}$ it is attached to. Since the vertices of the decoding gadget $Y^i_{r, j}$ lie directly after the vertices of the bundle $B^i_{r, j}$ in the linear arrangement, such edges incident to different decoding gadgets do not overlap.
	Therefore, such edges totally contribute at most $2t3^t$ to the size of any cut.
	\item A clause gadget $Z_{r, j}$ is a single vertex $w$ adjacent to some vertices on
	the right side of the matching edges of some decoding gadgets $Y^i_{r, j}$. 
	The clause $C_j$ contains at most $d$ literals so the variables in $C_j$ belong to at most $d$ groups of variables.
	Hence, $w$ has edges to right end vertices from at most $d$ decoding gadgets. Since any decoding gadget attached 
	contains $3^t$ matching edges, it has $3^t$ right end vertices. So $w$ is adjacent to at most $d3^t$ vertices. 
	Recall that all neighbors of $w$ belong to the column $Q_{r,j}$. 
	Since we constructed the linear arrangement ``columnwise'', the edges incident to different clause gadgets do not overlap. So such edges contribute at most $d3^t$ to the size of any cut. 
    \end{itemize}
    Altogether, we obtain the upper-bound of 
    \[
    	3 + n' + 2t3^t+d3^t = 3+n'+(2t+d)3^t
    \]
	 on the \ctwdth{} of
    $\ell'$. 
    
    Now we will extend $\ell'$ to a linear arrangement $\ell$ of $G$. Recall that the order of vertices on $\rootpath$ has not been fixed yet. First, we add the vertices in $\rootpath'$ to the linear arrangement as follows: We iterate through the vertices in $\ell'$ and for each root-connected vertex $v$, we insert its private neighbor $u = \rootpar{v}$ on $\rootpath$ 
    directly before $v$ into $\ell'$ and we insert the private terminal attached to $u$ directly before $u$. Thus, we obtain the desired linear arrangement $\ell$ of $G$.
    After that, we choose the order of the vertices on $\rootpath$ to coincide with their relative order in the constructed $\ell$.
    This way, we ensure that the edges of $\rootpath$ do not overlap on $\ell$ and contribute at most one edge to every cut. 
    Therefore, the addition of the vertices $\rootpath'$ increased the \ctwdth{} of the linear arrangement by at most 2: at most 1 by the edges from $\rootpath$ and at most one by edges from $V(\rootpath)$ to their private degree-1 terminals. So the \ctwdth{}
    of $\ell$ is at most $n' + (2t+d)3^t + 5$. 
    Note that together with $(G, T, k)$, this linear arrangement can be constructed from $I$ in polynomial time.
    \end{proof}

    \begin{lemma}
	If $I$ is satisfiable, then $G$ has a \Steinertree{} of size at most $k$ spanning all 
        vertices in $T$.
    \end{lemma}
    \begin{proof}
        For a path gadget $X$, we define the sets $A_{\scrD} = A_{\scrD}(X), A_{\scrC} = A_{\scrC}(X), A_{\scrO} = A_{\scrO}(X) \subseteq V(X)$ associated with the states $\scrC, \scrD, \scrO$ respectively as $A_{\scrD} = \{v, w_2, \verC, \verO, w'_1, v'\}$, $A_{\scrC} = \{v, w_1, \verD, \verO, w'_2, v'\}$, and $A_{\scrO} = \{u, w_1, \verD, \verC, w'_1, u'\}$.  We refer to \cref{fig:append-st-states} for graphical representation.

    \begin{figure}[t]
    \centering

    \begin{subfigure}[t]{.49\linewidth}
    \includegraphics[width=.99\linewidth]{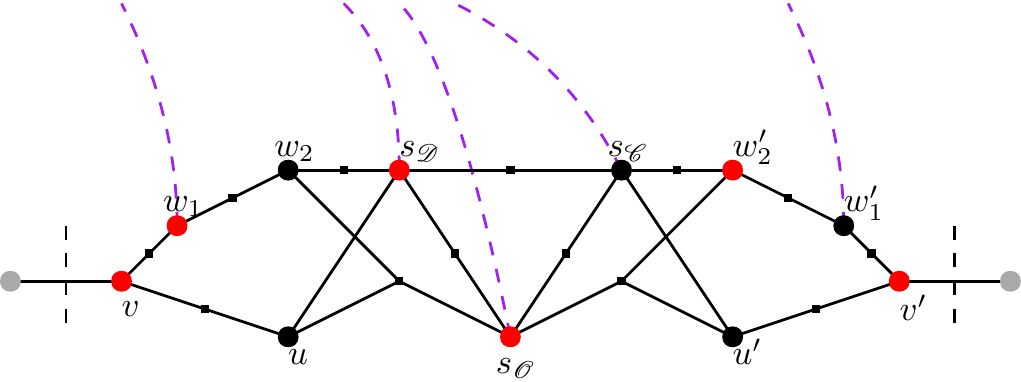}
    \caption{The set $A_{\scrC}$ corresponding to the state $\scrC$.}
    \label{app:fig-A-large-r}
    \end{subfigure}
    \hfill
    \begin{subfigure}[t]{.49\linewidth}
    \includegraphics[width=.99\linewidth]{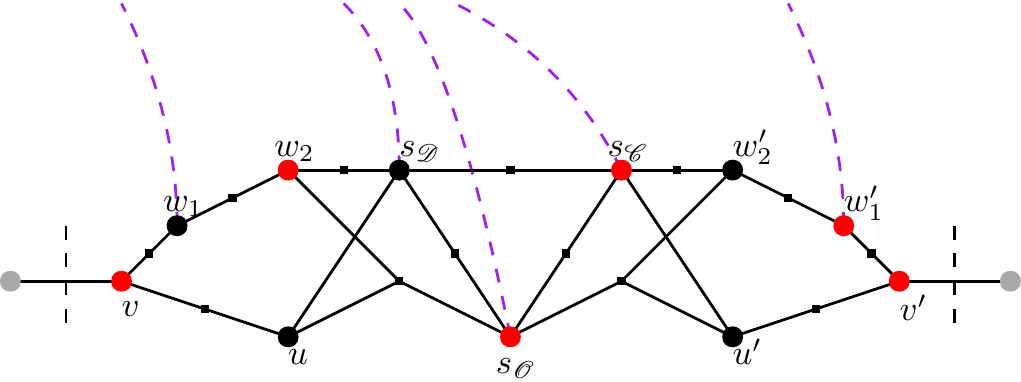}
    \caption{The set $A_{\scrD}$ corresponding to the state $\scrD$.}
    \label{app:fig-A-small-r}
    \end{subfigure}

    \begin{subfigure}[t]{.49\linewidth}
    \includegraphics[width=.99\linewidth]{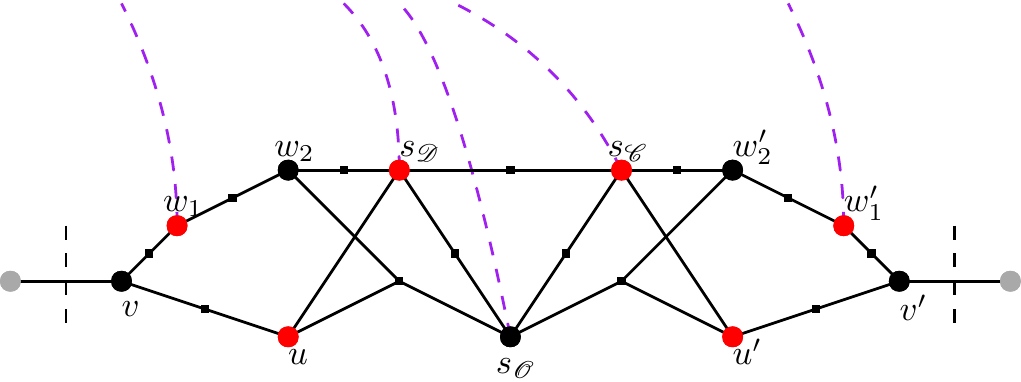}
    \caption{The set $A_{\scrO}$ corresponding to the state $\scrO$.}
    \label{app:fig-A-O}
    \end{subfigure}

    \caption{The sets $A_x, x\in\{\scrC, \scrD, \scrO\}$ are packings of vertices of a path gadget in a solution, corresponding to different state assignments. We depict the vertices in $A_x$ in red.}
    \label{fig:append-st-states}
    \end{figure}

        Let $\pi: \Var(I) \to \{0, 1\}$ be an
        assignment satisfying $I$. We will construct a Steiner tree $S$ of $G$ spanning $T$ of size $k$.
        Initially, $S$ is empty. For each $i$, we will add the following vertices from the $i$th path group and the decoding gadgets attached to it to $S$. 
        Consider the restriction $\pi_{|_{U_i}}$ of $\pi$ to the $i$th variable group $U_i$. Let $\delta^i = \Phi(\pi_{|_{U_i}})$ denote the state assignment for the path gadgets of a bundle corresponding to $\pi_{|_{U_i}}$ under $\Phi$. 
For each path $P_{i,q}$ and each path gadget $X^{i, q}_{r,j}$ on $P_{i,q}$ we proceed as follows. Let $x = \delta^i_q \in \{\scrC, \scrD, \scrO\}$ be the state assigned to this gadget in $\delta^i$.
        We add to $S$ all vertices of the set $A_x(X^{i,q}_{r,j})$ in $X^{i,q}_{r,j}$ and we say that $X$ has the state $x$.
        Note that the set $A_x$ only depends on $i$ and $q$ so the same set $A_x$ of vertices is added to $S$ for all path gadgets of the same path.
        In other words, we choose the same state for all path gadgets on the same path. 
        After that, for each $r,j$, we consider the decoding gadget $Y = Y^i_{r,j}$. For every matching edge $e_\sigma$ of $Y$, we add $u_\sigma$ to $S$, if $\sigma = \delta^i$ holds, and we add $v_\sigma$ to $S$ 
        otherwise. 
        Note that since there is one matching edge for each state assignment of $Y$, there exists exactly one matching edge $e_\sigma$ in $Y$ such that $\sigma = \delta^i$. 
        After all path and decoding gadgets are processed, we also add $V(\rootpath), g, g'$, and all terminals in the graph to $S$.

        Observe that for each component $C$ of the packing $\mathcal{P}(G)$, we have added exactly $p(C)$ vertices to $S$. So we have $k = |S|$. 
        The set $S$ contains all terminals. 
        So to prove that $S$ is the desired \Steinertree{}, it remains to show that $G[S]$ is connected. 
        Since all vertices of $\rootpath$ belong to $S$, it suffices to show that each other
        vertex in $S$ is reachable from $\rootpath$ in $G[S]$. 
        We call a vertex $v$ \emph{connected} if there is a path from $v$ to some vertex on $\rootpath$ that only uses vertices of $S$.
        This is clearly the case for the private terminals adjacent to vertices from $V(\rootpath)$. 
        This is also the case for $g$ and $g'$ and their private terminals, since the guards are root-connected. 
        This is also true for all vertices in decoding gadgets: indeed, the endpoints of
        matching edges are root-connected and for every matching edge, there is one
        end-vertex in $S$ so the subdividing terminal is also connected.

        Now consider an arbitrary path gadget $X = X_{r,j}^{i, q}$. Since the clique vertices and the vertices $w_1$ and $w'_1$ are root-connected, it suffices to prove that the remaining vertices of the gadget are connected as well. We distinguish the cases depending on the state of the gadget. 
        
        In case of $\scrC$, every vertex, other than $v'$ and the terminals adjacent to
        it, of $X \cap S$, is connected to a root-connected vertex from $S$ using only the vertices of $X \cap S$.
        This can be verified in \cref{app:fig-A-large-r}. Now we show that $v' = v'(X)$ (and hence, also the terminals
        adjacent to it) are connected too. If the gadget $X$ is the rightmost on a path, then $v'$
        is adjacent to $g'$ which is connected and in $S$. Otherwise, it is adjacent to the
        left-end vertex $v(X')$ of the next path gadget $X'$ on the same path. Since all gadgets of the same path are in the same state, $X'$ is also in the state $\scrC$. Therefore, the vertex $v(X')$ adjacent to $v'(X)$ is in $S$ and it is connected (as again can be verified in \cref{app:fig-A-large-r}). 
        
        Similarly, we consider the case of $\scrD$. 
        Every vertex, other than $v$ and the terminals adjacent to it, of $X \cap S$ is connected to a root-connected vertex from $S$ using only the vertices of $X \cap S$. This can be verified in \cref{app:fig-A-small-r}. Now we show that $v = v(X)$ (and hence, also the terminals adjacent to it) are connected too. If the gadget $X$ is the leftmost on a path, then $v$ is adjacent to $g$ which is connected and in $S$. Otherwise, it is
        adjacent to the right-end vertex $v'(X')$ of the previous path gadget $X'$ on the same path. 
        Since all gadgets of the same path are in the same state, $X'$ is also in the state $\scrD$. Therefore, the vertex $v'(X')$ adjacent to $v(X)$ is in $S$ and it is connected (as again can be verified in \cref{app:fig-A-large-r}).
        
        The last case is the state $\scrO$. Here, every vertex of $X \cap S$ is connected to a root-connected vertex from $S$ using only the vertices of $X \cap S$. This can also be verified in \cref{app:fig-A-O}.
    
        Now we show that the single terminal in a clause gadget is connected as well. 
        Consider an arbitrary clause gadget $Z_{r, j}$. Recall that it represents the clause $C_j$ and it is attached to the column $Q_{r, j}$.
        Since $\pi$ satisfies $I$, the clause $C_j$ contains a true literal, i.e., there exists an index $i^* \in [s]$ such that $\pi_{|_{U_{i^*}}}$ satisfies $C_j$. 
        Let then $\delta = \Phi(\pi_{U_{i^*}})$ be the state
        assignment of the bundle $B_{r,j}^{i^*}$ 
        and let $w$ be the terminal in $C_{r,j}$.
        Since $\pi_{|_{U_{i^*}}}$ satisfies $C_j$, there is an edge between $w$ and 
        the right end $u_{\delta}(Y^{i^*}_{r,j})$ of the matching edge corresponding to $\delta$ in the decoding gadget $Y^{i^*}_{r,j}$ of the bundle $B_{r,j}^{i^*}$. 
        By the construction of $S$ (recall $\delta = \Phi(\pi_{U_{i^*}})$), the vertex $u_{\delta}(Y^{i^*}_{r,j})$ belongs to $S$. 
        Therefore, the vertex $w$ is connected as well. Altogether, we have shown that all vertices in $S$ are connected and hence, the set $S$ is a \Steinertree{} of $G$ spanning $T$ of size $k$.
    \end{proof}

    \begin{lemma}\label{lem:st-lb-correctness-2}
        If $(G,T)$ has a \Steinertree{} of size $k$, then $I$ is satisfiable.
    \end{lemma}
    \begin{proof}
        Let $S_0$ be a \Steinertree{} of $G$ spanning $T$ of size exactly $k$. 
        Recall that by the choice of $k$, such a \Steinertree{} contains exactly $p(C)$ vertices from every component $C$ of the packing $\mathcal{P}(G)$.
        In particular, we have $V(\rootpath) \subset S$.
        Let $S = S_0 \setminus T$ be the set of non-terminals in $S_0$. 
        We will construct an assignment $\pi$ satisfying $I$. Let
        $\sigma$ be the mapping that assigns to a path gadget $X^{i, q}_{r,j}$ its state defined by $S_0$.
        We recall that for a gadget $X = X^{i, q}_{r,j}$, the set $\{\verC, \verD, \verO\}$ belongs to the packing $\mathcal{P}(G)$ and we have $p\bigl(\{\verC, \verD, \verO\}\bigr) = 2$. Therefore, there exists exactly one $x \in \{\scrC, \scrD, \scrO\}$ with $s_x(X) \notin S$ and the state of $X$ in $S$ is defined as $x$. So $\sigma$ maps every gadget $X$ to the state $x$ defined like this.
        We define the linear-ordering $\preceq$ over the states as $\scrD \preceq \scrO \preceq \scrC$. We prove the following claim. 
        \begin{claim}
        		Let $i, q$ be arbitrary but fixed. Then for any two consecutive path gadgets $X$ and $X'$ on the path $P_{i,q}$, it holds that $\sigma(X) \preceq \sigma(X')$. 
        \end{claim}
        
        \begin{proof}
        We refer to \cref{fig:append-st-states} for illustration.
        Note that the statement is equivalent to proving that $\sigma(X) = \scrC$ implies $\sigma(X') = \scrC$ and $\sigma(X') = \scrD$ implies $\sigma(X) = \scrD$.
        
        So first, let $\sigma(X) = \scrC$. By definition of $\scrC$, we have $\verC(X) \notin S$. So $w'_2(X) \in S$ holds due to the terminal between them. Since $\{w'_1, w'_2\} \in
        \mathcal{P}(G)$ and $p(\{w'_1, w'_2\}) = 1$ hold, the set $S$ does not contain the vertex $w'_1(X)$. This, in turn, implies that $v'(X) \in S$ holds because of the terminal between them. Again, due to $\{u', v'\} \in \PP(G)$ and $p(\{u', v'\}) = 1$, it holds that
        $u'(X) \notin S$. 
        So we have $u'(X), w_1'(X) \notin S$, $v'(X) \in S$, and $V(\rootpath) \subset S$. Since $S$ is a \Steinertree{}, the subgraph $G[S]$ is connected. Therefore, there is a path $P^*$ from $v'(X)$ to a vertex from $V(\rootpath)$ in $G[S]$ not using the vertices $u'(X)$ and $w_1'(X)$. Hence, the path $P^*$ goes through $v(X')$ and in particular, we have $v(X') \in S$.
        Again, since $\{u, v\} \in \PP(G)$ and $p\left(\{u, v\}\right) = 1$, we have $u(X') \notin S$. Therefore, the path $P^*$ going to some vertex on $\rootpath$ must use the vertex $w_1(X')$, i.e., $w_1(X') \in S$.
        Because of $\{w_1, w_2\} \in \PP(G)$ and $p\left(\{w_1, w_2\}\right) = 1$, we have $w_2(X') \notin S$. And again, we must have $\verD(X') \in S$ due to the terminal between them.
        Similarly, due to the terminal between $u(X'), w_2(X') \notin S$ and $\verO(X')$, we must have $\verO(X') \in S$. Finally, since $\{\verC, \verD, \verO\} \in \PP(G)$ and $p\left(\{\verC, \verD, \verO\}\right) = 2$, we have $\verC(X') \notin S$ and hence, $\sigma(X') = \scrC$ as claimed.
 
        The second part of the claim holds due to the symmetry. At \cref{fig:main-st-app}~(a) it can be seen that the gadget admits a vertical symmetry: If we reflect it this way, then the
clique vertices $\verD$ and $\verC$, non-clique vertices (with respect to apostrophe), as well as the
roles of the previous and the next gadgets get swapped. 
For the sake of completeness, we still provide the proof.
So first, let $\sigma(X') = \scrD$. By definition of $\scrD$, we have $\verD(X') \notin S$. So $w_2(X) \in S$ holds due to the terminal between them. Since $\{w_1, w_2\} \in
        \mathcal{P}(G)$ and $p(\{w_1, w_2\}) = 1$ hold, the set $S$ does not contain the vertex $w_1(X')$. This, in turn, implies that $v(X') \in S$ holds because of the terminal between them. Again, due to $\{u, v\} \in \PP(G)$ and $p(\{u, v\}) = 1$, it holds that
        $u(X') \notin S$. 
        So we have $u(X'), w_1(X') \notin S$, $v(X') \in S$, and $V(\rootpath) \subset S$. Since $S$ is a \Steinertree{}, the subgraph $G[S]$ is connected. Therefore, there is a path $P^*$ from $v(X')$ to a vertex from $V(\rootpath)$ in $G[S]$ not using the vertices $u(X')$ and $w_1(X')$. Hence, the path $P^*$ goes through $v'(X)$ and in particular, we have $v'(X) \in S$.
        Again, since $\{u', v'\} \in \PP(G)$ and $p\left(\{u', v'\}\right) = 1$, we have $u'(X) \notin S$. Therefore, the path $P^*$ going to some vertex on $\rootpath$ must use the vertex $w'_1(X)$, i.e., $w'_1(X) \in S$.
        Because of $\{w'_1, w'_2\} \in \PP(G)$ and $p\left(\{w'_1, w'_2\}\right) = 1$, we have $w'_2(X) \notin S$. And again, we must have $\verC(X) \in S$ due to the terminal between them.
        Similarly, due to the terminal between $u'(X), w'_2(X) \notin S$ and $\verO(X)$, we must have $\verO(X) \in S$. Finally, since $\{\verC, \verD, \verO\} \in \PP(G)$ and $p\left(\{\verC, \verD, \verO\}\right) = 2$, we have $\verD(X) \notin S$ and hence, $\sigma(X) = \scrD$ as claimed.
\end{proof}
        
        So, we have proven that for any path $P_{i, q}$, from left to right, the states of the gadgets on it do not decrease along $\preceq$ in $\sigma$ .
        In particular, for any path $P_{i,q}$, if we consider its gadgets in left to right order, the state in $\sigma$ can change at most twice. Since there are $n'$ such paths, there might occur at most $2n'$ such state changes in total. Therefore, there exists an index $r^* \in [2n'+1]$ such that for every $i$ and $q$, the gadgets $X^{i, q}_{r^*, 1}, \dots, X^{i ,q}_{r^*, m}$ (i.e., the gadgets of the $r^*$th path segment of the path $P_{i,q}$) are in the same state in $\sigma$. 
        Let $\sigma_{r^*}$ be the restriction of $\sigma$ to the gadgets of some fixed column $Q_{r^*, j}$ (note that by the choice of $r^*$, the assignment $\sigma_{r^*}$ is independent of $j$).
        
        For every $i$, let $\sigma^i_{r^*}$ be the restriction of $\sigma_{r^*}$ to the bundle $B^i_{r^*, j}$.
        Recall that $U_1, \dots, U_s$ is a partition of $U$. Hence, we can define a truth-value assignment $\pi$ of $U$ by defining its restrictions $\pi_{|U_i}$ for every $U_i$.
        So for a variable group $U_i$, we set $\pi_{U_i} = \phi^{-1} (\sigma^i_{r^*})$ if $\sigma^i_{r^*}$ is in the image of $\phi$ (recall that $\phi$ is injective), or $\pi$ assigns the
        value $0$ to all variables in $U_i$ otherwise.
    
        Now we show that $\pi$ satisfies $I$. Consider an arbitrary clause $C_j$. Since the terminal $t$ in the clause gadget $Z_{r, j}$ must be connected to the remaining vertices of the \Steinertree{} $S$, there exists a decoding gadget $Y = Y^{\widetilde{i}}_{r^*, j}$ for some $\widetilde i \in[s]$ and a matching edge $\widetilde{e}$ in it such that for the right side side $\widetilde{u}$ of $\widetilde{e}$ we have that $\{\widetilde{u},t\}$ is an edge of $G$ and $\widetilde{u} \in S$ holds. Recall that for the left side $\widetilde{v}$ of $\widetilde{e}$ we have that $\{\widetilde{u}, \widetilde{v}\}$ belongs to the packing $\PP(G)$ with $p\bigl(\{\widetilde{u}, \widetilde{v}\}\bigr) = 1$. 
        Therefore, it holds that $\widetilde{v} \notin S$. 
        Let $\widetilde{\sigma}$ be the state assignment of the bundle $B = B^{\widetilde{i}}_{r^*, j}$ such that $e_{\widetilde{\sigma}} = \widetilde{e}$. Suppose that $\widetilde{\sigma} \neq \sigma_{|_B}$. So there exists an index $\widetilde{q} \in [t]$ such that $\widetilde{\sigma}(X^{\widetilde{i}, \widetilde{q}}_{r^*, j}) \neq \sigma_{|_B}(X^{\widetilde{i}, \widetilde{q}}_{r^*, j})$. 
        Let $X = X^{\widetilde{i}, \widetilde{q}}_{r^*, j}$ and let $x = \sigma(X) \in \{\scrC, \scrD, \scrO\}$.
        Then by definition of $\sigma$, we have that $s_{x}(X) \notin S$. 
        Because of $\widetilde{\sigma}(X) \neq x$, by construction of a decoding gadget, the left side $v_{\widetilde{\sigma}}$ has an incident edge subdivided by a terminal $\widetilde{t}$ with the other end-vertex being $s_x$. But then since both $s_x$ and $v_{\widetilde{\sigma}} = \widetilde{v}$ are not in $S$, the terminal $\widetilde{t}$ is isolated in $G[S]$ -- a contradiction to the fact that $S$ is a \Steinertree{}. So it holds that $\widetilde{\sigma}=\sigma_{|B}$.
        Recall that by construction of a clause gadget $t$, the existence of an edge $\{u_{\sigma_{|_B}},t\}$ implies that $\sigma_{|_B}$ is in the image of $\Phi$ and $\Phi^{-1}(\sigma_{|_B})$ satisfies $C_j$. Finally, recall that by construction of $\pi$, we have $\pi_{U_{\widetilde{i}}} = \Phi^{-1}(\sigma_{|_B})$. Therefore, $\pi$ satisfies $C_j$.
        Altogether, we obtain that $\pi$ is an assignment satisfying $I$.
    \end{proof}
   
\begin{proof}[Proof. (\cref{app:theo-st})]
  Suppose there exists an algorithm $\mathcal A$ that given a linear arrangement of the input graph of cutwidth $\ctw$ solves \Pst{} in time $\ostar((3-\varepsilon)^{\ctw})$ 
  for some positive real $\varepsilon$. 
  We show that there exists a positive real $\delta$ such that for any positive integer $d$, there exists an algorithm solving the $d$-SAT problem in time $\ostar\bigl((2-\delta)^n\bigr)$ contradicting SETH. 
  Note that $\log_3 (3-\varepsilon) < 1$ holds. 
  Let $0 < \sigma < 1$ be such that $1-\sigma = \log_3 (3-\varepsilon)$. 
  We choose $t_0 \in \NN$ to be some constant \emph{integer} strictly greater than $\frac{1-\sigma}{\sigma \log_3 2}$ (for example the smallest such integer). This is possible, since $\sigma$ is constant as well. For this value of $t_0$, we construct the instance $(G, T, k)$ and a linear arrangement of $G$ of cutwidth $n' + \mathcal{O}(1)$ as described above in polynomial time. 
  We apply $\mathcal A$ to $(G, T, k)$ and output its result. 
  Above, we have shown that the graph $G$ admits
  a \Steinertree{} of size at most $k$ spanning $T$ if and only if $I$ is a yes-instance so this algorithm is correct.
  Now we bound the running time
  of the algorithm. 
  It holds that
  \begin{align*}
    n' &= st = \left\lceil \frac{n}{t_0}\right\rceil t \leq \left(\frac{n}{t_0} +1\right)t \leq \frac{n t}{t_0} + t\\
    &\leq n\frac{ \lceil t_0 \log_3 2\rceil}{t_0} + t \leq n \frac{t_0 \log_3 2 + 1 }{t_0} + t\\
    &= n (\log_3 2 + 1/t_0) + t
  \end{align*}
  Hence we get an algorithm for the $d$-SAT problem running in time
  \[\ostar((3-\varepsilon)^{n(\log_3 2 +1/t_0) + t + \mathcal{O}(1)}) =
  \ostar((3-\varepsilon)^{n(\log_3 2 + 1/t_0)}).\]
  We claim that there is a positive value $\delta$ such that
  \[
  (3-\varepsilon)^{n(\log_3 2 + 1/t_0)} \leq (2-\delta)^n. 
  \]
  This would contradict SETH and hence, complete our proof. In order
  to prove this we show that $(3-\varepsilon)^{\log_3 2 + 1/t_0} < 2$ or in other words $(\log_3 2 + 1/t_0)
  \log_3(3-\varepsilon) < \log_3 2$. We have
  \begin{align*}
    &(\log_3 2 + 1/t_0) \log_3(3-\varepsilon) < \left( \log_3 2 + \frac{\sigma \log_3 2}{1-\sigma}\right) (1-\sigma)\\
    &= \frac{\log_3 2 \cdot (1-\sigma)+ \sigma \log_3 2}{1-\sigma}(1-\sigma) = \log_3 2,
  \end{align*}
  where the first strict inequality follows from the value of $t_0$.
\end{proof}

\subsection{Connected Dominating Set}\label{app:cds-lb}
In this section, we show that the \Pcds{} problem cannot be solved in time $\ostar\bigl((3-\varepsilon)^{\ctw}\bigr)$. This proves that our
upper bound of $\ostar(3^{\ctw})$ from \cref{app:cds} is essentially optimal. To achieve this, we slightly modify the reduction for the \Pst{} problem from \cref{app:theo-st} and show
that this results in a reduction for the \Pcds{} problem.
A \emph{dominating set} of a graph is a set of vertices such that for every vertex of the graph not in this set, at least one of its neighbors belongs to this set.

\begin{quote}
	\Pcds{}

	\textbf{Input}: A graph $G = (V, E)$ and an integer $k$.

	\textbf{Question}: Is there a dominating set $S \subseteq V$ of $G$ of cardinality at most $k$ such that $G[S]$ is connected.
\end{quote}

\begin{lemma}\label{lem:cdstost}
	Let $d$ be a positive integer and let $I$ be an instance of the $d$-\textsc{SAT} problem.
	Let $\bigl(G = (V, E), T, k\bigr)$ be the instance of \Pst{} resulting from $I$ using the construction from \cref{app:theo-st}. 
	The graph $\hat{G} = (V, \hat{E})$ on the same vertex set as $G$ is defined similarly to $G$ as follows.
	We proceed the same way as in the construction from \cref{app:theo-st} but:
	\begin{enumerate}
		\item Each time we split some edge $\{u, v\}$ by some terminal $w \in T$, we additionally add the edge $\{u, v\}$ so that we obtain a triangle $u, v, w$. \label{hatg:first}
		\item In every path gadget, we create the triangles $\verO, w_2, u$ and $\verO, w'_2, u'$ by adding the new edges $\{\verO, w_2\}$, $\{w_2, u\}$, $\{u, \verO\}$ and $\{\verO, w'_2\}$, $\{w'_2, u'\}$, $\{u', \verO\}$ .
		\label{hatg:second}
	\end{enumerate}
	Then the graph $\hat G$ admits a \cds{} of size at most $k - |T|$ if and only if $G$ admits a \Steinertree{} spanning $T$ of size at most $k$. 
	So the instance $\bigl(\hat G, k - |T|\bigr)$ of \pcds{} is equivalent to $I$ as well.
	Moreover, the graph $\hat G$ together with a linear arrangement of $\hat G$ of cutwidth at most $n' + \mathcal{O}(1)$ can be computed in polynomial time from $I$ (here $n'$ is the number of paths as in the construction of \cref{app:theo-st}).
\end{lemma}
\begin{proof}
	First of all, note that the definition of \pcds{} does not contain terminals, i.e., they are just regular vertices. 
	Nevertheless, we refer to vertices from $T$ as \emph{terminals} to simplify the notation in the proof.

	Next, note that $\hat G$ can also be constructed from $I$ in polynomial time. 
	Also observe that $\hat G$ results from $G$ merely by adding edges. First, let $S$ be a \emph{minimum} connected dominating set of $\hat G$ of size  at most $k - |T|$. 
	First, we claim that $S$ cannot contain vertices in $T$. 
	To prove this, suppose there is a vertex $w \in S\cap T$. 
	Recall that no two terminals are adjacent in $G$ or $\hat G$. 
	So since $\hat G[S]$ is connected, the set $S$ must contain some neighbor $v \in S \setminus T$ of $w$.
	We show that $S \setminus \{w\}$ is a \cds{} of $\hat G$ contradicting the assumption that $S$ is minimum \cds{}. 
	First, recall that by construction of the graphs $G$ and $\hat G$ the neighborhood $N_G(w) = N_{\hat G}(w)$ forms a clique of size one or two so every vertex in $N_{\hat G}(w)$ is dominated by $v$ and $S \setminus \{w\}$ is a dominating set of $\hat G$.
	It remains to prove that $\hat G\bigl[S \setminus \{w\}\bigr]$ is still connected. 
	We distinguish the cases depending on the degree of $w$ in $G$.
	If $w$ has degree one, then removing it from $\hat G[S]$ does not disconnect the graph. 
	If $w$ has degree two in $\hat G$, then this terminal was added to subdivide an edge $uv$ of $\hat G$ for some vertex $u$ and by the construction of $\hat G$, this graph also contains the edge $uv$. Thus, every path in $\hat G[S]$ using $w$ and not ending in $w$ uses a subpath $u, w, v$ or $v, w, u$. This subpath can be replaced by the edge $uv$ to obtain a path in $\hat G\bigl[S \setminus \{w\}\bigr]$ with the same end-vertices. 
	Otherwise, $w$ has degree three in $G$ (recall that by construction, no terminal has degree greater than three) so it is a terminal adjacent to the vertices $x, y, z$ of some path gadget where either $\{x, y, z\} = \{\mathcal{O}, w_2, u\}$ or $\{x, y, z\} = \{\mathcal{O}, w'_2, u'\}$ holds. 
	Further, recall that by the construction of $\hat G$, the vertices form a triangle.
	Now, an argument similar to the previous case can be applied as follows. Every path in $\hat G[S]$ using $w$ and not ending in $w$ uses a subpath $a, w, b$ for some $a \neq b \in \{x, y, z\}$. This subpath can be replaced by the edge $\{a, b\}$ to obtain a path in $\hat G\bigl[S \setminus \{w\}\bigr]$ with the same end-vertices. 
	So in every case, we obtain that $S \setminus \{w\}$ is a smaller \cds{} of $\hat{G}$ than $S$ -- a contradiction to the minimality of $S$. Therefore, it holds that $S \cap T = \emptyset$.
	
	Now we claim that $S \cup T$ is a \Steinertree{} of $G$. 
	Clearly it contains all terminals. We prove that $G[S \cup T]$ is connected. 
	Since $S$ is a dominating set of $\hat G$ with $S \cap T = \emptyset$ and for every vertex $t \in T$, we have $N_G(t) = N_{\hat G}(t)$, we obtain that $T \subseteq N(S)$, i.e., every vertex of $T$ has a neighbor in $S$ in the graph $G[S \cup T]$. Therefore, to show that $G[S \cup T]$ is connected, it suffices to prove that for all $s_1, s_2 \in S$, there is a path between them in $G[S \cup T]$.
	Recall that $\hat G[S]$ is connected. So to show that $G[S \cup T]$ is connected, it suffices to show that for any $s_1, s_2 \in S$ that are adjacent in $\hat G[S]$ (i.e., adjacent in $\hat G$), there is a path in $G[S \cup T]$ between them.
	This is trivially true, if $s_1$ and $s_2$ are adjacent in $G$. Otherwise, we have that $\{s_1, s_2\}$ is an edge of $\hat G$ but not of $G$. By construction, this is only possible if there is a terminal $w \in T$ adjacent to both $s_1$ and $s_2$. Then $s_1, w, s_2$ is a path between $s_1$ and $s_2$ in $G[S \cup T]$. So $S \cup T$ is indeed connected and therefore, it is a \Steinertree{} of $G$ spanning $T$. Moreover, we have
	\[
		|S \cup T| = |S| + |T| \leq \bigl(k - |T|\bigr) + |T| = k.
	\]
	So $S \cup T$ is a \Steinertree{} of the desired size.

	For the other direction consider \Steinertree{} $S'$ of $G$ spanning $T$ of size at most $k$ in $G$.
	We claim that $S = S' \setminus T$ is a \cds{} of $\hat G$. 
	For shortness, writing about a domination we always refer to the domination in the graph $\hat G$.
	First, we prove that $S$ is a dominating set. 
	We know that $G[S']$ is connected and it holds that $T \subseteq S'$. Since no two terminals are adjacent in $G$, each terminal has a neighbor in $S$ in the graph $G[S']$. 
	Since all edges present in $G$ are also present in $\hat G$, all vertices of $T$ are dominated by $S$. 
	By the construction of $\hat G$, for several types of terminals, the neighborhood of such a terminal forms a clique. Therefore, if such a terminal is dominated by $S$, then each of its neighbor is dominated as well. 
	First, such terminals are those added to subdivide some edge. And second, these are the terminals of degree 3 in path gadgets. Therefore, every vertex of $\hat G$ adjacent to at least one such terminal is dominated by $S$. It remains to prove that the remainder is dominated by $S$ as well.
	Further, recall that every vertex of $\rootpath$ has an adjacent terminal of degree 1 and hence, all vertices of $\rootpath$ belong to every \Steinertree{}, i.e., they also belong to $S$ and therefore, are dominated by $S$.
	Similarly, the guards $g$ and $g'$ also belong to $S$ due to their private terminals so they are dominated as well.
	And by construction, these cases already cover all vertices of $\hat G$. Therefore, $S$ is indeed a dominating set of $\hat G$.
	
	Now we prove that the graph $\hat G [S]$ is connected. 
	Recall that all vertices of the path $\rootpath$ belong to $S$.
	We show that each vertex in $S$ is connected to some vertex on $\rootpath$ in $G[T \cup S] = G[S']$ through a path $P$ that does not intersect clause gadgets. 
	Clearly, such path can not use a degree one terminal. 
	The remaining terminals (i.e., having degree larger than one and not belonging to any clause gadget) are of two types: either a terminal has been introduced to subdivide some edge or it is a terminal of degree three in some path gadget.
	In both cases, the neighborhood of such a terminal $t$ forms a clique in $\hat G$ and every subpath $a, t, b$ of a path in $G[S']$ can be replaced with the edge $\{a, b\}$ in $\hat G[S]$.
	So it indeed suffices to show that each vertex in $S$ is connected to some vertex on $\rootpath$ in $G[T \cup S]$ through a path $P$ that does not intersect clause gadgets. 
	Clearly, this is the case for all vertices of $\rootpath$ and all root-connected vertices of $S$.
	Next, let $x$ be a vertex in a path gadget that is not root-connected, and let $P$ be a path from $x$ to some vertex on $\rootpath$ in $G[S']$ (the path $P$ exists since $G[S']$ is connected and it contains all vertices of $\rootpath$). 
	If $P$ does not visit a clause-gadget, then we are done. 
	Otherwise, $P$ passes through a vertex that is a clause gadget. Let $w$ be the first such vertex on $P$ (closest to $x$). 
	By the construction of a clause gadget, the vertex $w$ must be preceded by a vertex $u = u_{\sigma}(Y)$ on the path $P$ for some state assignment $\sigma$ and a decoding gadget $Y$, since $w$ is only adjacent to such vertices. But since $u$ is root connected, we can define the path $P' = P[x, u]$, i.e. the subpath of $P$ from $x$ to $u$. By appending the private neighbor $b$ of $u$ on $\rootpath$ to $P'$, we get a path in $G[S']$ from $x$ to the vertex $b$ on $\rootpath$ that does not use any clause gadget.
	Hence, each vertex in $S$ is connected to some vertex on $\rootpath$ in $G[S']$ through a path $P$ that does not intersect clause gadgets and hence (as argued above), $\hat G[S]$ is connected. Therefore, $S$ is a connected dominating set of $\hat G$. And we also have
	\[
		|S| = |S'| - |T| \leq k - |T|.
	\]
	So $S$ is indeed the desired connected dominating set of $G$.
	
	Let $\ell$ be the linear arrangement of $G$ constructed in the proof of \cref{lem:app-st-ctw}. Recall that $V(G) = V(G')$ so $\ell$ is also a linear arrangement of $\hat G$. Recall that $\ell$ can be constructed from $I$ in polynomial time.
	Now we show that the cutwidth of $\hat G$ on $\ell$ is only larger by some constant than the cutwidth of $G$ on $\ell$.
	Observe that every edge of $E(\hat G) \setminus E(G)$ is of one of the following three types: either both end-vertices belong to the same path gadget, or to the same decoding gadget, or one end-vertex belongs to a path gadget and the other to a decoding gadget in the same column. 
	Since the vertices of distinct path gadgets are not interleaved in $\ell$ and each path gadget contains only a constant edges of $E(\hat G) \setminus E(G)$, the number of edges of each cut in $\ell$ is totally increased by at most a constant due to the edges of the first type. 
	The analogous argument applies to the edges of the second type since the vertices of distinct decoding gadgets are not interleaved in $\ell$ and every decoding gadget only has a constant (recall that $t$ in the construction is a constant) number of vertices.
	Similarly, for any edge of the third type, both end-vertices belong to the same column. Vertices of distinct columns are not interleaved in $\ell$ and every clause gadget has a constant (recall that $d$ and $t$ are constant in the construction) degree. Therefore, the size of every cut again increases by at most a constant. Recall that the cutwidth of $G$ on $\ell$ is $n' + \mathcal{O}(1)$ so the cutwidth of $\hat G$ on $\ell$ is $n' + \mathcal{O}(1)$ as well.
\end{proof}

\begin{theorem}
	Assuming SETH, there is no algorithm that solves the \Pcds{} problem in time $\ostar((3-\varepsilon)^{\ctw})$ for
    any positive real $\varepsilon$, even when a linear arrangement $\ell$ of the this width of the graph is given.
\end{theorem}
\begin{proof}
	Assume $\mathcal A$ an algorithm that solves \Pcds{} in time
  $\ostar((3-\varepsilon)^{\ctw})$ for some positive real $\varepsilon$. We show that there exists a
  positive real $\delta$ such that for every positive integer $d$, there exists an algorithm solving
  the $d$-SAT problem in time $\ostar((2-\delta)^n)$. 
  This would contradict SETH. Let $d$ be arbitrary but fixed and let $t_0$ be the value depending on $d$ as given by the construction in \cref{app:theo-st}. 
  Let $I$ be an arbitrary instance of the $d$-SAT problem. We compute $\hat G$, $T$, and $k$ as described above. 
  We run the algorithm $\mathcal A$ to decide whether $\hat G$ admits a connected dominating set of size $k + |T|$ and output its answer. 
  By \cref{lem:cdstost} and \cref{append:st}, this algorithm outputs yes if and only if $I$ is satisfiable, i.e., it solves the $d$-SAT problem.
  This algorithm has the running time of $\ostar((3-\varepsilon)^{n' + \mathcal{O}(1)})$.   As shown in \cref{app:theo-st}, this can be upper bounded by $\ostar((2-\delta)^n)$ for some positive value $\delta$, contradicting SETH.
\end{proof}

\subsection{Connected Vertex Cover}
\label{app:cvc-lb} 
We show that the \Pcvc{} problem cannot be solved in time $\ostar\big((2-\varepsilon)^{\ctw}\big)$ for any positive value $\varepsilon$ assuming SETH. In \cref{thm:cvc-ub} we have provided an algorithm solving this problem in $\ostar(2^{\ctw})$ so this bound is essentially tight. 
A \emph{vertex cover} of a graph is a set of vertices such that every edge of the graph is incident to some vertex in this set.

\begin{quote}
	\textsc{Connected Vertex Cover}

	\textbf{Input}: A graph $G = (V, E)$ and an integer $k$.
	
	\textbf{Question}: Is there a vertex cover $S \subseteq V$ of cardinality at most $k$ such that $G[S]$ is connected.
\end{quote}

\begin{theorem}\label{theo:cvc}
	Assuming SETH, there is no algorithm that solves the \Pcvc{} problem in time $\ostar((2-\varepsilon)^{\ctw})$ for
    any positive real $\varepsilon$, even if a linear arrangement of $G$ of this width is given.
\end{theorem}
\subparagraph*{Construction.}

    Let $d \in \NN$ be arbitrary but fixed and let $I$ be an instance of $d$-SAT.
    We may assume that every clause of $I$ contains exactly $d$ literals. Let $v_1, \dots, v_n$ be the variables and let $C_1, \dots, C_m$ be the clauses in $I$. For all $j \in [m]$ and $z \in [d]$, let $b_{j, z}$ denote the $z$th literal in $C_j$ and let $v_{j,z}$ denote the variable in this literal. Let $s = m \cdot d$. 
    In the following, we will use the following indices: $i \in [n]$, $r \in [n+1]$, $j \in [m]$, $h \in \{0, 1\}$, and $z \in [d]$. These indices will always belong to these intervals so for shortness, we omit the domain when we use any of these indices.
    For illustration of our construction, we refer to
    \cref{fig:cvc-lb}. We construct an equivalent instance $(G = (V, E), k)$ of the \Pcvc\ problem as follows.
    Initially $G$ is an empty graph. For each variable $v_i$, we add to $G$ one path 
    \begin{alignat*}{3}
    	P_i:=\: &u^i_{1,1,0}, u^i_{1,1,1}, u^i_{1,2,0},& &\dots,& &u^i_{1,m,1},\\
                &u^i_{2,1,0}, u^i_{2,1,1}, u^i_{2,2,0},& &\dots,& &u^i_{2,m,1},\\
                & & &\quad\vdots& & \\
                &u^i_{n+1,1,0}, u^i_{n+1,1,1}, u^i_{n+1,2,0}, & &\dots, & & u^i_{n+1,m,1}
    \end{alignat*}
consisting of $2m(n+1)$ new vertices. More formally, the path $P_i$ consists of vertices $u^i_{r, j, h}$ so that the
indices $(r, j, h)$ appear in lexicographic order. We partition each path into $m(n+1)$ pairs of
consecutive vertices, i.e., the $(r, j)$th \emph{pair} of the $i$th path consists of the vertices
$u^i_{r, j, 0}$ and $u^i_{r, j, 1}$. These are the \emph{first} vertex and the \emph{second} vertex
of a pair respectively.
    
    For each $j$ and $r$, we create a clique $Q_{r,j} = \{q_{r,j}^1, \dots, q_{r,j}^d\}$ consisting
    of $d$ new vertices. The clique $Q_{r, j}$ will only be adjacent to vertices from the $(r, j)$th
    pairs of some of the paths $P_1, \dots, P_n$. Namely, for each $j$, $r$, and $z$, we do the following. Recall that
    $v_{j,z}$ is the variable in the $z$th literal $b_{j, z}$ of $j$th clause $C_j$.
    Let $i$ be such that $v_i = v_{j, z}$.
    We add an edge from $q_{r,j}^z$ to $u^{i}_{r,j,0}$ if $b_{j, z}$ is a
    negative literal of $v_{j, z}$ and to $u^{i}_{r,j,1}$ otherwise.

    Next, we add a path $\rootpath$ consisting of $(n+1)m(2n+d)$ new vertices denoted with
   \[
   	\rootpath = w_{1,1,1}, \dots w_{1,1,2n+d}, w_{1,2,1}, \dots, 
   	w_{1, m, 2n+d}, \dots, w_{n+1, 1, 1}, \dots, 
   	w_{n+1,m,2n+d}.
   \]
   More formally, the vertices of $\rootpath$ are $w_{r, j, p}$ (for $p \in [2n+d]$) so that
    indices $(r, j, p)$ appear in lexicographic order. We call $\rootpath$ the \emph{root-path}.
    For each vertex $w_{r, j, p}$ of this path, we add a new vertex $w'_{r, j, p}$ which is adjacent
    only to $w_{r,j,p}$. Let $\rootpath'$ be the set of all such vertices $w'_{r, j, p}$. Now every
    connected vertex cover of the arising graph necessarily contains all vertices of $\rootpath$.
    For all $r$ and $j$ we do the following. First, for every $i$, we add an edge between $w_{r, j,
    2i-1}$ and $u^i_{r, j, 0}$, and an edge between $w_{r,j,2i}$ and $u^i_{r,j,1}$.
    Next, for every $z$, we add an edge between $w_{r,j,2n+z}$ and $q_{r,j,z}$.
    This concludes the construction of the graph $G$ (see \cref{fig:cvc-lb} for illustration).
    Clearly, the instance $(G, k)$ can be constructed in polynomial time.

\begin{figure}[t]
    \centering
    \includegraphics{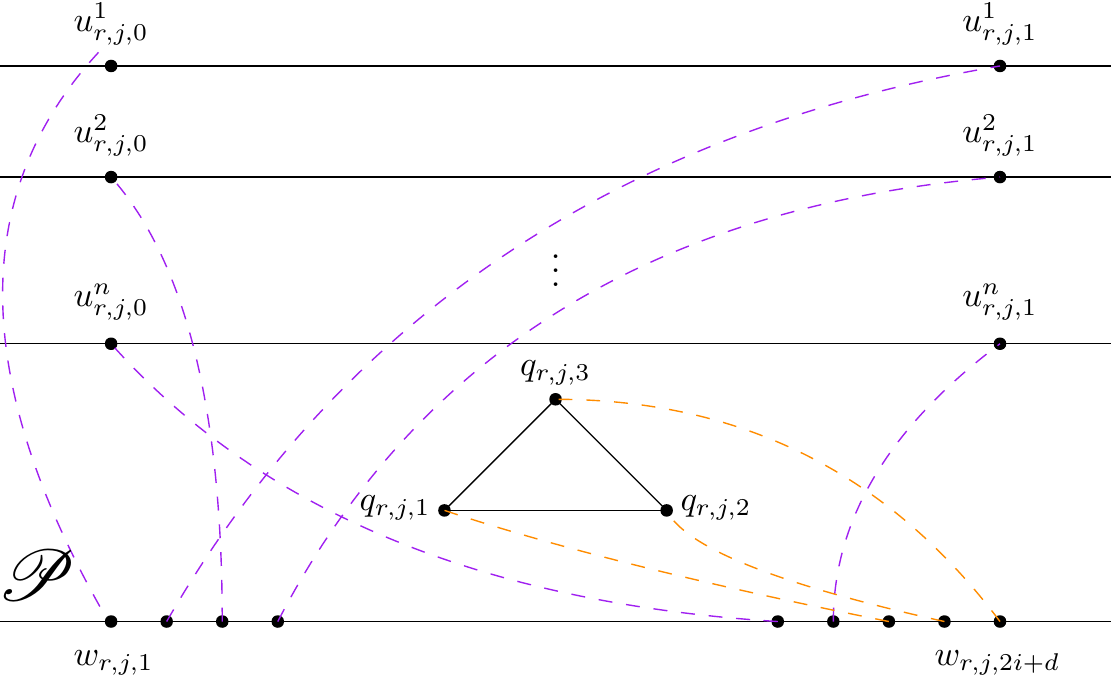}
    \caption{
	A part of the graph $G$ for a fixed pair of indices $r, j$, here $d = 3$.
    }
    \label{fig:cvc-lb}
\end{figure}

\subparagraph*{Budget.}    
    Recall that every connected vertex cover contains all vertices of $\rootpath$. Further, it
    contains at least half the vertices, i.e., at least $m(n+1)$ vertices, from a path $P_i$ for
    every $i$. Finally, it also contains at least $d - 1$ vertices of the clique $Q_{r,j}$ for every
    $r$ and $j$. Hence, any connected vertex cover of $G$ contains at least 
    \begin{equation}\label{eq:lb-on-cvc}
    	k := m(n+1)(2n+d) + m(n+1)n + m(n+1)(d-1)
    \end{equation}
    vertices. We fix the \emph{budget} (i.e., the size of the desired connected vertex cover) $k$ to this value matching the given lower bound. Hence, any connected
    vertex cover of size exactly $k$ consists of the following vertices: all vertices of $V(\rootpath)$, half of the vertices from each path $P_i$, and $d-1$ vertices from each clique $Q_{r,j}$.
    Clearly, the instance $(G, k)$ can be constructed in polynomial time.
    
 \begin{lemma}\label{lem:cvc-ctw}
    Let $(G, k)$ be the instance of \Pcvc\ arising from $I$ as described above.
    Then $G$ has the \ctwdth{} of at most $n + \mathcal{O}(1)$.
    Moreover, a linear arrangement of $G$ of this cutwidth can be constructed in polynomial time from $I$.
\end{lemma}
\begin{proof}
    We emphasize that $d$ is a constant.
    We describe a linear arrangement $\ell$ of $G$ of \ctwdth{} at most $n + \binom{d}{2} + d + 2 = n + \mathcal{O}(1)$. Recall that a linear arrangement of a graph is just a total ordering of its vertices. The linear arrangement $\ell$ initially consists of $(n+1)$ consecutive segments, and each of them consists of $m$ consecutive subsegments.
    By that we mean that for every $r < n+1$ the vertices of the $r$th segment appear before the vertices of the $(r+1)$st segment; and for every $j < m$, the vertices of the $j$th subsegment of some segment appear before the vertices of the $(j+1)$st segment of the same segment.
    For fixed $r$ and $j$, the following vertices belong to the $j$th subsegment of the
    $r$th segment in order from left to right. First, it contains all vertices $u^i_{r,j,0}$ and
    $u^i_{r,j,1}$ in the order of increasing $i \in [n]$ (and $u^i_{r,j,0}$ occurs before
    $u^i_{r,j,1}$). After that, all vertices of $Q_{r,j}$ occur in arbitrary order.

    After all segments are created, we insert the vertices of $V(\rootpath) \cup V(\rootpath')$ into the linear arrangement as follows.
    For every vertex $v$ of $G \setminus (\rootpath \cup
    \rootpath')$, let $w \in \rootpath$ denote the unique neighbor of $v$ in $\rootpath$ and
    let $w' \in \rootpath'$ be the unique neighbor of $w$ in $\rootpath'$. We insert $w$
    directly before $v$ and $w'$ directly before $w$ into our linear arrangement.
    
    Now we prove that $\ell$ has the cutwidth of at most $n + \binom{d}{2} + d + 2$. 
    Here, writing about a cut, we always refer to the cuts of $\ell$.
    For each
    pair $u^i_{r,j,0}, u^i_{r,j,1}$ of vertices, there is at most one incident edge to the previous
    subsegment and at most one incident edge to the following subsegment (if such subsegments
    exist). So totally, there are at most $n$ edges between any two consecutive subsegments that
    belong to any of the paths $P_1, \dots, P_n$, and there are no such edges between non-consecutive subsegments. Further,
    all edges of paths $P_i$ whose both end-vertices belong to the same subsegment do not overlap.
    More formally, for any fixed $r$ and $j$, no two distinct edges from $\{u^1_{r,j,0},
    u^1_{r,j,1}\}$, $\{u^2_{r,j,0}, u^2_{r,j,1}\}$, $\dots$, $\{u^n_{r,j,0}, u^n_{r,j,1}\}$ overlap
    and all of these edges belong to the $j$th subsegment of the $r$th segment. Therefore, the edges
    of all paths $P_i$ contribute at most $n$ edges to every cut of the linear arrangement. Each
    clique $Q_{r,j}$ has $\binom{d}{2}$ edges inside the clique and $d$ incident edges to some of
    the paths $P_i$. This contributes at most $\binom{d}{2} + d$ edges to every cut:  
    All end-vertices of the edges incident to a clique $Q_{r,j}$ belong to the $j$th subsegment of the $r$th segment, therefore no two edges incident to different cliques overlap.
    By construction, the vertices of $\rootpath$ appear in the same order in
    $\rootpath$ and on $\ell$. Hence, the edges of $\rootpath$ do not overlap and they
    contribute at most one edge to every cut. Finally, the edges incident to $\rootpath'$ do not
    overlap as well so they also contribute at most one edge to every cut. Altogether, the
    linear arrangement $\ell$ has the cutwidth of at most $n + \binom{d}{2} + d + 2$. 
    Also note that together with $(G, k)$, this linear arrangement can be constructed from $I$ in polynomial time.    
\end{proof}

\begin{lemma}
    If $I$ is satisfiable, then $G$ has a \cvc{} of size at most $k$ spanning all 
    the vertices from $T$.
\end{lemma}
\begin{proof}
    Let $\pi \colon \{v_1, \dots, v_n\} \to \{0, 1\}$ be an assignment satisfying $I$. We construct
    a set $S \subseteq V$ as follows. For each $i$, we add all vertices $u_{r,j,\pi(v_i)}^i$ for all
    $r$ and $j$ to $S$. Since $\pi$ is a satisfying assignment, every clause contains a true
    literal, i.e., for every $j$, there exists an index $z^* \in [d]$ such that $\pi(b_{j, z^*}) =
    1$ holds. Then for all $r$ and all $z \neq z^* \in [d]$, we add the vertex $q_{r,j}^{z}$ to $S$.
    We also add all vertices in $V(\rootpath)$ to $S$. It is easy to verify that $|S| = k$ holds: we
    have taken $m(n+1)$ vertices from each path $P_i$, $d-1$ vertices from each clique $Q_{r,j}$,
    and all $m(n+1)(2n + d)$ vertices of $\rootpath$.
    
    Now we prove that $S$ is a connected vertex cover of $G$. Since all vertices of the path
    $\rootpath$ are in $S$, and every other vertex in the graph is adjacent to a vertex in
    $\rootpath$, the subgraph $G[S]$ is connected. So it remains to prove that $S$ is a vertex
    cover of $G$. Again, since all vertices of $\rootpath$ are in $S$, all edges incident to
    $V(\rootpath)$ are covered. Further, from each path $P_i$ we have picked every second vertex,
    hence the edges of $P_1, \dots, P_n$ are covered as well. From each clique $Q_{r,j}$, we have
    picked all but one vertex. Therefore, for any fixed $r,j$, all edges incident to vertices in $Q_{r,j}$ are covered by the
    vertices in $Q_{r, j}$, except for the edge between $q = q_{r,j}^{z^*}$ and a vertex $u = u_{r,
    j, x}^i$, where $v_i = v_{j, z^*}$ and $x \in \{0, 1\}$ depending on whether $b_{j, z^*}$ is a
    positive or a negative literal. Recall that $\pi$ satisfies $b_{j, z^*}$ by the choice of $z^*$.
    First assume that the vertex $q_{r,j}^{z^*}$ corresponds to a negative literal $b_{j, z^*}$ of a
    variable $v_i$. Then, it holds that $x = 0$ (by construction) and by the choice of $z^*$, we
    have $\pi(v_i) = 0$. So by construction, it holds that $u_{r,j,0}^{i} \in S$ and hence, the edge
    $\{q, u\}$ is covered. If $q_{r,j}^{z^*}$ corresponds to the positive literal $b_{j, z^*}$ of a
    variable $v_i$, then $x = 1$, and $\pi(v_i) = 1$. Hence, $u_{r,j,1}^i \in S$ and hence, $\{q,
    u\}$ is again covered. So the set $S$ is indeed a connected vertex cover of $G$ of size $k$.
\end{proof}

\begin{lemma}
    If $G$ has a \cvc{} of size $k$, then $I$ is satisfiable.
\end{lemma}
\begin{proof}
    Let $S$ be a connected vertex cover of $G$ of size $k$. We will construct an assignment $\pi$
    satisfying $I$. Since the size of $S$ is equal to the lower bound on the size of a connected
    vertex cover of $G$ provided in \eqref{eq:lb-on-cvc}, $S$ must match this bound exactly for each
    part of $G$. In particular, it holds that $V(\rootpath)\subseteq S$, and $V(\rootpath') \cap S = \emptyset$. 
    Then from every 
    path $P_i$ exactly $m(n+1)$ vertices, i.e.\ the half of the vertices of the path, belong to
    $S$. And finally, for every $r$ and $j$, exactly $d-1$ vertices of the clique $Q_{r,j}$ belong to $S$.
    Note that for an arbitrary but fixed $i$, the set of pairs $\{u_{r,j,0}^i, u_{r,j,1}^i\}$ (for all $r$ and
    $j$) forms a matching of size $m(n+1)$ and all of these edges belong to the path $P_i$. Hence,
    from each of these pairs exactly one vertex is taken. Let $e_1, \dots, e_{(n+1) \cdot m}$ be the
    order in which these edges occur in $P_i$.
    
    The crucial observation is that if there would exist an index $z \in \bigl[(n+1)m - 1\bigr]$
    such that from $e_z$ the first vertex of the pair 
    and from $e_{z+1}$, the second
    vertex belongs to $S$, then the edge $e'$ that lies between $e_z$ and $e_{z+1}$ on $P_i$ would
    be not covered by $S$, which contradicts the assumption that $S$ is a vertex cover. So
    if for some $z \in \bigl[(n+1)m\bigr]$, the first vertex of $e_z$ belongs to $S$, then this also
    holds for every $z' \in \bigl[(n+1)m\bigr]$ such that $z \leq z'$.    
    Informally speaking, on each path $P_i$, there might happen at most one
    ``parity change''. Since there $n$ such paths $P_1, \dots, P_n$, at most $n$ parity changes
    could happen in total. Therefore, there exists an index $r^* \in [n+1]$ such that no parity
    change occurs in the $r^*$th ``interval'', i.e., for every $i \in [n]$, there exists an index
    $q(i) \in \{0, 1\}$ such that 
    \[
    	S \cap \bigl\{u_{r^*,j,x}^i \bigm\vert j\in[m],x\in \{0, 1\}\bigr\} = \bigl\{u_{r^*,j,q(i)}^i \bigm\vert j\in[m]\bigr\}.
    \]
    Then we define an assignment $\pi$ by setting  $\pi(v_i) = q(i)$ for each $i$.

    Now we show that $\pi$ satisfies $I$. Let $j$ be arbitrary but fixed. By the choice of $k$ (see
    \eqref{eq:lb-on-cvc}), we know that $|S \cap Q_{r^*, j}| = d - 1$ holds. Therefore, there exists
    an index $z^* \in [d]$ such that $q_{r^*, j}^{z^*} \notin S$. First, assume that $q_{r^*,
    j}^{z^*}$ corresponds to a negative literal $b_{j, z^*}$ of $v_i = v_{j, z^*}$.     
    Then there exists an edge $\{q_{r^*, j}^{z^*}, u_{r^*, j, 0}^i\}$ and it is covered by $S$ (recall that $S$ is a vertex cover). Due
    to $q_{r^*, j}^{z^*} \notin S$, we have $u_{r^*, j, 0}^{i} \in S$. By the construction of $\pi$,
    it holds that $\pi(v_{i}) = 0$ so $C_j$ is satisfied by $\pi$. 
    Now let $q_{r^*, j}^{z^*}$ correspond to a positive literal $b_{j, z^*}$ of $v_i = v_{j, z^*}$.  
    Then there exists an edge $\{q_{r^*, j}^{z^*}, u_{r^*, j, 1}^i\}$ and it is covered by $S$ (recall that $S$ is a vertex cover). Due
    to $q_{r^*, j}^{z^*} \notin S$, we have $u_{r^*, j, 1}^{i} \in S$. By the construction of $\pi$,
    it holds that $\pi(v_{i}) = 1$ so $C_j$ is satisfied by $\pi$. 
\end{proof}

\begin{proof}[Proof. (\cref{theo:cvc})]
	Suppose such $\varepsilon > 0$ exists, and let $\mathcal{A}$ be an algorithm that solves the
    \Pcvc{} problem in time $\ostar((2-\varepsilon)^{\ctw})$. Let $d$ be an arbitrary but fixed
    positive integer. Let $I$ be an instance of $d$-\textsc{SAT}. We build an instance $(G, k)$ of
    the \Pcvc{} problem equivalent to $I$ along with a linear arrangement of $G$ of cutwidth at most
    $n + \mathcal{O}(1)$ using the construction above in polynomial time. 
    After that, we run $\mathcal{A}$ on $(G, k)$ and output its answer. The
    correctness of the algorithm follows from the equivalence of $I$ and $(G, k)$. The described
    process runs in time $\ostar((2-\varepsilon)^{n + \mathcal{O}(1)}) = \ostar((2-\varepsilon)^n)$. So for every
    $d \in \NN$, we get an algorithm solving $d$-\textsc{SAT} in $\ostar((2-\varepsilon)^n)$
    contradicting SETH.
\end{proof}

\subsection{Odd Cycle Transversal}\label{app:oct-lb}

In this section, we show that the \Poct{} problem cannot be solved in time $\ostar\bigl((2 - \varepsilon)^{\ctw}\bigr)$ for any positive value $\varepsilon$. This proves that the $\ostar(2^{\ctw})$ algorithm from \cref{app:oct} is essentially tight.

\begin{quote}
	\Poct{}
	
	\textbf{Input}: A graph $G = (V, E)$ and an integer $k$.
	
	\textbf{Question}: Is there a subset $S \subseteq V$ of cardinality at most $k$ such that $G - S$ is bipartite.
\end{quote}

An instance of $d$-\textsc{NAE}-\textsc{SAT} is a pair $(U, C)$ where $U$ is a set of Boolean variables and $C$ is a set of clauses (each consists of at most $d$ literals). The question is whether there exists an assignment such that each clause contains a true as well as a false literal. 
We say that such an assignment \emph{NAE-satisfies} $(U, C)$.
Cygan et al.\ \cite{CyganDLMNOPSW16} have proven that the following equivalent formulation of SETH.

\begin{theorem}[\cite{CyganDLMNOPSW16}] \label{thm:hardness-of-nae-sat}
	The SETH is equivalent to the following statement: For every $\delta > 0$, there exists a constant $d \in \mathbb{N}$ such that $d$-\textsc{NAE}-\textsc{SAT} cannot be solved in time $\mathcal{O}\bigl((2 - \delta)^n\bigr)$ where $n$ denotes the number of variables. 
\end{theorem}

\begin{theorem}
	Assuming SETH, \Poct{} cannot be solved in time $\ostar((2-\varepsilon)^{\ctw})$ for any positive real $\varepsilon$, even if a linear arrangement of the input graph of width at most $\ctw$ is given with the input.
\end{theorem}

\begin{proof}
	For the sake of contradiction, suppose such $\varepsilon > 0$ and such algorithm exist. Let $0 < \delta < \varepsilon$ and let $d \in \mathbb{N}$ be a constant such that $d$-\textsc{NAE}-\textsc{SAT} cannot be solved in $\mathcal{O}((2 - \delta)^n)$ (exists by \cref{thm:hardness-of-nae-sat}). Let $\left(U, C\right)$ be an instance of $d$-\textsc{NAE}-\textsc{SAT}. Let $n = |U|$ and $m = |C|$. Further, let $U = \left\{v_1, \dots, v_n\right\}$, $C = \left\{c_1, \dots, c_m\right\}$, and for $j \in [m]$, let $c_j = \left\{l_j^1, l_j^2, \dots, l_j^{|c_j|}\right\}$ be the set of literals in $c_j$. 
	We may assume that every clause contains at least two literals: otherwise, the instance is trivially not satisfiable. 
	
	We construct an instance $(G, k)$ of \textsc{Odd Cycle Transversal} as follows. First, for every $i \in [n]$, we introduce a path 
	\[
		P_i = u_{i, 1}^1 u_{i, 1}^2 u_{i, 2}^1 u_{i, 2}^2 \dots u_{i, m}^1 u_{i, m}^2
	\]
	consisting of $2m$ new vertices. Later, we will choose the size of an odd cycle transversal in such a way that no vertex from these paths belongs to it. Since after the removal of an odd cycle transversal, the graph becomes properly 2-colorable, every path gets one of only two possible 2-colorings. The intuition behind it is that the color of the first vertex on $P_i$ determines the assignment of $v_i$.
	
\begin{figure}
	\center
	\includegraphics{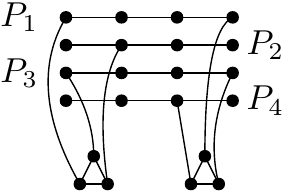}
	\caption{A graph constructed from an instance of $3$-\textsc{NAE}-\textsc{SAT} consisting of 4 variables and clauses $v_1 \lor \overline{v_2} \lor v_3$ and $\overline{v_1} \lor \overline{v_3} \lor v_4$.}
	\label{fig:oct-lb}
\end{figure}

	Further, for every $j \in [m]$, we introduce a set of new vertices $Q_j = \left\{q_j^1, \dots, q_j^{|c_j|}\right\}$ forming a clique. Note that every odd cycle transversal of the arising graph must contain at least $|c_j| - 2$ vertices of $Q_j$ since each three vertices of $Q_j$ form a triangle.

	Finally, for every $t \in \bigl[|c_j|\bigr]$, let $v_i$ be the variable appearing in the literal $l_j^t$. We add an edge $q_j^t r_j^t$ where
	\[
		r_j^t = 
		\begin{cases}
			u_{i, j}^1 & \mbox{if $l_j^t = v_i$,} \\
			u_{i, j}^2 & \mbox{if $l_j^t = \overline{v_i}$.}
		\end{cases}
		.
	\]
	Let $G$ denote the arising graph (see \cref{fig:oct-lb} for an example). Note that $G$ can be constructed in polynomial time. Finally, we define the size of the odd cycle transversal we are looking for as 
	\[
		k = \sum_{j = 1}^m \bigl(|c_j| - 2\bigr).
	\]
	
	\begin{lemma}
		If $(U, C)$ has an NAE-satisfying truth-value assignment, then $G$ admits an odd cycle transversal of size $k$.
	\end{lemma}
	
	\begin{proof}
		Let $\pi$ be an NAE-satisfying truth-value assignment of $(U, C)$. Then for every $j \in [m]$, there exist $1 \leq a_j, b_j \leq |c_j|$ such that 
		\[
			\pi\left(l_j^{a_j}\right) = 1, \pi(l_j^{b_j}) = 0. 
		\]
		We set 
		\[
			L_j = \Bigl\{q_j^{t} \Bigm\vert t \in \bigl[|c_j|\bigr] \setminus \{a_j, b_j\} \Bigr\}
		\]
		(note: $|L_j| = |c_j| - 2$) and 
		\[
			L = \bigcup_{j = 1}^{m} L_j
		\]
		(note: $|L| = k$). We now show that $L$ is an odd cycle transversal of $G$. For this purpose, we provide a proper 2-coloring $\beta$ of vertices of the graph $G - L$ with colors $0$ and $1$. First, for every $i \in [n]$:
		\begin{itemize}
			\item If $\pi(v_i) = 1$, for every $j \in [m]$, we set:
			\[
				\beta(u_{i, j}^1) = 1 \text{ and } \beta(u_{i, j}^2) = 0.
			\]
			\item If $\pi(v_i) = 0$, for every $j \in [m]$, we set:
			\[
				\beta(u_{i, j}^1) = 0 \text{ and } \beta(u_{i, j}^2) = 1.
			\]
		\end{itemize}
		In particular, this ensures that there are no conflicts along the edges of $P_1 , \dots, P_n$.
		Further, for every $j \in [m]$, we set 
		\[
			\beta(q_j^{a_j}) = 0 \text{ and } \beta(q_j^{b_j}) = 1.
		\]
		Note that the edges of $G\left[\bigcup\limits_{j=1}^{m} Q_j - L\right]$ induce a matching so there are no conflicts along these edges in $\beta$.
		Finally, observe that by construction, since $\pi(l_j^{a_j}) = 1$, we have $\beta(r_j^{a_j}) = 1$ so there is no conflict along the edge $q_j^{a_j} r_j^{a_j}$. Similarly, there is no conflict along the edge $q_j^{b_j} r_j^{b_j}$. Altogether, $\beta$ is a proper 2-coloring of $G-L$ and hence, $L$ is an odd cycle transversal of $G$ of size $k$.
	\end{proof}
	
	\begin{lemma}
		If $G$ admits an odd cycle transversal of size $k$, then $(U, C)$ has an NAE-satisfying assignment.
	\end{lemma}
	
	\begin{proof}
		Let $L$ be an odd cycle transversal of $G$ of size $k = \sum\limits_{j = 1}^m \bigl(|c_j| - 2\bigr)$ and let $\beta$ be a proper 2-coloring of vertices of the graph $G - L$ with colors $0$ and $1$.
		We define a truth-value assignment $\pi$ as
		\[
			\pi(v_i) = 
			\begin{cases}
				1 & \mbox{if $\beta(u_{i,1}^1) = 1$,} \\
				0 & \mbox{if $\beta(u_{i,1}^1) = 0$,} \\
			\end{cases}
		\]
		for every $i \in [n]$. We prove that $\pi$ NAE-satisfies $(U, C)$. 
		
		For every $j \in [m]$, let $L_j = L \cap Q_j$. Recall that the graph $G[Q_j]$ is a clique of size $|c_j| \geq 2$. Therefore, $|L_j| \geq c_j - 2$ since otherwise $G[Q_j - L]$ would contain a triangle. Then by the choice of $k$, for every $j \in [m]$, we have $|L_j| = c_j - 2$. In particular, for every $i \in [n]$, it holds that
		\[
			L \cap V(P_i) = \emptyset.
		\]
		Since $P_i$ is a path, there are only two ways how it can be colored in $\beta$:
		\begin{itemize}
			\item Either for all $j \in [m]$, we have $\beta(u_{i, j}^1) = 1$ and $\beta(u_{i, j}^2) = 0$. This holds exactly if $\pi(v_i) = 1$.
			\item Or for all $j \in [m]$, we have $\beta(u_{i, j}^1) = 0$ and $\beta(u_{i, j}^2) = 1$. This holds exactly if $\pi(v_i) = 0$.
		\end{itemize}
		
		For $j \in [m]$, it holds that $|Q_j - L| = 2$ and there is an edge between the two vertices of this set. Thereby, these vertices have different colors in $\beta$. So let $a_j, b_j \in \bigl[|c_j|\bigr]$ be such that $Q_j - L = \{q_j^{a_j}, q_j^{b_j}\}$ and $\beta(q_j^{a_j}) = 0, \beta(q_j^{b_j}) = 1$. 
		Since $\beta$ is proper, we have $\beta(r_j^{a_j}) = 1$ and $\beta(r_j^{b_j}) = 0$. Hence, by construction, $\pi(l_j^{a_j}) = 1$ and $\pi(l_j^{b_j}) = 0$ and $\pi$ NAE-satisfies the clause $c_j$. Since $j \in [m]$ is arbitrary, $\pi$ NAE-satisfies the clause $(U, C)$.
	\end{proof}

	\begin{lemma}
		The cutwidth of the graph $G$ is at most $n + \mathcal{O}(1)$.
	\end{lemma}
	
	\begin{proof}
		For $j \in [m]$, let the \emph{block} $B_j$ be defined as
		\[
			B_j = \left(\bigcup_{i = 1}^{n} \{u_{i,j}^1, u_{i,j}^2\}\right) \cup \left(\bigcup_{t = 1}^{|c_j|} \left\{q_j^t\right\}\right).
		\]
		Note that $B_1, B_2, \dots, B_m$ is a partition of $V(G)$.
		
		Consider a linear arrangement of $G$ such that for every $j \in [m-1]$, every vertex of $B_j$ occurs before every vertex of $B_{j+1}$ and the vertices of $B_j$ are ordered as:
		\[
			u_{1,j}^1, u_{2,j}^1, \dots, u_{n,j}^1, 
			q_j^1, q_j^2, \dots, q_j^{|c_j|},
			u_{1,j}^2, u_{2,j}^2, \dots, u_{n,j}^2.
		\]
		We show that this linear arrangement has the desired cutwidth.
		
		Observe that the end-vertices of every edge either belong to the same block or to two consecutive blocks. Now consider a cut before a vertex $v \in B_j$ for some $j \in [m]$. 
		We partition the edges of this cut into four types:
		\begin{enumerate}
			\item \label{item:type-1} $u_{i,j}^1 u_{i,j}^2$ for some $i \in [n]$,
			\item \label{item:type-2} $u_{i,j-1}^2 u_{i,j}^1$ for some $i \in [n]$ (exist only if $j \geq 2$),
			\item \label{item:type-3} $u_{i,j}^2 u_{i,j+1}^1$ for some $i \in [n]$ (exist only if $j \leq m-1$),
			\item \label{item:type-4} $q_j^t u_{i,j}^g$ for some $i \in [n]$, $t \in \bigl[|c_j|\bigr]$, $g \in [2]$,
			\item \label{item:type-5} and $q_j^{t_1} q_j^{t_2}$ for some $j \in [m]$, $t_1, t_2 \in \bigl[|c_j|\bigr]$.
		\end{enumerate}
		Observe that every cut contains at most ${\binom d 2} \in \mathcal{O}(1)$ edges of type \ref{item:type-5} and at most $d \in \mathcal{O}(1)$ edges of type \ref{item:type-4}. Now we make a case distinction to bound the number of edges of the first three types in the cut.	
		\begin{itemize}
			\item $u = u_{i,j}^1$ for some $i \in [n]$. Then the cut contains
			\begin{enumerate}
				\item $i - 1$ edges of type \ref{item:type-1},
				\item $n - (i - 1)$ edges of type \ref{item:type-2} if $j \geq 2$ and no such edges otherwise,
				\item and $0$ edges of type \ref{item:type-3}.
			\end{enumerate}
			
			\item $u = q_{j}^t$ for some $t \in \bigl[|c_j|\bigr]$. Then the cut contains
			\begin{enumerate}
				\item $n$ edges of type \ref{item:type-1},
				\item $0$ edges of type \ref{item:type-2},
				\item and $0$ edges of type \ref{item:type-3}.
			\end{enumerate}
			
			\item $u = u_{i,j}^2$ for some $i \in [n]$. Then the cut contains
			\begin{enumerate}
				\item $n - (i - 1)$ edges of type \ref{item:type-1},
				\item $0$ edges of type \ref{item:type-2},
				\item and $i - 1$ edges of type \ref{item:type-3} if $j \leq m-1$ and no such edges otherwise.
			\end{enumerate}
		\end{itemize}
		Thus, in every case, the cut contains at most $n$ edges of types \ref{item:type-1}, \ref{item:type-2}, and \ref{item:type-3}. So the graph $G$ indeed has cutwidth at most $n + \mathcal{O}(1)$.
	\end{proof}
	Altogether, the  instance $(G, k)$ of \textsc{Odd Cycle Transversal} can be constructed in polynomial time, it is equivalent to the instance $(U, C)$ of $d$-\textsc{NAE-SAT}, and $G$ has cutwidth of at most $n + \mathcal{O}(1)$. 
	Recall that $0 < \delta < \varepsilon$ and $d \in \mathbb{N}$ is a constant such that $d$-\textsc{NAE}-\textsc{SAT} cannot be solved in $\mathcal{O}((2 - \delta)^n)$.
	Hence, if \textsc{Odd Cycle Transversal} could be solved in $\ostar\bigl((2-\varepsilon)^{\ctw}\bigr)$, i.e., in $\mathcal{O}\bigl((2 - \varepsilon)^{\ctw}|V|^{\mathcal{O}(1)}\bigr)$, then $d$-\textsc{NAE-SAT} could be solved in $\mathcal{O}(2 - \delta)^n$ as follows.
	Given an instance $(U, C)$ of $d$-\textsc{NAE-SAT}, we first construct an equivalent instance $\bigl(G = (V, E), k\bigr)$ as described before in polynomial time and then solve it in $\mathcal{O}\bigl((2 - \varepsilon)^{n + \mathcal{O}(1)}|V|^{\mathcal{O}(1)}\bigr)$. Since $\delta < \varepsilon$ holds, this way we solve $d$-\textsc{NAE-SAT} in $\mathcal{O}\bigl((2 - \delta)^{n}\bigr)$ -- a contradiction to the choice of $d$. 
\end{proof} 

\subsection{Feedback Vertex Set}\label{app:fvs-lb}
Now we prove that assuming SETH, \Pfvs{} cannot be solved in time  $\ostar((2 - \varepsilon)^{\ctw})$ for any positive value $\varepsilon$. This proves the tightness of our $\ostar(2^{\ctw})$ upper bound from \cref{app:fvs-ub}. 
We achieve this with a very similar reduction to the one provided for the \Pcvc{} problem in \cref{app:cvc-lb}. In contrast to \pcvc{}, in \pfvs{} we seek to ``minimize the connectivity'' of the rest of the graph instead of ``maximizing the connectivity'' of the solution.
This imposes different kind of restrictions in the constructions. 
For example, instead of adding private neighbors to the root-path and forcing its vertices to belong to the solution, we seek to forbid solutions containing vertices on the root-path. 
Altogether, we get a construction that follows the same general scheme as for \Pcvc{} but with some problem-specific modifications.

\begin{quote}
	\Pfvs

	\textbf{Input}: An undirected graph $G = (V, E)$ and an integer $k$.

	\textbf{Question}: Is there a set $Y \subseteq V$ of cardinality at most $k$ such that $G - Y$ is a forest.

\end{quote}

\begin{theorem}\label{theo:fvs-lb}
	Assuming SETH, there is no algorithm that solves the \Pfvs{} problem in time $\ostar((2-\varepsilon)^{\ctw})$ for any positive real $\varepsilon$, even when a linear arrangement $\ell$ of $G$ of this width is given.
\end{theorem}

\subparagraph*{Construction.}
Let $d \in \NN$ be arbitrary but fixed and let $I$ be an instance of $d$-SAT.
    We may assume that every clause contains exactly $d$ literals. 
    Let $v_1, \dots, v_n$ be the variables and let $C_1, \dots, C_m$ be the clauses in $I$. For $j \in [m], t \in [d]$, let
	$b_{j, t}$ denote the $t$th literal in $C_j$, and let $v_{j, t}$ denote the variable in this
	literal. Let $s = m \cdot d$. 
	In the following, we will use the following indices: $i \in [n]$, $r \in [n+1]$, $j \in [m]$, $h \in \{0, 1\}$, $t \in [d]$, and $x \in \{0, 1\}$. These indices will always belong to these intervals so for shortness, we omit the domain when we use any of these indices.
    For illustration of our construction, we refer to
    \cref{fig:append-fvslb}.
	We construct an instance $(G = (V, E), k)$ of the \Pfvs\ problem
	as follows. Initially $G$ is an empty graph. For each variable $v_i$, we add to $G$ a path 
     \begin{alignat*}{3}
    	P_i:=\: &u^i_{1,1,0}, u^i_{1,1,1}, u^i_{1,2,0},& &\dots,& &u^i_{1,m,1},\\
                &u^i_{2,1,0}, u^i_{2,1,1}, u^i_{2,2,0},& &\dots,& &u^i_{2,m,1},\\
                & & &\quad\vdots& & \\
                &u^i_{n+1,1,0}, u^i_{n+1,1,1}, u^i_{n+1,2,0}, & &\dots, & & u^i_{n+1,m,1}
    \end{alignat*}
    consisting of $2m(n+1)$ new vertices. 
    More formally, the path $P_i$ consists of vertices $u^i_{r, j, h}$ so that the
indices $(r, j, h)$ appear in lexicographic order.
    We partition each path into $m(n+1)$ pairs of consecutive
    vertices, i.e.\ for $r \in [n+1]$, $j \in [m]$, and $i \in [n]$, the $(r, j)$th \emph{pair} of
    the $i$th path consists of the vertices $u^i_{r, j, 0}$ and $u^i_{r, j, 1}$. These are the first
    vertex and the second vertex of a pair respectively. For all $i$, $r$, $j$, we
    add the vertex $z^i_{r, j}$ and make it adjacent to $u^i_{r, j, 0}$ and $u^i_{r, j, 1}$ only.
    This turns each path into an alternating sequence of single edges and triangles. For simplicity with a slight abuse of notation, we still refer to this construction as a path $P_i$.

	For each $j$ and $r$, we create a clique $Q_{r,j} = \{q_{r,j}^1, \dots,
	q_{r,j}^{d+1}\}$ consisting of $d + 1$ new vertices. 
	The clique $Q_{r, j}$ will only be adjacent to vertices from the $(r, j)$th pairs of some of the paths $P_1, \dots, P_n$. 
	Namely, for each $j$, $r$, and $t$ we do the following. Let $i$ be such that $v_i =
	v_{j, t}$ and let $u = u^{i}_{r,j,0}$ if $b_{j, t}$ is a negative literal of $v_i$ and
	$u = u^i_{r,j,1}$ otherwise. We add an edge between $q_{r,j}^t$ and $u$. We
	also add a new vertex $y_{r,j}^{t}$ and make it adjacent to $q_{r, j}^t$ and $u$ so that these three vertices form a
	triangle. Note that the only neighbors of the vertex $q_{r, j}^{d+1}$ are the other vertices of the clique $Q_{r,j}$.

    Next, we add a path $\rootpath$ consisting of $(n+1)m \cdot 2n$ new vertices denoted with
    \[
        \rootpath = w_{1,1,1}, w_{1,1,2}, \dots, w_{1,1,2n}, w_{1,2,1}, \dots, 
        w_{1, m, 2n},
        \dots w_{n+1, 1, 1}, \dots, 
        w_{n+1,m,2n}.
    \]
    More formally, the vertices of $\rootpath$ are $w_{r, j, p}$ (for $p \in [2n]$) so that
    indices $(r, j, p)$ appear in lexicographic order.
   We call it the \emph{root-path}.
   For each $r$, $j$, and $i$, we add an edge $\{w_{r, j, 2i-1}, u^i_{r, j, 0}\}$ and an edge
   $\{w_{r,j,2i}, u^i_{r,j,1}$\}. Hence, each vertex $u^i_{r,j,h}$ on the path $P_i$ has an adjacent
   vertex $w_{r,j, 2i-1+h}$ on $\rootpath$. See \cref{fig:append-fvslb} for graphical representation of the construction.

   \begin{figure}[htpb]
   	\centering
   	\includegraphics[width=0.8\textwidth]{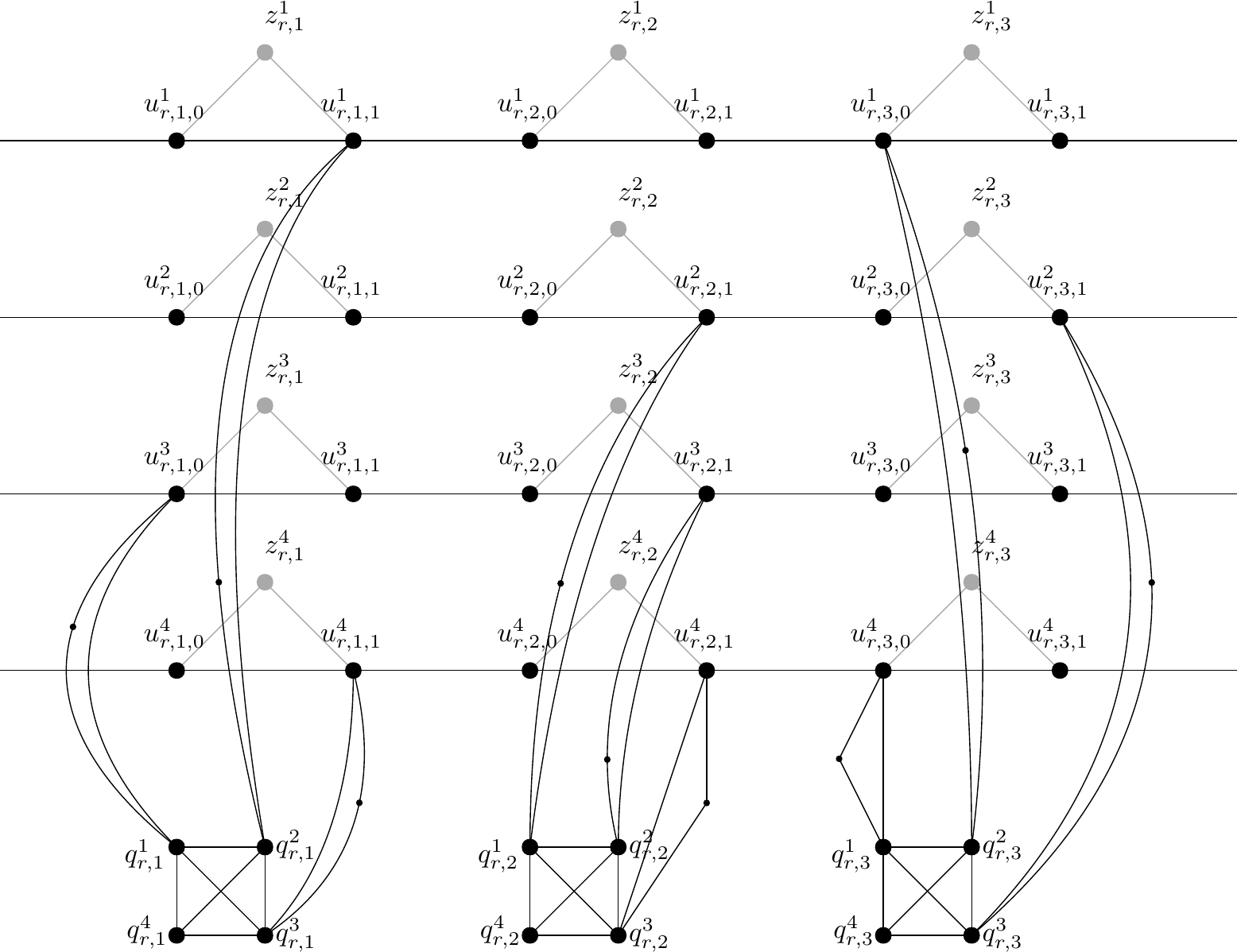}
   	\caption{
   	The $r$-th segment of the graph $G$ from arising in the reduction from an instance of $3$-SAT contains 4 variables and 3 clauses $v_1 \lor \overline{v_3} \lor v_4$, $v_2 \lor v_3 \lor v_4$, and $\overline{v_1} \lor v_2 \lor \overline{v_4}$.
    The path $\rootpath$ and its incident edges are omitted for readability. 
    }
   	\label{fig:append-fvslb}
   \end{figure}

\subparagraph*{Budget.}
For any fixed $i$, let $P'_i = V(P_i) \cup \bigl\{z^i_{r, j}\bigm\vert r\in[n+1], j\in[m]\bigr\}$. Since the subgraph induced by $P'_i$ contains $(n+1)m$ vertex-disjoint triangles, any \fvs{} must contain at least $(n+1)m$ vertices from $P'_i$.
Any \fvs{} must also contain at least $d-1$ vertices from each clique $Q_{r, j}$
of size $d+1$ (otherwise, there would remain a triangle). 
In total we get the lower bound of
\[
	n(n+1)m + (n+1) m (d-1) = (n+1)m(n+d-1)
\]
on the size of a \fvs{} of $G$. So we set $k = (n+1)m(n+d-1)$, matching this lower bound.
\begin{remark}\label{rem:fvs}
    Any \fvs{} of $G$ of size exactly $k$ consists of the following set of vertices: it contains exactly $(n+1)m$ vertices from each $P'_i$ and 
    exactly $d-1$ vertices from each clique $Q_{r,j}$. In particular, it does not include any vertex on $\rootpath$.
\end{remark}
Recall that $d$ is a constant so both $G$ and $k$ can be constructed in polynomial time from $I$.

\begin{lemma}\label{lem:fvs-ctw}
Let $(G, k)$ be the instance of \Pfvs\ arising from $I$ as described above.
    Then $G$ has the \ctwdth{} of at most $n + \mathcal{O}(1)$.
    Moreover, a linear arrangement of $G$ of this cutwidth can be constructed in polynomial time from $I$.
\end{lemma}
\begin{proof}
We emphasize that $d$ is a constant.
    We construct a linear arrangement $\ell$ of $G$ of width $n + {\binom{d+1}{2}} + 3d + 3 = n + \mathcal{O}(1)$. 
    Recall that a linear arrangement of a graph is just a total ordering of its vertices. The linear arrangement $\ell$ initially consists of $(n+1)$ consecutive segments, and each of them consists of $m$ consecutive subsegments.
    In other words, we build the linear arrangement from left to right as follows:
    For each $r$ from $1$ to $n+1$ we add the $r$th segment to the arrangement. We add the $r$th segment by adding the vertices of its $j$th subsegment for each $j$ from $1$ to $m$ one by one.
    For fixed $r$ and $j$, we add the vertices of the $j$th subsegment of the $r$th segment in the following order
    from left to right. First, add the vertices $u^i_{r,j,0}$, $z^i_{r,j}, u^i_{r,j,1}$ for all $i$ from $1$ to $n$. After that, we add all vertices of $Q_{r,j}$ and all vertices $y^t_{r,j}$ in an arbitrary order.
    When this process is finished, we add the
    vertices of $\rootpath$ to the linear arrangement as follows. For every vertex
    $u^i_{r,j,x}$ on a path $P_i$, let $w = w_{r, j, 2i-1+x}$ be the vertex on $\rootpath$ adjacent
    to $u^i_{r,j,x}$. We insert $w$ directly before $u^i_{r,j,x}$ into $\ell$.
	
    Now we prove that this $\ell$ has the cutwidth of at most $n + {\binom{d+1}{2}} + 3d + 3$.
    Here, writing about a cut, we always refer to the cuts of $\ell$.
    For any fixed $i$, the vertices of the path $P_i$ appear in the same order on $P_i$ and on $\ell$ so the edges of $P_i$ do not overlap.
    Hence, the paths $P_1, \dots, P_n$ totally contribute at most $n$ edges to any cut on $\ell$. 
    Since each vertex $z^i_{r, j}$ appears in $\ell$ between the vertices $u^i_{r, j, 0}$ and $u^i_{r,j,1}$ adjacent to it,
    and since no other vertex $u^{i'}_{r',j',x}$ or $z^{i'}_{r', j'}$ (for $i' \in [n]$, $j' \in [m]$, and $r' \in [n+1]$) appears between these three vertices, the edges incident to all vertices $z^i_{r, j}$ 
    do not overlap. Therefore, any cut of $\ell$ can contain at most one such edge. 
    Further, the vertices of $\rootpath$ appear in the same order on $\rootpath$ and on $\ell$. Hence, the edges of $\rootpath$ do not overlap so they can contribute at most one edge to a cut of $\ell$. Since each vertex of $\rootpath$ is added
    directly before the vertex $u^i_{r,j,x}$ adjacent to it, the edges between the vertices of
    $\rootpath$ and vertices $u^i_{r, j,x}$ do not overlap either, and add at most one to the
    \ctwdth{} of $\ell$.

    For fixed $r$ and $j$, let $Q'_{r, j} = Q_{r, j} \cup \bigl\{y^t_{r, j} \bigm\vert t \in [d]\bigr\}$. A clique $Q_{r, j}$ contains
    ${\binom{d+1}{2}}$ edges. It also has $d$ incident edges to vertices $y^t_{r,j}$ and $d$ incident edges to
    vertices $u^i_{r, j, h}$ for some $i$, $r$, $j$, and $h$. Each vertex $y^t_{r,j}$ is also adjacent to one
    vertex $u^i_{r, j, h}$ for some $i$, $r$, $j$, and $h$. So for fixed  there are totals in 
    \[
    	{\binom{d+1}{2}} + 2d + d = {\binom{d+1}{2}} + 3d
    \]
    edges incident to $Q'_{r, j}$. All 
    these edges run inside the $j$th subsegment of the $r$th segment. 
    Therefore, the edges incident to two different cliques do not overlap.
    So in total, the cutwidth of $\ell$ is bounded by 
    $n + {\binom{d+1}{2}} + 3d + 3$. 
    Also note that together with $(G, k)$, this linear arrangement can be constructed from $I$ in polynomial time.
\end{proof}

\begin{lemma}
    If $I$ is satisfiable, then $G$ has a \fvs{} of size at most $k$.
\end{lemma}
\begin{proof}
    Given a satisfying assignment $\pi$ of $I$, we build a vertex set $S$ and then show that it is a feedback vertex set of size $k$. 
    Initially $S$ is empty.
    Then for each $i\in [n], r \in [n+1], j \in [m]$, we add the vertex $u^i_{r, j, 0}$ to $S$, if $\pi(v_i) = 0$, and we add $u^i_{r,j, 1}$ to $S$ otherwise. 
    So for every vertex $u_{r, j, h}^i$, we have $u_{r, j, h}^i \in S$ if and only if $\pi(v_i) = h$.
    Next, for each clause $C_j$, let $t^* \in [d]$ be such that $\pi$ satisfies $b_{j,t^*}$ (it exists since $\pi$ is a satisfying
    assignment of $I$). 
    For each $r$ and $t' \in [d] \setminus \{t^*\}$ we add all vertices $q_{r, j}^{t'}$ to $S$, i.e., we add all vertices of the clique $Q_{r,j}$ except for $q_{r, j}^{t^*}$ and $q_{r, j}^{d+1}$. 
    This concludes the construction of the set $S$. 
    By the choice of $k$, it holds that $|S| = k$.
    
    Now we show that $S$ is a feedback vertex set of $G$, i.e., we show that
    $G - S$ is a forest. 
    Observe that for a vertex $v$ of degree at most one in some graph $H$, the graph $H - \{v\}$ is a forest if and only if $H$ is a forest.
    So we will iteratively remove vertices of degree at most one from $G - S$ until we obtain a graph, where it is easy to see that it is a forest.
    First of all, we remove all vertices $z_{r, j}^i$ for all $i, r, j$. This can be done since exactly one of two neighbors of $z_{r, j}^i$ belongs to $S$. 
    Second, we remove all vertices $q_{r, j}^{d+1}$. This is allowed since from the $d$ neighbors of $q_{r, j}^{d+1}$, exactly $d-1$ belong to $S$.
    
    Now for every fixed $r, j$, we consider the unique vertex $q_{r, j}^{t^*} \notin S$ with $t^* \in [d]$. Let $u_{r, j, h^*}^{i^*}$ be its unique neighbor in $V(P_1) \cup \dots \cup V(P_n)$. By the construction of $S$, the literal $b_{j, t^*}$ satisfies $C_j$ in $\pi$. 
    Now we distinguish on how $v_{i^*}$ occurs in $C_j$.
    If $b_{j, t^*} = v_{i^*}$, then we have $h^* = 1$ by the construction of $G$ and $\pi(v_{i^*}) = 1$ (since $b_{j, t^*} = v_{i^*}$ satisfies $C_j$ in $\pi$). The property $\pi(v_{i^*}) = 1$ implies that by the construction of $S$, we have $u_{r, j, 1}^{i^*}\in S$.
    On the other hand, if $b_{j, t^*} = \overline{v_{i^*}}$, then we have $h^* = 0$ by the construction of $G$ and $\pi(v_{i^*}) = 0$ (since $b_{j, t^*} = v_{i^*}$ satisfies $C_j$  in $\pi$). The property $\pi(v_{i^*}) = 0$ implies that by the construction of $S$, we have $u_{r, j, 0}^{i^*}\in S$.
    Hence, in any case we have that $u_{r, j, h^*}^{i^*} \in S$.
    So now the vertex $q_{r, j}^{t^*}$ has exactly one neighbor in the current graph, namely the vertex $y^{t^*}_{r,j}$. Moreover, the vertex $y^{t^*}_{r,j}$ now only has the neighbor $q_{r, j}^{t^*}$. Therefore, we can remove all such vertices $q_{r, j}^{t^*}$ and $y^{t^*}_{r,j}$.
    Also for all $t \in [d] \setminus \{t^*\}$, we can remove the vertex $y^{t}_{r, j}$ since it originally had degree 2 in $G$ and one of its neighbors (namely, $q_{r, j}^t$) belongs to $S$. 
    
    Observe that until now we have already removed all clique vertices $q_{r, j}^{t}$, all vertices $y_{r, j}^t$, and all vertices $z^i_{r, j, h}$. Also recall that from every path $P_i$, either exactly the vertices $u^i_{r, j, 0}$ or exactly the vertices $u^i_{r, j, 1}$ (for all $r, j$) belong to $S$. Therefore, every vertex $u^i_{r, j, h} \notin S$ currently has exactly one neighbor, namely the one on $\rootpath$. So we can also remove the remaining vertices from $V(P_1) \cup \dots \cup V(P_n)$ as well. Altogether, the remaining vertices are exactly the vertices of $\rootpath$ which induce a path by construction and in particular, they induce a forest. Therefore, the graph $G - S$ also induces a forest as claimed.
\end{proof}

\begin{lemma}
    If $G$ has a \fvs{} $S$ of size $k$, then $I$ is satisfiable.
\end{lemma}
\begin{proof}
   By \cref{rem:fvs}, the set $S$ consists of exactly one vertex from every triangle in $G\bigl[V(P_1) \cup V(P_n)\bigr]$ and exactly $d-1$ vertices from each clique $Q_{r, j}$.
   In particular, it does not contain any vertex of $\rootpath$.
   This implies that for any two consecutive vertices $u^i_{r, j, x}, u^i_{r', j', x'}$ on a path $P_i$, at least one of them belongs to $S$: these vertices form a cycle with a subpath of $\rootpath$. 
   This has the following implications. 
   First, for each pair $u^i_{r, j, 0}, u^i_{r, j, 1}$, exactly one of the vertices 
   belongs to the $S$. 
   In particular, no vertex $z^i_{r, j}$ and no vertex $y_{r, j}^t$ belongs to the solution.
   And second, for a path $P_i$, if there exists $r$ and $j$ such that $u^i_{r, j, 0}$ belongs to $S$, then for every $r' \in [n+1]$ and $j' \in [m]$ such that $(r, j) < (r', j')$ holds, we have $u^{i}_{r', j', 0} \in S$.
   Hence, on each path at most one parity change 
   occurs. 
   Since we
   have $n$ paths in total, at most $n$ such changes appear.
   
   Therefore, there exists an index $r^* \in [n+1]$ such that
   either all vertices $u^i_{r^*, j ,0}$ or all vertices $u^i_{r^*,j,1}$ belong to $S$. 
   We claim that the assignment $\pi$, that assigns the value of $0$ to $v_i$ if $u^i_{r, 1, 0} \in S$ and the value of $1$ otherwise, is a satisfying assignment of $I$. 
   
   Let $t^* \in [d]$ be such that $q_{r^*, j^*}^{t^*} \notin S$ holds (recall that $S$ contains $d-1$ vertices from $Q_{r,j}$ so such $t$ exists). 
   Let $x^* \in \{0,1\}$ and $i^* \in [n]$ be such that $v_{i^*} = v_{j^*, t^*}$ holds and $u^{i^*}_{r^*, j^*, x^*}$ is the vertex on $P_i$ adjacent to $q_{r^*, j^*}^{t^*}$. 
   Since $q_{r^*, j^*}^{t^*} \notin S$ and $y_{r^*, j^*}^{t^*} \notin S$, the vertex $u^{i^*}_{r^*, j^*, x^*}$ forming a triangle with this to vertices belongs to $S$. 
   By the choice of $r^*$, it then also holds that $u^{i^*}_{r^*, 1, x^*}$ belongs to $S$.
   Hence, we have $\pi(v_i) = x^*$ by the choice of $\pi$. 
   Also recall that by the construction of $G$, we have $x^* = 0$ if $b_{j, t^*} = \overline{v_{i^*}}$, and $x^* = 1$ if $b_{j, t^*} = v_{i^*}$. So $v_i$ satisfies $C_j$ in $\pi$ and $\pi$ is the desired satisfying assignment of $I$.
\end{proof}

\begin{proof}[Proof. (\cref{theo:fvs-lb})]
	Suppose such $\varepsilon > 0$ exists. 
	So let $\mathcal{A}$ be an algorithm that solves the \Pfvs{} problem in time $\ostar((2-\varepsilon)^{\ctw})$. 
	Let $d$ be an arbitrary but fixed positive integer. 
	Let $I$ be an instance of $d$-\textsc{SAT}. 
	First, in polynomial time, we compute an instance $(G, k)$ of the \Pfvs{} problem equivalent to $I$ along with a linear arrangement $\ell$ of $G$ of cutwidth at most $n + \mathcal{O}(1)$ using the above construction.
    After that, we run $\mathcal{A}$ on $(G, k)$ and output its answer. 
    The correctness of the algorithm follows from the equivalence of $I$ and $(G, k)$. 
    The described process then runs in time $\ostar\bigl((2-\varepsilon)^{n + \mathcal{O}(1)}\bigr) = \ostar\bigl((2-\varepsilon)^n\bigr)$. 
    So for every $d$, we get an algorithm solving $d$-\textsc{SAT} in $\ostar\bigl((2-\varepsilon)^n\bigr)$ contradicting SETH.
\end{proof}

\subsection{Connected Odd Cycle Transversal}\label{app:coct}
Cygan et al.\ \cite{CyganKN18} provided an $\ostar(4^{\tw})$ algorithm solving the \Pcoct{} problem. By \cref{lem:ctw-ub}, we can therefore solve this problem in $\ostar(4^{\ctw})$.
In this section we prove that under SETH the \Pcoct{} problem cannot be solved in time $\ostar\bigl((4-\varepsilon)^{\ctw}\bigr)$ for any positive $\varepsilon$ thus proving that the above algorithm is essentially optimal. 

\begin{quote}
	\Pcoct{}
	
	\textbf{Input}: A graph $G = (V, E)$ and an integer $k$.
	
	\textbf{Question}: Is there a subset $S \subseteq V$ of cardinality at most $k$ such that $G - S$ is bipartite and $G[S]$ is connected.
\end{quote}

\begin{theorem}\label{theo:coct}
	Assuming SETH, there is no algorithm that solves the \Pcoct{} problem in time $\ostar((4-\varepsilon)^{\ctw})$ for
    any positive real $\varepsilon$, even when a linear arrangement $\ell$ of $G$ of this width is given.
\end{theorem}

\subparagraph{Construction.}
	Let $d \in \NN$ and let $I$ be an instance of $d$-SAT. Let $n$ be the number of variables of $I$ and $m$ be the number of clauses in $I$. 
	We may assume that $n$ is even: otherwise, we could add one new variable that does not appear in any clause.
	So let $n' = n / 2$.
	In the following, we will use the following indices: $i \in [n']$, $r \in [2n'+1]$, $j \in [m]$, $t \in [d]$, and $x \in \{0, 1\}$. These indices will always belong to these intervals so for shortness, we omit the domain when we use any of these indices.
	We partition the set of variables of $I$ into $n'$ groups $U_1, \dots U_{n'}$ of size two each. 
	For each group $U_{i}$, we add a path-like gadget $P_i$ described next and we call it a \emph{path} for simplicity. 
	The path $P_i$ is a sequence of $(2n' + 1)m$ copies of a fixed gadget $X$ called a \emph{path gadget} described later. Each path gadget $X$ will have two special vertices $v$ and $v'$ called the \emph{left-end} and the \emph{right-end} vertex, respectively. The right-end vertex of each path gadget of $P_i$ (except for the last one) is adjacent to the left-end vertex of the next path gadget on $P_i$. 
	Therefore, the path $P_i$ is a simple path of length $(2n' + 1)m$ each of whose vertices is blown up to a path gadget.
	The path $P_i$ is partitioned into $2n'+1$ groups of $m$ consecutive path gadgets each. For every $r$, the $r$th such group is called the $r$th \emph{path segment} of $P_i$ denoted by $S^i_r$. 
	For every $r$, we call the set $S_r = \bigl\{S^i_r \mid i \in [n']\bigr\}$ the \emph{$r$th segment}.
	Next, for every $j$, we denote the $j$th path gadget of $S^i_r$ with $X^i_{r,j}$.
	Finally, for every $r$, we call the set $Q_{r, j} = \bigl\{X^i_{r,j} \bigm\vert i \in [n]\bigr\}$ the $(r, j)$th \emph{column}.

	By \emph{adding a triangle} at a vertex $v$, we mean adding a pair of new vertices $w$ and $w'$ and creating
	a clique on vertices $v$, $w$, and $w'$. We denote this triangle by $T(v)$. The graph $G$ also contains two simple vertex-disjoint paths $\rootpath$ and $\colorpath$ called the \emph{root-path} and the \emph{color-path}, respectively. 
	These paths consist of new vertices whose number and order will be determined
	later. We add a triangle at each vertex $u$ of $\rootpath$ and with $\rootpath'$, we denote the set of vertices of these triangles. In particular, we have $V(\rootpath) \subseteq \rootpath'$. 
	These triangles ensure that if $G$ admits a \coct{} of size $S$ at most $k \in \NN$, then it also admits a \coct{} $S'$ of size at most $k \in \NN$ with $V(\rootpath) \subseteq S'$.
	We call a vertex \emph{root-connected} if it has a neighbor on $\rootpath$. Moreover, distinct root-connected vertices are meant to have distinct neighbors on $\rootpath$ and every root-connected vertex $v$ has exactly one neighbor on $\rootpath$ denoted by $\rootpar{v}$. To ensure that this is possible, we choose the length of $\rootpath$ to be equal to the number of root-connected vertices in $G$.
	
	The set $\blackset \subset V(\colorpath)$ is a subset of the vertices of $\colorpath$. We will choose the order vertices on $\colorpath$ such that the vertices of $\blackset$ will have odd positions, meanwhile the rest of the vertices will only separate the vertices of $\blackset$, i.e. the length of $\colorpath$ will be exactly $2|\blackset| -1$, where each other vertex belongs to $\blackset$ starting from the first vertex of $\colorpath$.
	Every vertex $w$ in $\blackset$ will be assigned to one so-called \emph{black} or \emph{white} vertex $v$ in such a way that this assignment is a bijection. We choose the size of $\blackset$ to be equal to the sum of the number of black and white vertices in $G$, also called \emph{colored} vertices, to ensure that such an assignment is possible. 
	In this case we also write $w = \colpar{v}$. 
	By making a vertex $v$ black, we mean adding a new vertex $w$ adjacent to $v$ and $\colpar{v}$.
	By making a vertex $v$ white, we mean adding an edge between $v$ and $\colpar{v}$.
	As we mentioned before, the order of vertices in $\rootpath$ and $\colorpath$ is not fixed yet. It will later be defined in such a way that the cutwidth of $G$ is upper bounded as desired.
	We also add two new root-connected vertices $g$ and $g'$ to the graph and we call them \emph{guards}. We add triangles at $g$ and $g'$. We add an edge between $g$ and the left-end vertex of the first path gadget $X^i_{1,1}$ on each path $P_i$. Similarly, we add an edge between $g'$ and the right-end vertex of the last path gadget $X^i_{2n'+1, m}$ on each path $P_i$.
	
	By adding a triangle between two vertices $u$ and $v$, we mean adding a new vertex $w$ and creating a clique on vertices $u$, $v$, and $w$. We denote this clique by $T(u, v)$. 
	Now we describe the structure of a path gadget. A path gadget $X$ initially consists of the left-end vertex $v$ and the right-end vertex $v'$ mentioned above. Then we add four root-connected vertices $\verB$, $\verW$, $\verC$, and
	$\verD$ and we add a triangle between each pair of these vertices. 
	We call these vertices the \emph{clique vertices} and denote them by $K = \{\verB, \verW, \verC, \verD\}$. We make $\verB$ black and $\verW$ white. Next, we add the vertices $u, w_1, w_2, u', w'_1, w'_2$ and make $w_1$ and
	$w'_1$ root-connected. 
	We add a triangle between each of the following pairs of vertices: $\{v,u\}$, $\{v, w_1\}$, $\{w_1,w_2\}$, $\{w_2, \verD\}$. Next we
	add the edges $\{u, \verD\}$, $\{u, \verW\}$, $\{u, \verB\}$, $\{w_2, \verW\}$, $\{w_2,
	\verB\}$, and $\{u, w_2\}$. We also add two edges $\{v, \verW\}$ and $\{v, \verB\}$ and subdivide them. Symmetrically,
	we add triangles between the pairs $\{v', u'\}, \{v', w'_1\}$, $\{w'_1, w'_2\}$, $\{w'_2,
	\verC\}$ of vertices, and we add the edges $\{u', \verC\}$, $\{u', \verW\}$, $\{u', \verB\}$, $\{w'_2,
	\verW\}$, $\{w'_2, \verB\}$, and $\{u', w'_2\}$. We also add the edges $\{v', \verW\}$ and $\{v',
	\verB\}$ but this time we subdivide each of them \textbf{twice}. This concludes the construction of a path
	gadget. See \cref{fig:append-coct-pg} for graphical representation.
	If not clear from the context, we write the name of a gadget after the name of a vertex in parentheses to indicate to which gadget it belongs, e.g., the vertex $v(X)$ is the left-end vertex of a path gadget $X$.
    If the gadget is clear from the context, we omit this clarification.
	Recall that between each pair of clique vertices of $X$ we have created a triangle. Therefore, for any such clique, every \coct{} of $G$ must contain at least three vertices among the clique vertices together with vertices added to create triangles between them. Moreover, any \coct{} containing exactly three vertices among these, must have all three vertices as clique vertices, since otherwise, an \oct{} must choose a triangle vertex between two unchosen clique vertices, which would make the \oct{} not connected. We will fix the size of the desired \coct{} in such a way that it can contains at most three vertices in each clique together with the vertices added to create triangles between its vertices. This allows us to identify the \emph{state} of $X$ in a \coct{} with the clique vertex not contained in it, defining the states $\scrB$, $\scrW$, $\scrC$, and $\scrD$ respectively.

	\begin{figure}[t]
		\centering
		\includegraphics[width=.9\linewidth]{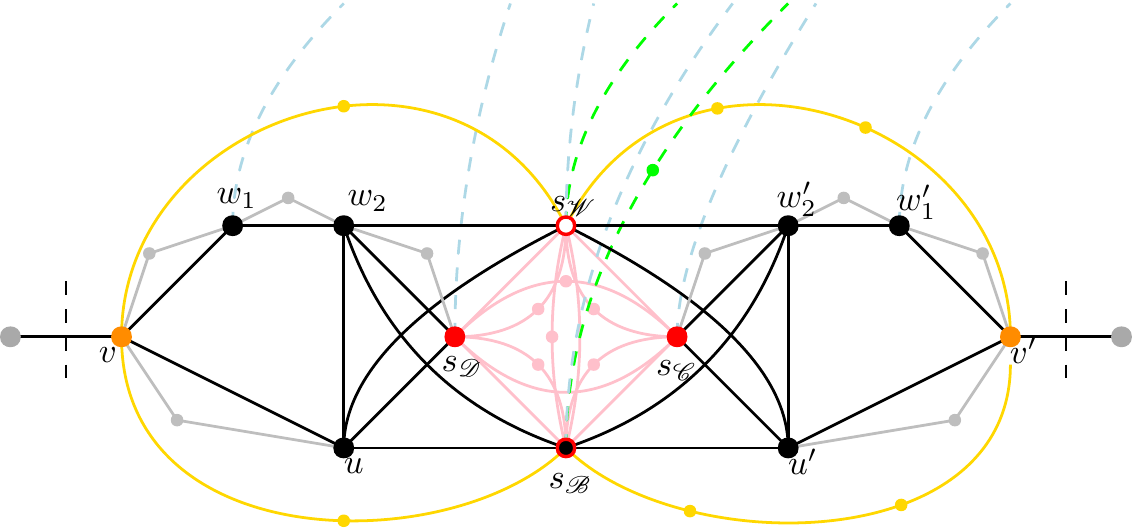}
		\caption{A graphical representation of a path gadget. 
		Clique vertices are depicted in red. 
		Edges from root-connected vertices to the root-path $\rootpath$ are depicted with dashed light blue rays.
		The subdivision-vertices and vertices added to create a triangle between two vertices are depicted smaller than normal vertices.
		Edges from colored vertices to the color-path $\colorpath$ are depicted with dashed green rays. 
		The vertical dashed lines mark the connections to the preceding and the following gadget.
		}
		\label{fig:append-coct-pg}
	\end{figure}

	We attach a decoding gadget $Y^i_{r, j} = Y$ (described next) to each path gadget $X^i_{r, j} = X$. The decoding gadget $Y$
	consist of eight root-connected vertices $\bigl\{u_x, v_x \bigm \vert x \in \{\scrB, \scrW, \scrC, \scrD\} \bigr \}$. For each $x \in \{\scrB, \scrW, \scrC, \scrD\}$, we say that $u_x$ and $v_x$ form a \emph{pair}, we call $v_x$ its \emph{left side} and $u_x$ its \emph{right side}, and we add a triangle $T(u_x, v_x)$ between them. Further, for every $x \in \{\scrB, \scrW, \scrC, \scrD\}$ and every $y = s_{x'}(X)$ for $x' \in \{\scrB, \scrW, \scrC, \scrD\} \setminus \{x\}$, we add the triangle $T(v_x, y)$ between $v_x$ and $y$.
	
	Finally, we describe the construction of a clause gadget. 
	We fix an arbitrary bijective mapping 
	$\phi: \{0, 1\}^2 \rightarrow
	\{\scrB, \scrW, \scrC, \scrD\}$ 
	that assigns a state to each partial assignment of two variables. 
	We also fix an ordering of variables in each group $U_i$: now with a slight abuse of notation, we can apply $\phi$ to partial assignments of each $U_i$.
	Clause gadgets represent the
	clauses of $I$. 
	For every $r, j$, we attach a clause gadget $Z_{r, j}$ to the column $Q_{r, j}$ as follows.
	Informally speaking, a clause gadget $Z_{r, j}$
	is merely a simple cycle through certain vertices of different decoding gadgets in the column $Q_{r,
	j}$ and possibly, through one more new vertex. 
	More formally, for $i \in [n']$, let $\Pi_i$ be the set of all partial assignments to the pair $U_i$ satisfying the clause
	$C_j$ and let $\Sigma_i = \phi(\Pi_i)$. 
	Finally, let $D_i = \bigl\{u_x(Y^i_{r,j}) \mid x \in \Sigma_i\bigr\}$ be the set of vertices on the right side of the decoding gadget $Y^i_{r, j}$ corresponding to states in $\Sigma_i$. Let $D
	= \bigcup_{i\in[n']} D_i$ and let $q_1, \dots, q_{|D|}$ be such that $D = \{q_1, \dots, q_{|D|}\}$.
	If $|D|$ is odd, we add the edges
	\[
		\{q_1, q_2\}, \{q_2, q_3\}, \dots, \{q_{|D|-1}, q_{|D|}\}, \{q_{|D|}, q_1\}
	\]
	so that the vertices of $D$ form a (simple) odd cycle.
	Otherwise, $D$ is even. In this case, we introduce a new vertex $w$ (called a \emph{parity-vertex}) and the edges
	\[
		\{q_1, q_2\}, \{q_2, q_3\}, \dots, \{q_{|D|-1}, q_{|D|}\}, \{q_{|D|}, w\}, \{w, q_1\}
	\]
	so that the vertices of $D \cup \{w\}$ form a (simple) odd cycle.
	In both cases, the created odd cycle is the clause gadget $Z_{r,j}$.
	This concludes the construction of the graph $G$ (strictly speaking, we still have not fixed the order of vertices on $\rootpath$ and $\colorpath$ but we will do that in the proof of \cref{lem:coct-lb-ctw} for the sake of clarity).
	
	Clearly, $G$
	can be constructed in time polynomial in the size of $I$.

	\subparagraph{Budget.} Now we fix the budget $k$ (i.e., the size of the desired \coct{}) of the resulting instance. 
	We choose the budget $k$ matching a lower bound on the size of any \coct{} of the graph $G$.
	We define this lower bound using a so-called \emph{packing} $\mathcal{P}(G)$ of $G$. 
       A packing is a family of pairwise disjoint sets of vertices (also called \emph{components}) of $G$ such that each component $C$ is assigned a value $p(C)$ such that each \coct{} contains at least $p(C)$ vertices from $C$. Since components are disjoint, the sum of values $p(C)$ over all components $C$ of a packing is a lower bound on the size of any \coct{}.
       Moreover, if a \coct{} of this size exists, then it contains exactly $p(C)$ vertices from every component $C$.
	First, for every path gadget $X$ in $G$ we define several components of $\mathcal{P}(G)$ denoted with $\mathcal{P}(X)$ as follows. Let $K'(X) = \bigcup_{x,y \in K(X)} T(x, y)$ be the set of clique vertices together with the triangles between the clique vertices. We let $\mathcal{P}(X) =$ $\bigl\{T(v, u), T(w_1, w_2),$ $T(v', u'), T(w'_1, w'_2), K'\bigr\}$. 
	Note that these four components are vertex-disjoint.
	First, we set $p(K') = 3$: indeed, $K'$ induces four triangles (i.e., odd cycles) and it is easy to see that it is impossible to cover them with less than three vertices.
	We also set $p(C) = 1$ for every $C \in \mathcal{P}(X)\setminus K'$: by construction, $C$ induces a triangle so each (connected) odd cycle transversal contains at least one vertex from $C$.	
	In total, we get $p(X) = 7$. Next, for every decoding gadget $Y$ in $G$, we define the packing $\mathcal{P}(Y) = \bigl\{T(u_x, v_x) \bigm \vert x \in \{\scrC,\scrB,\scrW,\scrD\}\bigr\}$, and set $p(C) = 1$ for every $C \in \mathcal{P}(Y)$: indeed, $C$ induces a triangle by construction. 
	So we get $p(Y) = 4$. 
	We also define the packing $\mathcal{P}(\rootpath') = \bigl\{T(v) \bigm \vert v\in V(\rootpath)\bigr\}$, and we set $p(C) = 1$ for all $c \in
	\mathcal{P}(\rootpath')$: the triangles in  $\mathcal{P}(\rootpath')$ are vertex-disjoint.
	Similarly, we define $\mathcal{P}\bigl (T(g) \cup T(g')\bigr) = \bigl\{T(g), T(g')\bigr\}$, and set $p\bigl(T(g)) = p(T(g')\bigr) = 1$.
	Note that the components we have defined are pairwise vertex-disjoint. Therefore, the union of all these components is a valid packing $\mathcal{P}(G)$. So to define the budget $k$, we sum up their $p(\cdot)$ values.
	There are $n' (2n' + 1) m$ path gadgets, $n' (2n' + 1) m$ decoding gadgets, and two guards. 
	Also, there are $2 + n'(2n'+1)m \cdot 6$ root-connected vertices (namely 2 guards and 6 root-connected vertices in each path gadget), i.e., there are $2 + n'(2n'+1)m \cdot 6$ vertices in $\rootpath$.
	Altogether, we define
	\[
		k = n'(2n'+1)m \cdot 7 + n'(2n'+1)m \cdot 4 + 2 + (2 + n'(2n'+1)m \cdot 6).
	\]
	
	\begin{lemma}\label{lem:coct-lb-ctw}
		Let $d, n$ be positive integers. Let $I$ be an instance of the $d$-SAT problem with $n$ variables and let $G$ be the graph resulting from $I$ as described above. Then $G$ has the \ctwdth{} of at most $n' + \mathcal{O}(1)$ (recall: $n' = n/2$).
	\end{lemma}
	\begin{proof}
		We describe a linear arrangement $\ell$ of $G$ of \ctwdth{} at most $n'+\mathcal{O}(1)$. We start with an empty arrangement $\ell$ and insert the vertices of $G$ into it from left to right as follows (i.e., a vertex is always inserted to the right end of $\ell$ unless explicitly stated otherwise). 
		We iterate over the pairs $(r,j)$ in the lexicographic order. 
		For a fixed column $Q_{r, j}$, we iterate through $i \in [n']$.
		For a fixed $i$, we first insert all vertices of $X^i_{r, j}$ (in arbitrary relative order) and then all vertices of $Y^i_{r, j}$ (in arbitrary relative order) to the very right of $\ell$.
		After processing all $i$, if the clause $C_j$ contains an even number of literals, we also add the vertex $w$ of the clause gadget in the column $Q_{r, j}$ to the end of $\ell$.
		After processing all pairs $(r, j)$, we add $g$ at the beginning of $\ell$ and $g'$ at the end of $\ell$.
		Now all vertices not belonging to $\rootpath$ and $\colorpath$ have been added to $\ell$.
		
		Recall that the order of vertices in both $\rootpath$ and $\colorpath$ are not fixed yet. We mentioned before that these orders are chosen in such a way that the cutwidth ``does not increase too much''. Now, after fixing the order of the remaining vertices in $\ell$, we can do this.
		First we order the vertices on $\rootpath$ according to the order of their
		root-connected neighbors, i.e. for two root-connected vertices $u, v$, the vertex $\rootpar{u}$ appears before $\rootpar{v}$ on $\rootpath$ if and only if $u$ appears before $v$ on $\ell$. 
		To insert $V(\rootpath)$ into $\ell$, for every root-connected vertex $u$, we add its neighbor $\rootpar{u}$ directly before $u$ on $\ell$ and we add the vertices of $T(\rootpar{u}) \setminus \{\rootpar{u}\}$ 	directly before $\rootpar{u}$. 
		Similarly, we order the vertices on $\colorpath$ in such a way that the following properties hold:
		\begin{enumerate}
			\item As mentioned before, the vertices of $\blackset$ appear at odd positions of $\colorpath$.
			\item The vertices of $\blackset$ appear in $\colorpath$ in such a way that 
			for two colored vertices $u, v$, the vertex $\colpar{u}$ appears before $\colpar{v}$ on $\colorpath$ if and only if $u$ appears before $v$ on $\ell$.
			\item As mentioned before, the rest of the vertices only separate the vertices of $\blackset$, so we add one such vertex between each two consecutive vertices of $\blackset$ on $\colorpath$.
		\end{enumerate}
		To insert $V(\colorpath)$ into $\ell$, for every colored vertex $u$, we proceed as follows. If $u$ is black, let $w$ denote the (unique) vertex adjacent to $u$ and $\blackset(u)$, then we insert the vertex $w$ directly before $u$ and $\blackset(u)$ directly before $w$ into $\ell$.
		If $u$ is white, we insert $\blackset(u)$ directly before $u$ into $\ell$. For a vertex $v \in V(\colorpath)\setminus \blackset$, let $u, u' \in \blackset$ be the vertices directly preceding and following it on $\colorpath$. We add $v$ on an arbitrary position between $u$ and $u'$ on $\ell$.
		
		Now all vertices of $G$ have been inserted into $\ell$.
		We show that the resulting linear arrangement $\ell$ has the \ctwdth{} of at most $n' + \mathcal{O}(1)$. 
		First, let us forget about the vertices of $\rootpath$ and $\colorpath$ and edges incident to them and let us denote with $\ell'$ the restriction of $\ell$ obtained this way. 
		Then, for every $(r, j)$ the vertex set $B_{r, j}$ consisting of $X_{r, j} \cup Y_{r, j}$ and (if exists) the vertex $w$ of the clause gadget in the column $Q_{r, j}$ forms a consecutive block in $\ell'$. 
		Therefore, for any $(r, j) \neq (r', j')$, no two edges induced by $E(B_{r, j})$ and $E(B_{r', j'})$, respectively overlap on $\ell'$ (and hence, on $\ell$ as well). 
		Since for any $(r, j)$, the set $B_{r, j}$ has only a constant size, all such edges totally contribute only a constant to the size of any cut in $\ell$.
		Further, the path gadgets of the same path $P_i$ appear in the same order in $P_i$ and in $\ell$. Therefore, the edges between different gadgets of the path $P_i$ do not overlap on $\ell$, i.e., two such edges can only overlap if they run along distinct paths. Since there are $n'$ such paths, all such edges totally contribute at most $n'$ to the size of any cut (we call these edges \emph{interpath} for a moment). 
		Similarly, the edges between $g$ (resp.\ $g'$) and the left-end (resp.\ right-end) vertices of the leftmost (resp.\ rightmost) gadgets of every path do not overlap with interpath edges. 
		Finally, the above-mentioned edges incident to $g$ do not overlap with the above-mentioned edges incident to $g'$.
		Therefore, in every cut of $\ell$, there are at most $n'$ (at most one per path $P_i$) edges of the last three types.
		Next, we consider a clause gadget $Z_{r,j}$ in some column $Q_{r, j}$. This gadget is an odd cycle using some vertices of some gadgets $X_{r,j}^i$ (and possibly an auxiliary vertex $w$). 
		The clause $C_j$ contains at most $d$ literals. Therefore, there are at most $d$ such $i$ where $Z_{r, j}$ contains a vertex of $X_{r, j}^i$.
		Further, recall that every path $P_i$ corresponds to two boolean variables of $I$. There are four partial assignments of these variables and hence, for every $i$, $Z_{r, j}$ contains at most $4$ vertices of $X_{r, j}^i$.
		Altogether, we obtain that $Z_{r, j}$ contains at most $4d$ vertices from path gadgets and possibly the vertex $w$. So $Z_{r, j}$ contains at most $4d+1$ edges. 
		Recall that by construction, the columns $Q_{r, j}$ form consecutive blocks in $\ell$ so the edges of distinct clause gadgets do not overlap. So the edges of clause gadgets totally contribute at most $4d+1$, i.e., a constant, to the size of every cut. 
		Altogether, the cutwidth of $\ell'$ is bounded by $n' + \mathcal{O}(1)$.
		Now we show that after adding the vertices of $\rootpath$ and $\colorpath$ to $\ell'$ (to obtain $\ell$), the cutwidth can increase by a constant only.
		
		For a root-connected vertex $v$, let $B_{\rootpath}(v)$ be the set defined as $B_{\rootpath}(v) = \{v\} \cup T(\rootpar{v})$.
		Similarly, for a black vertex $v$, let $w$ be the (unique) vertex adjacent to $v$ and $\blackset(v)$, then with $B_{\colorpath}(v)$ we denote the set $B_{\colorpath}(v) = \{v, w, \blackset(v)\}$. And for a white vertex $w$, with $B_{\colorpath}(v)$ we denote the set $B_{\colorpath}(v) = \{v, \blackset(v)\}$. 
		Let $x \in \{\rootpath, \colorpath\}$. 
		Since the vertices in $V(x)$ appear in the same order on $x$ and
		on $\ell$, the edges of $x$ do not overlap. Therefore, every cut of $\ell$ contains at most 1 edge of $x$.
		Further, for two different vertices $u, v$, where both sets $a = B_x(u), b = B_x(v)$ exist, by the construction of $\ell$, we have that the elements of $a$ and $b$ do not interleave in $\ell$, and hence a cut of $\ell$ can only contain edges induced by one of the sets $a$ and $b$. By construction, a cut of $\ell$ contains at most 2 edges induced by a set $B_x(v)$ when $x=\rootpath$ and at most one edge when $x = \colorpath$. Altogether, a cut of $\ell$ can contain at most one edge of $\rootpath$, one edge of $\colorpath$, at most two edges induced by some $B_{\rootpath}(v)$ and at most one edge induced by some $B_{\colorpath}(w)$ adding at most 5 to the total \ctwdth{} of the graph.
		Now it can be verified that we have taken into account all edges of $G$ and so every cut of $\ell$ contains at most $n' + \mathcal{O}(1)$ edges.
	\end{proof}

	\begin{remark}\label{rem:coct}
		Let $S$ be a \coct{} of $G$ of size $k$ and let $X = X_{r,j}^i$ be an arbitrary path gadget in $G$. 
		As argued before, $S$ contains exactly three clique vertices of $X$.
		Let $s_x$ be the unique clique vertex of $X$ not contained in $S$. 
		Hence, $x$ is the \emph{state} of the gadget $X$ (in $S$).
		Recall that for every $y \in \{\scrB, \scrW, \scrC, \scrD\} \setminus \{x\}$, there is a triangle $T(s_x, v_y)$ in $G$ where $v_y$ belongs to the decoding gadget $Y_{r,j}^i$. Since $S$ is a \coct{} of $G$ not containing $x$, and the third vertex of this triangle is only adjacent to $v_y$ and $s_x$, we can assume that the set $S$ contains $v_y$.
		By the choice of $k$, this implies that $S$ does not contain $u_y$. So from the right side of $Y_{r,j}$, only the vertex $u_x$ can be contained in $S$.
	\end{remark}

	\begin{lemma}
		If $I$ is satisfiable, then $G$ admits a COCT of size at most $k$.
	\end{lemma}
	\begin{proof}
	Let $\pi$ be an assignment satisfying $I$ and for every $i$, let $\pi_{i}$ be the restriction of $\pi$ to the variable pair $U_i$.  Let $\phi(\pi_i)$ be the state of a path gadget to which the partial assignment $\pi_i$ is mapped. We define the following sets (See \cref{fig:append-coct-states}):
	\begin{itemize}
		\item $S_{\scrC} = \{\verW, \verB, \verD, v, v', w_1, w'_2\}$,
		\item $S_{\scrD} = \{\verW, \verB, \verC, v, v', w'_1, w_2\}$.
		\item $S_{\scrW} = \{\verB, \verC, \verD, u, u', w_1, w'_1\}$,
		\item $S_{\scrB} = \{\verW, \verC, \verD, u, u', w_1, w'_1\}$,
	\end{itemize}
	Now we describe a set $S$ of size $k$ and later we will show that it is a \coct{} of $G$.
	We start with an empty set and add the following vertices to $S$. First, we add all vertices from $V(\rootpath)$ and the vertices $g$ and $g'$. 
	Then, for each path gadget $X^i_{r, j}$ we add all vertices from $S_{\phi(\pi_i)}(X^i_{r, j})$. 
	We say that the path gadget $X^i_{r, j}$ is now in the state $\phi(\pi_i)$.
	Also from the decoding gadget $Y^i_{r, j}$ attached to it, we add the vertex $u_{\phi(\pi_i)}(Y^i_{r, j})$ and the vertices $v_{x'}(Y^i_{r, j})$ for all $x' \in \{\scrB, \scrW, \scrC, \scrD\} \setminus \phi(\pi_i)$. 
	This concludes the construction of the set $S$.
	Observe that for every component $C$ of the packing $\mathcal{P}(G)$, the set $S$ contains exactly $p(C)$ vertices, and it does not contain any further vertices. Therefore, we have $|S| = k$.

	\begin{figure}[t]
	\centering

	\begin{subfigure}[t]{.49\linewidth}
	\includegraphics[width=.99\linewidth]{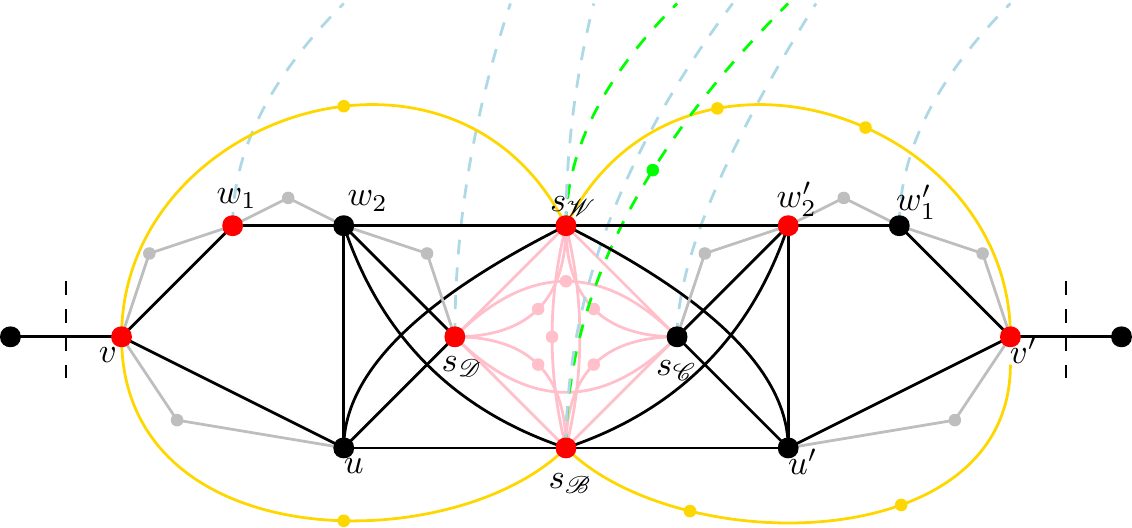}
	\caption{The set $S_{\scrC}$ corresponding to the state $\scrC$.}
	\end{subfigure}
	\hfill
	\begin{subfigure}[t]{.49\linewidth}
	\includegraphics[width=.99\linewidth]{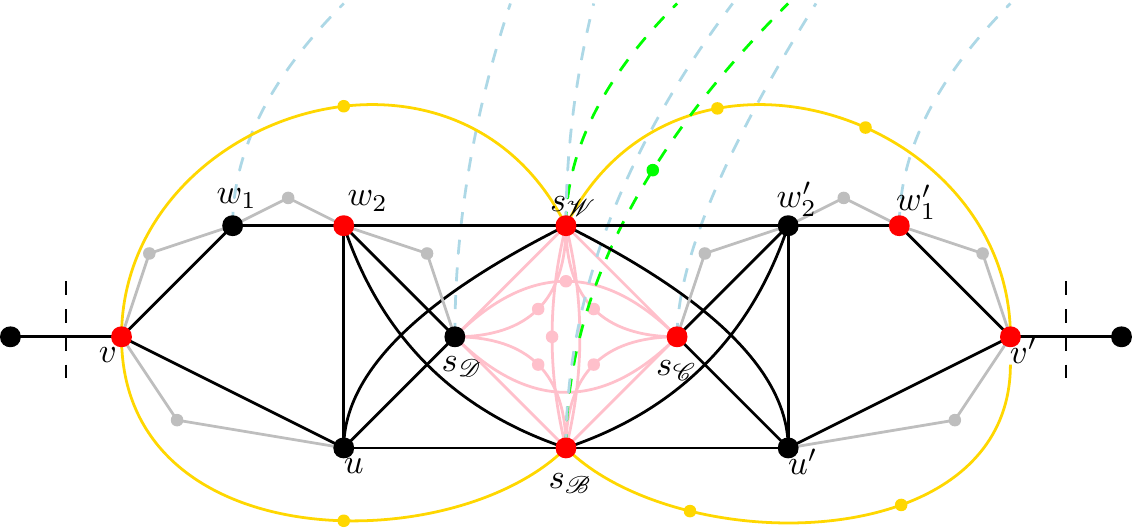}
	\caption{The set $S_{\scrD}$ corresponding to the state $\scrD$.}
	\end{subfigure}

	\begin{subfigure}[t]{.49\linewidth}
	\includegraphics[width=.99\linewidth]{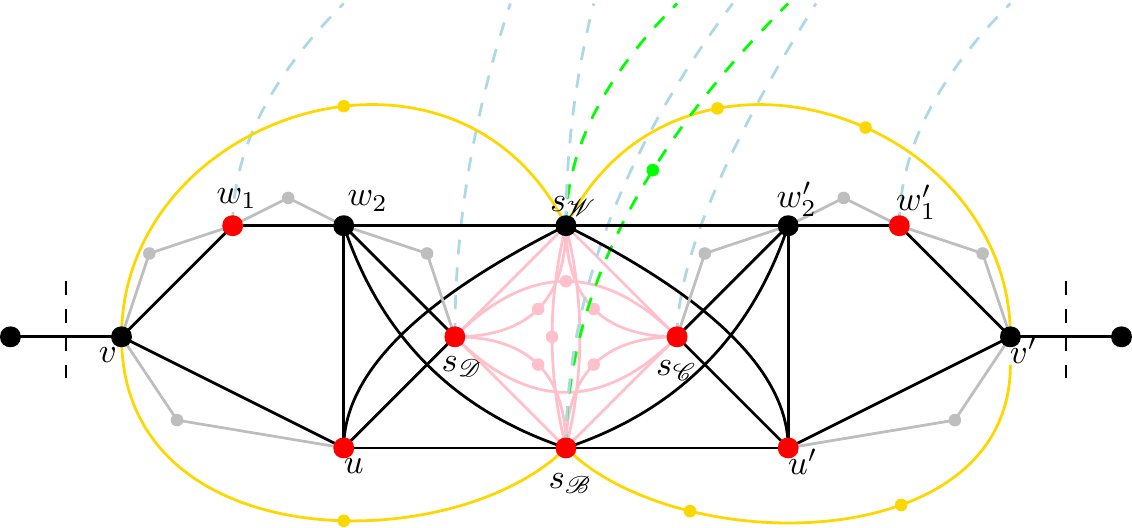}
	\caption{The set $S_{\scrW}$ corresponding to the state $\scrW$.}
	\end{subfigure}
	\hfill
	\begin{subfigure}[t]{.49\linewidth}
	\includegraphics[width=.99\linewidth]{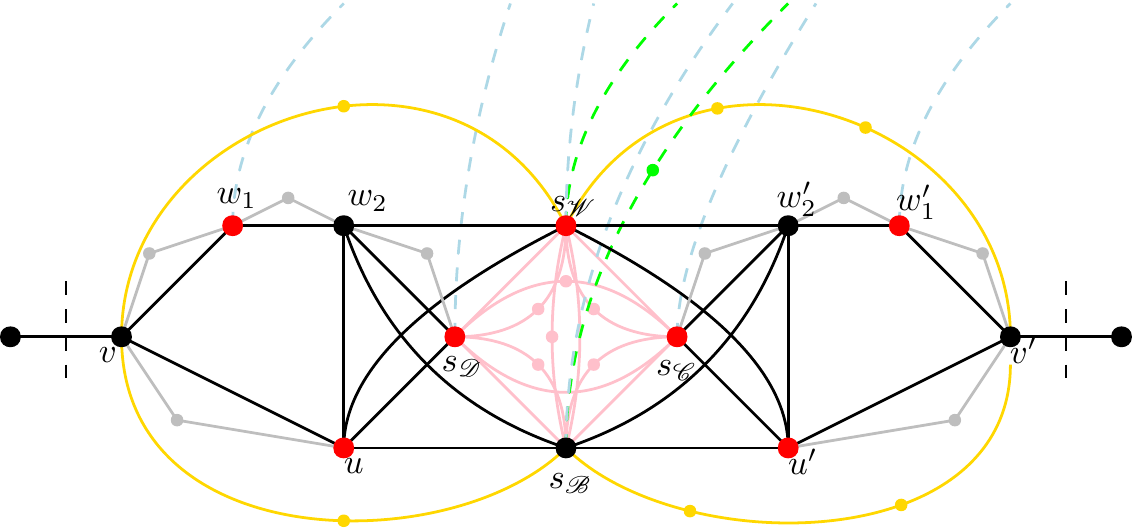}
	\caption{The set $S_{\scrB}$ corresponding to the state $\scrB$.}
	\end{subfigure}

	\caption{The sets $S_x, x\in\{\scrD, \scrC, \scrW, \scrB\}$ corresponding to different states of a path gadget. We depict the vertices of $S_x$ in red.}
	\label{fig:append-coct-states}
	\end{figure}
    
    To prove that $S$ is a \coct{} of $G$, we first show that it is an odd cycle transversal of $G$.
	First of all, observe that for every state $x \in \{\scrD, \scrC, \scrW, \scrB\}$, the set $S_x$ hits all odd cycles in any subgraph induced by a single path gadget (this can be verified in \cref{fig:append-coct-states}). Since for every path gadget, we have added a set $S_x$ (for some $x \in \{\scrD, \scrC, \scrW, \scrB\}$) of vertices of this path gadget to $S$, this set hits all odd cycles in any subgraph induced by a single path gadget.
	Next, recall that by construction, for every pair $u_x, v_x$ of every decoding gadget, exactly one of these vertices belongs to $S$. 
	Therefore, the set $S$ also hits all odd cycles in any subgraph induced by a single decoding gadget. 
	Now consider an arbitrary clause gadget $Z_{r, j} = Z$.
	Since $\pi$ is an assignment satisfying $I$, there is a pair of variables $U_{i(j)}$ (for some $i(j) \in [n'])$ such that the partial assignment $\pi_{i(j)}$ satisfies $C_j$. 
	So by the construction of the graph $G$, the vertex $u_{\phi(\pi_{i(j)})}(Y^{i(j)}_{r, j})$ belongs to $Z$. 
	Since the path gadget $X_{r, j}^{i(j)}$ is in the state $\phi(\pi_{i(j)})$, we also have $u_{\phi(\pi_{i(j)})}(Y^{i(j)}_{r, j}) \in S$ by the construction of $S$. 
	Therefore, the odd cycle $Z$ is hit by $S$.
	Altogether, we obtain that $S$ hits all odd cycles in every subgraph induced by a single gadget. Next we show that it also hits the remaining odd cycles in $G$.
	
	We observe the following. Vertices of decoding gadgets have the three following types of neighbors. 
	First, the vertices of the root-path $\rootpath$ (all belong to $S$). Second, vertices of decoding gadgets (and parity vertices of clause gadgets which only have decoding gadget as their neighbors). 
	And finally, vertices of the path gadget the corresponding decoding gadget is attached to. 
	Now suppose there is an odd cycle in $G - S$ that uses at least one vertex from a path gadget and at least one vertex from a decoding gadget.
	Then by the above observation, such a cycle needs to contain an edge from some vertex $a$ of a path gadget $X$ to some vertex $b$ of a decoding gadget $Y$.
	By construction of the graph $G$, this is only possible if $Y$ is attached to $X$, $a = s_x(X)$ (for some state $x$), and $b = v_y(Y)$ for some $y \neq x \in \{\scrB, \scrW, \scrC, \scrD\}$. But by the construction of the set $S$, exactly one of vertices $a$ and $b$ belongs to $S$ so such a cycle does not exist in $G - S$ -- a contradiction. 
	Therefore, any odd cycle in $G - S$ either does not use vertices from decoding gadgets at all or it only uses vertices of (possibly several) decoding gadgets and parity vertices. 
	
	\begin{figure}[t]
	\centering

	\begin{subfigure}[t]{.49\linewidth}
	\includegraphics[width=.99\linewidth]{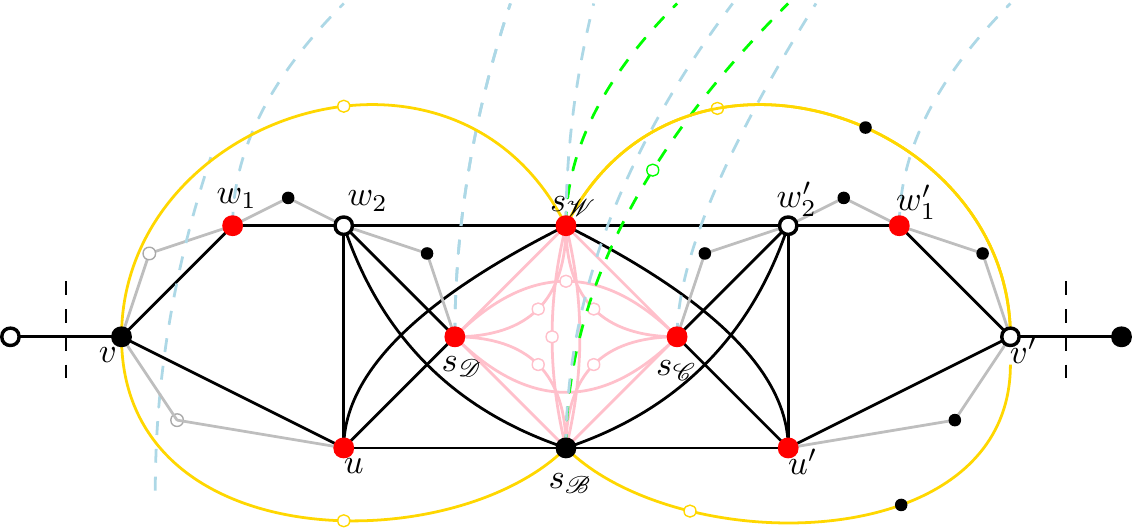}
	\end{subfigure}
	\hfill
	\begin{subfigure}[t]{.49\linewidth}
	\includegraphics[width=.99\linewidth]{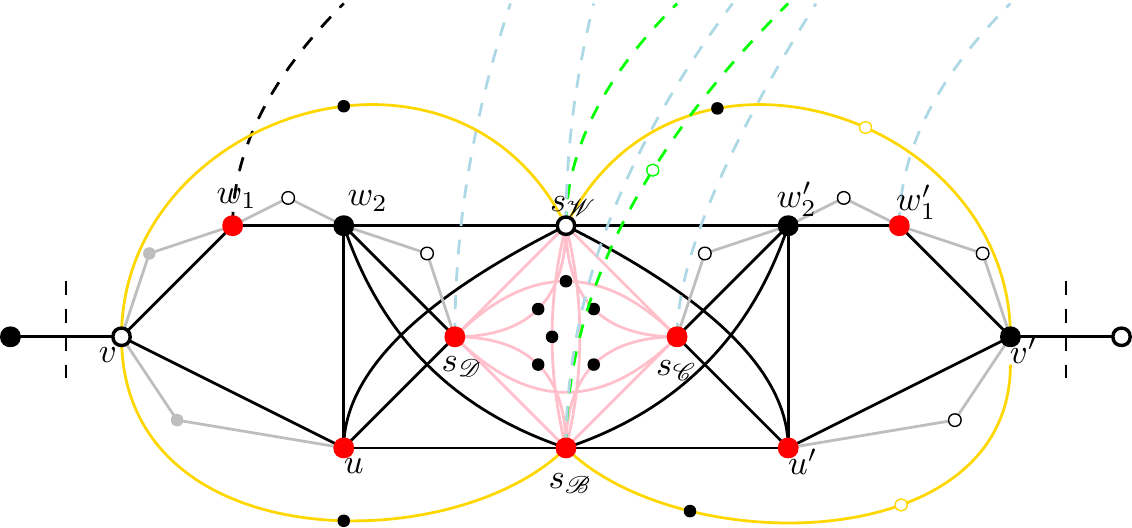}
	\end{subfigure}

	\caption{A proper 2-coloring of a path gadget in the state $\scrB$ (on the left) and $\scrW$ (on the right) after the removal of $S$. The vertices of $S$ are depicted in red and the 2-coloring of the remainder uses the colors black and white.}
	\label{fig:append-coct-states-coloring}
	\end{figure}

	First, suppose the former case occurs.
	Such a cycle can only use vertices of path gadgets, vertices from the color-path $\colorpath$, and subdivision-vertices between any black vertex $a$ and the corresponding vertex $\colpar{a}$ on the color-path, i.e. it contains no vertices from decoding gadgets and no vertex from $V(\rootpath) \cup \{g, g'\} \subseteq S$.
	Let us consider some vertex $a$ from a path gadget $X$ on this cycle. 
	Observe that $X$ can only be in the state $\scrW$ or $\scrB$ since otherwise the set of vertices of $X - S$ is separated from the remainder of $G - S$ (recall \cref{fig:append-coct-states} for illustration).
	Now we provide a proper 2-coloring with colors $\black$ and $\white$ of the subgraph of $G-S$ induced by the vertices of $\colorpath$ and vertices of all path gadgets in states $\scrW$ and $\scrB$.
	The existence of such a coloring would contradict the existence of an odd cycle in it.
	The vertices of path gadgets are colored in the natural way: If the state is $\scrW$, then we color the white vertex $\verW$ with $\white$ and extend the coloring to the remaining vertices outside $S$ to an ``almost unique'' proper 2-coloring (see \cref{fig:append-coct-states-coloring} on the right).
	And if the state is $\scrB$, then we color the black vertex $\verB$ with $\black$ and extend the coloring to the remaining vertices outside $S$ to an ``almost unique'' proper 2-coloring (see \cref{fig:append-coct-states-coloring} on the left). 
	We say ``almost unique'' since there is a set of (after the removal of $S$) isolated subdivision-vertices whose color does not matter.
	Note that black (resp.\ white) vertices are colored with $\black$ (resp.\ white) color in this construction.
	Next, for every white vertex $a$, we color the vertex $\colpar{a}$ with $\black$; and for every black vertex $a$, we color the vertex $\colpar{a}$ with $\black$ and the subdivision-vertex between $a$ and $\colpar{a}$ with $\white$.
	Note that this last operation only assigns color $\black$ to some vertices from $\blackset$. 
	Therefore, this coloring can naturally be extended to the whole color-path $\colorpath$ so that it is proper on $G[\colorpath]$: namely, we color all vertices in $\blackset$ with $\black$ and all vertices in $\colorpath \setminus \blackset$ with $\white$. 
	Now we claim that this 2-coloring is indeed proper on the above-mentioned subgraph.
	First, for edges having both end-vertices in the same path gadget, this property can be verified in \cref{fig:append-coct-states-coloring}.
	Second, for edges having both end-points in distinct path gadgets, we have that this edge connects the right-end point of some path gadget and the left-end of the next gadget on the same path.
	So these two gadgets are in the same state (by the construction of $S$) and hence, the colorings of this gadgets are equal. 
	So again in \cref{fig:append-coct-states-coloring}, it can again be verified that in both possible states, the end-vertices of this edge have different colors.
	Third, for edges such that at least one end-vertex belongs to the color-path $\colorpath$ or to a subdivision-vertex adjacent to some vertex on the color-path, the construction of the coloring directly ensures that the end-vertices have different colors.
	Finally, for the edges of $\colorpath$, this property holds by construction.
	Therefore, we obtain a proper 2-coloring of this subgraph so that no odd cycle of the former case can occur.

	Now we move to the latter case, where an odd cycles uses only vertices of decoding gadgets and parity vertices.
	Recall that for every pair $u_x, v_x$ from any decoding gadget, exactly one vertex belongs to $S$.
	Let $u_x, v_x$ be a pair in a decoding gadget $Y$, and let $w$ be the private vertex added to create the triangle $T(u_x, v_x)$. Note that $u_x$ separates $v_x$ and $w$ from all other vertices of all decoding gadgets and parity vertices in $G$. Hence, either $v_x \in S$ and $w$ has degree one in $G - S$, or $u_x \in S$, and both $v_x$ and $w$ are separated from all other vertices of path gadgets and parity vertices. In both cases, both $v_x$ and $w$ cannot be included in an odd cycle of $G - S$ over path gadgets and parity vertices only.
	Therefore, a simple odd cycle of this type in $G - S$ can only use vertices from the right side of decoding gadgets and parity vertices.
	Further, each vertex from $u_x$ on the right side of a decoding gadget $Y^i_{r,j}$ can be included in at most one clause gadget, namely $Z_{r,j}$. Hence, in $G$, $u_x$ is adjacent only to its other pair vertex $v_x$, the private neighbor $w$ added to create $T(v_x, u_x)$ and at most two other vertices on $Z_{r, j}$. Each parity vertex belongs also to a clause gadget $Z_{r, j}$ and is adjacent to exactly two vertices on this clause gadget. Since the vertices of different clause gadgets are disjoint, an odd cycle of $G - S$ over decoding gadgets and parity vertices can only be a clause gadget. But as we have already shown above, each such cycle is hit by $S$.
	Altogether, we obtain that $G - S$ contains no odd cycles, and hence, $S$ is an odd cycle transversal of $G$.

	Finally, we show that the subgraph $G[S]$ is connected. Since $V(\rootpath)
	\subseteq S$, it suffices to prove that there is a path from each
	vertex in $S$ to $\rootpath$ in $G[S]$. This is clearly the case for the root-path itself and for decoding gadgets, the guards $g$ and $g'$, and clique vertices in each
	path gadget since these vertices are root-connected. So it remains to prove that this is also true for vertices in path gadgets.
	In states $\scrW$ and $\scrB$, each vertex of a path gadget that belongs to $S$ is either root-connected itself or it has a root-connected neighbor in $S$ in the same path gadget (this can be verified in \cref{fig:append-coct-states}. 
	For states $\scrC$ (resp.\ $\scrD$), this is also true for all vertices but possibly $v'$ (resp.\ $v$). 
	However, we recall that all path gadgets on the same path are in the same state.
	So for the state $\scrC$ (resp.\ $\scrD$) and vertex $v'$ (resp.\ $v$), it holds that either the vertex $v$ (resp.\ $v'$) of the next (resp.\ previous) gadget on the same path or the guard $g'$ (resp. $g$) belongs to $S$ and for this vertex we already know that it is connected to the root-path $\rootpath$ in $G[S]$ (see \cref{fig:append-coct-states} for illustration again).
	Altogether, we obtain that $S$ is indeed a \coct{} of $G$ of size $k$.
	\end{proof}

	\begin{lemma}
		If $G$ has a \coct{} of size $k$, then $I$ is satisfiable.
	\end{lemma}
	\begin{proof}
	Let $S$ be an \coct{} of $G$. 
	We show how to construct an assignment satisfying $I$ from it.
	We make some observations concerning the structure of $S$. 
	Recall that by the construction of the packing $\mathcal{P}(G)$, the value $k$ is a lower bound on the size of any \coct{} of $G$. 
	Therefore, if we remove at least one vertex from $S$, then it is not a \coct{} of $G$ anymore. We will strongly rely on this property.
	We also recall that by the choice of $k$, the set $S$ contains exactly $p(C)$ vertices from every component $C$ of the packing $\mathcal{P}(G)$ and no further vertices.
	First, we may assume that $S$ contains no vertices added to create a triangle between two vertices for the following reason.
	If a vertex $v \in S$ was added to create a triangle between some vertices $v_1$ and $v_2$ (which are the only neighbors of $v$), then due to the connectivity of $G[S]$, the set $S$ also contains the vertex $v_q$ for some $q \in [2]$. 
	But then since $v$ has the degree of 2 in $G$ and there is an edge $\{v_1,v_2\}$ in $G$, the set $S \setminus \{v\}$ is a \coct{} of $G$ of smaller size -- a contradiction.
	We also assume that the set $S$ does not contain any vertex added to create a triangle at some vertex $v$.
	If a vertex $w$, added to create a triangle at some vertex $v$, was in $S$, then since the packing $\mathcal{P}(G)$ contains a component $T(v)$ with $p\bigl(T(v)\bigr) = 1$, the set $S$ contains neither $v$ nor the third vertex of that triangle which makes $w$ isolated in $G[S]$ -- a contradiction to the connectivity of $G[S]$.
	Further, the set $S$ does not contain vertices of path gadgets that have been introduced to subdivide some edge twice and parity vertices: none of such vertices belongs to a component of $\mathcal{P}(G)$.
    Also recall that every vertex of $V(\rootpath) \cup \{g, g'\}$ has a triangle at it and this triangle is a component of $\mathcal{P}(G)$.
    So by the above observation, we have $V(\rootpath) \cup \{g, g'\} \subseteq S$. 
    Similarly, recall that for every path gadget $X$, the set $K'(X)$ of all clique vertices and subdivision-vertices between them forms a component of $\mathcal{P}(G)$ with $p\bigl(K'(X)\bigr) = 3$. By the above observation, none of the subdivision-vertices is contained in $S$. So exactly three clique-vertices of $X$ belong to $S$. For the unique clique-vertex $x \notin S$ of $X$, we say that $X$ is in the state $x$.  
    
    Now we claim that there exists an index $r^* \in [2n'+1]$ such that for 
	every $i$ and every $j_1, j_2 \in [m]$, the path gadgets $X^i_{r^*, j_1}$ and $X^i_{r^*, j_2}$ have the same state in $S$. 
	Before proving this claim, let us show how it concludes the proof.
	For every $i$, let $x_i$ denote the state of $X^i_{r^*, 1}$ in $S$.
	We define an assignment $\Pi$ in such a way that for every $i$, it assigns the the value $\phi^{-1}(x_i)$ to the variable group $U_i$. Since $U_1, \dots, U_{n'}$ is a partition of the variable set, $\Pi$ is well-defined.
	Consider an arbitrary $j$.
	Recall that the clause gadget $Z_{r^*, j}$ is an odd cycle by construction. 
    By the observation above, the set $S$ does not contain a parity vertex of $Z_{r^*,j}$ (if such a vertex exists). So the connected odd cycle transversal $S$ contains some other vertex $u^*$ from it, i.e. by the construction of $Z_{r^*, j}$, there must exist a decoding gadget $Y^{i^*}_{r^*, j}$ (for some $i^* \in [n']$) and a vertex $u_{x^*}(Y^{i^*}_{r^*, j}) = u^*$ (for some $x^* \in \{\scrB, \scrW, \scrC, \scrD\}$) in it such that $u_{x^*} \in S$ and the partial assignment $\phi^{-1}(x^*)$ of $U_{i^*}$ satisfies $C_j$. 
	By \cref{rem:coct}, the property $u_{x^*} \in S$ implies that the set $S$ defines the state $x^*$ in $X^{i^*}_{r^*,j}$. 
	Recall that $S$ defines the state $x_i$ in $X^i_{r^*, 1}$.
	By the choice of $r^*$, the states of gadgets $X^i_{r^*, 1}$ and $X^i_{r^*, j}$ are the same, i.e., $x^* = x_i$. Therefore, $\pi_i = \phi^{-1}(x_i) = \phi^{-1}(x^*)$ (and hence, also $\pi$) satisfies $C_j$. 
	Since $j$ was chosen arbitrarily, all clauses of $I$ are indeed satisfied by $\pi$ and $I$ admits a satisfying assignment.

	Now it remains to prove the above claim.
	For this, we first consider a partial ordering $\preccurlyeq$ on the set of states defined as $\scrD \preccurlyeq \scrW \preccurlyeq \scrC, \scrD \preccurlyeq \scrB \preccurlyeq \scrC$.
	Let $X_1$ and $X_2$ be two consecutive path gadgets on some path and let $x_1$ and $x_2$ be their states in $S$, respectively.
	We show that $x_1 \preccurlyeq x_2$ holds.
	
	First, suppose that $x_1 = \scrC$ holds. We then show that $x_2 = \scrC$ holds as well.
	So we have $\verC(X_1) \notin S$. 
	We recall that triangles are odd cycles and hence, the set $S$ hits every triangle.
	Due to the triangle between $\verC$ and $w_2'$, it then holds that $w_2'(X_1) \in S$. By the construction of the packing $\mathcal{P}(G)$, we have $w_1'(X_1) \notin S$. Due to the triangle between $w_1'$ and $v'$, we have $v'(X_1) \in S$. Again, by the construction of $\mathcal{P}(G)$, we have $u'(X_1) \notin S$. Since $G[S]$ is connected and $v'(X_1) \in S$, and $w_1'(X_1), u'(X_1) \notin S$, the vertex $v'(X_1)$ must be connected to the remainder of the graph $G[S]$ via $v(X_2)$, i.e., $v(X_2) \in S$.
	Again, by the construction of $\mathcal{P}(G)$, we have $u(X_2) \notin S$. Since $v'(X_1)$ and $v(X_2)$ still need to be connected to the remainder of $G[S]$, it then holds that $w_1(X_2) \in S$ and (by $\mathcal{P}(G)$) also $w_2(X_2) \notin S$. Then due to the triangle between $w_2$ and $\verD$, we have $\verD(X_2) \in S$. 
	There are also triangles $u, w_2, \verW$ and $u, w_2, \verB$. Together with $u(X_2), w_2(X_2) \notin S$, we obtain $\verW(X_2), \verB(X_2) \in S$. Finally, $\verD(X_2), \verW(X_2), \verB(X_2) \in S$ and the construction of $\mathcal{P}(G)$ imply that $\verC(X_2) \notin S$, i.e., the state of $X_2$ is $x_2 = \scrC$ as claimed. 
	
	Very similarly, we show that if $x_2 = \scrD$ holds, then we also have $x_1 = \scrD$. 
	So let $\verD(X_2) \notin S$.
	Due to the triangle between $w_2$ and $\verD$, we have $w_2(X_2) \in S$. So due to packing $\mathcal{P}(G)$, we also have $w_1(X_2) \notin S$. Triangle between $w_1$ and $v$ implies $v(X_2) \in S$ and $\mathcal{P}(G)$ also implies $u(X_2) \notin S$. Since $G[S]$ is connected, the vertex $v(X_2)$ is connected to the remainder of this subgraph. Recall that $w_1(X_2), u(X_2) \notin S$. So $v(X_2)$ must be connected through $v'(X_1)$ to the remainder of the graph, i.e., we have $v'(X_1) \in S$. By $\mathcal{P}(G)$, we have $u'(X_1) \notin S$ and due to the connectivity of $v'(X_1)$ with the remainder of $G[S]$, we then need to have $w_1'(X_1) \in S$. The packing $\mathcal{P}(G)$ then implies $w_2'(X_1) \notin S$. Due to the triangle between $w_2'$ and $\verC$, we get $\verC(X_1) \in S$. The triangles $w_2', u', \verW$ and $w_2', u', \verB$ imply $\verB(X_1), \verW(X_1) \in S$. Finally, $\verC(X_1), \verB(X_1), \verW(X_1) \in S$ and the packing $\mathcal{P}(G)$ imply $\verD(X_1) \notin S$. So the state of $X_1$ is $x_1 = \scrD$ as claimed.
	
	Now we show that no transition from $x_1 = \scrB$ to $x_2 = \scrW$ occurs.
	So let $x_1 = \scrB$ and $x_2 = \scrW$. We then have $\verB(X_1), \verW(X_2) \notin S$.
	Due to the triangle $\verB, w'_2, u'$, one of the vertices $u'(X_1)$ and $w'_2(X_1)$ must be in $S$.
	First, suppose that $u'(X_1) \notin S$ holds. Then we have $w'_2(X_1) \in S$. 
	Due to $\mathcal{P}(G)$, we have $w'_1(X_1) \notin S$. Due to the triangle between $w_1'$ and $v'$, we have $v'(X_1) \in S$ and by $\mathcal{P}(G)$, we also have $u'(X_1) \notin S$.
	Since $G[S]$ is connected and $v'(X_1) \in S$, $w_1'(X_1), u'(X_1) \notin S$, the vertex $v'(X_1)$ is connected to the remainder of the graph $G[S]$ via $v(X_2)$, i.e., $v(X_2) \in S$.
	Again, by the construction of $\mathcal{P}(G)$, we have $u(X_2) \notin S$. Since $v'(X_1)$ and $v(X_2)$ still need to be connected to the remainder of $G[S]$, it then holds that $w_1(X_2) \in S$ and (by $\mathcal{P}(G)$) also $w_2(X_2) \notin S$. Then due to the triangle between $w_2$ and $\verD$, we have $\verD(X_2) \in S$. 
	There is also a triangle $u, w_2, \verW$. Together with $u(X_2), w_2(X_2) \notin S$, we obtain $\verW(X_2) \in S$ -- a contradiction to $x_2 = \scrW$.

	So we may assume that $u'(X_1)$ belongs to $S$. Then due to $\pack$, the vertex $v'(X_1)$ does not belong to $S$. 
	First, suppose that $v(X_2) \in S$ holds. By $\pack$, we then have $u(X_2) \notin S$. Together with $v'(X_1) \notin S$, the connectivity of $G[S]$ now implies that $w_1(X_2) \in S$ and due to $\pack$, also $w_2(X_2) \notin S$.
	Now recall that $x_2 = \scrW$, i.e., $\verW(X_2) \notin S$. Then the triangle $w_2, u, \verW$ in $X_2$ is not hit by $S$ -- a contradiction. 
    So it holds that $v(X_2) \notin S$.
    Due to triangles between $v$ and $w_1$ and between $v$ and $u$, we have $w_1(X_2), u(X_2) \in S$. Due to $\pack$, we obtain $w_2(X_2) \notin S$.
	Since $S$ is an \oct{} of $G$, the graph $G - S$ admits a proper 2-coloring. 
	Consider such a fixed coloring with colors $\black$ and $\white$. 
	We may assume that $\verB(X_1)$ is colored $\black$. 
	Recall that by the construction of $\mathcal{P}(G)$, no vertex of the path $\colorpath$ belongs to $S$.
	Also recall that $\verB$ is black so the vertex $\colpar{\verB(X_1)} \in \blackset$ must have the color $\black$ as well. 
	Since $\colorpath$ is a path, and the vertices of $\blackset$ are exactly the vertices of $\colorpath$ on odd positions, all vertices of $\blackset$ are colored with $\black$, and all vertices of $V(\colorpath) \setminus \blackset$ are colored $\white$.
	Therefore, the white vertex $\verW(X_2)$ having a neighbor $\colpar{\verW(X_2)} \in \blackset$ is colored $\white$. 
	Since the edge $\{v',\verB\}$ is subdivided twice and $\verB(X_1)$ is colored $\black$, the vertex $v'(X_1)$ must be colored $\white$.
	On the other hand, since the edge $\{\verW, v\}$ is subdivided once and $\verW(X_2)$ is colored $\white$, the vertex $v(X_2)$ must also be colored $\white$. So the edge $\{v'(X_1), v(X_2)\}$ has two end-vertices both colored $\white$ -- a contradiction to a proper 2-coloring of $G - S$. Hence, no transition from $x_1 = \scrB$ to $x_2 = \scrW$ occurs.
	
	The transition from $x_1 = \scrW$ to $x_2 = \scrB$ also does not occur by symmetry.
	However, for the sake of completeness, we provide the complete proof of this statement.
	Let $x_1 = \scrW$ and $x_2 = \scrB$. We then have $\verW(X_1), \verB(X_2) \notin S$.
	Due to the triangle $\verB, w_2, u$, one of the vertices $u(X_2)$ and $w_2(X_2)$ must be in $S$.
	First, suppose that $u(X_2) \notin S$ holds. Then we have $w_2(X_2) \in S$. 
	Due to $\mathcal{P}(G)$, we have $w_1(X_2) \notin S$. 
	Due to the triangle between $w_1$ and $v$, we have $v(X_2) \in S$ and by $\mathcal{P}(G)$, we also have $u(X_2) \notin S$.
	Since $G[S]$ is connected and $v(X_2) \in S$, $w_1(X_2), u(X_2) \notin S$, the vertex $v(X_2)$ is connected to the remainder of the graph $G[S]$ via $v'(X_1)$, i.e., $v(X_1) \in S$.
	Again, by the construction of $\mathcal{P}(G)$, we have $u'(X_1) \notin S$. Since $v(X_2)$ and $v'(X_1)$ still need to be connected to the remainder of $G[S]$, it then holds that $w_1'(X_1) \in S$ and (by $\mathcal{P}(G)$) also $w'_2(X_1) \notin S$. Then due to the triangle between $w'_2$ and $\verC$, we have $\verC(X_1) \in S$. 
	There is also a triangle $u', w'_2, \verW$. Together with $u'(X_1), w'_2(X_1) \notin S$, we obtain $\verW(X_2) \in S$ -- a contraction to $x_1 = \scrW$.
	
	So we may assume that $u(X_2)$ belongs to $S$. Then due to $\pack$, the vertex $v(X_2)$ does not belong to $S$. 
	First, suppose that $v'(X_1) \in S$ holds. By $\pack$, we then have $u'(X_1) \notin S$. Together with $v(X_2) \notin S$, the connectivity of $G[S]$ now implies that $w'_1(X_1) \in S$ and due to $\pack$, also $w'_2(X_1) \notin S$.
	Now recall that $x_1 = \scrB$, i.e., $\verB(X_1) \notin S$. Then the triangle $w'_2, u', \verB$ in $X_2$ is not hit by $S$ -- a contradiction. 
    So it holds that $v'(X_1) \notin S$.
    Due to triangles between $v'$ and $w'_1$ and between $v'$ and $u'$, we have $w'_1(X_1), u'(X_1) \in S$. And due to $\pack$, we obtain $w'_2(X_1) \notin S$.
	Since $S$ is an \oct{} of $G$, the graph $G - S$ admits a proper 2-coloring. 
	Consider such a fixed coloring with colors $\black$ and $\white$. 
	We may assume that $\verW(X_2)$ is colored $\white$. 
	Recall that by the construction of $\mathcal{P}(G)$, no vertex of the path $\colorpath$ belongs to $S$.
	Also recall that $\verW$ is white so the vertex $\colpar{\verW(X_2)} \in \blackset$ must have the color $\black$. 
	Since $\colorpath$ is a path, and the vertices of $\blackset$ are exactly the vertices of $\colorpath$ on odd positions, all vertices of $\blackset$ are colored with $\black$, and all vertices of $V(\colorpath) \setminus \blackset$ are colored $\white$.
	Therefore, the black vertex $\verB(X_1)$ having a neighbor $\colpar{\verB(X_1)} \in \blackset$ is colored with $\black$. 
	Due to the edge $\{v,\verB\}$ subdivided once and $\verB(X_2)$ colored with $\black$, the vertex $v(X_2)$ is also colored with $\black$.
	On the other hand, due to the edge $\{\verW, v'\}$ subdivided twice and $\verW(X_1)$ colored with $\white$, the vertex $v'(X_1)$ must be colored with $\black$. So the edge $v'(X_1) v(X_2)$ has two end-vertices colored with $\black$ -- a contradiction to a proper 2-coloring of $G - S$. Hence, no transition from $x_1 = \scrW$ to $x_2 = \scrB$ occurs.
	
    Altogether, we have shown that for the states $x_1$ and $x_2$ of two consecutive path gadgets on the same path, we have $x_1 \preccurlyeq x_2$. Hence, along any fixed path, the state can change at most twice. There are $n'$ paths so totally, at most $2n'$ state changes occur. Therefore, there exists an index $r* \in [2n'+1]$ such that ``no state change occurs on the $r^*$th segment'', or formally: for every $i$ and every $j_1, j_2 \in [m]$, it holds that the states of $X^i_{r^*, j_1}$ and $X^i_{r^*, j_2}$ are the same in $S$. As explained before, this concludes the proof.
\end{proof}

\begin{proof}[Proof. (\cref{theo:coct})]
	Suppose there exists an algorithm $\mathcal A$ that given a linear arrangement of the input graph of cutwidth $\ctw$ solves \Pcoct{} in time $\ostar((4-\varepsilon)^{\ctw})$ for some positive real $\varepsilon$. 
	Given an instance $I$ of the $d$-SAT problem, we construct an equivalent instance $(G, k)$ of \Pcoct{} together with a linear arrangement $\ell$ of $G$ of cutwidth at most $n' + \mathcal{O}(1)$ in polynomial time.
	Then we run $\mathcal A$ on $G$ and output its answer. 
	The equivalence of instances implies the correctness of this algorithm solving $d$-SAT.
	The running time of the algorithm is bounded by $\ostar((4-\varepsilon)^{\ctw(G)})$. It holds that 
	\begin{align*}
		(4-\varepsilon)^{\ctw} &= (4-\varepsilon)^{n'+\mathcal{O}(1)} = (4-\varepsilon)^{\frac{n}{2} + \mathcal{O}(1)}\\
		&\leq c \sqrt{4-\varepsilon}^n < (2-\delta)^n
	\end{align*}
	for some constant $c$ and a positive value $\delta$ contradicting SETH.
\end{proof}

\section{Conclusion}
\label{sec:conclusion}
We have initiated the study of the exact complexity of hard connectivity
problems parameterized by cutwidth (under SETH) and we provided tight bounds for
six problems, namely \Pcvc, \Pcds, \Poct, \Pfvs, \Pst, and \Pcoct. 

One specific question that remains open is the exact complexity of \Phc\
parameterized by cutwidth (also open for treewidth). For pathwidth, it is known that $(2 +
\sqrt{2})$ is the optimal base assuming SETH \cite{CyganKN18}.
For more general questions, recall that cutwidth is an edge-separator analogue of pathwidth.
It would be interesting to study tight bounds for connectivity problems when parameterized by edge-separator analogues of treewidth such as tree-cut
width \cite{Wollan15}, edge-treewidth \cite{abs-2112-07524}, and edge-cut width
\cite{BrandCGHK22}.

\bibliography{lipics-paper.bib}

\begin{thebibliography}{10}

\bibitem{BjorklundHKK07}
Andreas Bj{\"{o}}rklund, Thore Husfeldt, Petteri Kaski, and Mikko Koivisto.
\newblock Fourier meets {M}{\"{o}}bius: fast subset convolution.
\newblock In David~S. Johnson and Uriel Feige, editors, {\em Proceedings of the
  39th Annual {ACM} Symposium on Theory of Computing, San Diego, California,
  USA, June 11-13, 2007}, pages 67--74. {ACM}, 2007.
\newblock \href {https://doi.org/10.1145/1250790.1250801}
  {\path{doi:10.1145/1250790.1250801}}.

\bibitem{BodlaenderCKN15}
Hans~L. Bodlaender, Marek Cygan, Stefan Kratsch, and Jesper Nederlof.
\newblock Deterministic single exponential time algorithms for connectivity
  problems parameterized by treewidth.
\newblock {\em Inf. Comput.}, 243:86--111, 2015.
\newblock \href {https://doi.org/10.1016/j.ic.2014.12.008}
  {\path{doi:10.1016/j.ic.2014.12.008}}.

\bibitem{BrandCGHK22}
Cornelius Brand, Esra Ceylan, Robert Ganian, Christian Hatschka, and Viktoriia
  Korchemna.
\newblock Edge-cut width: An algorithmically driven analogue of treewidth based
  on edge cuts.
\newblock In Michael~A. Bekos and Michael Kaufmann, editors, {\em
  Graph-Theoretic Concepts in Computer Science - 48th International Workshop,
  {WG} 2022, T{\"{u}}bingen, Germany, June 22-24, 2022, Revised Selected
  Papers}, volume 13453 of {\em Lecture Notes in Computer Science}, pages
  98--113. Springer, 2022.
\newblock \href {https://doi.org/10.1007/978-3-031-15914-5\_8}
  {\path{doi:10.1007/978-3-031-15914-5\_8}}.

\bibitem{CurticapeanLN18}
Radu Curticapean, Nathan Lindzey, and Jesper Nederlof.
\newblock A tight lower bound for counting {H}amiltonian cycles via matrix
  rank.
\newblock In Artur Czumaj, editor, {\em Proceedings of the Twenty-Ninth Annual
  {ACM-SIAM} Symposium on Discrete Algorithms, {SODA} 2018, New Orleans, LA,
  USA, January 7-10, 2018}, pages 1080--1099. {SIAM}, 2018.
\newblock \href {https://doi.org/10.1137/1.9781611975031.70}
  {\path{doi:10.1137/1.9781611975031.70}}.

\bibitem{CyganDLMNOPSW16}
Marek Cygan, Holger Dell, Daniel Lokshtanov, D{\'{a}}niel Marx, Jesper
  Nederlof, Yoshio Okamoto, Ramamohan Paturi, Saket Saurabh, and Magnus
  Wahlstr{\"{o}}m.
\newblock On problems as hard as {CNF-SAT}.
\newblock {\em {ACM} Trans. Algorithms}, 12(3):41:1--41:24, 2016.
\newblock \href {https://doi.org/10.1145/2925416} {\path{doi:10.1145/2925416}}.

\bibitem{CyganFKLMPPS15}
Marek Cygan, Fedor~V. Fomin, Lukasz Kowalik, Daniel Lokshtanov, D{\'{a}}niel
  Marx, Marcin Pilipczuk, Michal Pilipczuk, and Saket Saurabh.
\newblock {\em Parameterized Algorithms}.
\newblock Springer, 2015.
\newblock \href {https://doi.org/10.1007/978-3-319-21275-3}
  {\path{doi:10.1007/978-3-319-21275-3}}.

\bibitem{CyganKN18}
Marek Cygan, Stefan Kratsch, and Jesper Nederlof.
\newblock Fast {H}amiltonicity checking via bases of perfect matchings.
\newblock {\em J. {ACM}}, 65(3):12:1--12:46, 2018.
\newblock \href {https://doi.org/10.1145/3148227} {\path{doi:10.1145/3148227}}.

\bibitem{CyganNPPRW11}
Marek Cygan, Jesper Nederlof, Marcin Pilipczuk, Michal Pilipczuk, Johan M.~M.
  van Rooij, and Jakub~Onufry Wojtaszczyk.
\newblock Solving connectivity problems parameterized by treewidth in single
  exponential time.
\newblock In Rafail Ostrovsky, editor, {\em {IEEE} 52nd Annual Symposium on
  Foundations of Computer Science, {FOCS} 2011, Palm Springs, CA, USA, October
  22-25, 2011}, pages 150--159. {IEEE} Computer Society, 2011.
\newblock \href {https://doi.org/10.1109/FOCS.2011.23}
  {\path{doi:10.1109/FOCS.2011.23}}.

\bibitem{GroenlandMNS22}
Carla Groenland, Isja Mannens, Jesper Nederlof, and Krisztina Szil{\'{a}}gyi.
\newblock Tight bounds for counting colorings and connected edge sets
  parameterized by cutwidth.
\newblock In Petra Berenbrink and Benjamin Monmege, editors, {\em 39th
  International Symposium on Theoretical Aspects of Computer Science, {STACS}
  2022, March 15-18, 2022, Marseille, France (Virtual Conference)}, volume 219
  of {\em LIPIcs}, pages 36:1--36:20. Schloss Dagstuhl - Leibniz-Zentrum
  f{\"{u}}r Informatik, 2022.
\newblock \href {https://doi.org/10.4230/LIPIcs.STACS.2022.36}
  {\path{doi:10.4230/LIPIcs.STACS.2022.36}}.

\bibitem{HegerfeldK22}
Falko Hegerfeld and Stefan Kratsch.
\newblock Towards exact structural thresholds for parameterized complexity.
\newblock In Holger Dell and Jesper Nederlof, editors, {\em 17th International
  Symposium on Parameterized and Exact Computation, {IPEC} 2022, September 7-9,
  2022, Potsdam, Germany}, volume 249 of {\em LIPIcs}, pages 17:1--17:20.
  Schloss Dagstuhl - Leibniz-Zentrum f{\"{u}}r Informatik, 2022.
\newblock \href {https://doi.org/10.4230/LIPIcs.IPEC.2022.17}
  {\path{doi:10.4230/LIPIcs.IPEC.2022.17}}.

\bibitem{ImpagliazzoP01}
Russell Impagliazzo and Ramamohan Paturi.
\newblock On the complexity of k-{SAT}.
\newblock {\em J. Comput. Syst. Sci.}, 62(2):367--375, 2001.
\newblock \href {https://doi.org/10.1006/jcss.2000.1727}
  {\path{doi:10.1006/jcss.2000.1727}}.

\bibitem{ImpagliazzoPZ01}
Russell Impagliazzo, Ramamohan Paturi, and Francis Zane.
\newblock Which problems have strongly exponential complexity?
\newblock {\em J. Comput. Syst. Sci.}, 63(4):512--530, 2001.
\newblock \href {https://doi.org/10.1006/jcss.2001.1774}
  {\path{doi:10.1006/jcss.2001.1774}}.

\bibitem{JansenN18}
Bart M.~P. Jansen and Jesper Nederlof.
\newblock Computing the chromatic number using graph decompositions via matrix
  rank.
\newblock In Yossi Azar, Hannah Bast, and Grzegorz Herman, editors, {\em 26th
  Annual European Symposium on Algorithms, {ESA} 2018, August 20-22, 2018,
  Helsinki, Finland}, volume 112 of {\em LIPIcs}, pages 47:1--47:15. Schloss
  Dagstuhl - Leibniz-Zentrum f{\"{u}}r Informatik, 2018.
\newblock \href {https://doi.org/10.4230/LIPIcs.ESA.2018.47}
  {\path{doi:10.4230/LIPIcs.ESA.2018.47}}.

\bibitem{Kinnersley92}
Nancy~G. Kinnersley.
\newblock The vertex separation number of a graph equals its path-width.
\newblock {\em Inf. Process. Lett.}, 42(6):345--350, 1992.
\newblock \href {https://doi.org/10.1016/0020-0190(92)90234-M}
  {\path{doi:10.1016/0020-0190(92)90234-M}}.

\bibitem{Lampis20}
Michael Lampis.
\newblock Finer tight bounds for coloring on clique-width.
\newblock {\em {SIAM} J. Discret. Math.}, 34(3):1538--1558, 2020.
\newblock \href {https://doi.org/10.1137/19M1280326}
  {\path{doi:10.1137/19M1280326}}.

\bibitem{LokshtanovMS18}
Daniel Lokshtanov, D{\'{a}}niel Marx, and Saket Saurabh.
\newblock Known algorithms on graphs of bounded treewidth are probably optimal.
\newblock {\em {ACM} Trans. Algorithms}, 14(2):13:1--13:30, 2018.
\newblock \href {https://doi.org/10.1145/3170442} {\path{doi:10.1145/3170442}}.

\bibitem{LokshtanovMS182}
Daniel Lokshtanov, D{\'{a}}niel Marx, and Saket Saurabh.
\newblock Slightly superexponential parameterized problems.
\newblock {\em {SIAM} J. Comput.}, 47(3):675--702, 2018.
\newblock \href {https://doi.org/10.1137/16M1104834}
  {\path{doi:10.1137/16M1104834}}.

\bibitem{abs-2112-07524}
Lo{\"{\i}}c Magne, Christophe Paul, Abhijat Sharma, and Dimitrios~M. Thilikos.
\newblock Edge-treewidth: Algorithmic and combinatorial properties.
\newblock {\em CoRR}, abs/2112.07524, 2021.
\newblock URL: \url{https://arxiv.org/abs/2112.07524}, \href
  {http://arxiv.org/abs/2112.07524} {\path{arXiv:2112.07524}}.

\bibitem{MarxSS21}
D{\'{a}}niel Marx, Govind~S. Sankar, and Philipp Schepper.
\newblock Degrees and gaps: Tight complexity results of general factor problems
  parameterized by treewidth and cutwidth.
\newblock In Nikhil Bansal, Emanuela Merelli, and James Worrell, editors, {\em
  48th International Colloquium on Automata, Languages, and Programming,
  {ICALP} 2021, July 12-16, 2021, Glasgow, Scotland (Virtual Conference)},
  volume 198 of {\em LIPIcs}, pages 95:1--95:20. Schloss Dagstuhl -
  Leibniz-Zentrum f{\"{u}}r Informatik, 2021.
\newblock \href {https://doi.org/10.4230/LIPIcs.ICALP.2021.95}
  {\path{doi:10.4230/LIPIcs.ICALP.2021.95}}.

\bibitem{MulmuleyVV87}
Ketan Mulmuley, Umesh~V. Vazirani, and Vijay~V. Vazirani.
\newblock Matching is as easy as matrix inversion.
\newblock {\em Comb.}, 7(1):105--113, 1987.
\newblock \href {https://doi.org/10.1007/BF02579206}
  {\path{doi:10.1007/BF02579206}}.

\bibitem{PiecykR21}
Marta Piecyk and Pawel Rzazewski.
\newblock Fine-grained complexity of the list homomorphism problem: Feedback
  vertex set and cutwidth.
\newblock In Markus Bl{\"{a}}ser and Benjamin Monmege, editors, {\em 38th
  International Symposium on Theoretical Aspects of Computer Science, {STACS}
  2021, March 16-19, 2021, Saarbr{\"{u}}cken, Germany (Virtual Conference)},
  volume 187 of {\em LIPIcs}, pages 56:1--56:17. Schloss Dagstuhl -
  Leibniz-Zentrum f{\"{u}}r Informatik, 2021.
\newblock \href {https://doi.org/10.4230/LIPIcs.STACS.2021.56}
  {\path{doi:10.4230/LIPIcs.STACS.2021.56}}.

\bibitem{PilipczukW18}
Michal Pilipczuk and Marcin Wrochna.
\newblock On space efficiency of algorithms working on structural
  decompositions of graphs.
\newblock {\em {ACM} Trans. Comput. Theory}, 9(4):18:1--18:36, 2018.
\newblock \href {https://doi.org/10.1145/3154856} {\path{doi:10.1145/3154856}}.

\bibitem{GeffenJKM20}
Bas A.~M. van Geffen, Bart M.~P. Jansen, Arnoud A. W.~M. de~Kroon, and Rolf
  Morel.
\newblock Lower bounds for dynamic programming on planar graphs of bounded
  cutwidth.
\newblock {\em J. Graph Algorithms Appl.}, 24(3):461--482, 2020.
\newblock \href {https://doi.org/10.7155/jgaa.00542}
  {\path{doi:10.7155/jgaa.00542}}.

\bibitem{Wollan15}
Paul Wollan.
\newblock The structure of graphs not admitting a fixed immersion.
\newblock {\em J. Comb. Theory, Ser. {B}}, 110:47--66, 2015.
\newblock \href {https://doi.org/10.1016/j.jctb.2014.07.003}
  {\path{doi:10.1016/j.jctb.2014.07.003}}.

\end{thebibliography}
\end{document}